\documentclass[11pt, letterpaper]{article}
\title{Near Optimal Alphabet-Soundness Tradeoff PCPs}
 \author{
 Dor Minzer\thanks{Department of Mathematics, Massachusetts Institute of Technology, Cambridge, USA. Supported by NSF CCF award 2227876 and 
 NSF CAREER award 2239160.}
 \and
 Kai Zhe Zheng\thanks{Department of Mathematics, Massachusetts Institute of Technology, Cambridge, USA. Supported by the NSF GRFP DGE-2141064.}
 }
\date{\vspace{-5ex}}
\usepackage{fullpage}
\usepackage{amsthm}
\usepackage{mathtools}
\usepackage{amsmath,amssymb,amsfonts,nicefrac}
\usepackage{xspace}
\usepackage{color}
\usepackage{url}
\usepackage{hyperref}
\usepackage{stackengine}
\usepackage{bm}
\usepackage{bbm}
\usepackage{times}
\usepackage{thmtools,thm-restate}
\usepackage{enumitem}
\usepackage[capitalise,nameinlink,noabbrev]{cleveref}

\newcommand{\qbin}[2]{\begin{bmatrix}{#1}\\ {#2}\end{bmatrix}_q}
\newcommand{\amb}{\mathsf{amb}}

\DeclareMathOperator{\codim}{codim}
\DeclareMathOperator{\Tr}{Tr}

\DeclareMathOperator{\Gras}{{\sf Grass}}
\DeclareMathOperator{\poly}{poly}
\DeclareMathOperator{\val}{{\sf val}}

\DeclareMathOperator{\spa}{span}
\DeclareMathOperator{\var}{var}

\DeclareMathOperator{\rank}{rank}
\DeclareMathOperator{\im}{im}

\newcommand{\GapLin}{{\sf Gap3Lin}}
\newcommand{\Lin}{{\sf 3Lin}}
\newcommand{\NP}{{\sf NP}}
\newcommand{\SAT}{{\sf SAT}}

\newcommand{\Eq}{{\sf Eq}}

\newcommand{\Ug}{\mathcal{U}_{{\sf good}}}
\newcommand{\Ql}{\mathcal{Q}_{{\sf lucky}}}
\newcommand{\Qs}{\mathcal{Q}_{{\sf smooth}}}
\newcommand{\Us}{\mathcal{U}_{{\sf sat}}}

\newcommand{\Ws}{W_{{\sf second}}}

\newcommand{\V}{\mathcal{V}}
\DeclarePairedDelimiter\floor{\lfloor}{\rfloor}

\newcommand{\A}{\mathcal{A}}
\newcommand{\T}{\mathcal{T}}

\newcommand{\B}{\mathcal{B}}
\newcommand\inner[2]{\langle{#1},{#2}\rangle}
\newcommand{\norm}[1]{\left\lVert#1\right\rVert}
\newcommand{\Zoom}{{\sf Zoom}}
\newcommand{\Grass}{{\sf Grass}}

\newcommand{\E}{\mathop{\mathbb{E}}}
\newcommand{\Ff}{\mathbb{F}}
\newcommand{\D}{\mathcal{D}}

\newcommand{\U}{\mathcal{U}}

\DeclarePairedDelimiter{\ceil}{\lceil}{\rceil}

\newcommand{\C}{\mathcal{C}}
\newcommand{\Rcal}{\mathcal{R}}
\newcommand{\Lcal}{\mathcal{L}}
\newcommand{\Cl}{\textsf{Clique}}

\newcommand{\W}{\mathcal{W}}
\newcommand{\ind}{\mathbbm{1}}

\newcommand{\Lc}{\mathcal{L}}

\newcommand{\eps}{\varepsilon}
\renewcommand{\epsilon}{\eps}

\newcommand\skipi{{\vskip 10pt}}
\renewcommand\leq{\leqslant}
\renewcommand\geq{\geqslant}

\newcommand{\mc}{\mathcal}
\newcommand\stackleftrightarrow[1]{%
    \mathrel{{\stackon[4pt]{$\longleftrightarrow$}{$\scriptscriptstyle#1$}}}}
\newcommand{\minrep}{{\sf MinRep}}
 \newcommand\stackrightarrow[1]{%
    \mathrel{{\stackon[4pt]{$\longrightarrow$}{$\scriptscriptstyle#1$}}}}   
\theoremstyle{plain} 

   \newtheorem{thm}{Theorem}[section]
   
   \newtheorem{fact}[thm]{Fact}
   \newtheorem{lemma}[thm]{Lemma}
   \newtheorem{corollary}[thm]{Corollary}
   \newtheorem{claim}[thm]{Claim}
   \newtheorem{remark}[thm]{Remark}
   \newtheorem{definition}[thm]{Definition}

\begin{document}
\maketitle
\begin{abstract}
    We show that for all $\eps>0$, for sufficiently large $q\in\mathbb{N}$ that is a power of $2$, for all $\delta>0$, it is NP-hard 
    to distinguish whether a given $2$-Prover-$1$-Round  projection
    game with alphabet size $q$ 
    has value at least $1-\delta$, or value at 
    most $1/q^{1-\eps}$. This establishes a nearly 
    optimal alphabet-to-soundness tradeoff for $2$-query PCPs with alphabet size $q$, 
    improving upon a result of [Chan, J. ACM 2016]. 
    Our result has the following implications:
    \begin{enumerate}
        \item Near optimal hardness for Quadratic Programming: it is NP-hard to approximate 
        the value of 
        a given Boolean Quadratic Program within 
        factor $(\log n)^{1 - o(1)}$ under
        quasi-polynomial time reductions.
        This improves upon a result of [Khot, Safra, ToC 2013] and nearly matches
        the performance of the best known  algorithms due to [Megretski, IWOTA 2000], [Nemirovski, Roos, Terlaky, Mathematical Programming 1999] and [Charikar, Wirth, FOCS 2004] that achieve $O(\log n)$ approximation ratio.
        \item Bounded degree $2$-CSPs: under randomized reductions, 
        for sufficiently large $d>0$,
        it is NP-hard to approximate the value of $2$-CSPs in which each variable appears in at most
        $d$ constraints to within a factor of $(1-o(1))\frac{d}{2}$, improving upon a result of [Lee,  Manurangsi, ITCS 2024].
        \item Improved hardness results for connectivity problems: using results 
        of [Laekhanukit, SODA 2014] and [Manurangsi, Inf.~Process.~Lett., 2019], we deduce improved 
        hardness results for the Rooted $k$-Connectivity Problem, the Vertex-Connectivity Survivable Network Design Problem 
        and the Vertex-Connectivity $k$-Route Cut Problem.
    \end{enumerate}
\end{abstract}

\section{Introduction}
The PCP theorem is a fundamental result in theoretical computer science with many equivalent formulations~\cite{FGLSS,AroraSafra,ALMSS}. One of the formulations asserts 
that there exists $\eps>0$ such that given a satisfiable $3$-SAT formula $\phi$, it is NP-hard to find an assignment that satisfies at least $(1-\eps)$ fraction of the constraints. The PCP theorem has a myriad
of applications within theoretical computer science, 
and of particular interest to this paper are its applications to hardness of approximation.

The vast majority of hardness of approximation results
are proved via reductions from the PCP theorem above. 
Oftentimes, to get a strong hardness of approximation
result, one must first amplify the basic PCP theorem into a result with stronger  parameters~\cite{Hastad,HastadClique,Feige,KP} 
(see~\cite{trevisan2014inapproximability} for a survey). To discuss these parameters, it is
often convenient to view the PCP through the problem of $2$-Prover-$1$-Round Games, which we define next.\footnote{Strictly speaking, 
the notion below is often referred to as projection $2$-Prover-$1$-Round games. As we only deal with projection games in this paper we omit the more general definition.}
\begin{definition}
An instance $\Psi$ of $2$-Prover-$1$-Round Games consists of a bipartite
graph $G = (L\cup R, E)$, alphabets $\Sigma_L$, $\Sigma_R$, and a collection of constraints 
$\Phi = \{\phi_e\}_{e\in E}$ specified by maps $\phi_e\colon \Sigma_L\to\Sigma_R$.
\begin{enumerate}
    \item The alphabet size of $\Psi$ is defined to 
    be $|\Sigma_L|+|\Sigma_R|$.
    \item The value of $\Psi$ is defined to be
    the maximum fraction of edges $e\in E$ 
    that can be satisfied by any assignment. 
    That is, 
    \[
    {\sf val}(\Psi) 
    = \max_{
    \substack{
    A_L\colon L\to\Sigma_L\\ 
    A_R\colon R\to\Sigma_R}}\frac{|\{e = (u,v)\in E~|~\phi_e(A_L(u)) = A_R(v)\}|}{|E|}.
    \]
\end{enumerate}
\end{definition}

The combinatorial view of $2$-Prover-$1$-Round 
Games has its origins in an equivalent, 
active view of PCP in terms
of a game between a verifier and two all-powerful
provers. The verifier and the two provers have access 
to an instance $\Psi$ of $2$-Prover-$1$-Round Games, 
and the provers may agree beforehand on a strategy; 
after this period they are not allowed to communicate. 
The verifier then 
picks a random edge, $e = (u,v)$, from the $2$-Prover-$1$-Round game, sends $u$ to the first prover and $v$ to the second prover, and then receives a label in response from each one of them. Finally, the verifier checks that these labels satisfy the constraint $\phi_e$, and if so accepts (and otherwise rejects). 
It is easy to see that the value of the $2$-Prover-$1$-Round game is equal to the acceptance probability of 
the verifier under the best strategy of the provers. 
This view will be useful for us later.

The majority of hardness of 
approximation results are proved by combining 
the basic PCP theorem of~\cite{FGLSS,AroraSafra,ALMSS} 
with Raz's parallel repetition theorem~\cite{Raz}, which together imply the following result in the language of $2$-Prover-$1$-Round Games:
\begin{thm}\label{thm:soundness_alphabet_tradeoff_pc}
    There exists $\gamma>0$ such that for sufficiently 
    large $R$, given a $2$-Prover-$1$-Round game
    $\Psi$ with alphabet size $R$, it is NP-hard to distinguish between the following two cases:
    \begin{enumerate}
        \item YES case: ${\sf val}(\Psi) = 1$. 
        \item NO case: ${\sf val}(\Psi) \leq \frac{1}{R^{\gamma}}$.
    \end{enumerate}
\end{thm}
While the soundness guarantee in~\cref{thm:soundness_alphabet_tradeoff_pc} is often sufficient, some applications require strengthenings of it. In particular, one could require the soundness to be as small as possible with respect to the alphabet size (as opposed to just being small in absolute terms). The tradeoff between the soundness error of the PCP
and the alphabet size of the PCP is the main focus
of this paper.

Since a random
assignment to $\Psi$ satisfies in expectation at least
$\frac{1}{R}$ fraction of the constraints, the best alphabet to soundness tradeoff one could hope for in~\cref{thm:soundness_alphabet_tradeoff_pc} is with $\gamma=1-o(1)$. Khot and Safra~\cite{KS} 
were the first to achieve a reasonable bound for $\gamma$, and their result achieves $\gamma = 1/6$ with imperfect completeness (i.e., ${\sf val}(\Psi)\geq 1-o(1)$ instead of ${\sf val}(\Psi) = 1$ in the YES case). Subsequently, Chan~\cite{Chan} obtained an improvement (using different techniques), showing that one can get $\gamma = 1/2-o(1)$ (again with imperfect 
completeness).

\subsection{Main Results}
\subsubsection{Near Optimal Alphabet vs Soundness Tradeoff}
The main result of this work improves upon all prior results, and shows that one may take $\gamma = 1-o(1)$
in~\cref{thm:soundness_alphabet_tradeoff_pc}, again with imperfect completeness. Formally, we show:
\begin{thm}\label{thm:main}
For all $\eps,\tau>0$, for sufficiently large $R$, 
given a $2$-Prover-$1$-Round game $\Psi$, it is NP-hard
to distinguish between the following two cases:
\begin{enumerate}
        \item YES case: ${\sf val}(\Psi) \geq 1-\tau$. 
        \item NO case: ${\sf val}(\Psi) \leq \frac{1}{R^{1-\eps}}$.
    \end{enumerate}
\end{thm}

As discussed below,~\cref{thm:main} has several applications to combinatorial
optimization problems. These
applications require additional features from the
instances produced in~\cref{thm:main} which
are omitted from the formulation for the sake of clarity.
For instance, one application requires a good tradeoff between the size of the instance and the size of the alphabet, 
which our construction achieves (see the discussion following~\cref{thm:QP}). Other applications require 
the underlying constraint graph to be bounded-degree and biregular, which our construction also achieves (after mild modifications; see~\cref{thm: biregular csp}).

\subsubsection{Application: NP-Hardness of Approximating Quadratic Programs}

An instance of the Quadratic Programming problem consists of a 
quadratic form $Q(x) = \sum\limits_{i,j=1}^{n} a_{i,j} x_{i} x_j$ where $a_{i,i} = 0$ for all $i$, and the goal is to maximize $Q(x)$ over $x\in \{-1,1\}^n$. It was previously known that this problem admits an $O(\log n)$ approximation 
algorithm~\cite{megretski2001relaxations,nemirovski1999maximization,charikar2004maximizing} and is NP-hard to approximate
within factor $(\log n)^{1/6 - o(1)}$~\cite{ABHKS,KS} under quasi-polynomial reductions. 
As a first application of~\cref{thm:main}, we improve the hardness of approximation result for the Quadratic Programming problem:
\begin{thm}\label{thm:QP}
It is NP-hard to approximate 
Quadratic Programming to within a factor of $(\log n)^{1 - o(1)}$ under quasi-polynomial reductions.
\end{thm}
\cref{thm:QP} is proved via a connection between
$2$-Prover-$1$-Round Games and Quadratic Programming
due to Arora, Berger, Hazan, Kindler, and Safra~\cite{ABHKS}. This connection requires a 
good tradeoff between the alphabet size, the soundness error
of the PCP, and the size of the PCP, and the construction in~\cref{thm:QP} has
a sufficiently good tradeoff: letting 
$N$ be the size of the instance, 
the alphabet size is
$(\log N)^{1-o(1)}$ and
the soundness error is
$(\log N)^{-1+o(1)}$.\footnote{We remark that, due to the use of the long-code, Chan~\cite{Chan} does not achieve a good enough trade-off between the alphabet size and the instance size. This is the reason his result does not yield an improvement over the inapproximability result for Quadratic Programming of Khot and Safra~\cite{KS}.}

\paragraph{Relevance to the sliding scale conjecture?}  We do not know 
how to use our techniques to achieve soundness error that is 
smaller than inversely poly-logarithmic in the instance size.
In particular, our techniques have no bearing 
on the sliding scale conjecture, 
which asserts the existence of a PCP with $O(1)$ queries, polynomial size, polynomial alphabet size, and inverse polynomial soundness~\cite{sliding_scale_survey}. 
A key feature of our PCP reduction is the so-called ``covering property'', which seems to inherently limit the soundness of our construction to be at most inversely logarithmic in the instance size.

\subsubsection{Application: NP-hardness of Approximating Bounded Degree $2$-CSPs}
A $2$-Constraint Satisfaction Problem ($2$-CSP) instance  $\Psi = (X,C,\Sigma)$ consists of a set
of variables $X$, a set of constraints $C$, 
and an alphabet $\Sigma$. Each
constraint in $C$ has the form $P(x_i,x_j) = 1$
where $P\colon \Sigma\times\Sigma \to\{0,1\}$ 
is a predicate (which may be different in distinct
constraints). The degree of the instance $\Psi$ is defined to be the maximum, over variables
$x\in X$, of the number of constraints that $x$ appears in. 
The goal is to find an assignment $A\colon X\to\Sigma$ 
that satisfies as many of the constraints as possible.

There is a simple $\frac{d+1}{2}$ approximation algorithm for the $2$-CSP problem for instances
with degree at most $d$. Lee and Manurangsi proved a 
nearly matching $\left(\frac{1}{2} - o(1)\right) d$ 
hardness of approximation result assuming the
Unique-Games Conjecture~\cite{LeeManurangsi}. Unconditionally, they show the problem to be
NP-hard to approximate within factor $\left(\frac{1}{3} - o(1)\right)d$ under randomized reductions.
Using the ideas of Lee and Manurangsi, our main 
result implies a nearly matching NP-hardness 
result for bounded degree $2$-CSPs:
\begin{thm}\label{thm:2CSPs}
  For all $\eta>0$, for sufficiently large $d$, 
  approximating the value of $2$-CSPs with 
  degree at most $d$ within factor
  $\left(\frac{1}{2} - \eta\right)d$ is NP-hard under 
  randomized reductions.
\end{thm}

As in~\cite{LeeManurangsi},~\cref{thm:2CSPs} has a further
application to finding independent sets in 
claw free graphs. A $k$-claw is the $(k+1)$ vertex graph which consists of one center vertex connected to
$k$ leaves and has no other edges. A graph $G$ 
is said to be $k$-claw-free if $G$ does not contain
a $k$-claw as an induced subgraph. There is a polynomial time algorithm that approximates the size of the largest
independent set in a given $k$-claw-free graph $G$ within 
factor $\frac{k}{2}$~\cite{berman2000d, thiery2023improved}, and a quasi-polynomial time
algorithm due to~\cite{cygan2013sell} that achieves a better  
$\left(\frac{1}{3} + o(1)\right) k$.
As in~\cite{LeeManurangsi}, using ideas from~\cite{dvovrak2023parameterized},~\cref{thm:2CSPs} implies that for all $\eps>0$, for sufficiently large $k$, it is NP-hard (under randomized reductions) to approximate the size of 
the largest independent set in a given $k$-claw-free graph within factor
$\left(\frac{1}{4} + \eta\right)k$. This 
improves upon the result of~\cite{LeeManurangsi} who
showed that the same result holds assuming the Unique-Games Conjecture.

\subsubsection{Application: NP-hardness of Approximating Connectivity Problems}
Using ideas of Laekhanukit~\cite{laekhanukit2014parameters} and some subsequent improvements by Manurangsi~\cite{Manurangsi19}, ~\cref{thm:main} implies improved hardness of 
approximation results for several graph connectivity
problems. More specifically, it 
implies new improvements over each of the 
results outlined in table $1$ of~\cite{laekhanukit2014parameters}, with the exception of Rooted-$k$-Connectivity on directed graphs where the improvement is already implied by~\cite{Manurangsi19}. Below, we briefly discuss one of these graph connectivity problems --- the Rooted $k$-Connectivity Problem --- but defer the reader to~\cite{laekhanukit2014parameters} for 
a detailed discussion of the remaining graph connectivity problems considered.

In the Rooted-$k$-Connectivity problem there is a graph $G = (V,E)$, edge costs $c\colon E \to \mathbb{R}$, a root vertex $r\in V$ and a set 
of terminals $T\subseteq V\setminus\{ r\}$. 
The goal is to find a subgraph $G'$ of smallest 
cost that for each $t\in T$, has at least $k$ 
vertex disjoint paths from $r$ to $t$. The problem
admits a trivial $|T|$-approximation algorithm (by 
applying minimum cost $k$-flow algorithm for 
each vertex in $T$), as well as an $O(k\log k)$
approximation algorithm~\cite{nutov2012approximating}.
Using the ideas of~\cite{laekhanukit2014parameters},~\cref{thm:main} implies the following improved
hardness of approximation results:
\begin{thm}\label{thm:rooted_conn}
For all $\eps>0$, for sufficiently large $k$ it 
is NP-hard under randomized reductions to approximate the Rooted-$k$-Connectivity problem on undirected graphs to within a factor of $k^{1/5-\epsilon}$, the Vertex-Connectivity Survivable Network Design Problem with connectivity parameters at most $k$ to within a factor of $k^{1/3-\epsilon}$, and the Vertex-Connectivity $k$-Route Cut Problem to within a factor of $k^{1/3-\epsilon}$.
\end{thm}
We remark that in~\cite{chakraborty2008network}, a weaker form of hardness for the Vertex-Connectivity Survivable Network 
problem is proved. More precisely, they show an $\Omega(k^{1/3}/\log k)$ integrality gap for the set-pair 
relaxation of the problem. Our hardness result of 
$k^{1/3 - \eps}$ shows that, unless 
$\text{NP}\subseteq\text{BPP}$, no polynomial-size relaxation can yield a better than 
 $k^{1/3-\eps}$ factor approximation algorithm.

\subsection{Our Techniques}
\cref{thm:main} is proved by composing 
an inner PCP and an outer PCP. 
Both of these components incorporate ideas from the proof of the $2$-to-$1$ Games Theorem. The outer PCP is constructed using smooth parallel repetition~\cite{KS,KMS} while the inner PCP
is based on a codeword test using the Grassmann graph~\cite{KMS,DKKMS1, DKKMS2, KMS2}. 

The novelty in this paper, in terms of techniques, is twofold. First, in order to get such a strong alphabet to soundness tradeoff, we must consider a Grassmann test in a different and extreme regime of parameters. Additionally, we have to work with much lower soundness error than prior works. These differences complicate matters considerably. Second, our soundness analysis is more involved than that of the $2$-to-$1$-Games Theorem. As is the case in~\cite{KMS,DKKMS1, DKKMS2, KMS2}, we too use global hypercontractivity, but we do so more extensively. 
We also require quantitatively stronger versions of 
global hypercontractivity over the Grassmann graph 
which are due to~\cite{EvraKL}.
Our analysis also incorporates ideas from the plane versus
plane test and direct product testing~\cite{RS,IKW,MZ}, from classical PCP 
theory~\cite{KS}, as well as from error correcting codes~\cite{GRS}. All of these tools are necessary to prove our main technical statement ---~\cref{lem:main_technical_intro} below --- which
is a combinatorial statement perhaps of independent interest.

We now elaborate on each one of the components separately.

\subsubsection{The Inner PCP}
Our inner PCP is based on the subspace vs subspace
low-degree test. Below, we first give a general
overview of the objective in low-degree testing. 
We then discuss
the traditional notion of soundness as well as a 
non-traditional notion of soundness for low-degree tests. Finally, we explain the low-degree test 
used in this paper, the notion of soundness that we 
need from it, and the way that this notion of soundness is used.

\paragraph{Low-degree tests in PCPs.} Low-degree tests have been a vital component in PCPs since their inception, and
much attention has been devoted to improving their
various parameters. The goal in low-degree testing is
to encode a low-degree function $f\colon \mathbb{F}_q^n\to\mathbb{F}_q$ via a table (or a few tables) of values in a locally testable way. Traditionally, one picks a parameter $\ell\in\mathbb{N}$ (which is thought of as a constant and is most often just $2$) and encodes the function $f$ by the table $T$ of restrictions of $f$ to $\ell$-dimensional affine subspaces of $\mathbb{F}_q^n$. For the case $\ell=2$, the test associated with this encoding is known as the Plane vs Plane test~\cite{RS}. The Plane vs Plane test proceeds by picking two planes $P_1$, $P_2$ intersecting on a line, and then checking that 
$T[P_1]$ and $T[P_2]$ agree on $P_1\cap P_2$. It is easy to see that the test has perfect completeness, namely that a valid table of restrictions $T$ passes
the test with probability $1$. In the other direction, the soundness error of the 
test --- which is a converse
type statement --- is much less clear (and is crucial towards applications in PCP). In the context of the Plane vs Plane
test, it is known that if a table $T$, that assigns to
each plane a degree $d$ function, passes the Plane vs
Plane test with probability $\eps\geq q^{-c}$ (where
$c>0$ is a small absolute constant), then there is a degree 
$d$ function $f$ such that $T[P] \equiv f|_{P}$ on at
least $\Omega(\eps)$ fraction of the planes.

Nailing down the value of the constant $c$ for which
soundness holds is an interesting open problem which
is related to soundness vs alphabet size vs instance 
size tradeoff in PCPs~\cite{MR,BDN,MZ}. Currently, 
the best known analysis for the Plane vs Plane test~\cite{MR} shows that one may take $c=1/8$. Better analysis is known for higher dimensional encoding~\cite{BDN,MZ}, and for the $3$-dimensional 
version of it a near optimal soundness result is known~\cite{MZ}.

\paragraph{Low-degree tests in this paper.} In the context of the current paper, we wish to encode
linear functions $f\colon \mathbb{F}_q^n\to\mathbb{F}_q$, and we do so by 
the subspaces encoding. Specifically, we set 
integer parameters $\ell_1 \geq \ell_2$, and encode 
the function $f$ using the table $T_1$ of the restrictions of $f$ to all $\ell_1$-dimensional linear subspaces of $\mathbb{F}_q^n$, and the table $T_2$ 
of the restrictions of $f$ to all $\ell_2$-dimensional 
linear subspaces of $\mathbb{F}_q^n$. The test we consider 
is the natural inclusion test:
\begin{enumerate}
    \item Sample a random $\ell_1$-dimensional subspace
    $L_1\subseteq \mathbb{F}_q^n$ and a random $\ell_2$-dimensional 
    subspace $L_2\subseteq L_1$.
    \item Read $T_1[L_1]$, $T_2[L_2]$ and accept if they agree on $L_2$.
\end{enumerate}
As is often the case, the completeness of the test --- namely the fact that valid tables $T_1,T_2$ pass the
test with probability $1$ --- is clear. The question
of most interest then is with regards to the soundness
of the test. Namely, what is the smallest $\eps$ such that any two tables $T_1$ and $T_2$ that assign linear 
functions to subspaces and pass the test with 
probability $\eps$, must necessarily ``come from'' a legitimate linear function $f$?

\paragraph{Traditional notion of soundness.} As the alphabet vs soundness tradeoff is key to the discussion herein, we begin by remarking that the alphabet size of the above encoding is $q^{\ell_1} + q^{\ell_2} = \Theta(q^{\ell_1})$ (since there 
are $q^{\ell}$ distinct linear functions on a linear
space of dimension $\ell$ over $\mathbb{F}_q$). Thus, ideally 
we would like to show that the soundness error 
of the above test is $q^{-(1-o(1))\ell_1}$. Alas, 
this is false. Indeed, it turns out that one may 
construct assignments that pass the test with 
probability at least $\Omega(\max(q^{-\ell_2}, q^{\ell_2-\ell_1}))$ that do not have significant 
correlation with any linear function $f$:
\begin{enumerate}
    \item Taking $T_1,T_2$ randomly by assigning to each subspace a random linear function, one can easily see that the test passes with probability $\Theta(q^{-\ell_2})$.
    \item Taking linear subspaces $W_1,\ldots, W_{100q^{\ell_1}}\subseteq\mathbb{F}_q^n$ of codimension $1$ randomly, and a random linear
    function $f_i\colon W_i\to\mathbb{F}_q$ for each
    $i$, one may choose $T_1$ and $T_2$ as follows. 
    For each $L_1$, pick a random $i$ such that $L_1\subseteq W_i$ (if such $i$ exists) and assign $T_1[L_1] = f_i|_{L_1}$. For each $L_2$, pick a random $i$ such that $L_2\subseteq W_i$ (if such $i$ exists) and assign $T_2[L_2] = f_i|_{L_2}$. Taking 
    $L_2\subseteq L_1$ randomly, one sees that with 
    constant probability $L_2$ has $\Theta(q^{\ell_1-\ell_2})$ many possible $i$'s, $L_1$ has $\Theta(1)$ many possible $i$'s and furthermore 
    there is at least one $i$ that is valid for both 
    of them. With probability $\Omega(q^{\ell_2-\ell_1})$ this common $i$ is chosen for both $L_1$ and $L_2$, and in this case, the test on $(L_1,L_2)$ passes. It follows that, in expectation, $T_1,T_2$ pass the test with probability $\Omega(q^{\ell_2-\ell_1})$.
\end{enumerate}
In light of the above, it makes sense that the best possible alphabet vs soundness tradeoff we may achieve
with the subspace encoding is by taking $\ell_2 = \ell_1/2$. Such a setting of the parameters would give alphabet size $R = q^{\ell_1}$ and (possibly) soundness error $\Theta(1/\sqrt{R})$. 
This tradeoff is not good enough for our purposes, and so we must venture beyond the traditional
notion of soundness.

\paragraph{Non-traditional notion of soundness.} 
The above test was first considered in the context 
of the $2$-to-$1$ Games Theorem, wherein one takes $q=2$ and $\ell_2 = \ell_1-1$. In this setting, the 
test is not sound in the traditional sense; instead, the test is shown to satisfy a non-standard notion of soundness, which
nevertheless is sufficient for the purposes of constructing a PCP. More specifically, in~\cite{KMS2} 
it is proved that for all $\eps>0$ there is $r \in\mathbb{N}$ such that for sufficiently large $\ell$ and for tables $T_1,T_2$ as above, there are subspaces $Q\subseteq W\subseteq \mathbb{F}_q^n$ with ${\sf dim}(Q) + {\sf codim}(W) \leq r$ and a linear function $f\colon W\to\mathbb{F}_q$ such that 
\[
\Pr_{Q\subseteq L_1\subseteq W}[T_1[L_1]\equiv f|_{L_1}]\geq \eps'(\eps)>0.
\] 
We refer to the set
\[
\{L\subseteq \mathbb{F}_q^n~|~{\sf dim}(L)=\ell_1, Q\subseteq L\subseteq W\}
\]
as the zoom-in of $Q$ and zoom-out of $W$. While this result is sufficient for the
purposes of $2$-to-$1$ Games, the dependency between 
$\ell$ and $\eps$ (and thus, between the soundness 
and the alphabet size) is not good enough for us.

\paragraph{Our low-degree test.} It turns out that
the proper setting of parameters for us is $\ell_2 = (1-\delta) \ell_1$ where $\delta>0$
is a small constant. With these parameters, we are 
able to show that for $\eps\geq q^{-(1-\delta')\ell_1}$ (where $\delta' = \delta'(\delta)>0$ 
is a vanishing function of $\delta$), if $T_1$, $T_2$ 
pass the test with probability at least $\eps$, then
there are subspaces $Q\subseteq W$ with ${\sf dim}(Q) + {\sf codim}(W)\leq r = r(\delta)\in\mathbb{N}$, and 
a linear function $f\colon W\to\mathbb{F}_q$ such 
that 
\[
\Pr_{Q\subseteq L_1\subseteq W}[T_1[L_1]\equiv f|_{L_1}]\geq \eps'(\eps) = \Omega(\eps).
\] 

Working in the very small soundness regime of 
$\eps\geq q^{-(1-\delta')\ell_1}$ entails many challenges, however. First, 
dealing with such small soundness requires us to use a strengthening of the 
global hypercontractivity result of~\cite{KMS2} in
the form of an optimal level $d$ inequality due to 
Evra, Kindler and Lifshitz~\cite{EvraKL}. Second, in the context of~\cite{KMS2}, $\eps'$ could be
    any function of $\eps$ (and indeed it ends up
    being a polynomial function of $\eps$). In the 
    context of the current paper, it is crucial that
    $\eps' = \eps^{1+o(1)}$, as opposed to, say, $\eps' = \eps^{1.1}$. The reason is that, as we 
    are dealing with very small $\eps$, the result would be trivial for $\eps' = \eps^{1.1}$ and not
    useful towards the analysis of the PCP (as then
    $\eps'$ would be below the threshold $q^{-\ell_1}$
    which represents the agreement a random linear 
    function $f$ has with $T_1$).

\subsubsection{Getting List Decoding Bounds}
As is usually the case in PCP reductions, we require a list 
decoding version for our low-degree test. Indeed, 
using a standard argument we are able to show 
that in the setting that $\ell_2 = (1-\delta)\ell_1$
and $\eps\geq q^{-(1-\delta')\ell_1}$, there is 
$r = r(\delta,\delta')\in\mathbb{N}$ such that 
for at least $q^{-\Theta(\ell_1)}$ fraction of
subspaces $Q\subseteq \mathbb{F}_q^n$ of dimension $r$, there
exists a subspace $W$ with codimension at most $r$  and $Q\subseteq W\subseteq\mathbb{F}_q^n$, as well as 
a linear function $f\colon W\to\mathbb{F}_q$, such
that 
\begin{equation}\label{eq1_intro}
\Pr_{Q\subseteq L_1\subseteq W}[T_1[L_1]\equiv f|_{L_1}]\geq \eps'(\eps) = \Omega(\eps).
\end{equation}

This list decoding version theorem alone is not enough. 
In our PCP construction, we compose the inner
PCP with an outer PCP (that we describe below), and 
analyzing the composition requires decoding 
global linear functions (from a list decoding version
theorem as above) in a coordinated manner between two non communicating parties. Oftentimes, the
number of possible global functions that may be 
decoded is constant, in which case randomly sampling
one among them works. This is not the case 
for us, though: if $(Q, W)$ and $(Q', W')$ 
are distinct zoom-in and zoom-out pairs for which 
there are linear functions $f_{Q, W}$ and $f_{Q', W'}$
satisfying~\eqref{eq1_intro}, then the functions $f_{Q,W}$ 
and $f_{Q',W'}$ could be completely different. Thus, 
to achieve a coordinated decoding procedure, we must:
\begin{enumerate}
    \item Facilitate a way for the two parties to agree on a zoom-in and zoom-out pair $(Q,W)$ with noticeable probability.
    \item Show that for each $(Q,W)$ there are at most 
    ${\sf poly}(1/\eps)$ linear functions $f_{Q,W}$ for which
    \[
        \Pr_{Q\subseteq L_1\subseteq W}[T_1[L_1]\equiv f_{Q,W}|_{L_1}]\geq \eps'.
    \]
\end{enumerate}
The second item is precisely the reason 
we need $\eps'$ to be $\eps^{1+o(1)}$; any worse dependency, such as $\eps' = \eps^{1.1}$ would lead 
to the second item being false. We also remark that
the number of functions being ${\sf poly}(1/\eps)$ 
is important to us as well. There is some slack in this bound,
but a weak quantitative bound such as ${\sf exp}({\sf exp}(1/\eps))$ would have been insufficient for some
of our applications. Luckily, such bounds can be deduced from~\cite{GRS} for the case of linear functions.\footnote{In the case of higher degree functions (even quadratic functions) some bounds are known~\cite{gopalan2010fourier,bhowmick2015list} but they would not have been good enough for us.}

We now discuss the first item. 
It turns out that agreeing on the zoom-in $Q$ 
can be delegated to the outer PCP, in the sense that our outer PCP game will provide the two parties with a coordinated zoom-in $Q$ (while still being sound). 
Morally speaking, this works because in our list decoding theorem, 
the fraction of zoom-ins $Q$ that work is significant. 
Coordinating zoom-outs is more difficult, and this 
is where much of the novelty in our analysis lies.

\subsubsection{Coordinating Zoom-outs}
For the sake of simplicity and to focus on the main ideas, we ignore zoom-ins for now and assume that the 
list decoding statement holds with no $Q$. Thus, 
the list decoding theorem asserts that there exists
a zoom-out $W$ of constant codimension on which 
there is a global linear function agreeing with the table $T_1$. However, there 
could be many such zoom-outs $W$, say $W_1,\ldots,W_m$ 
and say all of them were of codimension $r = O_{\delta,\delta'}(1)$. If the number $m$ were sufficiently large ---  say at least $q^{-{\sf poly}(\ell_1)}$ fraction of all codimension $r$ subspaces --- then we would have been able to coordinate them in the same
way as we coordinate zoom-ins. If the number $m$ were sufficiently small --- say 
$m = q^{{\sf poly}(\ell_1)}$, then 
randomly guessing a zoom-out would work well enough. 
The main issue is that the number 
$m$ is intermediate, say $m = q^{\sqrt{n}}$.

This issue had already appeared in~\cite{KMS,DKKMS1}. 
Therein, this issue is resolved by showing that 
if there 
are at least $m\geq q^{100\ell_1^2}$ zoom-outs 
$W_1,\ldots,W_{m}$ of codimension $r$, and linear functions $f_1,\ldots,f_{m}$ on $W_1,\ldots,W_m$ respectively 
such that 
\[
\Pr_{L\subseteq W_i}[T_1[L]\equiv f_i|_{L}]\geq \eps'
\]
for all $i$, then there exists a zoom-out $W$ of codimension \emph{strictly less than $r$} and a linear
function $f\colon W\to\mathbb{F}_q$ such that
\[
\Pr_{L\subseteq W}[T_1[L]\equiv f|_{L}]\geq \Omega(\eps'^{12}).
\]
Thus, if there are too many zoom-outs of a certain 
codimension, then there is necessarily a zoom-out
of smaller codimension that also works. In that case, the 
parties could go up to that codimension. 

This result
is not good enough for us, due to the polynomial
gap between 
the agreement of $T_1$ and $f_i$, and
the agreement of $T_1$ and $f$. Indeed, in our
range of parameters, $\eps'^{12}$ will be below 
the trivial threshold $q^{-\ell_1}$ which is the agreement
a random linear function $f$ has with $T_1$, and 
therefore the promise on the function $f$ above
is meaningless.

We resolve this issue by showing a stronger, essentially optimal version of the above assertion
still holds. 
Formally, we prove:
\begin{lemma}\label{lem:main_technical_intro}
For all $\delta>0$, $r\in\mathbb{N}$ there is $C>1$ such that 
the following holds for $\eps'\geq q^{-(1-\delta)\ell_1}$ and large enough $\ell_1$.
Suppose that $T_1$ is a table that assigns to each 
subspace $L$ of dimension $\ell_1$ a linear function,
and suppose that there are at least $m\geq q^{C\ell_1}$ subspaces $W_1,\ldots,W_m$ of codimension $r$ and
linear functions $f_i\colon W_i\to\mathbb{F}_q$ such that for all $i=1,\ldots,m$,
\[
\Pr_{L\subseteq W_i}[T_1[L] \equiv f_i|_{L}]\geq \eps'.
\]
Then, there exists a zoom-out
$W$ of codimension strictly smaller than $r$ and
a linear function $f\colon W\to\mathbb{F}_q$ such
that
\[
\Pr_{L\subseteq W}[T_1[L]\equiv f|_{L}]\geq \Omega(\eps').
\]
\end{lemma}
We defer a detailed discussion about~\cref{lem:main_technical_intro} 
and its proof to~\cref{sec: bounded zoom-outs}, 
but remark that our proof of~\cref{lem:main_technical_intro} is very different from the arguments in~\cite{DKKMS1} and 
is significantly more involved.

\subsubsection{The Outer PCP}
Our outer PCP game is the same as that of~\cite{KMS,DKKMS1}, which is a smooth
parallel repetition of the equation versus 
variables game of Hastad~\cite{Hastad} (or of~\cite{KP} for the application to Quadratic Programming). As in there, we equip this game 
with the ``advice'' feature to facilitate zoom-in
coordination (as discussed above). For the sake of completeness we elaborate 
on the construction of the outer PCP below.

We start with an instance of $\GapLin$ that has a large
gap between the soundness and completeness. Namely,
we start with an instance $(X,\Eq)$ of linear equations over
$\mathbb{F}_q$ in which each equation has the form 
$a x_{i_1} + b x_{i_2} + c x_{i_3} = d$. It is known~\cite{Hastad}
that for all $\eta>0$, it is NP-hard to distinguish between the following two cases:
\begin{enumerate}
    \item YES case: ${\sf val}(X,\Eq)\geq 1-\eta$.
    \item NO case: ${\sf val}(X,\Eq)\leq \frac{1.1}{q}$.
\end{enumerate}
The standard \textit{Variable versus Equation game} proceeds as follows:
\begin{enumerate}
    \item The verifier picks an equation $e$ uniformly at random from ${\sf Eq}$ and lets $U =  \{x_{i_1},x_{i_2},x_{i_3}\}$ be the variables in it.
    \item The verifier samples $V\subseteq U$ of size $1$ uniformly.
\end{enumerate} 
The verifier then sends $U$ to the first prover, and $V$ to the second prover. The verifier expects to receive from each prover an assignment to the variables they received, and accepts if their assignments are consistent and satisfy the equation $e$. It is easy to see that if 
${\sf val}(X,\Eq)\geq 1-\eta$, then the provers have a strategy that makes the verifier accept with probability at least $1-\eta$, and if ${\sf val}(X,\Eq)\leq \frac{1.1}{q}$, then for any provers' strategy, the verifier accepts with probability at most $1-\Omega(1)$.

Now, we modify the standard Variable versus Equation game into the \textit{smooth Variable versus Equation game with $r$-advice}, which proceeds as follows. 
The verifier has a smoothness parameter $\beta>0$ and picks a random equation $e$, say $a x_{i_1} + b x_{i_2} + c x_{i_3} = d$, from $(X,\Eq)$. Then:
\begin{enumerate}
    \item With probability $1-\beta$ 
the verifier takes $U = V = \{x_{i_1},x_{i_2},x_{i_3}\}$ and vectors 
$u_1=v_1,\ldots,u_r=v_r\in \mathbb{F}_q^{U}$ sampled uniformly and independently.
    \item With probability $\beta$, the verifier 
    sets $U = \{x_{i_1}, x_{i_2}, x_{i_3}\}$, 
    chooses a set consisting of a single variable $V\subseteq U$ uniformly at random. The verifier picks $v_1,\ldots,v_r\in\mathbb{F}_q^V$ uniformly and independently 
    and appends to each $v_i$ the value $0$ in 
    the coordinates of $U\setminus V$ to get $u_1,\ldots,u_r$.
\end{enumerate} 
After that, the verifier sends $U$ and $u_1,\ldots,u_r$
to the first prover and $V$ and $v_1,\ldots,v_r$ to the
second prover. The verifier expects to get from them $\mathbb{F}_q$ assignments to the variables in $U$
and in $V$, and accepts if and only if these assignments 
are consistent, and furthermore the assignment to $U$ satisfies the equation $e$.

Denoting the smooth Variable versus Equation game above by ${\sf \Psi}$, it is easy to see that if ${\sf val}(X,\Eq)\geq 1-\eta$, then ${\sf val}(\Psi)\geq 1-\eta$, and 
if ${\sf val}(X,\Eq)\leq 1.1/q$, then ${\sf val}(\Psi)\leq 1-\Omega(q^{-r}\beta)$. 
Indeed, to see the last assertion, note that the probability that $V$ is a proper subset of $U$ and that $v_1,\ldots,v_r$ are all $0$, is $q^{-r}\beta$. In that case, the game being played is the standard Variable versus Equation game, which has value at most $1-\Omega(1)$ as discussed earlier. 
The gap between $1-\eta$ 
and $1-\Omega(q^{-r}\beta)$ is too weak for us, and thus we 
apply parallel repetition.

In the parallel repetition of the smooth Variable versus Equation game with advice, denoted by $\Psi^{\otimes k}$, the verifier picks $k$ equations
uniformly and independently $e_1,\ldots,e_k$, and picks $U_i$, $u^i_{1},\ldots, u^i_{r}$ and $V_i$, $v^i_{1},\ldots, v^i_{r}$ for each $i=1,\ldots, k$ from $e_i$ independently. Thus, the questions 
of the provers may be seen as $U = U_1\cup\ldots\cup U_k$ and $V = V_1\cup\ldots\cup V_k$ and their advice 
is $\vec{u}_j = (u^1_{j},\ldots,u^k_{j})\in\mathbb{F}_q^U$ for $j=1,\ldots, r$ and
$\vec{v}_j = (v^1_{j},\ldots,v^k_{j})\in\mathbb{F}_q^V$ for $j=1,\ldots, r$ respectively. The verifier expects to get from the first prover a vector in $\mathbb{F}_q^{U}$ which specifies an $\mathbb{F}_q$ assignment to $U$, and from the second prover a vector
in $\mathbb{F}_q^{V}$ which specifies an $\mathbb{F}_q$ assignment to $V$. The verifier accepts if and only if
these assignments are consistent and the assignment
of the first prover satisfies all of $e_1,\ldots,e_k$.
It is clear that if ${\sf val}(X,\Eq)\geq 1-\eta$, 
then ${\sf val}(\Psi^{\otimes m})\geq 1-k\eta$. 
Using the parallel repetition theorem of Rao~\cite{Rao} we argue that if ${\sf val}(X,\Eq)\leq \frac{1.1}{q}$, then ${\sf val}(\Psi^{\otimes k})\leq 2^{-\Omega(\beta q^{-r} k)}$. The game $\Psi^{\otimes k}$ is 
our outer PCP game.

\begin{remark}
In the case
of the Quadratic Programming application we require 
a hardness result in which the completeness is very 
close to $1$ (see~\cref{th: 3lin hardness}). The differences between the reduction in 
that case and the reduction presented above are 
mostly minor, and amount to picking the parameters 
a bit differently. There is one significant difference
in the analysis, namely we require a much sharper form of the ``covering property'' used in~\cite{KMS,DKKMS1}; see~\cref{sec: covering intro} for a discussion.
\end{remark}

\subsubsection{Composing the Outer PCP and the Inner PCP Game}
To compose the outer and inner PCPs, we take the outer PCP game, only keep the questions $U$ to the first 
prover and consider an induced $2$-Prover-$1$-Round
game on it. The alphabet is $\mathbb{F}_q^{3k}$, and we think of a response to a question $U$ as an $\mathbb{F}_q$ 
assignment to the variables of $U$. There is
a constraint between $U$ and $U'$ if there is a question $V$ to the second prover such that $V\subseteq U\cap U'$. Denoting the assignments to
$U$ and $U'$ by $s_{U}$ and $s_{U'}$, the constraint
between $U$ and $U'$ is that $s_{U}$ satisfies all
of the equations that form $U$, $s_{U'}$ satisfies 
all of the equations that form $U'$, and 
$s_U$, $s_{U'}$ agree on $U\cap U'$.

The composition amounts to replacing each question $U$ with a copy of our inner PCP. Namely, we identify between the question $U$ and the space $\mathbb{F}_q^{U}$, and then replace $U$ by a copy of the $\ell_2, \ell_1$ subspaces graph of $\mathbb{F}_q^U$. The answer $s_U$ is naturally identified with the linear function $f_U(x) = \langle s_U, x\rangle$, which is then encoded by the subspaces encoding via tables of assignments $T_{1,U}$ and $T_{2,U}$.

The constraints of the composed PCP must check that:
(1) side conditions: the encoded vector $s_U$ satisfies the equations 
of $U$, and (2) consistency: $s_U$ and $s_{U'}$ agree
on $U\cap U'$. 

The first set of constraints is addressed by the folding technique, which we omit from this discussion. 
The second set of constraints is addressed by the 
$\ell_1$ vs $\ell_2$ subspace test, except that we
have to modify it so that it works across blocks $U$ 
and $U'$. This completes the description of the 
composition step of the other PCP and the inner PCP, 
and thereby the description of our reduction.
\vspace{-2ex}
\paragraph{Analyzing the Composed PCP.} Traditionally, the analysis of the composed PCP relates its soundness with the soundness of the outer PCP and the soundness of the inner PCP. Such analysis, however, fails in our case. Instead, and similarly to prior works~\cite{KMS,DKKMS2}, our analysis relates the soundness of the composed PCP with the non-standard notion of soundness of the inner PCP (as discussed above) and the soundness of a modification of the (standard) outer PCP game, as presented above. The goal of this modification is to allow the provers to correlate their choice for the zoom-in set $Q$, and it is achieved by the additional vectors $\{\vec{u}_j\}_{j=1,\ldots,r}$ and $\{\vec{v}_j\}_{j=1,\ldots,r}$ that the verifier sends to the provers in the outer PCP game.

\subsubsection{The Covering Property}\label{sec: covering intro}
We end this introductory section by discussing the 
covering property, which is an important feature of
our outer PCP construction enabling the composition step to go through. The covering property 
first appeared in~\cite{KS} and later more extensively 
in the context of the $2$-to-$1$ Games~\cite{KMS,DKKMS1}. 
To discuss the covering property, let
$k\in \mathbb{N}$ be thought of as large, let 
$\beta\in (0,1)$ be thought of as $k^{-0.99}$
and consider pairwise disjoint sets $U_1,\ldots, U_k$ where each $U_i$ has size $3$ (in our context, $U_i$ 
will be the set of variables in the $i$th equation the
verifier chose). Let $U = U_1\cup\ldots\cup U_k$, and consider
the following two distributions over tuples in $\mathbb{F}_q^U$:
\begin{enumerate}
    \item Sample 
    $x_1,\ldots,x_{\ell}\in\mathbb{F}_q^U$ uniformly.
    \item For each $i$ independently, take $V_i = U_i$ with
    probability $1-\beta$ and otherwise take $V_i\subseteq U_i$
    randomly of size $1$, then set $V = V_1\cup\ldots\cup V_k$. 
    Sample $x_1,\ldots,x_{\ell}\in\mathbb{F}_q^V$ uniformly and lift them to points
    in $\mathbb{F}_q^{U}$ by appending $0$'s in $U\setminus V$.
    Output the lifted points.
\end{enumerate}
In~\cite{KMS} it is shown that the two distributions above are 
$q^{3\ell}\beta\sqrt{k}$ close in statistical distance, which
is good enough for~\cref{thm:main}. This is insufficient for~\cref{thm:QP}, though.\footnote{
The reason is that, letting $N$ be the size of the instance 
we produce, it holds that $k$ is roughly $\log N$
and $q^{\ell}$ is the alphabet size. To have small statistical 
distance, we must have $k \leq q^{6\ell}$, hence the soundness error 
could not go lower than $(\log N)^{-1/6}$.} Carrying out a 
different analysis, one can show that the two 
distributions are close with better parameters and in 
a stronger sense: there exists a set $E$ of $\ell$ tuples
which has negligible measure in both distributions, such 
that each tuple not in $E$ is assigned the same probability 
under the two distributions up to factor $(1+o(1))$. 
We are able to prove this statement provided that $k$ is 
only slightly larger than $q^{2\ell}$, which is still insufficient. 

The issue with the above two distributions is that they are 
actually far from each other if, say, $k = q^{1.9\ell}$. 
To see that, one can notice that the expected number of $i$'s 
such that each one of $x_1,\ldots,x_{\ell}$ has the form 
$(a,0,0)\in\mathbb{F}_q^3$ on coordinates corresponding to $U_i$
is very different. In the first distribution, this expectation 
is $\Theta(q^{-2\ell} k)$ which is less than 1, whereas in the second distribution it is at least $\beta k/3\geq k^{0.01}$.

To resolve this issue and to go all the way through in the Quadratic Programming application, we have to modify the 
distributions in the covering property so that (a) they will be
close even if $k = q^{1.01\ell}$, and (b) we can still use 
these distributions in the composition step in our analysis
of the PCP construction. Indeed, this is the route we take, 
and the two distributions we use are defined as follows:
\begin{enumerate}
    \item Sample 
    $x_1,\ldots,x_{\ell}\in\mathbb{F}_q^U$ uniformly.
    \item For each $i$ independently, take $V_i = U_i$ with
    probability $1-\beta$ and otherwise take $V_i\subseteq U_i$
    randomly of size $1$, then set $V = V_1\cup\ldots\cup V_k$. 
    Sample $x_1,\ldots,x_{\ell}\in\mathbb{F}_q^V$ uniformly, 
    and let $w_i\in\mathbb{F}_q^U$ be the vector of the coefficients of the equation $e_i$ forming $U_i$. 
    Lift the points $x_1,\ldots,x_{\ell}$
    to $x_1',\ldots,x_{\ell}'\in \mathbb{F}_q^{U}$ by appending $0$'s in $U\setminus V$ and take $y_j = x_j + \sum\limits_{i=1}^{k} \alpha_{i,j}w_i$ where $\alpha_{i,j}$ 
    are independent random elements from $\mathbb{F}_q$. 
    Output $y_1,\ldots,y_{\ell}$.
\end{enumerate}
We show that for a suitable choice of $k$ and $\beta$, these distributions are close even in the case that 
$k = q^{1.01\ell}$.\footnote{More specifically, one takes a small $c>0$ and chooses $\beta = k^{2c/3 - 1}$, $k = q^{(1+c)\ell}$.} 
Indeed, as a sanity check one could count the expected number
of appearances of blocks of the form $(0,a,0)\in\mathbb{F}^{3}$
and see they are very close ($q^{-2\ell} k$ versus $(1-\beta)q^{-2\ell}k + \beta k q^{-\ell}$).
In this setting of parameters, $k$ is roughly equal to the 
alphabet size --- which can be made to be equal $(\log N)^{1-o(1)}$ under
quasi-polynomial time reductions.
\begin{remark}
A tight covering property is crucial for obtaining the tight hardness of approximation factor in~\cref{thm:QP}. In the reduction of~\cite{ABHKS} from $2$-Prover-$1$-Round 
games to Quadratic Programming, the size of the resulting instance is exponential in the alphabet size and the soundness error remains roughly the same. In our case the alphabet size is roughly $k$, hence the 
instance size is dominated by $N = 2^{\Theta(k^{1+o(1)})}$. 
If our analysis required $k = q^{C\ell}$, then even showing an optimal soundness of $q^{-(1-o(1))\ell}$ for the $2$-Prover-$1$-Round game would only yield a factor of $(\log N)^{1/C-o(1)}$ hardness for Quadratic Programming.
\end{remark}

\subsection{Open Problems}
We end this introductory section by mentioning a few open problems for future research.
\begin{enumerate}
    \item {\bf The Low-degree CSP Conjecture.}  It would be interesting to see if our techniques have any bearing on the conjecture of~\cite{chuzhoy2022new}. This conjecture
is a sort of a mixture between the $d$-to-$1$ 
games and sliding scale conjectures, and it focuses
on the relation between the 
instance size and the soundness error, while allowing the alphabet to be quite
large. The conjecture is motivated by improved hardness 
results for densest $k$-subgraph 
style problems, and below we give a rough description of it.

The $2$-to-$1$ (or $2$-to-$2$) Theorem asserts that the soundness of a PCP can be arbitrarily small constant in a $d$-to-$1$ (or $d$-to-$d$) game even if $d$ is fixed; it does not apply well in the sub-constant soundness regime, though. The parallel repetition theorem asserts that for every $d\colon \mathbb{N}\to\mathbb{N}$, one can get a $d(n)$-to-$d(n)$ PCP with soundness which is polynomial in $1/d(n)$. The Low-degree CSP conjecture asks for a PCP that has $d(n)$ growing with the instance size, whose soundness decays more rapidly than polynomial in $d$. Specifically, it asks for $d(n) = 2^{\log^{\delta} n}$ and soundness $2^{-\log^{1/2+\delta} n}$.

\item {\bf Perfect completeness.} Our reduction starts with the $\GapLin$ problem and thus inherently has imperfect completeness. It would be very interesting if one could prove a version of~\Cref{thm:main} with perfect completeness. As far as we know no result better than~\Cref{thm:soundness_alphabet_tradeoff_pc} is known to date. While we are not aware of any direct applications of it, we believe such a result may lead to techniques which are able to give strong PCPs with perfect completeness (which are quite rare). 
\item {\bf Other applications.} As shown above, it is sometimes the case that results proved assuming the Unique-Games Conjecture can be translated into NP-hardness results using~\Cref{thm:main}. We suspect there should be other such problems.
\end{enumerate}

\section{Preliminaries}
\subsection{The Grassmann Domain}
In this section we present the Grassmann domain and associated Fourier analytic tools that are required for our analysis of the inner PCP. Throughout this section we fix parameters $n, \ell$ with $1 \ll \ell \ll n$, and a prime power $q$ (which we will eventually pick to be $2$). 
\subsubsection{Basic Definitions}
The Grassmann domain $\Grass_q(n, \ell)$ consists of the set of all $\ell$-dimensional subspaces $L \subseteq \Ff_q^n$. At times we will have a vector space $W$ over $\mathbb{F}_q$, and we may write $\Grass_q(W, \ell)$ to mean $\Grass_q(n, \ell)$, where $n$ is the dimension of $W$ (under some identification of $W$ with $\Ff_q^n$). When the field size $\Ff_q$ is clear from context we drop the subscript and simply write $\Grass(n, \ell)$ or $\Grass(W, \ell)$.

The number of $\ell$-dimensional subspaces of $\Ff_q^n$ is counted by the Gaussian binomial coefficient, $\qbin{n}{\ell}$. The following
standard fact gives a formula for the Gaussian
binomial coefficients, and we omit the proof.
\begin{fact}
    Suppose $1 \leq \ell \leq \frac{n}{2}$, then the number of vertices in $\Grass_q(n, \ell)$ is given by 
    \[
    \qbin{n}{\ell} = \prod_{i=0}^{\ell-1}\frac{q^n - q^i}{q^\ell - q^i}.
    \]
\end{fact}

\paragraph{Bipartite Inclusion Graphs} We will often consider bipartite inclusion graphs between $\Grass_q(n, \ell)$ and $\Grass_q(n, \ell')$ for $\ell' < \ell$. This graph has an edge between every $L \in \Grass_q(n,\ell)$ and $L' \in \Grass_q(n, \ell')$ such that $L \supseteq L'$. We refer to the family of such graphs as \emph{Grassmann graphs}. The edges of these graphs are the motivation for the test used in our inner PCP. While we never refer to this graph explicitly, it will be helpful to have it in mind.
\paragraph{Zoom-ins and Zoom-outs.} 

For subspaces $Q \subseteq W \subseteq \Ff_q^n$ and a dimension $\ell \in \mathbb{N}$, let 
\[
\Zoom_{\ell}[Q,W] = \{L \in \Grass_q(n, \ell) \; | \; Q \subseteq L \subseteq W\}.
\]
We refer to $Q$ as a zoom-in and $W$ as a zoom-out. When $W = \Ff_q^n$, $\Zoom_{\ell}[Q,W]$ is the zoom-in on $Q$, and when $Q = \{0\}$, $\Zoom_{\ell}[Q,W]$ is the zoom-out on $W$.

\subsubsection{Pseudo-randomness over the Grassmann graph}
One notion that will be important to us is $(r,\epsilon)$-pseudo-randomness, which measures how large the fractional size of a set of subspaces $\mc{L}$ can be in a zoom-in/zoom-out restrictions of ``size $r$''.\footnote{The results we mention can be stated in greater generality, but we avoid it for the sake of
simplicity.}  The measure of a set of subspaces $\mc{L} \subseteq \Grass_q(n, \ell)$ is the fraction of $\ell$-dimensional subspaces it contains, denoted as
\[
\mu(\mc{L}) = \Pr_{L \in \Grass_q(n, \ell)}[L \in \mc{L}].
\]
We remark that the dimension $\ell$ in the probability above will always be clear from context based on $\mc{L}$.

For subspaces $Q \subseteq W \subseteq \Ff_q^n$, the measure of $\mc{L}$ under zoom-in/zoom-out pair $(Q,W)$ is defined as 
\begin{equation} \label{eq: measure in zoom-in zoom-out}
\mu_{Q, W}(\mc{L}) = \Pr_{L \in \Grass_q(n, \ell)}[L \in \mc{L} \; | \; Q \subseteq L \subseteq W].
\end{equation}
We also define a similar notation for measure inside of only a zoom-in or only a zoom-out. 

For a zoom-in $Q$, we write
\[
\mu_{Q, \circ}(\mc{L}) = \Pr_{L \in \Grass_q(n, \ell)}[L \in \mc{L} \; | \; Q \subseteq L].
\]
and for a zoom-out $W$, we write
\[
\mu_{\circ, W}(\mc{L}) = \Pr_{L \in \Grass_q(n, \ell)}[L \in \mc{L} \; | \; L \subseteq W].
\]
Here we abuse notation and note that when we use the notation $\mu_Q$ or $\mu_W$, it will always be clear from context which of the above two notations we are using and whether the subspace in the subscript is a zoom-in or zoom-out.

\begin{definition}
  We say that a set of subspaces $\mc{L} \subseteq \Grass_q(n,\ell)$ is \textit{$(r, \epsilon)$-pseudo-random} if for all $Q \subseteq W \subseteq \Ff_q^n$ satisfying $\dim(Q) + \codim(W) = r$, we have $\mu_{Q, W}(\mc{L})  \leq \epsilon$.
\end{definition}
It turns out that a pseudo-random $\mc{L}$ as above 
has the same edge expansion (and other related edge-counts) as a random set of the same density, and we 
require such an assertion in our inner PCP. Below is a precise formulation. 

\begin{lemma} \label{lm: pseudorandom edges}
    Let $\mc{L} \subseteq \Gras_q(n, 2\ell)$ and $\mc{R} \subseteq \Gras_q(n, 2(1-\delta)\ell)$ be sets of subspaces with fractional sizes $\Pr_{L \in \Grass_q(n, 2\ell)}[L \in \mc{L}] = \alpha$ and $\Pr_{R \in \Grass_q(n,2(1-\delta)\ell)}[R \in \mc{R}] = \beta$, and suppose that $\mc{L}$ is $(r,\eps)$ pseudo-random.
Then for all $t\geq 4$ that are powers of $2$,
\[ 
\Pr_{L \supseteq R}[L \in \mc{L}, R \in \mc{R}]
\leq q^{O_{t,r}(1)} \beta^{(t-1)/t} \eps^{(t-2)/t} + q^{-r\delta\ell}\sqrt{\alpha\beta}.
\]
In the probability above, $L \in \Grass_q(n, 2\ell)$ is chosen uniformly at random and $R \in \Grass_q(n, 2(1-\delta)\ell)$ is chosen uniformly at random conditioned on being contained in $L$. 
\end{lemma}
\begin{proof}
    Deferred to~\cref{app: level d}.
\end{proof}

We also need the following lemma, asserting that if a not-too-small set $\mc{L}$ 
is highly pseudo-random,
then its density remains nearly the same on all zoom-ins of dimension $1$. For a point $z \in \Ff_q^n$ we denote by $\mu_z(\cdot)$ the measure $\mu_Q(\cdot)$ as above where $Q = \spa(z)$.
\begin{lemma} \label{lm: fourier even covering}
For all $\xi>0$, the following holds for
sufficiently large $\ell$. Suppose that $\ell' \geq \frac{\xi}{3}\ell$, 
$\delta_2 = \frac{\xi}{100}$, and $W$ is a subspace of dimension $\dim(W)\geq \ell'^2$. Let $\Lc \subseteq \Grass_q(W, \ell')$ have measure $\mu(\Lc) = \eta \geq q^{-2\ell}$ and set $Z = \{z \in W \; | \; |\mu_z(\Lc) - \eta| \leq \frac{\eta}{10} \}$. If $\Lc$ is $(1, q^{\delta_2 \ell/100}\eta)$-pseudo-random, then 
\[
|Z| \geq \left(1- q^{-\frac{\ell'}{2}}\right)|W|.
\]
\end{lemma}
\begin{proof}
    The proof is deferred to~\cref{app: even covering}.
\end{proof}

\subsection{Hardness of 3LIN}
In this section we cite several hardness of approximation results for the problem of solving linear equations over 
finite fields, which are the starting point of our reduction.
We begin by defining the $\Lin$ and the $\GapLin$ problem.

\begin{definition}
    For a prime power $q$, an instance of $\Lin$  
    $(X, \Eq)$ consists of a set of variables $X$ and 
    a set of linear equations $\Eq$ over $\mathbb{F}_q$. 
    Each equation in $\Eq$ depends on exactly three variables in $X$, each variable appears in at most $10$ equations, and any two distinct equations in $\Eq$ share at most a single variable. 
\end{definition}
The goal in the $\Lin$ problem is to find an assignment 
$A\colon X\to\mathbb{F}_q$ satisfying as many of the equations
in $E$ as possible. The maximum fraction of equations that 
can be satisfied is called the value of the instance.
We remark that usually in the literature, 
the condition that two equations in 
$E$ share at most a single variable is not included 
in the definition of $\Lin$, as well as the bound on the 
number of occurrences of each variable.

For $0<s<c\leq 1$, the problem $\GapLin[c,s]$ is the promise
problem wherein the input is an instance $(X,\Eq)$ of $\Lin$
promised to either have value at least $c$ or at most $s$, 
and the goal is to distinguish between these two cases. 
The problem $\GapLin[c,s]$ with various settings of $c$ and $s$ 
will be the starting point for 
our reductions.

 In the proof of~\cref{thm:main}, we will use a classical result of H\r{a}stad~\cite{Hastad}, 
 stating that for 
 general $\Lin$ instances (i.e., without the additional conditions
 that two equations share at most a single variable and each variable occurs in at most $10$ equations), 
 the problem $\GapLin[1-\eps,1/q+\eps]$ is NP-hard for all 
 constant $q\in\mathbb{N}$ and $\eps>0$. This result implies 
 the following theorem by elementary reductions (similarly to~\cref{th: sat to 3lin}; see also~\cite[Footnote 14]{minzer2022monotonicity}):
\begin{thm} \label{th: 3lin hardness}
There exists $s<1$ such that for every $\eta > 0$ and prime $q$, $\GapLin\left[1-\eta, s\right]$ is $\NP$-hard.
\end{thm}

 In the proof of~\cref{thm:QP} we will need a hardness result for $\Lin$ with completeness close to $1$, and we will use a hardness result of Khot and Ponnuswami~\cite{KP}. Once 
 again, their result does not immediately guarantee the fact
 that any two equations share at most a single variable and that each variable occurs in at most $10$ equations,
 however it can be achieved by 
 an elementary reduction.

\begin{thm} \label{th: sat to 3lin}
    For any $q$ power of $2$, there is a reduction from $\SAT$ with size $n$ to a $\GapLin[1-\eta, 1-\epsilon]$ instance with size $N$ over $\Ff_q$ with $\{0,1\}$-coefficients, where 
    \begin{itemize}
        \item Both $N$ and the running time of the reduction are bounded by $2^{O(\log^2n)}$;
        \item $\eta \leq 2^{-\Omega(\sqrt{\log N})}$;
        \item $\epsilon \geq \Omega\left(\frac{1}{\log^3N} \right)$.
    \end{itemize}
\end{thm}
\begin{proof}
    Using~\cite[Theorem 4]{KP} we get the statement above without the additional property that any two distinct equations share at most a single variable.
    To get this property we reduce an instance $(X,{\sf Eq})$ to an instance $(X',{\sf Eq}')$ as
    follows: for each equation $e\in {\sf Eq}$ given as $x+y+z = b$ for $x,y,z\in X$, we produce the new variables $x_e,y_e,z_e$ and add the equations 
    $x+y_e+z_e = b$, $x_e+y+z_e = b$ and $x_e+y_e+z = b$. We note that each variable in $(X',{\sf Eq}')$ appears in at most $10$ equations, and that if 
    ${\sf val}(X,{\sf Eq})\geq 1-\eta$, then ${\sf val}(X',{\sf Eq}')\geq 1-\eta$, so the completeness is clear. For the soundness analysis suppose that $A$ is an assignment to $(X',{\sf Eq}')$ satisfying at least $1-\eps$ fraction of the equations. By an averaging argument for at least $1-3\eps$ fraction of $e\in {\sf Eq}$, we have that $A$ satisfies all of the equations $x+y_e+z_e = b$, $x_e+y+z_e = b$ and $x_e+y_e+z = b$. By adding them up and using the fact that $\mathbb{F}_q$ has characteristic $2$ we get that $A$ also satisfies the equation $x+y+z = b$, giving that ${\sf val}(X,{\sf Eq})\geq 1-3\eps$ and completing the soundness analysis.
\end{proof}

\section{The Outer PCP}\label{sec:outer}
In this section, we describe our outer PCP game.
In short, our outer PCP is a smooth parallel repetition of a ``variable versus equation'' game with advice. This outer PCP was 
first considered in~\cite{KS} without the advice feature, and
then in~\cite{KMS} with the advice feature.

\subsection{The Outer PCP construction}
Let $\eps_1 < \eps_2$ be parameters. Our
 reduction starts with an instance $(X,\Eq)$ of the $\GapLin[1-\eps_1, 1-\eps_2]$ 
 problem, and we present it gradually. We begin by presenting the
 basic Variable versus Equation Game, then equip
 it with the additional features of smoothness and advice, and finally perform parallel repetition.

\skipi
\subsubsection{The Variable versus Equation Game}
We first convert the instance $(X,\Eq)$ into an instance 
of $2$-Prover-$1$-Round Games, and it will be convenient 
for us to describe it in the active view with a verifier 
and $2$ provers. 

In the Variable versus Equation game, the verifier picks an equation $e \in \Eq$ uniformly at random, and then chooses a random variable $x \in e$. The verifier sends the question $e$, i.e.\ the three variables appearing in $e$, to the first prover, and sends the variable $x$ to the second prover. The provers are expected to answer with assignments to their received variables, and the verifier accepts if and only if the two assignments agree on $x$ and the first prover's assignment satisfies the equation $e$. If the verifier accepts then we also say that the provers pass. This game has the following completeness and soundness, which are both easy to see (we omit the formal proof):
\begin{enumerate}
    \item \textbf{Completeness:} If $(X, \Eq)$ has an assignment satisfying $(1-\epsilon)$-fraction of the equations, then the provers have a strategy that passes with probability at least $1-\epsilon$. 
    \item \textbf{Soundness:} If $(X, \Eq)$ has no assignment satisfying more than $(1-\epsilon)$-fraction of the equations, then the provers can pass with probability at most $1 - \frac{\epsilon}{3}$.
\end{enumerate}

\subsubsection{The Smooth Variable versus Equation Game}
We next describe a smooth version of the  Variable versus Equation game. In this game, the verifier has a parameter 
$\beta\in (0,1]$, and it proceeds as follows:
\begin{enumerate}
    \item The verifier chooses an equation $e \in \Eq$ uniformly, 
    and lets $U$ be the set of variables in $e$.
    \item With probability $1-\beta$, the verifier chooses 
    $V = U$. With probability $\beta$, the verifier chooses
    $V\subseteq U$ randomly of size $1$.
    \item The verifier sends $U$ to the first prover, and $V$ to the second prover.
    \item The provers respond with assignments to the variables they receive, and the verifier accepts if and only if their 
    assignments agree on $V$ and the assignment to $U$ satisfies
    the equation $e$.
\end{enumerate}
The smooth Variable versus Equation game has the following 
completeness and soundness property, which are again easily 
seen to hold (we omit the formal proof).
\begin{enumerate}
    \item \textbf{Completeness:} If $(X, \Eq)$ has an assignment satisfying $(1-\epsilon)$-fraction of the equations, then the provers have a strategy that passes with probability at least $1-\epsilon$. 
    \item \textbf{Soundness:} If $(X, \Eq)$ has no assignment satisfying more than $(1-\epsilon)$-fraction of the equations, then the provers can pass with probability at most $1 - \frac{\beta\epsilon}{3}$.
\end{enumerate}

\subsubsection{The Smooth Variable versus Equation Game with Advice}
Next, we introduce the feature of advice into the smooth 
Variable versus Equation Game. This ``advice'' acts as shared randomness which may help the provers in their strategy;  we show though that it does not considerably change the soundness. The game is denoted by $G_{\beta,r}$ for $\beta\in (0,1]$ and $r\in\mathbb{N}$, and proceeds as follows:
\begin{enumerate}
    \item The verifier chooses an equation $e \in \Eq$ uniformly, 
    and lets $U$ be the set of variables in $e$.
    \item With probability $1-\beta$, the verifier chooses 
    $V = U$. With probability $\beta$, the verifier chooses
    $V\subseteq U$ randomly of size $1$.
    \item The verifier picks vectors $v_1,\ldots,v_r\in \mathbb{F}_q^V$ 
    uniformly and independently. If $U = V$ the verifier takes $u_i = v_i$ for all $i$, and 
    otherwise the verifier takes the vectors $u_1,\ldots,u_r\in\mathbb{F}_q^U$
    where for all $i=1,\ldots,r$, the vector $u_i$ agrees with $v_i$ on the coordinate of $V$, and is 
    $0$ in the coordinates of $U\setminus V$.
    \item The verifier sends $U$ and $u_1,\ldots,u_r$ to the first prover, and $V$ and $v_1,\ldots,v_r$ to the second prover.
    \item The provers respond with assignments to the variables they receive, and the verifier accepts if and only if their 
    assignments agree on $V$ and the assignment to $U$ satisfies
    the equation $e$.
\end{enumerate}

Below we state the completeness and soundness of this game:
\begin{enumerate}
    \item \textbf{Completeness:} If $(X, \Eq)$ has an assignment satisfying $(1-\epsilon)$-fraction of the equations, then the provers have a strategy that passes with probability at least $1-\epsilon$. This is easy to see.
    \item \textbf{Soundness:} If $(X, \Eq)$ has no assignment satisfying more than $(1-\epsilon)$-fraction of the equations, then the provers can pass with probability at most $1 - \frac{q^{-r}\beta\epsilon}{3}$. Indeed, suppose that the
    provers can win the game with probability at least 
    $1-\eta$. Note that with probability at least $\beta q^{-r}$ it holds 
    that $U\neq V$ and all the vectors $u_1,\ldots,u_r$ 
    and $v_1,\ldots,v_r$ are all $0$, in which case the provers
    play the standard Variable versus Equation game. Thus,
    the provers' strategy wins in the latter game with 
    probability at least $1-\frac{\eta}{q^{-r}\beta}$, so we must have that $ \frac{\eta}{q^{-r}\beta} \geq \frac{\eps}{3}$.
\end{enumerate}

\subsubsection{Parallel Repetition of the Smooth Variable versus Equation Game with Advice} \label{sec: final outer pcp}
Finally, our outer PCP is the $k$-fold parallel repetition of $G_{\beta,r}$, which we denote by $G^{\otimes k}_{\beta,r}$. 
Below is a full description of it:
\begin{enumerate}
    \item The verifier chooses equations $e_1,\ldots,e_k \in \Eq$ uniformly and independently, and lets $U_i$ be the set of variables in $e_i$.
    \item For each $i$ independently, with probability $1-\beta$, the verifier chooses 
    $V_i = U_i$. With probability $\beta$, the verifier chooses
    $V_i\subseteq U_i$ randomly of size $1$.
    \item For each $i$ independently, the verifier picks vectors $v_1^{i},\ldots,v_r^{i}\in \mathbb{F}_q^{V_i}$ 
    uniformly and independently. If $U_i = V_i$ the verifier takes $u_j^{i} = v_j^{i}$ for $j=1,\ldots,r$, and 
    otherwise the verifier takes the vectors $u_1^{i},\ldots,u_r^{i}\in\mathbb{F}_q^{U_i}$
    where for all $j=1,\ldots,r$, the vector $u_j^{i}$ agrees with $v_j^{i}$ on the coordinate of $V_i$, and is 
    $0$ in the coordinates of $U_i\setminus V_i$.
    \item The verifier sets $U = \bigcup_{i=1}^k U_i$ and $u_j = (u_j^1,\ldots,u_j^k)$ for each $j=1,\ldots r$, and 
    $V = \cup_{i=1}^{k} V_i$ and $v_j = (v_j^1,\ldots,v_j^k)$ 
    for each $j=1,\ldots,r$. The verifier sends $U$ and 
    $u_1,\ldots,u_r$ to the first prover, and $V$ and 
    $v_1,\ldots,v_r$ to the second prover.
    \item The provers respond with assignments to the variables they receive, and the verifier accepts if and only if their 
    assignments agree on $V$ and the assignment to $U$ satisfies
    the equations $e_1,\ldots,e_k$.
\end{enumerate}
Next, we state the completeness and the soundness of the game 
$G_{\beta,r}^{\otimes k}$, and we begin with its completeness.
\begin{claim}
If $(X,\Eq)$ has an assignment satisfying at least $1-\eps$
of the equations, then the provers can win 
$G_{\beta,r}^{\otimes k}$ with probability at least $1-k\eps$.
\end{claim}
\begin{proof}
Let $A$ be an assignment that satisfies at least $(1-\eps)$-fraction of the equations in $E$, and consider the strategy
of the provers that assigns their variables according to $A$.
Note that whenever each one of the equations $e_1,\ldots,e_k$
the verifier chose is satisfied by $A$, the verifier accepts.
By the union bound, the probability this happens is at least
$1-k\eps$.
\end{proof}
Next, we establish the soundness of the game $G_{\beta,r}^{\otimes k}$.
\begin{claim}\label{claim:soundness_of_outerpcp}
If there is no assignment to $(X,\Eq)$ satisfying at least $(1-\eps)$-fraction
of the equations, then the provers can win 
$G_{\beta,r}^{\otimes k}$ with probability at most 
$2^{-\Omega(\eps^2 q^{-r}\beta k)}$.
\end{claim}
\begin{proof}
We appeal to the parallel repetition theorem for projection games
of Rao~\cite{Rao}, but we have to do so carefully. That theorem
states that if $\Psi$ is a $2$-Prover-$1$-Round game with 
${\sf val}(\Psi)\leq 1-\eta$, then ${\sf val}(\Psi^{\otimes k})\leq 2^{-\Omega(\eta^2 k)}$. We cannot apply the theorem
directly on $G_{\beta,r}$ (as the square is too costly for us).
Instead, we consider the game $\Psi = G_{\beta,r}^{\ceil{q^r/\beta}}$ and note that it has value bounded away from $1$. 

Write ${\sf val}(\Psi) = 1-\eta$. 
For each coordinate $i$, the probability that $U_i=V_i$ or $U_i\neq V_i$ but $v_1^{i},\ldots,v_{r}^{i}$ are all $0$, is equal to $1-\beta q^{-r}$. Thus, by independence the probability that this happens for $\ceil{q^r/\beta}$ coordinates is at most 
$(1-\beta q^{-r})^{\ceil{q^r/\beta}}\leq e^{-1+o(1)}\leq 0.9$.
In particular, with probability 
at least $0.1$ there exists at least one coordinate $i$
in which $U_i\neq V_i$ and all of the advice vectors $v_1^i,\ldots,v_r^i$ and $u_1^i,\ldots,u_r^i$ are all $0$. 
Thus, there exists a coordinate $i$ and a fixing for the questions of the provers outside $i$ so that the answers of
the players to the $i$th coordinate win the standard Variable
versus Equation game with probability at least $1-10\eta$. 
It follows that $1-10\eta\leq 1-\frac{\eps}{3}$, and so 
$\eta\geq \frac{\eps}{30}$.

We conclude from Rao's parallel repetition theorem that
\[
{\sf val}(G_{\beta,r}^{\otimes k})
={\sf val}(\Psi^{\otimes \frac{k}{\ceil{q^r/\beta}}})
\leq 2^{-\Omega(\eps^2 q^{-r}\beta k)}.
\qedhere
\]
\end{proof}

\paragraph{Viewing the advice as subspaces.} 
Due to the fact that each variable appears in at most $O(1)$
equations, it can easily be seen that with probability $1-O(k^2/n)$, all variables in $e_1,\ldots,e_k$ are distinct.
In that case, note that the $r$ vectors of advice 
to the second prover, $v_1, \ldots, v_r \in \Ff_q^V$, 
are uniform, and the second prover may consider their span $Q_V$. Note that, conditioned on (the likely event of) the advice vectors being linearly independent, $Q_V$ is a uniformly random $r$ 
dimensional subspace of $\mathbb{F}_q^V$.
As for the first prover, the vectors $u_1,\ldots,u_r\in \mathbb{F}_q^U$ are not uniformly distributed. Nevertheless,
as shown by the covering property from~\cite{KS,KMS} (and presented below), the distribution of $u_1,\ldots,u_r$ is 
close to uniform over $r$-tuple of vectors from $\mathbb{F}_q^U$.
Thus, the first prover can also take their span, call it 
$Q_U$, and think of it as a random $r$-dimensional subspace
of $\mathbb{F}_q^U$ (which is highly correlated to $Q_V$).

\section{The Composed PCP Construction}\label{sec:pcp_construct}
In this section we describe our final PCP construction, 
which is a composition of the outer PCP from~\cref{sec:outer} with the inner PCP 
based on the Grassmann consistency test.

\subsection{The Underlying Graph}
Our reduction starts from an instance $(X, \Eq)$ of $\GapLin$. Consider the game 
$G_{\beta,r}^{\otimes k}$ from~\cref{sec:outer}, and 
let $\U$ denote the set of questions asked to the first prover. Thus $\U$ consists of all $k$-tuples of equations $U = (e_1, \ldots, e_k) \in \Eq^{k}$ from the $\GapLin$ instance $(X, \Eq)$. 
It will be convenient to only keep the $U = (e_1, \ldots, e_k)$ that satisfy the following properties:
\begin{itemize}
    \item The equations $e_1,\ldots, e_k$ are distinct and do not share variables.
    \item For any $i \neq j$ and pair of variables $x \in e_i$ and $y \in e_j$, the variables $x$ and $y$ do not appear together in any equation in the instance $(X, \Eq)$. 
\end{itemize}

The fraction of $U = (e_1, \ldots, e_k)$ that do not satisfy the above is $O(k^2/n)$ which is negligible for us, and dropping
them will only reduce our completeness by $o(1)$. This will not affect our analysis, and henceforth we will assume that all $U = (e_1, \ldots, e_k)$ satisfy the above properties. We now describe the $2$-Prover-$1$-Round Games instance $\Psi = (\mathcal{A}, \mathcal{B}, E, \Sigma_1, \Sigma_2, \Phi)$. All vertices in the underlying graph will correspond to subspaces of $\Ff_q^X$.

{\bf Notation:} for $e \in \Eq$ let $v_e \in \Ff_q^{X}$ be the vector of coefficients of the equation $e$; we remark that as our instance $(X,\Eq)$ will come either from either~\cref{th: 3lin hardness} with $q=2$ or~\cref{th: sat to 3lin}, the vector $v_e$ will be $1$ on coordinates corresponding to variables in $e$, and $0$ on all other coordinates.

\subsubsection{The Vertices}
For each question $U = (e_1, \ldots, e_k)$, let $H_U = \spa(v_{e_1}, \ldots, v_{e_k})$, and think of it as a subspace of $\Ff_q^U$. By the first property described above, $\dim(H_U) = k$ and $\dim(\Ff_q^{U}) = 3k$. The vertices of $\Psi$ are: 
\begin{align*}
&\A = \{L \oplus H_U \; | \; U \in \U, L \subseteq \Ff_q^U, \dim(L) = 2\ell, L \cap H_U = \{0\} \},\\
&\B = \{R \; | \; \exists U \in \U, \text{ s.t. } R \subseteq \Ff_q^U, \dim(R) = 2(1-\delta)\ell  \}.
\end{align*}
Morally, the vertices in $\mathcal{A}$ are all $2\ell$-dimensional subspaces of some $\mathbb{F}_q^{U}$ for some $U\in \mathcal{U}$. For technical reasons, we require them to intersect
$H_U$ trivially (which is the case for a typical $2\ell$-dimensional space) and add to them the space $H_U$.\footnote{This has the effect of collapsing $L$ and $L'$ such that $L\oplus H_U = L'\oplus H_U$ to a single vertex.}
The vertices in $\mathcal{B}$ are all $2(1-\delta)\ell$
dimensional subspaces of $\mathbb{F}_q^U$.

\subsubsection{The Alphabets} 
The alphabets $\Sigma_1, \Sigma_2$ have sizes $|\Sigma_1| = q^{2\ell}$ and $|\Sigma_2| = q^{2(1-\delta)\ell}$. For each vertex $L \oplus H_U \in \A$, let $\psi: H_U \xrightarrow[]{} \Ff_q$ denote the linear function that satisfies the side conditions given by the equations in $U$. In notations, writing $e_i \in U$ as $\langle x, v_{e_i} \rangle = b_i$ for $x \in \Ff_q^U$, we set $\psi(v_{e_i}) = b_i$. We say a linear function 
$f\colon L\oplus H_U\to\mathbb{F}_q$ satisfies the side conditions of $U$ if $f|_{H_{U}}\equiv \psi$. 
In this language, for a vertex $L \oplus H_U$ we identify $\Sigma_1$ with
\[
\{f : L \oplus H_U \xrightarrow[]{} \Ff_q \; | \; f \text{ is linear function satisfying the side conditions of $U$} \}.
\]
As $L \cap H_U = \{0\}$ and $\dim(L) = 2\ell$, it is easy to see that the above set indeed has size $q^{2\ell}$. For each right vertex $R$, we identify $\Sigma_2$ with
\[
\{f : R \xrightarrow[]{} \Ff_q \; | \; f \text{ is linear}\}.
\]

\subsubsection{The Edges} 
To define the edges, we first need the following relation on $\mathcal{A}$. Say that $(L \oplus H_U) \sim (L' \oplus H_{U'})$ if 
\[
L + H_U + H_{U'} = L' + H_U + H_{U'}.
\]
As all subspaces above are in $\Ff_q^{X}$, the equality above is well defined. We next show that this is an equivalence relation. The reflexivity and symmetry are clear, and the following lemma establishes transitivity.

\begin{lemma} \label{lm: transitive}
    If $L_1 + H_{U_1} + H_{U_2} = L_2 + H_{U_1} + H_{U_2}$, and $L_2 + H_{U_2} + H_{U_3} = L_3 + H_{U_2} + H_{U_3}$, then
    \[
    L_1 + H_{U_1} + H_{U_3} = L_3 + H_{U_1} + H_{U_3}.
    \]
\end{lemma}
\begin{proof}
    We ``add'' $H_{U_1}$ to the second equation to obtain,
    \[
    L_2 + H_{U_1} + H_{U_2} + H_{U_3} = L_3 + H_{U_1} + H_{U_2} + H_{U_3}.
    \]
    Next, write $H_{U_2} = A + B$, where $A$ is the span of all vectors $v_e$ for equations $e$ in $U_2$ that are also in $U_1$ or $U_3$, while $B$ is the span of all vectors $v_e$ for equations $e \in U_2$ that are in neither $U_1$ nor $U_3$. It follows that $A \cap B = \{0\}$. Now note that any equation $e \in B$ has at most one variable that appears in an equation in $U_1$, and at most one variable that appears in an equation in $U_3$. Thus, each $e \in B$, has a ``private variable'', and as the equations in $B$ are over disjoint sets of variables, this private variable does not appear in $U_1 \cup U_3 \cup (U_2 \setminus e)$. It follows that 
    \begin{equation}\label{eq:transitivity}
    B \cap (L_1 + L_3 + H_{U_1} + H_{U_3}) = \{0\} \subset \Ff_q^X.
    \end{equation}
    Indeed, by the above discussion any nonzero vector in $B \subseteq \Ff_q^X $ is nonzero on at least one coordinate of $X$ (corresponding to a private variable), and no vector in $\Ff_q^{U_1}$ or $\Ff_q^{U_2}$ is supported on this coordinate.

    Substituting $H_{U_2} = A + B$ into the original equation yields
    \[
    L_1 + (H_{U_1} + H_{U_3} + A) + B =  L_3 +( H_{U_1} + H_{U_3} + A) + B,
    \]
    and as $A \subset H_{U_1} + H_{U_3}$ we get $L_1 + H_{U_1} + H_{U_3} + B =  L_3 + H_{U_1} + H_{U_3} + B$.
    Using~\eqref{eq:transitivity} 
    it follows that 
    \[L_1 + H_{U_1} + H_{U_3}=  L_3 + H_{U_1} + H_{U_3}.
    \qedhere
    \] 
\end{proof}
By~\cref{lm: transitive} the relation $\sim$ is an equivalence relation, so we may partition $\mathcal{A}$ into equivalence classes. We denote the equivalence class of a vertex $L\oplus H_U$ by $[L \oplus H_U]$. We often refer to each equivalence class as a clique because this relation partitions $\mathcal{A}$ into cliques:
\[
\mathcal{A} = \textsf{Clique}_1 \sqcup \cdots \sqcup \textsf{Clique}_{m}.
\]
The actual number of cliques, $m$, will not be important, but it is clear that such a number exists. The edges of our graph will
be between vertices $L\oplus H_U$ and $R$ if there exists 
$L'\oplus H_{U'}\in [L\oplus H_U]$ such that $L'\supseteq R$.
The edges will be weighted according to a sampling process 
that we describe in the next section, which also defines 
the constraints on $\Psi$. For future reference, the following lemma will be helpful in defining the constraints:

\begin{lemma} \label{lm: clique extension}
    Suppose $L \oplus H_U \sim L' \oplus H_{U'}$ and let $f: L \oplus H_U \xrightarrow[]{} \Ff_q$ be a linear function satisfying the side conditions on $U$. Then there is a unique linear function $f': L' \oplus H_{U'} \xrightarrow[]{} \Ff_q$  satisfying the side conditions on $U'$ such that there exists a linear function $g:  L + H_U + H_{U'} \xrightarrow[]{} \Ff_q$ satisfying the side conditions of both $U$ and $U'$ such that
    \[
    g|_{ L \oplus H_U} \equiv f \quad \text{ and } \quad g|_{ L' \oplus H_{U'}} \equiv f'.
    \]
    In words, the above equations say $g$ is a linear extension of both $f$ and $f'$.
\end{lemma}
\begin{proof}
    Note that there is only one way to extend $f$ to $L + H_U + H_{U'}$ in a manner that satisfies the side conditions given by $U'$. Let this function be $g$. We take $f'$ to be $g|_{L' \oplus H_{U'}}$.
\end{proof}

\subsubsection{The Constraints} \label{sec: constraint graph}
Suppose that $T_1$ is an assignment to $\mathcal{A}$ 
that assigns, to each vertex $L\oplus H_U$, a linear function $T_1[L\oplus H_U]$ satisfying the side conditions. Further 
suppose that $T_2$ is an assignment that assigns to each vertex $R\in \mathcal{B}$ a linear function on $R$. The verifier performs the following test, which also describes 
the constraints of $\Psi$:
\begin{enumerate}
    \item Choose $U$ uniformly at random from $\U$.
    \item Choose $L \oplus H_U$ uniformly, where $\dim(L) = 2\ell$ and $L \cap H_U = \{0\}$, and choose $R \subseteq L$ of dimension $2(1-\delta)\ell$ uniformly.
    \item Choose $L' \oplus H_{U'} \in [L \oplus H_U]$ uniformly 
    \item As in~\cref{lm: clique extension}, 
    extend $T_1[L' \oplus H_{U'}]$ to $L' + H_{U'} + H_U$ in the unique manner that respects the side conditions and let $\Tilde{T}_1[L \oplus H_U]$ be the restriction of this extension to $L \oplus H_U$.
    \item Accept if and only if $\Tilde{T}_1[L \oplus H_U]|_{R} = T_2[R]$. 
\end{enumerate}
This finishes the description of our instance $\Psi$. It is 
clear that the running time and instance size is $n^{O(k)}$
and that the alphabet size is $O(q^{2\ell})$. 
As is often the case, showing the completeness of the reduction is relatively easy. The soundness analysis is much more
complicated, and in the section below we develop 
some tools.

\section{Tools for Soundness Analysis} \label{sec:tools_for_pcp_analysis} 
In this section we will present all of the tools needed to analyze the soundness of our PCP. 
\subsection{The $2\ell$ versus $2\ell(1-\delta)$ subspace agreement test}
We begin by discussing the $2\ell$ versus $2\ell(1-\delta)$ 
test and our decoding theorem for it. In our setting, we 
have a question $U\in \U$ for the first prover, and we consider the $2\ell$ versus $2\ell(1-\delta)$ test inside 
the space $\mathbb{F}_q^U$. We assume that this test passes with probability at least $\eps \geq q^{-2\ell(1-\delta')}$
(where $\delta'$ is, say, $\delta ' = 1000\delta$) and we want to use this fact to devise a strategy for the first prover. 
Below, we first state and prove a basic decoding theorem, and 
then use it to deduce a quantitatively stronger version that 
also incorporates the side conditions.

Let $T_1$ be a table that assigns, to each $L \in \Grass(\Ff_q^U, 2\ell)$, a linear function $T_1[L]: L \xrightarrow[]{} \Ff_q$, and let $T_2$ be a table that assigns, to each $R \in \Grass(\Ff_q^{U},2(1-\delta)\ell)$, a linear function $T_2[R]\colon R\to\mathbb{F}_q$. We recall that $|U| = 3k \gg 2\ell$. 
In this section, we show that if tables $T_1$ and $T_2$ are $\epsilon$-consistent, namely
\[
\Pr_{\substack{L \in \Grass(\Ff_q^U, 2\ell)\\ R \in \Grass(\Ff_q^U, 2(1-\delta)\ell)}}[T_1[L]|_R \equiv T_2[R] \; | \; R \subseteq L] \geq \epsilon.
\]
for $\epsilon \geq q^{-2\ell(1-1000\delta)}$, then the table $T_1$ must have non-trivial agreement with a linear function on some zoom-in and zoom-out combination of constant dimension. 
The proof uses~\cref{lm: pseudorandom edges} along with an 
idea from~\cite{BKS}.
\begin{thm} \label{th: consistency}
    Let $0<\delta<\frac{1}{1000}$ and let $\ell$ be sufficiently large relative to $1/\delta$. Suppose that tables $T_1$ and $T_2$ are $\epsilon$-consistent where $\epsilon \geq q^{-2(1-1000\delta)\ell}$. Then there exist subspaces $Q \subset W$ and a linear function $f: W \xrightarrow[]{} \Ff_q$ such that:
    \begin{enumerate}
    \item $\dim(Q) + \codim(W) = \frac{10}{\delta}$.
    \item $f|_L \equiv T_1[L]$  for $\epsilon'$-fraction of $2\ell$-dimensional $L \in {\Zoom}_{2\ell}[Q,W]$, where $\epsilon' = q^{-2(1-1000\delta^2)\ell}$.
\end{enumerate}
\end{thm}
\begin{proof}
Consider the bipartite Grassmann inclusion graph, $G$, whose sides are the domains
$\Grass(n, 2\ell)$ and $\Grass(n, (1-\delta)2\ell)$, and denote its edge set by $E$.
Choose a linear function $f: \Ff_q^n \xrightarrow[]{} \Ff_q$ uniformly at random and define the (random) sets of vertices
\[
S_{L,f} = \{L \in \Grass(n, 2\ell) \; | \; f|_L \equiv T_1[L] \} \quad \text{and} \quad S_{R,f} = \{R \in \Grass(n, 2(1-\delta)\ell)\; | \; f|_R \equiv T_2[R] \}.
\]
Denote by $E(S_{L,f}, S_{R,f})$ the set of edges with endpoints in $S_{L,f}$ and $S_{R,f}$. We lower bound the expected size of $E(S_{L,f}, S_{R,f})$ over the choice of $f$. Note that for each edge $(L, R)\in E$ 
such that $T_1[L]|_{R}\equiv T_2[R]$, we have 
that $(L,R)\in E(S_{L,f}, S_{R,f})$ with probability $q^{-2\ell}$.
Indeed, with probability $q^{-2\ell}$ we have that $T_1[L] \equiv f|_{L}$, and in that case we automatically get that $T_2[R] \equiv T_1[L]|_{R} \equiv (f|_{L})|_{R} \equiv f|_R$. As the number of edges
$(L,R)$ such that $T_1[L]|_{R} \equiv T_2[R]$ is at least $\eps |E|$, we conclude that
\[
\E_{f}\left[\left|E(S_{L,f}, S_{R,f})\right|\right] \geq \epsilon q^{-2\ell} |E|.
\]
Note that we also have that 
\[
\E_f[\mu(S_{R,f})]
=\E_{f}\left[\frac{|S_{R,f}|}{|R|}\right] 
= q^{-2\ell(1-\delta)}.
\]
Using Linearity of Expectation, we get that
\[
\E_{f}\left[\left|E(S_{L,f}, S_{R,f})\right|
-
\frac{1}{2}\epsilon q^{-2\delta\ell} \mu(S_{R,f})|E|
\right] \geq \frac{1}{2}\epsilon q^{-2\ell} |E|,
\]
thus there exists $f$ for which the random variable on the 
left hand side is at least $\frac{1}{2}\epsilon q^{-2\ell} |E|$, and we fix $f$ so that
\begin{equation}\label{eq:decode1}
\left|E(S_{L,f}, S_{R,f})\right|
\geq
\frac{1}{2}\epsilon q^{-2\delta\ell} \mu(S_{R,f})|E|
+\frac{1}{2}\epsilon q^{-2\ell} |E|.
\end{equation}

We claim that $S_{L,f}$ is $(r,\epsilon')$-pseudo-random for $r = \frac{10}{\delta}$ and $\epsilon' = q^{-2\ell(1-1000\delta^2)}$. 
Suppose for the sake of contradiction that this is not the case, and that $S_{L,f}$ is $(r,\epsilon')$-pseudo-random. Denote $\alpha = \mu(S_{L, f})$ and $\beta = \mu(S_{R,f})$. Since $S_{L, f}$ is not $(r, \eps')$-pseudo-random, we may apply \cref{lm: pseudorandom edges} and get:
\begin{equation}\label{eq:decode2}
    \frac{1}{|E|}|E(S_{L,f}, S_{R,f})| = \Pr_{L \supseteq R}[L \in S_{L,f}, \in S_{R, f}]\leq q^{O_{t,r}(1)}\beta^{\frac{t-1}{t}}\epsilon'^{\frac{t-2}{t}}+q^{-r\delta \ell} \sqrt{\alpha \beta} 
    \leq q^{O_{t,r}(1)}\beta^{\frac{t-1}{t}}\epsilon'^{\frac{t-2}{t}},
\end{equation}
for any $t \geq 4$ that is a power of $2$. In the last inequality, we used the fact that by~\eqref{eq:decode1}
\[
\beta|E|
=
|S_{R,f}| \frac{|E|}{|R|}
\geq
|E(S_{L,f}, S_{R,f})|
\geq 
\frac{1}{2}\epsilon q^{-2\ell} |E|,
\]
so $\beta\geq \frac{1}{2}\eps q^{-2\ell}\geq q^{-4\ell}$, and thus 
the second term in the middle of~\eqref{eq:decode2} is
negligible compared to the first term there. Combining~\eqref{eq:decode1} and~\eqref{eq:decode2} gives 
us that
\[
\frac{1}{2}\eps q^{-2\delta\ell} \beta \leq q^{O_{t,r}(1)}\beta^{\frac{t-1}{t}}\eps'^{\frac{t-2}{t}}.
\]
Simplifying and using the definition of $\eps'$ along with the fact that $\eps\geq q^{-2\ell(1-1000\delta)}$ and $\beta \geq q^{-4\ell}$ we get
\[
\frac{1}{2}q^{-2\delta\ell} \leq q^{O_{t,r}(1)} q^{\frac{4\ell}{t}}q^{\left(\frac{4}{t} + 2000 \left(\delta^2\frac{t-2}{t}-\delta \right) \right) \ell}.
\]
Investigating the last two exponents of $q$, we have that for $t\geq \frac{1}{\delta - \delta^2} \geq 2$,
\[
\left(\frac{4}{t} + \frac{4}{t} + 2000 \left(\delta^2\frac{t-2}{t}-\delta \right)\right)\ell \leq -1992(\delta - \delta^2)\ell.
\]
This implies that
\[
\frac{1}{2}q^{-2\delta\ell}\leq q^{O_{\delta}(1)}
q^{-1992(\delta - \delta^2)\ell},
\]
and contradiction. It follows that $S_{L,f}$ is not $(r,\eps')$-pseudo-random, and unraveling the definition of not being
pseudo-random gives the conclusion of the theorem.
\end{proof}

\subsubsection{Finding a Large Fraction of Successful Zoom-Ins}
\cref{th: consistency} asserts the existence of a good pair of zoom-in and zoom-out $(Q,W)$ inside which the table $T_1$ has good agreement with a global linear function. As discussed in the introduction, our argument requires a quantitatively stronger version asserting that there is a good fraction of zoom-ins that work for us. Below, we state a corollary of~\cref{th: consistency} which achieves this. We defer its proof to~\cref{app: inner pcp}.

\begin{thm} \label{th: consistency many zoom in}
     Suppose that tables $T_1$ and $T_2$ are $\epsilon$-consistent for $\epsilon \geq 2q^{-2\ell(1-1000\delta)}$. Then there exist positive integers $r_1$ and $r_2$ satisfying $r_1 + r_2 = r = \frac{10}{\delta}$, such that for at least $q^{-5\ell^2}$-fraction of the $r_1$-dimensional subspaces $Q$, there exists a subspace $W \supseteq Q$ of codimension $r_2$ and a linear function $g_{Q,W}$ such that
    \[
    \Pr_{L \in \Grass(\Ff_q^U, 2\ell)}\left[g_{Q,W}|_{L} = T_1[L] \; | \; Q \subseteq L \subseteq W\right] \geq \frac{q^{-2\ell(1-1000\delta^2)}}{2}.
    \]
\end{thm}
\begin{proof}
    The proof is deferred to Section~\ref{app: inner pcp}.
\end{proof}
\subsubsection{Incorporating Side Conditions for Zoom-Ins}
Next, we require a version of~\cref{th: consistency many zoom in} which also takes the side conditions into account.

\begin{thm}\label{th: consistent with side}
Let $U$ be a question to the first prover, let $T_1$ be the first prover's table
and suppose that 
\[
\Pr_{\substack{L \in \Grass(\Ff_q^U, 2\ell), L \cap H_U = \{0\}\\ R \in \Grass(\Ff_q^U, 2(1-\delta)\ell)}}[T_1[L\oplus H_U]|_R = T_2[R] \; | \; R \subseteq L] := \epsilon \geq 4q^{-2(1-1000\delta)\ell}.
\]
Then there are parameters $r_1$ and $r_2$ such that $r_1 + r_2 \leq \frac{10}{\delta}$, such that for at least $q^{-6\ell^2}$ fraction of the $r_1$-dimensional subspaces $Q \subseteq \Ff_q^{U}$, there exists $W \subseteq \Ff_q^U$ of codimension $r_2$ containing $Q \oplus H_U$, and a global linear function $g_{Q,W}: W \xrightarrow[]{} \Ff_q$ that respects the side conditions on $H_U$ such that
    \[
    \Pr_{L \in \Grass(\Ff_q^U, 2\ell), L \cap H_U = \{0\}}[g_{Q,W}|_{L\oplus H_U} = T_1[L\oplus H_U] \;|\; Q \subseteq L \subseteq W] \geq \frac{q^{-2(1-1000\delta^2)\ell}}{2}.
    \]
\end{thm}
\begin{proof}
For any $2k$-dimensional subspace $A \subseteq \Ff_q^U$ such that $A \cap H_U = \{0\}$ and $H_U \oplus A = \Ff_q^U$,
let $T_A$ be the table over $\Grass(A, 2\ell)$ defined as follows. For the $2\ell$-dimensional subspaces $L \subseteq A$ such that $L \cap H_U = \{0\}$, define $T_A[L] = T_1[L \oplus H_U]|_L$.  Throughout the proof, $R$ is always used to denote a random subspace of dimension $2(1-\delta)\ell$, so we will not specify this further in the distribution of probabilities.

Let 
\[
p'(A) = \Pr_{L \in \Grass(A, 2\ell), R \subseteq L}[T_{A}[L]|_R \equiv T_2[R]].
\]
Noting that
\begin{equation*}
        \E_{A}[p'(A)] =  \Pr_{L \in \Grass(\Ff_q^U, 2\ell), L \cap H_U = \{0\}, R \subseteq L}[T_1[L \oplus H_U]|_{R} \equiv T_2[R]] = \eps,
\end{equation*}
we get by an averaging argument that $p'(A) \geq \frac{\epsilon}{2}$ for at least $\frac{\epsilon}{2}$-fraction of $A$'s. Fix any such $A$. By~\cref{th: consistency many zoom in} there exist positive integers $r_1$ and $r_2$ such that for at least $q^{-5\ell^2}$-fraction of $r_1$-dimensional zoom-ins $Q \subseteq A$, there exists a zoom-out $W' \supset Q$ of codimension $r_2$ and a linear function $g_{Q,W'}$ such that, 
\begin{equation} \label{eq: agreement A}  
\Pr_{L \in \Grass(A, 2\ell) }[ g_{Q,W'}|_L \equiv T_A[L] \; | \; Q \subseteq L \subseteq W'] \geq \frac{q^{-2(1-1000\delta^2)\ell}}{2}.
\end{equation}

Let $W = W' \oplus H_U$ and let $g_{Q,W}: W \xrightarrow[]{} \Ff_q$ be the unique extension of $g_{Q,W'}$ to $W$ satisfying the side conditions. We have
\begin{align*}
   & \Pr_{L \in \Grass(\Ff_q^U, 2\ell), L \cap H_U = \{0\}}[g_{Q,W}|_{L\oplus H_U} \equiv T_1[L\oplus H_U]\;|\; Q \subseteq L  \subseteq W] \\
    & =\Pr_{L \in \Grass(A, 2\ell)}[g_{Q,W}|_{L\oplus H_U} \equiv T_1[L\oplus H_U]\;|\; Q \subseteq L  \subseteq W']\\
    &= \Pr_{L \in \Grass(A, 2\ell)}[g_{Q,W'}|_{L} \equiv T_A[L]\;|\; Q \subseteq L  \subseteq W']  \\
    &\geq \frac{q^{-2(1-1000\delta^2)\ell}}{2}
\end{align*}
In the first transition we used the fact that $L \oplus H_U$ is distributed the same in the first two expressions. In the second transition we used the fact that $g_{Q, W}$ and $T[L \oplus H_U]$ agree on $H_U$ because both satisfy the side conditions on $H_U$, along with the fact that $g_{Q,W}|_A \equiv g_{Q,W'}$.

To conclude, we see that when sampling $A$ as above and then $Q\subseteq A$
of dimension $r_1$, the zoom-in $Q$ has a zoom-out $W \supseteq Q \oplus H_U$ and 
a function $g_{Q,W}: W \to \Ff_q$ satisfying the conditions in the theorem 
with probability at least $\frac{\eps}{2} q^{-5\ell^2}$. The theorem requires us to bound the probability that $Q$ satisfies this property when $Q$ is chosen uniformly in $\Ff_q^U$ of dimension $r_1$, which is a slightly different distribution than what we have. Indeed, when sampling $A$ as above and then $Q\subseteq A$ of dimension $r_1$, the subspace $Q$ is uniformly random $r_1$-dimensional subspace in $\Ff_q^U$ conditioned on $Q \cap H_U = \{0\}$. However, $\left(1-q^{r_1-2k}\right)$-fraction of $r_1$-dimensional subspaces $Q \subseteq \Ff_q^U$ satisfy $Q \cap H_U = \{0\}$, so overall we get that at least $\left(\frac{\eps}{2} q^{-5\ell^2} -q^{r_1-2k}\right) \geq q^{-6\ell^2}$-fraction of all $r_1$-dimensional $Q \subseteq \Ff_q^U$ satisfy the condition of the theorem. 
\end{proof}
\subsection{The Covering Property} \label{sec: covering properties}
In this section, we present the so-called ``covering property'', which is a feature of our PCP construction that allows us to move between the first prover's distribution over $2\ell$-dimensional subspaces of $\Ff_q^U$ and the second prover's distribution over $2\ell$-dimensional subspaces of $\Ff_q^V$. Similar covering properties are shown in \cite{KS, KMS}; however, obtaining the optimal quadratic-programming hardness result in~\cref{thm:QP} requires a stronger analysis that goes beyond the covering properties of \cite{KS, KMS}. We are able to obtain a covering property that holds with the following parameters set in the outer PCP:
\begin{equation} \label{eq: pcp parameters}
    k = q^{2(1+c)\ell} \quad , \quad \beta = q^{-2(1+2c/3)\ell},
\end{equation}
where $c > 0$ is a constant arbitrarily small relative to $\delta$.

\begin{remark}
We remark that the content
of this section is only necessary if we take $k = q^{(2+c)\ell}$ as in~\eqref{eq: pcp parameters}. This setting is used to get a good tradeoff between the size, alphabet size and soundness of the PCP, which is crucial in the proof of~\cref{thm:QP}. If one is only interested in the tradeoff between the alphabet size and the soundness, i.e.~only wants to prove~\cref{thm:main}, then one can afford to take $k$ as a much bigger function of $\ell$. For example, taking $k = q^{100\ell^2}$, one could use the covering property from~\cite{KS, KMS} instead of the more complicated covering property below, and replace the auxiliary lemma in~\cref{sec: V in U condition on Q} by a simple averaging argument.
\end{remark}

\subsubsection{The Basic Covering Property}
We start by stating a basic form of the improved covering property, and defer its proof to~\cref{app: covering}. Fix a question $U = (e_1, \ldots, e_k)$ to the first prover and recall that $H_U = \spa(v_{e_1}, \ldots, v_{e_k})$ where $v_{e_i}$ is the vector of coefficients of $e_i$. For $V \subset U$ we view $\Ff_q^V \subset \Ff_q^U$ as the subspace consisting of vectors that are zero on all indices in $U \setminus V$. The covering property we show will relate the following two distributions:

\noindent $\D:$
\begin{itemize}
    \item Choose $x_1, \ldots, x_{2\ell} \in \Ff_q^U$ uniformly.
    \item Output the list $(x_1,\ldots , x_{2\ell})$.
\end{itemize}

\noindent $\D':$
\begin{itemize}
    \item Choose $V \subseteq U$ according to the Outer PCP.
    \item Choose $x'_1,  \ldots, x'_{2\ell} \in \Ff_q^V$ uniformly.
    \item Choose $w_{1}, \ldots, w_{2\ell} \in H_U$ uniformly, and set $x_i = x'_i + w_i$ for $1 \leq i \leq 2\ell$.
    \item Output the list $(x_1, \ldots, x_{2\ell})$.
\end{itemize}
The covering property used in 
prior works asserts that the distribution $\mathcal{D}$ 
is statistically close to a variant of the distribution 
$\mathcal{D}'$. This closeness is not good enough for us, however, 
as we will want to consider events of rather small probability 
under $\mathcal{D}$ and still assert that their probability 
is roughly the same under $\mathcal{D}'$. 

Our version of the covering property deviates from those of prior works in two  ways. First, in these earlier works
the distribution $\mathcal{D}'$ was generated  without the addition of the random vectors $w_{1},\ldots,w_{2\ell}$ from $H_U$. As explained in the 
introduction, this distribution is not good enough
for the purpose of~\cref{thm:QP}, and we must consider
the distribution $\mathcal{D}'$ above instead. Second, the notion
of statistical closeness is too strict for us, and is in fact not achievable. To get around this hurdle, we instead show a ``contiguity-type'' statement, asserting that almost all inputs $x$ are assigned the same
probability under these two distributions up to factor $1+o(1)$. 

More precisely, setting $\eta = q^{-100\ell^{100}}$ throughout this subsection, we show:
\begin{restatable}{lemma}{basiccovering}
\label{lm: basic covering}
There exists $E \subseteq \left(\Ff_q^U\right)^{2\ell}$ such that both $\D(E)$ and $\D'(E)$ are at most $\eta^{40}$, and for all  $(x_1, \ldots, x_{2\ell}) \notin E$ we have
    \[
    0.9 \leq \frac{\mc{D}(x_1,\ldots, x_{2\ell}) }{\mc{D}'(x_1,\ldots, x_{2\ell})} \leq 1.1.
    \]
\end{restatable}
\begin{proof}
        The proof is deferred to~\cref{app: covering basic}.
\end{proof}

\subsubsection{The Covering Property with Zoom-ins}
\cref{lm: basic covering} represents the basic form
of the covering property, but for our application we require 
a version which incorporates zoom-ins and advice.
Namely, we will actually be interested in the case where $\D$ and $\D'$ are conditioned on some $r_1$-dimensional zoom-in $Q$, for an arbitrary dimension $r_1 \leq \frac{10}{\delta}$. To make notation simpler, let us write $x = (x_1,\ldots, x_{2\ell})$ and use $\spa_{r_1}(x)$ to denote $\spa(x_1,\ldots, x_{r_1})$. Additionally, for any set $\mc{L} \subseteq \left(\Ff_q^{U}\right)^{2\ell}$, define
\[
\D_Q(\Lcal) := \Pr_{x \sim \D} [x \in \Lcal \; | \; \spa_{r_1}(x) = Q] = \frac{\D(\{x\in \Lcal\; | \; \spa_{r_1}(x) = Q \})}{\D(\{x \; | \; \spa_{r_1}(x) = Q \})}.
\] 
Also define $\D'_Q$ similarly as
\[
\D'_Q(\Lcal) := \Pr_{x \sim \D'} [x \in \Lcal \; | \; \spa_{r_1}(x) = Q] = \frac{\D'(\{x\in \Lcal\; | \; \spa_{r_1}(x) = Q \})}{\D'(\{x \; | \; \spa_{r_1}(x) = Q \})}.
\] 

From~\cref{lm: basic covering} we can conclude that for 
any $\Lcal \subseteq (\Ff_q^U)^{2\ell}$ that is not too small, the measure $\D'_Q(\Lcal)$ is within at least a constant factor of $\D_Q(\Lcal)$ for nearly all $Q$:
\begin{lemma} \label{lm: smooth Q}
    For any $\Lcal \subseteq (\Ff_q^{U})^{2\ell}$, we have
    \[
    \Pr_{Q}\left[\D'_Q(\Lcal) \geq 0.8 \cdot \D_Q(\Lcal) - \eta^{20} \right] \geq 1 - 3\eta^{20},
    \]
    where $Q$ is the span of $r_1$ uniformly random vectors in $\Ff_q^U$.
\end{lemma}
\begin{proof}
Throughout the proof all of the expectations and probabilities over $Q$ choose $Q$ as in the lemma statement. Let $E$ be the small set of tuples from~\cref{lm: basic covering}, so that 
\begin{equation}\label{eq: covering zoom in expectations}  
\E_{x \sim \D, Q = \spa_{r_1}(x)}[\D_Q(E)] \leq \eta^{40} \quad \text{and} \quad \E_{x \sim \D', Q = \spa_{r_1}(x)}[\D'_Q(E)] \leq \eta^{40}.
\end{equation}

Let ${\sf Bad}_1 = \{x \in \left(\Ff_q^U\right)^{2\ell} \; | \; \D_Q(E) \geq \eta^{20}, Q:= \spa_{r_1}(x)\}$ and ${\sf Bad}_2 = \{x \in \left(\Ff_q^U\right)^{2\ell} \; | \; \D'_Q(E) \geq \eta^{20}, Q:= \spa_{r_1}(x)\}$. By the union bound we have
\[
\Pr_{x \sim \D}[x \in {\sf Bad}_1 \cup {\sf Bad}_2] \leq \Pr_{x \sim \D}[x \in {\sf Bad}_1 ] + \Pr_{x \sim \D}[x \in {\sf Bad}_2].
\]
By the first expectation \eqref{eq: covering zoom in expectations} and Markov's inequality, we have $\Pr_{x \sim \D}[x \in {\sf Bad}_1] \leq \eta^{20}$. To bound the second term, we have
\[
\Pr_{x \sim \D}[x \in {\sf Bad}_2] = \sum_{x \in {\sf Bad}_2} \D(x) \leq \sum_{x \in {\sf Bad}_2 \setminus E} \D(x) +  \sum_{x \in  E} \D(x) \leq  \sum_{x \in {\sf Bad}_2 \setminus E} 1.1\cdot \D'(x) + \eta^{40} \leq 1.1\cdot \eta^{20} + \eta^{40}.
\]
In the penultimate transition we are using \cref{lm: basic covering} to bound $\D(x) \leq 1.1\cdot \D'(x)$ if $x \notin E$ as well as to bound $\D(E) \leq \eta^{40}$. In the final transition we are using the fact $\Pr_{x \sim \D'}[ x \in {\sf Bad}_2] \leq \eta^{20}$, which follows by an application of Markov's inequality on the second expectation in \eqref{eq: covering zoom in expectations}.

Altogether, we get that with probability at least $1 - 3\eta^{20}$, we have $x \notin {\sf Bad}_1 \cup {\sf Bad}_2$, and $Q:= \spa_{r_1}(x)$ satisfies $\D_Q(E), \D'_Q(E) \leq \eta^{20}$. In this case we have,
\begin{align*}
   \sum_{\spa_{r_1}(x) = Q} \mc{D}(x) 
   \geq \sum_{\spa_{r_1}(x) = Q, x \notin E} \mc{D}(x) 
   &\geq \sum_{ \spa_{r_1}(x) = Q, x \notin E} 0.9 \cdot \D'(x) \\
   &= 0.9 \cdot \left(\sum_{ \spa_{r_1}(x) = Q} \D'(x) -  \sum_{ \spa_{r_1}(x) = Q, x \in E} \D'(x) \right),
\end{align*}
where we used~\cref{lm: basic covering} in the
second transition. Dividing both sides by $\sum_{ \spa_{r_1}(x) = Q} \D'(x)$ gives that
\begin{equation} \label{eq: ineq in covering step}
\frac{\sum_{\spa_{r_1}(x) = Q} \mc{D}(x)}{\sum_{\spa_{r_1}(x) = Q} \D'(x) } \geq 0.9(1 - \D'_Q(E)) \geq 0.9 ( 1 - \eta^{20}) \geq 0.89.
\end{equation}
It follows that
\begin{align*}
    \D'_Q(\Lcal) &= \frac{\sum_{x \in \Lcal, \spa_{r_1}(x) = Q} \D'(x) }{\sum_{\spa_{r_1}(x) = Q} \D'(x)} \\
    &\geq \frac{0.89\sum_{x \in \Lcal \cap \overline{E}, \spa_{r_1}(x) = Q} \D'(x)}{ \sum_{ \spa_{r_1}(x) = Q} \D(x)} \\
    &\geq \frac{0.9\cdot 0.89 \cdot \sum_{x \in \Lcal \cap \overline{E}, \spa_{r_1}(x) = Q} \D(x)}{\sum_{ \spa_{r_1}(x) = Q} \D(x)} \\
    &\geq 0.8 \cdot \D_Q(\mc{L}) - \frac{\sum_{\spa_{r_1}(x) = Q, x \in E}\D(x)}{\sum_{ \spa_{r_1}(x) = Q} \D(x)}\\
    &= 0.8 \cdot \D_Q(\mc{L})  - \D_Q(E) \\
    &\geq 0.8\cdot \D_Q(\mc{L}) - \eta^{20},
\end{align*}
where we used~\eqref{eq: ineq in covering step} in the second transition and $\D_Q(E) \leq \eta^{20}$ in the last transition.
\end{proof}

\subsubsection{The Covering Property for the Advice} \label{sec: covering advice}
We will also need a lemma that applies to $r_1$-dimensional subspaces for some constant $r_1 = O(\delta^{-1})$. This is to handle the fact that the zoom-in $Q$ is sampled uniformly from $\Ff_q^V$ after $V$ is chosen according to the outer PCP, and then lifted to a subspace over $\Ff_q^U$, instead of uniformly from $\Ff_q^U$. For a fixed question $U$ to the first prover, let $\D_{r_1}$ denote the uniform measure over $\left(\Ff_q^{U}\right)^{r_1}$ and let $\D'_{r_1}$ denote the following distribution over $\left(\Ff_q^U\right)^{r_1}$:
\begin{itemize}
    \item Choose $V \subseteq U$ according to the Outer PCP.
    \item Choose $x'_1,  \ldots, x'_{r_1} \in \Ff_q^V$ uniformly.
    \item Choose $w_1, \ldots, w_{r_1} \in H_U$ uniformly, and set $x_i = x'_i + w_i$ for $1 \leq i \leq r_1$.
    \item Output the list $(x_1,\ldots, x_{r_1})$.
\end{itemize}
The distributions $\D_{r_1}$ and $\D'_{r_1}$ have the following property.
\begin{lemma} \label{lm: covering zoom-in} 
For $\mathcal{Q} \subseteq \left(\Ff_q^U\right)^{r_1}$ such that $\D_{r_1}(\mc{Q})\geq q^{-10\ell^{10}}$, we have 
\[
\D'_{r_1}(\mc{Q}) \geq 0.8 \cdot \D_{r_1}(\mc{Q}).
\]
\end{lemma}
\begin{proof}
Take $E$ from \cref{lm: basic covering} and define
\[
\mc{L} = \{(x_1,\ldots, x_{2\ell}) \in \left(\Ff_q^U\right)^{2\ell} \; | \; (x_1,\ldots, x_{r_1}) \in \mc{Q}\}.
\]
Then $\D_{r_1}(\mc{Q}) = \D(\Lcal)$, $\D'_{r_1}( \mc{Q}) = \D'(\Lcal)$, and 
\[
\D'_{r_1}(\mc{Q}) = \D'(\Lcal) \geq \sum_{x \in \Lcal \setminus E} \D'(x) \geq 0.9 \cdot \sum_{x \in \Lcal \setminus E} \D(x) \geq 0.9\cdot \D(\Lcal) - \D(E) \geq 0.8\cdot \D(\Lcal) = 0.8\cdot \D_{r_1}(\Lcal),
\]
where we used \cref{lm: basic covering} in the second inequality, and the fact that $\D(E) \leq \eta^{40}$ and $\D(\Lcal) = \D_{r_1}(\mc{Q}) \geq q^{-10\ell^{10}}$ in the penultimate transition.
\end{proof}

\subsubsection{An Auxiliary Lemma} \label{sec: V in U condition on Q}
We conclude this section with an auxiliary lemma that will be used in the analysis. 
Throughout this subsection, we fix a question $U = (U_1, \ldots, U_k)$ to the first prover, where $U_i$ denotes the variables from the $i$th equation, say $e_i$, and a subspace $Q \subseteq \Ff_q^U$ of dimension $r$, and we consider the distribution over the second prover's question, $V$, conditioned on the first $r$ advice vectors having span $Q$. More formally, consider the distribution $\mc{A}(\cdot, \cdot)$ over $(V, Q')$ defined as follows:
\begin{itemize}
    \item For each $i \in [k]$, choose $V_i \subseteq U_i$ according to the outer PCP, and set $V = \bigcup_{i=1}^k V_i$.
    \item Choose $r$ vectors $v_1, \ldots, v_{r} \in \Ff_q^V$ uniformly at random.
    \item Output $(V, \spa(v_1, \ldots, v_{r}))$.
\end{itemize}
Then the distribution we are interested in is the marginal distribution $\mc{A}(\cdot, Q)$. After $V$ is chosen, we let $V_1, \ldots, V_k$ be the variables in $V$ from equations $e_1,\ldots, e_k$ respectively. We start with a helpful intermediate result, showing that in the sampling procedure, for any $Q$ and $i \in [k]$, the probability $V_i \neq U_i$ is roughly at most $\beta$ even when conditioned on $V_j$ already being chosen for $j \in \mc{I} \subseteq [k]$ (where $\mc{I}$ can be an arbitrary subset of indices). For a vector $v \in \Ff_q^U$, we use $v^{i}$ to denote its restriction to the variables in $U_i$, and think of $v^{i}$ as being in the $3$-dimensional space $\Ff_q^{U_i}$.

\begin{lemma} \label{lm: variable drop conditioned}
Fix a question $U = (U_1,\ldots, U_k)$, a dimension $r$ subspace $Q \subseteq \Ff_q^U$, a parameter $0 < \beta < q^{-2\ell}$, and let $\mc{A}$ be defined accordingly as above using $U$ and $\beta$. Let $\mc{I} \subseteq [k]$ be an arbitrary set of indices and suppose for $j \in \mc{I}$, $V_j \subseteq U_j$ has already been chosen. Then, for any $i \in [k] \setminus \mc{I}$
    \[
\Pr_{V \sim \mc{A}(\cdot, Q)}[V_i \neq U_i \; | \; \{V_j\}_{\forall j \in \mc{I}}] :=    
\Pr_{V_i \subseteq U_i \; Q' \subseteq \Ff_q^V}[V_i \neq U_i \; | \; Q' = Q, \{V_j\}_{\forall j \in \mc{I}}] \leq \frac{\beta}{(1-\beta)q^{-3r}}.
    \]
    In words, the above probability is the probability that variables are dropped in equation $i$, when sampling $V \subseteq U$, conditioned on the span of the first $r$ advice vectors equaling $Q$, and conditioned on $V_j$ being fixed for $j \in \mc{I}$.\footnote{We implicitly assume that the probability space we consider is nonempty, so that there exist $V_s \subseteq U_s$ for $s \notin \mc{I}$ such that $Q' \subseteq \Ff_q^{V}$ after setting $V = \bigcup_{i=1}^k V_i$, with $V_j$ fixed as above for $j \in \mc{I}$.}
\end{lemma}
\begin{proof}
View each vector in $\Ff_q^{U}$ as indexed by the variables appearing in $U$. It will be helpful to briefly recall how $V \subseteq U$ and the advice subspace $Q$ are chosen according to the conditional distribution we consider. First, for $j \in \mc{I}$, $V_j \subseteq U_j$ is already fixed. For each $j \in [k] \setminus \mc{I}$ independently we set $U_j = V_j$ with probability $1-\beta$, or with probability $\beta$, we set $V_j$ to be a uniformly random subset of $U_j$ of size $1$. After choosing $V_j$, we choose $r$ advice vectors $v^j_1,\ldots, v^j_{r} \in \Ff_q^{V_i}$ uniformly at random. We then set $V = \bigcup_{j=1}^k V_j$, for each $t \in [r]$ set the advice vector $v_t \in \Ff_q^V$ to be the concatenation of $v^{1}_t, \ldots, v^{k}_t$, and set $Q' = \spa(v_1, \ldots, v_{r})$. 

It is clear that there is some distribution over bases $(a_1,\ldots, a_r)$ of $Q$ such that 
\[
 \Pr_{V_j \subseteq U_j, \forall j \notin \mc{I},\;  Q' \subseteq \Ff_q^V}\left[V_i \neq U_i \; | \;  Q' = Q \right] = \E_{a_1,\ldots, a_r}\left[\Pr_{V_j \subseteq U_j, \forall j \notin \mc{I},\;  Q' \subseteq \Ff_q^V}\left[V_i \neq U_i \; | \;   v_t = a_t~ \forall t \in [r] \right] \right],
\]
where the expectation on the right is over this distribution of bases $(a_1,\ldots, a_r)$ of $Q$. We do not need to know what this distribution is. We only use the fact that, by its existence, there is some basis $a_1,\ldots, a_r$ which maximizes the inner probability on the right, and fixing this basis $(a_1,\ldots, a_r)$ of $Q$, we have
\begin{align*}
    \Pr_{V_j \subseteq U_j, \forall j \notin \mc{I},\;  Q' \subseteq \Ff_q^V}\left[V_i \neq U_i \; | \;  Q' = Q \right] &\leq \Pr_{V_j \subseteq U_j, \forall j \notin \mc{I},\;  Q' \subseteq \Ff_q^V}\left[V_i \neq U_i \; | \;   v_t = a_t~ \forall t \in [r] \right] \\
    &= \frac{\Pr\limits_{V_j \subseteq U_j, \forall j \notin \mc{I}, \; v_1,\ldots, v_{r} \in \Ff_q^V}\left[V_i \neq U_i \land (v_t = a_t~ \forall t \in [r])\right]}{\Pr\limits_{V_j \subseteq U_j, \forall j \notin \mc{I}, \; v_1,\ldots, v_{r} \in \Ff_q^V}[v_t = a_t ~ \forall t \in [r]]}.
\end{align*}
Since $V_j \subseteq U_j$ is chosen independently, the above probability becomes
\begin{align*}
&\frac{\Pr\limits_{V_i \subseteq U_i, v^{i}_t \in \Ff_q^{V_i}}\left[V_i \neq U_i \land (v^{i}_t = a^{i}_t ~ \forall t \in [r])\right]\prod \limits_{j \in [k] \setminus \{i\}}\Pr\limits_{V_j \subseteq U_j, v^j_t \in \Ff_q^{V_j}}\left[ v^j_t = a^{j}_t ~ \forall t \in [r] \right]}{\prod \limits_{j = 1}^{k}\Pr\limits_{V_j \subseteq U_j, v^j_t \in \Ff_q^{V_j}}\left[ v^j_t = a^{j}_t ~ \forall t \in [r] \right]} \\
&= \frac{\Pr\limits_{V_i \subseteq U_i, v^{i}_t \in \Ff_q^{V_i}}\left[V_i \neq U_i \land (v^{i}_t = a^{i}_t ~ \forall t \in [r])\right]}{\Pr\limits_{V_i \subseteq U_i, v^{i}_t \in \Ff_q^{V_i}}\left[ v^{i}_t = a^{i}_t \; \forall t \in [r]\right]} \\
&\leq \frac{\beta}{\Pr\limits_{V_i \subseteq U_i, v^{i}_t \in \Ff_q^{V_i}}\left[ v^{i}_t = a^{i}_t \; \forall t \in [r]\right]}\\
&\leq \frac{\beta}{(1-\beta)q^{-3r}}.
\end{align*}
In the third transition, we upper bound the numerator by $\Pr[V_i \neq U_i] \leq \beta$, and in the last transition we lower bound the denominator by $(1-\beta)(1/|\Ff_q^{U_i}|)^{r} = q^{-3r}$, which is the probability that $V_i = U_i$ and $v^{i}_t = a^{i}_t$ for all $t \in [r]$.
\end{proof}

In the soundness analysis we will find some zoom-in and zoom-out pair $(Q,W)$ for the larger prover, and we will be interested in projecting it into a question $V$ of the smaller prover. In particular, we will be interested in the distribution of the codimension $W\cap \mathbb{F}_q^V$. The following consequence of~\cref{lm: variable drop conditioned} upper bounds the probability its codimension is smaller by $j$ than the codimension of $W$.
\begin{lemma} \label{lm: V pj calc}
    For integer $r,s\geq 0$, the following holds for sufficiently large $k$. Let $U$ be a fixed question to the first prover in the outer PCP consisting of $3k$ distinct variables in some set of $k$ equations, let $Q \subseteq \Ff_q^U$ be a subspace of dimension $r$, and let $W \subseteq \Ff_q^U$ be a codimension $s$ subspace that contains both $H_U$ and $Q$.  Let $V \subseteq U$ be a random question to the second prover chosen according to the outer PCP, conditioned on $Q \subseteq \Ff_q^V$, and let $W[V] = \Ff_q^V \cap W$. Then, for each $0 \leq j \leq s$, we have
    \[
    \Pr_{V \sim \mc{A}(\cdot, Q)}[\codim(W[V]) = s-j] \leq (2s\beta q^{3r})^j,
    \]
    where the codimension is with respect to $\Ff_q^V$.
\end{lemma}
\begin{proof}
Let us label the variables in $U$ as $(x_1, \ldots, x_{3k})$ so that the entries of vectors in $\Ff_q^U$ are indexed by variables in $(x_1, \ldots, x_{3k})$. Recall that the equations making up $U$ are labeled as $e_1, \ldots, e_k$. We will choose $V$ gradually, where at each step, we choose a new equation from $\{e_1, \ldots, e_k\}$ and then choose which variables to drop from it according to the sampling procedure of the outer PCP. Specifically, at step $i$, let $e'_1, \ldots, e'_{i-1} \in \{e_1,\ldots, e_k\}$ be the equations that have been processed already. We choose $e'_i \in \{e_1,\ldots,e_k\} \setminus \{e'_1,\ldots, e'_{i-1}\}$. Then, with probability $1-\beta$ we do not drop any variables from $e'_i$, and with probability $\beta$ we drop a uniformly random two out of the three variables from $e'_i$. Recall that when a variable is dropped when obtaining $V$, every vector in $\Ff_q^V$ has coordinate zero at the index of the dropped variable.

Let $V_0 = U$ and let $V_i$ be the state of $V$ after step $i$ in this process, i.e.\ after we have processed equations $e'_1, \ldots, e'_i \in \{e_1, \ldots, e_k\}$ and chosen which variables to drop from these equations. Also define $W[V_i] := W \cap \Ff_q^{V_i}$ and say that $\codim(W[V_i])$ is the codimension of $W[V_i]$ with respect to $\Ff_q^{V_i}$. We will first show that it is possible to adaptively choose the first $k-s$ equations so that the codimension does not drop, i.e.\ $\codim(W[V_0]) = \cdots = \codim(W[V_{k-s}])$.  For each one of the remaining $s$ steps we will (roughly) show that the codimension drops by $1$ with probability at most $\beta$, and does not change with probability $1-\beta$.

At step $i$, call an equation $e$ saturated if the restriction of $W[V_{i-1}]$ to the variables in $e$ has dimension $3$. If we choose $e'_i = e$, it is straightforward to check that $\codim(W[V_{i}]) = \codim(W[V_{i-1}])$ with probability $1$. We argue inductively that for the first $k-s$ steps, we can always adaptively choose a saturated equation, and thus have $\codim(W[V_0]) = \cdots = \codim(W[V_{k-s}])$.

Fix an arbitrary $i$ such that $1 \leq i \leq k-s$ and suppose $\codim(W[V_{i-1}]) = s$. We claim that there can be at most $s$ unsaturated equations in $\{e_1,\ldots, e_k\} \setminus \{e'_1,\ldots, e'_{i-1}\}$. To see this, note that for each equation $e\in \{e_1,\ldots, e_k\} \setminus \{e'_1,\ldots, e'_{i-1}\}$ that is unsaturated there is a corresponding vector in the dual space of $W[V_{i-1}] \subset \Ff_q^{V_{i-1}}$ supported only on the indices corresponding to variables in $e$. Thus, for any set of distinct unsaturated equations, these vectors in the dual space are linearly independent. Since the dual space of $W[V_{i-1}]$ has dimension at most $s$, there are thus at most $s$ unsaturated equations. Therefore, if $1\leq i \leq k-s$, then there is always at least one saturated equation in $\{e_1,\ldots, e_k \} \setminus \{e'_1, \ldots, e'_{i-1}\}$, and we can adaptively choose $e'_i$ to be this equation.

We now show that for the unsaturated equations, the codimension can drop by $1$ with probability at most $\frac{\beta}{(1-\beta)q^{-3r}}$ and otherwise the codimension does not change. Let $e_i$ be an equation that is unsaturated. By \cref{lm: variable drop conditioned}, with probability at least $1-\frac{\beta}{(1-\beta)q^{-3r}}$ no variables in $e_i$ are dropped and in this case we have $V_i = V_{i+1}$, $W[V_i] = W[V_{i+1}]$, and hence $\codim(W[V_i]) = \codim(W[V_{i+1}])$.

With the remaining at most $\frac{\beta}{(1-\beta)q^{-3r}}$ probability, two variables are dropped, and in this case we show that the codimension decreases by at most $1$.
Suppose $a_1, a_2$ are the characteristic vectors for the variables that are dropped. Recall that $H_U \subseteq W$ at the start, and let $w_i$ be the coefficient vector of the $i$th equation so that $w_i \in H_U$. Then, $w_i \in W[V_i]$, so it must be the case that $a_1$ is not yet in the dual space of $W[V_i]$, and so, when viewing both $W[V_{i+1}]$ and $W[V_i]$ as subspaces of $\Ff_q^{V_i}$, we get that  $\dim(W[V_i]) - \dim(W[V_{i+1}]) \geq 1$. Combining this with the fact that $\dim\left(\Ff_q^{V_{i}}\right) - \dim\left(\Ff_q^{V_{i+1}}\right) = 2$, where again both spaces are viewed as subspaces of $\Ff_q^{V_i}$ it follows that the codimension of $W[V_{i+1}]$ relative to $\Ff_q^{V_{i+1}}$ is at most one more than that of  $W[V_{i}]$ relative to $\Ff_q^{V_{i}}$.

Altogether, this shows that for $k-s$ of the steps $i$, the codimension does not change, while for the remaining $s$ steps, the codimension drops by $1$ with probability at most $\frac{\beta}{(1-\beta)q^{-3r}}$ and remains the same otherwise. As a result,
\[
    \Pr_{V \sim \mc{A}(\cdot, Q)}[\codim(W[V]) = s-j]  \leq \binom{s}{j} \left(\frac{\beta}{(1-\beta)q^{-3r}}\right)^j \leq (2s\beta q^{3r})^j,
\]
where the second inequality is by union bounding over all possible combinations of $j$ out of $s$ steps where the codimension may drop by $1$.
\end{proof}

We finish this section by stating a corollary of~\cref{lm: V pj calc} which will be used in the soundness analysis. The setting one should have in mind is that we 
have a zoom-in and zoom-out pair $(Q,W)$ and an event
$\mathcal{L}\subseteq (\mathbb{F}_q^U)^{2\ell}$ inside it with noticeable probability, where $U$ is a question to the first prover. The lemma asserts that 
over the choice of the question $V$ to the smaller prover, the projection of the event $\mathcal{L}$ to 
the zoom-in and zoom-out pair $(Q,W[V])$ still has 
noticeable probability:

\begin{lemma} \label{lm: V preserve} 
The following holds for any integers $r,s > 0$, $\ell$ sufficiently large relative to $r$ and $s$, $k > \ell$, $0 < \beta < q^{-2\ell}$, and $C > 2sq^{-2\ell}$. Let $U$ be a fixed question to the first prover in the outer PCP consisting of $3k$ variables in some set of $k$ equations, let $Q \subseteq \Ff_q^U$ be a zoom-in of dimension $r$, let $W \subseteq \Ff_q^U$ be a zoom-out of codimension $s$ containing $H_U$ and $Q$. 

Also let $\mc{L} \subseteq \left(\Ff_q^{U}\right)^{2\ell}$ be such that  $\spa(x) \subseteq W$ for every $x \in \mc{L}$ and $\mc{D}'_Q(\mc{L}) \geq q^{-s(2\ell - r)} \cdot C$. With probability at least $\beta^{s+2}$ over $V$ (chosen according to the marginal of $\mc{A}$ defined above, with parameter $\beta$), we have 
   \[
   \Pr_{x'_i \in W[V], w_i \in H_U}\left[x' + w\in \mc{L} \; | \; \spa_r(x'+w) = Q\right] \geq\frac{C}{4(s+1)2^sq^{3rs}}.
   \]
   In the probability above, the sampling procedure is that of $\D'_Q$ with $V$ fixed and we use $x'$ and $w$ to denote $x' = (x'_1,\ldots, x'_{2\ell}), w = (w_1,\ldots, w_{2\ell})$ respectively.
\end{lemma}
\begin{proof}
For each $0 \leq j \leq s$, let $E_j$ be the set of $V \subseteq U$ such that $\codim(W[V]) = s-j$, and let $p_j := \Pr_{V \sim \mc{A}(\cdot, Q)}[V \in E_j]$, where $V$ is chosen according to the outer PCP. 
    Then $\D'_Q(\mc{L}) \geq q^{-s(2\ell - r)} \cdot C$ implies that
    \begin{equation}   \label{eq: aux sum}
        \sum_{j = 0}^{s} p_j \cdot \E_{V \sim E_j} \left[  \Pr_{x'_i \in W[V], w_i \in H_U}\left[x'+w \in \mc{L} \; | \; \spa_r(x'+w) = Q\right] \right] \geq q^{-s(2\ell - r)} \cdot C.
    \end{equation}
    In the above expectation $V \sim E_j$ means choosing $V$ according to the distribution $\mc{A}(\cdot, Q)$ conditioned on $V \in E_j$. For each $j$ and $V\in E_j$, we can bound the inner term using the fact that if $x \in \mc{L}$, then $\spa(x'+w) \subseteq W$, so as a result of $W \supseteq H_U$, we must have $\spa(x') \subseteq W[V]$. Thus,
    \begin{align*} 
   &\Pr_{x'_i \in W[V], w_i \in H_U}\left[x'+w \in \mc{L} \; | \; \spa_r(x'+w) = Q\right] \\
   &\leq \Pr_{x'_i \in \Ff_q^V, w_i \in H_U}[x'_{r+1},\ldots, x'_{2\ell} \in W[V] ]  \\
    &\leq q^{-(s-j)(2\ell - r)},
    \end{align*}
where in the second inequality we use the fact that $\codim(W[V]) = s-j$. Call $j$ rare if $p_j \leq \beta^{j+1}$. The contribution to the sum above from rare $j$ is at most
\[
\sum_{j \text{ rare}} \beta^{j+1} \cdot q^{-(s-j)(2\ell - r)} \leq s \cdot \beta \cdot q^{-s(2\ell-r)} \leq q^{-s(2\ell-r)} \frac{C}{2},
\]
where we used $\beta\leq q^{-2\ell}\leq \frac{C}{2s}$.
Removing this contribution from~\eqref{eq: aux sum} it follows that there exists $j$ that is not rare satisfying
\begin{equation}\label{eq:aux_lem_1}
p_j \cdot \E_{V \in E_j} \left[ \Pr_{x'_i \in W[V], w_i \in H_U}\left[x'+w \in \mc{L} \; | \; \spa_r(x'+w) = Q\right] \right] \geq q^{-s(2\ell-r)}  \cdot \frac{C}{2(s+1)}.
\end{equation}
Fix this $j$ henceforth. Using again the fact that if $x'+w \in \mc{L}$, then $\spa(x'+w) \subseteq W$ and $\spa(x') \subseteq W[V]$, we get

\begin{align*}   
&p_j \cdot q^{-(s-j)(2\ell-r)} \cdot \E_{V \in E_j} \left[ \Pr_{x'_i \in W[V], w_i \in H_U}\left[x'+w \in \mc{L} \; | \; \spa_r(x'+w) = Q\right] \right]\\
&=p_j \cdot q^{-(s-j)(2\ell-r)} \cdot \E_{V \in E_j} \left[ \Pr_{\substack{{x'_i \in \Ff_q^V, w_i \in H_U,}\\ {x_i = x'_i + w_i, i \in [2\ell]}}}[x'+w \in \mc{L} \; | \; \spa_{r}(x) = Q, \; \spa(x') \subseteq W[V] ] \right] \\
&= p_j \cdot \E_{V \in E_j} \left[ \Pr_{x'_i \in \Ff_q^V, w_i \in H_U}[x'+w \in \mc{L} \; | \; \spa_{r}(x'+w) = Q ] \right] \\
&\geq q^{-s(2\ell-r)}  \cdot \frac{C}{2(s+1)},
\end{align*}
where the last transition uses~\eqref{eq:aux_lem_1}. Dividing the last inequality by $p_j$ we conclude that
\begin{align*}
 \E_{V \in E_j} \left[ \Pr_{x'_i \in W[V], w_i \in H_U}[x'+w \in \mc{L} \; | \; \spa_{r}(x'+w) = Q] \right] &\geq \frac{q^{-j(2\ell-r)}}{p_j} \cdot \frac{C}{2(s+1)} \\
 &\geq \frac{q^{-j(2\ell-r)}}{(2s\beta q^{3r})^{j}} \cdot \frac{C}{2(s+1)}\\
 &\geq \frac{C}{2(s+1)2^sq^{3rs}}.
\end{align*}
where in the second inequality we apply the upper bound on $p_j$ from \cref{lm: V pj calc}, and in the last transition we use $\beta \leq q^{-2\ell}$. Now let 
\[
C' :=  \frac{C}{2(s+1)2^sq^{3rs}}.
\]
By an averaging argument, conditioned on $V\in E_j$ with probability at least $\frac{C'}{2}$ the probability in the expectation above is at least $\frac{C'}{2}$. 
Overall, as $j$ is not rare, we get that the probability that the expectation above is at least $\frac{C'}{2}$ over the choice of $V$ is at least $p_j\frac{C'}{2} \geq \beta^{s+2}$.
\end{proof}
\subsection{The Number of Maximal Zoom-Outs is Bounded} \label{sec: bounded zoom-outs statement}
\cref{th: consistency many zoom in} suggests that the two provers can agree on a zoom-in with reasonable probability using their advice. The same cannot be said for zoom-outs however, and to circumvent this issue we must
develop further tools. In this section, we define the notion of \emph{maximal} zoom-outs and show that for a fixed zoom-in $Q$, the number of maximal zoom-outs is bounded. 
\subsubsection{Generic Sets of Subspaces}
 To deal with large collections of zoom-outs
we define a special property of zoom-outs that is called ``genericness''. To motivate it, let $W_{\amb}$ be an ambient space and suppose  $W_1,W_2\subseteq W_{\amb}$
are distinct subspaces of codimension $r$. Then $W_1\cap W_2$ is a subspace
whose codimension is between $2r$ and $r+1$. For a typical pair
of subspaces, the intersection $W_1\cap W_2$ has codimension $2r$ though, 
and in this case we say the two subspaces are generic. Genericness is useful probabilistically because if $W_1,W_2$ are generic, then the event that a randomly chosen $2\ell$-dimensional subspace is contained in $W_1$, and the event the subspace is contained in $W_2$, are almost independent. Below is a formal definition which also applies to a set of subspaces rather than just two:
\begin{definition}
    We say that a set $\mathcal{S} = \{W_1, \ldots , W_N\}$ of codimension $r$ subspaces of $W_{\amb}$ is $t$-generic with respect to $W_{\amb}$ if for any $t$ distinct subspaces, say $W_{i_1}, \ldots, W_{i_t}\in\mathcal{S}$, we have $\codim(\bigcap_{1\leq j \leq t}W_{i_j}) = t\cdot r$. When the ambient space $W_{\amb}$ is clear from context we simply say that $\mathcal{S}$ is $t$-generic.
\end{definition}

Throughout this section, it will be convenient to denote by $\codim_V(W)$ the codimension of $W \cap W_{\amb}$ with respect to $W_{\amb}$. When no subscript is used, the identity of the subspace $W_{\amb}$ should be clear from context.

We remark that any set of subspaces that is $t$-generic with respect to $W_{\amb}$ is also $t'$-generic with respect to $W_{\amb}$ for any $t' < t$. In this section we show a sunflower-type lemma, stating that any large set of codimension $r$ subspaces inside $W_{\amb}$ contains a large set of subspaces that are $t$-generic with respect to $W'_{\amb}$ for some $W'_{\amb} \subseteq W_{\amb}$. 
Below is a formal statement.
\begin{lemma} \label{lm: generic}
    Let $t,r\in\mathbb{N}$ be integers and let $\mathcal{S}$ be a set of $N$ subspaces each with codimension $r$ inside of $W_{\amb}$. Then there exists a subspace $W'_{\amb} \subseteq W_{\amb}$ and a set of subspaces $\mathcal{S'}  \subseteq \mathcal{S}$ such that:
    \begin{itemize}
        \item $\left|\mathcal{S'}\right| \geq \frac{N^{\frac{1}{(r+1)\cdot (t-1)!}}}{q^{3r}}$.
        \item Each $W_i \in \mathcal{S'}$ is contained in $W'_{\amb}$ and has codimension $s$ with respect to $W'_{\amb}$, where $s \leq r$.
        \item $\mathcal{S'}$ is $t$-generic with respect to $W'_{\amb}$.
    \end{itemize}
\end{lemma}

To show~\cref{lm: generic} we need two auxiliary lemmas. The first,~\cref{lm: inductive generic}, states that for $j \geq 2$, any $j$-generic set of subspaces contains a large $(j+1)$-generic set of subspaces.
The second,~\cref{lm: reduce codimension}, states that a given collection of subspaces is either already $2$-generic, or else many subspaces in the collection are contained in a common hyperplane. 

\begin{lemma} \label{lm: inductive generic}
Let $\mathcal{S} = \{W_1, \ldots , W_{N}\}$ be a set of $N$-subspaces of codimension $r$ inside of $W_{\amb}$ that is $j$-generic with respect to $W_{\amb}$ for $j \geq 2$, then there is a subset $\{W_1,\ldots, W_{N'}\} \subseteq \mathcal{S}$ of size 
$N'\geq \frac{N^{1/j}}{q^r}$
that is $(j+1)$-generic with respect to $W_{\amb}$.
\end{lemma}
\begin{proof}
Fix any $j$ distinct subspaces in $\mathcal{S}$, say $W_1, \ldots, W_j$ and let $W = W_1 \cap \cdots \cap W_j$. Since $\mathcal{S}$ is $j$-generic, $\codim(W) = j\cdot r$. We claim that there are at most $q^{j\cdot r}$ subspaces $W_{i'} \in \mathcal{S} \setminus \{W_1,\ldots, W_j\}$ such that $\codim(W \cap W_{i'}) \leq (j+1)r - 1$. Call such subspaces bad, and suppose for the sake of contradiction that there are more than $q^{jr}$ bad subspaces $W_{i'}$. Then for each bad $W_{i'}$ we have,
\begin{align*}
\dim(W_{i'} + W) &= \dim(W_{i'}) + \dim(W) - \dim(W_{i'} \cap W) \\
&\leq (\dim(W_{\amb}) - r) + (\dim(W_{\amb}) - jr) - (\dim(W_{\amb}) - (j+1)r + 1) \\
&= \dim(W_{\amb}) - 1.
\end{align*}
Therefore, for each $W_{i'}$, the space $W + W_{i'}$ is contained in a hyperplane $H$ such that $H \supseteq W$. There are at most $q^{\codim(W)} - 1 = q^{j\cdot r} -1$ hyperplanes $H$ containing $W$, and by the pigeonhole principle it follows that there are two bad subspaces say $W_{i'_1}, W_{i'_2}$ that are both contained in the same hyperplane $H$. This is a contradiction however, as by the $j$-genericness of $\mathcal{S}$, we must have
\begin{align*}
\dim(W_{i'_1} + W_{i'_2}) &= \dim(W_{i'_1}) + \dim(W_{i'_2}) - \dim(W_{i'_1} \cap W_{i'_2})  \\
&= 2(\dim(W_{\amb})-r) - (\dim(W_{\amb})-2r)\\
&= \dim(W_{\amb}),
\end{align*}
and hence $W_{i'_1}$ and $W_{i'_2}$ cannot both be contained in the hyperplane $H$.

The lemma now follows from the claim we have just shown. Construct a subset $\mathcal{S'}$ greedily as follows: 
\begin{enumerate}
    \item Initialize $\mathcal{S'}$ by picking $j$ arbitrary 
    subspaces from $\mathcal{S}$ and inserting them to $\mathcal{S'}$.
    \item For any $j$ subspaces in $\mathcal{S'}$, say 
    $W_1,\ldots,W_j$, remove any $W'\in\mathcal{S}$ which
    is bad for them.
    \item If $\mathcal{S}$ is not empty, pick some $W\in\mathcal{S}$, insert it to $\mathcal{S'}$ and iterate.
\end{enumerate}
Note that trivially, the collection $\mathcal{S'}$ will be 
$(j+1)$-generic in the end of the process. To lower bound the
size of $\mathcal{S'}$, note that when $|\mathcal{S'}| = s$, the number 
of elements from $\mathcal{S}$ that have been deleted is 
at most $s^j q^{jr}$, and hence so long as this value is less than
$N$, we may do another iteration. Thus, we must have that 
$s\geq \left(\frac{N}{q^{jr}}\right)^{1/j} = \frac{N^{1/j}}{q^r}$
when the process terminates.
\end{proof}

\begin{lemma} \label{lm: reduce codimension}
    Let $\{ W_1, \ldots, W_N\}$ be a set of subspaces of $W_{\amb}$ of codimension $r$. Then for any integer $m \geq 1$, at least one of the following holds.
    \begin{itemize}
        \item There are $m$ subspaces, say $W_1, \ldots, W_m$ that are $2$-generic with respect to $W_{\amb}$.
        \item There is a subspace $W'_{\amb} \subseteq W_{\amb}$ of codimension $1$ that contains $N' \geq \frac{N}{mq^r}$ of these subspaces, say $W_1,\ldots, W_{N'}$.
    \end{itemize}
\end{lemma}
\begin{proof}
    Note that for any $1\leq i \neq j \leq N$, we have $\codim(W_i \cap W_j) \leq 2r$. Consider the graph with vertices $W_1, \ldots, W_N$ and $W_i$, $W_j$ are adjacent if and only if $i \neq j$ and $\codim(W_i \cap W_j) \leq 2r-1$. If every vertex in this graph has degree at most $\floor{\frac{N}{m}}$, then we are done as there is an independent set of size $m$ and these subspaces satisfy the first condition. Suppose this is not the case. Then there is a vertex, say $W_N$, that has at least $\ceil{\frac{N}{m}}$ neighbors, say $W_1, \ldots, W_{\ceil{\frac{N}{m}}}$. For $1\leq i \leq \ceil{\frac{N}{m}}$, we have $\codim(W_N \cap W_i) \leq 2r-1$, so 
    \[
    \dim(W_i + W_N) = \dim(W_i) + \dim(W_N) - \dim(W_i \cap W_N) \leq \dim(W_{\amb}) - 1.
    \]
    Thus $W_i + W_N$ is always contained in a codimension $1$ subspace of $W_{\amb}$ that contains $W_N$. Since the number of such subspaces is $q^{r}-1$, one of these codimension $1$ subspaces of $W_{\amb}$, say $W'_{\amb}$, contains at least 
    \[
    \frac{\ceil{N/m}}{q^r-1} \geq \frac{N}{mq^r}
    \]
    of the subspaces in the list $W_1, \ldots, W_{\ceil{\frac{N}{m}}}$. 
\end{proof}
Repeatedly applying~\cref{lm: reduce codimension} yields the following corollary.
\begin{corollary} \label{cor: two generic}
    Let $\{W_1, \ldots, W_N\}$ be a set of subspaces of $W_{\amb}$ of codimension $r$ with respect to $W_{\amb}$. Then, there exists a subspace $W'_{\amb} \subseteq W_{\amb}$, an integer $1 \leq s \leq r$, and a subset of $m \geq \frac{N^{\frac{1}{r+1}}}{q^{r}}$ subspaces from this set, say $\{W_1, \ldots, W_m\}$, which are all contained in $W'_{\amb}$ and satisfy the following two conditions.
    \begin{itemize}
        \item Each $W_i$, $1\leq i \leq m$, has codimension $s$ with respect to $W'_{\amb}$.
        \item $\{W_1, \ldots, W_m\}$ is $2$-generic with respect to $W'_{\amb}$.
    \end{itemize}
\end{corollary}
\begin{proof}
    To start set $W'_{\amb} = W_{\amb}$. If the $W_i$'s have codimension $1$ in $W'_{\amb}$ then the result holds. 
    
    Otherwise, if the conclusion does not hold, then apply~\cref{lm: reduce codimension} with $m = \frac{N^{\frac{1}{r+1}}}{q^{r}}$. Either the first condition of~\cref{lm: reduce codimension} holds and we are done, or we can find a new subspace, $W''_{\amb}$, of codimension $1$ inside the current $W'_{\amb}$ containing at least $\frac{N}{mq^r}$ of the $W_i$'s. Set $W'_{\amb} = W''_{\amb}$ and repeat. Note that the codimension of the $W_i$'s with respect to $W'_{\amb}$ drops by $1$ after every iteration, so we will repeat at most $r$ times before reaching the desired conclusion. This yields a list of $W_i$'s that satisfy the conditions of size at least
    \[
    \frac{N}{(mq^r)^r} = N^{\frac{1}{r+1}} \geq m.\qedhere
    \]
\end{proof}
Combining~\cref{cor: two generic} and~\cref{lm: inductive generic}, we get~\cref{lm: generic} as follows:
\begin{proof}[Proof of~\cref{lm: generic}]
    By~\cref{cor: two generic}, there is a collection $\mathcal{S}'$ of size $|\mathcal{S}'| \geq \frac{N^{\frac{1}{r+1}}}{q^{r}}$ and $W'_{\amb} \subseteq W_{\amb}$ such that $\mathcal{S'}$ is $2$-generic with respect to $W'_{\amb}$, and each $W_i$ has (the same) codimension $s\leq r$ with respect to $W'_{\amb}$. Applying~\cref{lm: inductive generic} 
    $t-2$ times, there is a set of $t$-generic subspaces relative to $W'_{\amb}$, $\mathcal{S}'' \subseteq \mathcal{S'}$, of size
    \[
   \left|\mathcal{S}'' \right| \geq \left(\left(\left(|\mathcal{S}'|^{\frac{1}{2}} \cdot \frac{1}{q^r} \right)^{\frac{1}{3}} \cdots \right) \cdot \frac{1}{q^r}\right)^{\frac{1}{t-1}}\cdot \frac{1}{q^r}
   \geq \frac{\left|\mathcal{S}' \right|^{\frac{1}{(t-1)!}}}{q^{2r}}  \
   \geq \frac{N^{\frac{1}{(r+1)\cdot (t-1)!}}}{q^{3r}}.\qedhere
    \]
\end{proof}

In addition to~\cref{lm: generic}, we state another useful feature of generic sets of subspaces, formalized in~\cref{lm: keep half generic} below. The lemma asserts that if a collection $\{W_1,\ldots,W_N\}$ is generic, and one zooms-outs from the ambient space $W_{\amb}$ into a hyperplane $H$, then
one gets an induced collection $\{W_1\cap H,\ldots,W_N\cap H\}$
which is almost as generic. 
\begin{lemma} \label{lm: keep half generic}
    Let $\mathcal{S} = \{W_1, \ldots, W_N\}$ be a collection of subspaces of codimension $r$ that is $t$-generic with respect to some space $W_{\amb}$ for an even integer $t$, and let $H$ be a hyperplane in $W_{\amb}$. Then the collection of subspaces
    $\mathcal{S}' = \{W_1 \cap H, \ldots, W_N \cap H\}$ can be made a $\frac{t}{2}$-generic set of subspaces with respect to $H$ with codimension $r$ inside of $H$ by removing at most $\frac{t}{2}$ subspaces $W_i\cap H$ from it.
\end{lemma}
\begin{proof}
Suppose that  $\mathcal{S}'$ is not $\frac{t}{2}$-generic with respect to $H$ with codimension $r$ inside of $H$, as otherwise we are done. In this case, there must exist $\frac{t}{2}$ distinct subspaces, say $W_1 \cap H,\ldots, W_{\frac{t}{2}} \cap H \in \mathcal{S}'$ such that the codimension of $W_1 \cap \cdots \cap W_{\frac{t}{2}}\cap H$ with respect to $H$, is less than $tr/2$. As $H$ has codimension $1$ with respect to $W_{\amb}$ it follows that
\begin{equation} \label{eq: codim aux}
\codim\left(W_1 \cap \cdots \cap W_{\frac{t}{2}} \cap H\right) \leq \frac{t}{2} \cdot r,
\end{equation}
where now the codimension is with respect to $W_{\amb}$. 
Since $\mathcal{S}$ is $t$-generic (and thus $\frac{t}{2}$-generic as well) with respect to $W_{\amb}$, this implies that $W_1 \cap \cdots \cap W_{\frac{t}{2}} \subseteq H$.
Indeed, this is because $W_1 \cap \cdots \cap W_{\frac{t}{2}}$ already has codimension $\frac{tr}{2}$, so by \eqref{eq: codim aux}, further intersecting with $H$ does not result in a smaller subspace, and hence this subspace must be inside $H$.

Now delete $W_1 \cap H, \ldots,  W_{\frac{t}{2}}\cap H$ from $\mathcal{S}'$. We claim that the resulting set is $\frac{t}{2}$-generic with respect to $H$. Suppose for the sake of contradiction that it is not. Then there must be another $\frac{t}{2}$ distinct subspaces, say $W_{\frac{t}{2}+1} \cap H,\ldots, W_{t} \cap H \in \mathcal{S}'$ such that $W_{\frac{t}{2}+1} \cap \cdots \cap W_{t} \subseteq H$, implying
\[
\left(W_1 \cap \cdots \cap W_{\frac{t}{2}}\right) + \left(W_{\frac{t}{2}+1} \cap \cdots \cap W_{t}\right) \subseteq H.
\]
This is a contradiction however, as $\mathcal{S}$ is $t$-generic with respect to $W_{\amb}$, so $\codim(W_1 \cap \cdots \cap W_t) = tr$, and
\begin{align*}
&\dim\left((W_1 \cap \cdots \cap W_{\frac{t}{2}}) + (W_{\frac{t}{2}+1} \cap \cdots \cap W_{t})\right)\\
&\qquad\qquad=\dim\left(W_1 \cap \cdots \cap W_{\frac{t}{2}}\right)+\dim\left(W_{\frac{t}{2}+1} \cap \cdots \cap W_{t}\right)
-\dim\left(W_{1} \cap W_{2} \cap \cdots \cap W_{t}\right)\\
&\qquad\qquad= 2\left(\dim(W_{\amb}) - \frac{t}{2}r\right) - \dim(W_{\amb}) + tr = \dim(W_{\amb}) > \dim(H),
\end{align*}
and contradiction.
\end{proof}
The following result is a version of~\cref{lm: keep half generic} for codimensions larger than $1$:
\begin{lemma} \label{lm: keep half generic iterated}
    Let $\W$ be a collection of subspaces of codimension $r$ that is $2^K$-generic with respect to $W_{\amb}$ and let $B$ be a subspace of codimension $j$. Then, there is a set of subspaces $\mc{W}_B$ such that the following holds
    \begin{itemize}
        \item Each $W_B \in \W_B$ is equal to $W \cap B$ for some $W \in \W$.
        \item The set of subspaces $\W_B$ is $2^{K-j}$-generic with respect to $B$ and each subspace in $\W_B$ has codimension $r$.
        \item $|\W| - 2^K \leq |\mc{W}_B| \leq |\W|$
    \end{itemize}
\end{lemma}
\begin{proof}
We will obtain $\W_B$ by removing a small number of subspaces from $\{W \cap B \; | \; W \in \W\}$. There is a sequence of subspaces $W_{\amb} = B_0 \supseteq B_1 \supseteq \cdots \supseteq B_j = B$, such that $B_{i+1}$ is a hyperplane inside of $B_i$. Do the following,
\begin{enumerate}
    \item Initialize $\W_{0} = \W$ and set $i = 1$.
    \item Set $\W_{i} = \{W \cap B_i \; | \; W \in \W_{ i - 1}\}$. Then, using \cref{lm: keep half generic}, remove at most $2^{K-i-1}$-many subspaces to turn $\W_{i}$ into a $2^{K - i}$-generic collection with respect to $B_i$.
    \item Stop if $i = j$, otherwise, increase $i$ by $1$ and return to step $2$.
\end{enumerate}

We verify the three properties. The first property holds by construction. For the second property, it is clear that the output is a set of subspaces $\W_{j} \subseteq \W_B$ that is $2^{K-j}$-generic with respect to $B$, and by \cref{lm: keep half generic}, the codimension of each $W \cap B \in \W_B$ with respect to $B$ is still $r$. Finally, for the third property, note that during each iteration, at most $2^{K-i-1}$ subspaces are removed by~\cref{lm: keep half generic}, so altogether we remove at most $\sum_{i=0}^{K-1} 2^i \leq 2^K$ subspaces.
\end{proof}

\subsubsection{The Sampling Lemma} \label{sec: sampling lemmas}
As explained earlier, the notion of genericness is useful probabilistically, 
and in this section we state and prove a sampling lemma about generic collections which is necessary for 
our analysis. We work inside a subspace $W_{\amb}$. Fix an arbitrary zoom-in $Q \subseteq W_{\amb}$ of dimension $a$, and let $\mathcal{S} = \{W_1,\ldots, W_m\}$ be a $2$-generic collection of subspaces of $W_{\amb}$ of codimension $r$ all containing $Q$. Also let $\mathcal{A}$ be a set of $j$-dimensional subspaces containing $Q$. Consider the following two probability measures over $\Zoom_j[Q,W_{\amb}]$:
\begin{enumerate}
    \item The distribution $\mu$ which is uniform over $\Zoom_j[Q,W_{\amb}]$.
    \item The distribution $\nu_{\mc{S}}$, wherein a subspace is sampled
    by first picking $i\in \{1,\ldots,m\}$ uniformly and then 
    sampling a subspace from $\Zoom_j[Q,W_i]$ uniformly.
\end{enumerate}
Throughout this section one should think of $j$ as much larger than $r$ and $a$, and $\dim(W_{\amb})$ as much larger than $j$. Concretely, we require $j \geq 2^{100}r^2 a^2$ and $\dim(W_{\amb}) \geq 2^{100}j$. The main content of this section is the following lemma, 
asserting that the measures 
$\mu$ and $\nu_{\mc{S}}$ are close in 
statistical distance provided that $m$ is large. More precisely:
\begin{lemma} \label{lm: nu vs mu}
Let $W_{\amb}$ be the ambient space and let $a,r, j$ be dimensions such that $j \geq 2^{100r^2a^2}$, $\dim(W_{\amb}) \geq 2^{100j}$. Let $Q \subseteq W_{\amb}$ be a subspace of dimension $\dim(Q) = a$, and let $\mc{S}$ be a $2$-generic collection of subspaces of $W_{\amb}$ such that $|\mc{S}| = m$, each subspace in $\mc{S}$ has codimension $r$ with respect $W_{\amb}$, and each subspace in $\mc{S}$ contains $Q$. Then, for any set of $j$-dimensional subspaces $\mathcal{L} \subset \Zoom_j[Q,W_{\amb}]$, we have
    \[
    \left|\nu_{\mc{S}}(\Lcal) - \mu(\Lcal)\right| \leq \frac{3q^{\frac{r}{2}(j-a)}}{\sqrt{m}}.
    \]
\end{lemma}
We now set up some notations for the proof of~\cref{lm: nu vs mu}. For $L \in \Zoom_j[Q,W_{\amb}]$ let 
\begin{equation} \label{eq: def N_W(L)}  
N_{\mc{S}}(L) = |\{W_i \in \mathcal{S} \; | \; L \subseteq W_i \}|,
\end{equation}
and for an arbitrary pair of distinct $W_i, W_{i'} \in \mathcal{S}$ define the following quantities: 
\begin{equation}   \label{eq: quantities sampling lemma}
\begin{split}
D &= | \{L \in \Zoom_j[Q, W_{\amb}] \; | \; L \subseteq W_i \}| \\ 
p_1 &= \Pr_{L \in \Zoom_j[Q, W_{\amb}]}[L \subseteq W_i] \\
p_2 &= \Pr_{L \in \Zoom_j[Q, W_{\amb}]}[L \subseteq W_i \cap W_{i'}] \\
\end{split}
\end{equation}
We note that all of these quantities are well defined as they do not
depend on the identity of $W_i$ and 
$W_{i'}$.
The following bound on $p_1$ will be useful in this section.
\begin{lemma} \label{lm: p_1 bound}
Let $p_1$ be as in~\eqref{eq: quantities sampling lemma} and suppose $W_{\amb}$ and the dimensions $a,r,j$ satisfy $j \geq 2^{100r^2a^2}$, $\dim(W_{\amb}) \geq 2^{100j}$. Then $|p_1 - q^{-r(j-a)}| \leq \frac{q^j}{|W_{\amb}|}$.
\end{lemma}
\begin{proof}
    We have,
    \begin{align*}
        p_1 &= \Pr_{L \in \Zoom_j[Q, W_{\amb}]}[L \subseteq W_i] \\
        &= \Pr_{v_1, \ldots, v_{j-a} \in V}[v_1, \ldots, v_{j-a} \in W_i \; | \; \dim(Q + \spa(v_1,\ldots, v_{j-a})) = j] \\
    \end{align*}
    Additionally, 
    \[
    \Pr_{v_1, \ldots, v_{j-a} \in W_{\amb}}[v_1, \ldots, v_{j-a} \in W_i] = q^{-r(j-a)},
    \]
    so 
    \[
    \left|p_1 - q^{-r(j-a)}\right| \leq \Pr_{v_1, \ldots, v_{j-a} \in V}[\dim(Q + \spa(v_1, \ldots, v_{j-a})) < j] \leq \frac{q^j}{|W_{\amb}|}.\qedhere
    \]
\end{proof}
 Using our notations and $|\Zoom_j[Q, W_{\amb}]| \cdot p_1 = D$ we have that
\begin{equation} \label{eq: measures sampling lemma}
\mu(L) = \frac{1}{|\Zoom_j[Q,W_{\amb}]|} = \frac{p_1}{D},
\qquad\qquad\qquad \nu_{\mc{S}}(L) = \frac{N_{\mc{S}}(L)}{m} \cdot \frac{1}{D} =  \frac{N_{\mc{S}}(L)}{mD},
\end{equation}
so comparing $\mu$ and $\nu_{\mc{S}}$ amounts to comparing $p_1$ and $N_{{\mc{S}}}/m$.
In the following claim we analyze the expectation and variance 
of $N_{\mc{S}}(L)$ when $L$ is chosen uniformly from $\Zoom_j[Q,W_{\amb}]$:
\begin{claim} \label{cl: var bound}
Let $r, a, j, Q, W_{\amb}, N_S$ be as defined above. 
We have,
\[
\E_{L \in \Zoom_j[Q,W_{\amb}]}[N_{\mc{S}}(L)] = p_1 m \quad \text{and} \quad \var(N_{\mc{S}}(L)) \leq p_1 m, \]
where the variance is over uniform $L \in \Zoom_j[Q,W_{amb}]$.
\end{claim} 
\begin{proof}
    By linearity of expectation 
    \[
    \E_{L \in \Zoom_j[Q,W_{\amb}]}[N_{\mc{S}}(L)] = \sum_{i = 1}^m \Pr_{L \in \Zoom_j[Q,W_{\amb}]}[L \subseteq W_i] = p_1m,
    \]
    and we move on to the variance analysis.
    To bound $\E_{L \in \Zoom_j[Q,W_{\amb}]}[N_{\mc{S}}(L)^2]$, write
    \begin{align*}
            \E_{L \in \Zoom_j[Q,W_{\amb}]}[N_{\mc{S}}(L)^2] &= \E_{L \in \Zoom_j[Q,W_{\amb}]}\left[\left(\sum_{i=1}^{m} \ind_{L \subseteq W_i} \right)^2\right] \\
            &\leq  m \cdot \Pr_{L \in \Zoom_j[Q,W_{\amb}]}[ L \subseteq W_i] + m^2 \cdot\Pr_{L \in \Zoom_j[Q,W_{\amb}]}[L \subseteq W_{i} \cap W_{i'}] \\
            &= p_1m + p_2 m^2.
    \end{align*}
    It follows that,
    \[
    \var(N_{\mc{S}}(L)) \leq p_1m + p_2m^2 - p_1^2 m^2.
    \] 
    Finally note that by the $2$-genericness $p_2$ and $p_1^2$ are nearly the same value, and it can be checked using~\eqref{eq: quantities sampling lemma} that $p_2 \leq p_1^2$ and so $\var(N_{\mc{S}}(L)) \leq p_1 m$.
\end{proof}

Combining Chebyshev's inequality with~\cref{cl: var bound}, we conclude the following lemma:
\begin{lemma} \label{lm: deviation bounds}
    For any $c>0$ it holds that
    \[
    \Pr_{L \in \Zoom_j[Q, W_{\amb}]}\Biggl[ \left|N_{\mc{S}}(L) - p_1 m \right| \geq c \cdot p_1 m \Biggr] \leq \frac{1}{c^2 p_1 m} \leq \frac{1.01 \cdot q^{r(j-a)}}{c^2 m},
    \]
where recall the chosen $L$ has dimension $j$, $\dim(W_{\amb}) \geq 2^{100}q^j$, $\mc{S}$ is a $2$-generic set of $m$ codimension $r$ subspaces all containing a fixed subspace of dimension $a$, and 
\[
p_1 = \Pr_{L \in \Zoom_j[Q, W_{\amb}]}[L \subseteq W_i] \in \left[q^{-r(j-a)} - \frac{q^j}{|W_{\amb}|}, q^{-r(j-a)} + \frac{q^j}{|W_{\amb}|}\right]
\]
and
\[
\E_{L \in \Zoom_j[Q, W_{\amb}]}[N_{\mc{S}}(L)] = p_1m.
\]
In words, $p_1$ the probability that a random $L \in \Zoom_j[Q, W_{\amb}]$ is contained in some fixed codimension $r$, $W_i$ subspace of $W_{\amb}$.
\end{lemma}
\begin{proof}
The first inequality is an immediate result of Chebyshev's inequality with the bounds from~\cref{cl: var bound}. To get the second inequality, we apply the lower bound on $p_1$ from \cref{lm: p_1 bound} to get
\[
p_1 \geq q^{-r(j-a)} - \frac{q^j}{|W_{\amb}|} \geq \frac{1}{1.01q^{r(j-a)}}.
\]
\end{proof}
Lastly, we use~\cref{cl: var bound} to prove~\cref{lm: nu vs mu}.
\begin{proof}[Proof of~\cref{lm: nu vs mu}]
We have,
\begin{align*}
|\mu(\Lcal) - \nu_{\mc{S}}(\Lcal)|
=\frac{1}{mD}\left|\sum\limits_{L\in \Lcal}N_{\mc{S}}(L)-p_1 m\right|
&\leq 
\frac{1}{mD}\sum\limits_{L\in \Lcal}\left|N_{\mc{S}}(L)-p_1 m\right|\\
&= \frac{|\Zoom_j[Q, W_{\amb}]|}{mD} \cdot \E_{L\in \Zoom_j[Q, W_{\amb}]}[|N_{\mc{S}}(L)-p_1 m|],
\end{align*}
and by Cauchy-Schwartz we get that
\[
|\mu(\Lcal) - \nu_{\mc{S}}(\Lcal)|
\leq \frac{|\Zoom_j[Q, W_{\amb}]|}{mD}\sqrt{\var(N_{\mc{S}}(L))}.
\]
Plugging in~\cref{cl: var bound} and using $|\Zoom_j[Q, W_{\amb}]| = D/p_1$ we get that
\[
|\mu(\Lcal) - \nu_{\mc{S}}(\Lcal)|
\leq \frac{1}{\sqrt{p_1 m}}\leq \frac{3q^{\frac{r}{2}(j-a)}}{\sqrt{m}}.
\]
In the second inequality, we apply the lower bound on $p_1$ from \cref{lm: p_1 bound} to get 
\[
p_1 \geq q^{-r(j-a)} - \frac{q^j}{|W_{\amb}|} \geq \frac{1}{9q^{r(j-a)}}.
\]
\end{proof}


\subsubsection{An Upper Bound on the Number of Maximal Zoom-outs}
For this subsection, we work in the second prover's space, $\Ff_q^V$, and make the assumption that $|V| \gg \ell$, say $|V| \geq 2^{100}q^{\ell}$ to be concrete. We first establish several results in the simplified setting where there is no zoom-in. After that we show how to deduce an analogous result with a zoom-in. Throughout this section, we fix 
$T$ to be a table that assigns, to each $L \in \Grass(\Ff_q^V, 2\ell)$, a linear function on $L$.

\begin{definition} \label{def: maximal}
    Given a table $T$ on $\Grass(\Ff_q^V, 2\ell)$ and a subspace $Q \subseteq \Ff_q^V$, we call a zoom-out, function pair, $(W, g_W)$, where $Q \subseteq W \subseteq \Ff_q^V$ and $g_W: W \xrightarrow[]{} \Ff_q$, $(C, \zeta)$-maximal with respect to $T$ on $Q$ if 
    \[
    \Pr_{L \in \Grass(\Ff_q^V, 2\ell)}[g_W|_L \equiv T[L] \; | \; Q \subseteq L \subseteq W] \geq C,
    \]
  and there does not exist another zoom-out function pair, $(W', g_{W'})$ such that $\Ff_q^V \supseteq W' \supsetneq W$, $g_{W'}: W' \xrightarrow[]{} \Ff_q$, $g_{W'}|_{W}\equiv g|_{W}$ and
    \[
    \Pr_{L \in \Grass(\Ff_q^V, 2\ell)}[g_{W'}|_L \equiv T[L] \; | \; Q \subseteq L \subseteq W'] \geq \zeta C.
    \]
    In the case that $Q = \{0\}$, we say that $(W, g_W)$ is $(C, \zeta)$-maximal with respect to $T$.
\end{definition}
In the above statement, $C$ should be thought of as small and 
$\zeta \in (0, 1]$ should be thought of as an absolute constant. With this in mind, one should think of a zoom-out function pair $(W, g_W)$ as being $(C, \zeta)$ maximal if the agreement between $g_W$ and the table $T$ on entries in $\Zoom_{2\ell}[Q, W]$ is not ``explained'' by the agreement of a larger zoom-out function pair, $(W', g_{W'})$, where $W' \supseteq W$ and $g_{W'}|_W = g_W$. By definition every zoom-out function pair $(W, f)$, where $f$ has non-trivial agreement with $T$ inside $W$, is in some sense contained in a maximal zoom-out function pair $(W', f')$ (possibly with slightly lower agreement). This point is encapsulated by the following lemma.

\begin{lemma} \label{lm: zoom out contained in maximal}
   Let $T$ be a table on $\Grass(\Ff_q^V, 2\ell)$, $Q \subseteq \Ff_q^V$, and $W \subseteq \Ff_q^V$ be a subspace of codimension $r$ containing $Q$. Suppose that there exists a linear function $g_{W}: W \xrightarrow[]{} \Ff_q$ such that 
    \[
    \Pr_{L \in \Grass(\Ff_q^V, 2\ell)}[g_{W}|_L \equiv T[L] \; | \; Q \subseteq L \subseteq W] \geq C.
    \]
Then there exists a subspace $W' \supseteq W$ and a linear function $g_{W'}: W' \xrightarrow[]{} \Ff_q$ such that $g_{W'}|_{W} \equiv g_W$ and $(g_{W'}, W')$ is $(C \zeta^{-r}, \zeta)$-maximal. 
\end{lemma}
\begin{proof}
  This is an immediate consequence of~\cref{def: maximal}. If $(W, g_W)$ is $(C, \zeta)$-maximal then we are done. Otherwise, there must exist $W_1, g_{W_1}$ such that $W_1 \supsetneq W$, $g_{W_1}|_W \equiv g_W$ and 
   \[
    \Pr_{L \in \Grass(\Ff_q^V, 2\ell)}[g_{W_1}|_L \equiv T[L] \; | \; Q \subseteq L \subseteq W_1] \geq \zeta C.
    \]
   We can repeat this argument at most $r$ times before obtaining some $(g_{W'}, W')$ that is $(C\zeta^{-r}, \zeta)$ maximal and satisfies $W' \supseteq W$ and $g_{W'}|_{W} \equiv g_W$.
\end{proof}

Next, we will want to upper bound the number of maximal zoom-outs. The following lemma (which is a formal version of~\cref{lem:main_technical_intro} from the introduction) is the key component in this upper bound.  We set the following parameters for the remainder of the section, which can all be considered constant:
\[
\xi = \delta^5, \quad \delta_2 = \frac{\xi}{100}, \quad t = \left(2^{2 + 10/\delta_2}\right)!\;.
\]
Recall that the parameter $\delta > 0$ is an arbitrary (small) constant, which we may fix freely.

\begin{lemma} \label{lm: bound on good zoom outs}
Fix $\delta > 0$, a finite field $\Ff_q$, and set $\xi, \delta_2, t$ as above. Then the following is true for sufficiently large $\ell$. Let $T$ be a table on $\Grass(\Ff_q^V, 2\ell)$ and set $r \leq \frac{10}{\delta}$, $C \geq q^{-2(1-\delta^{5})\ell}$, and 
 \[
 N \geq q^{100 \left(t-1\right)! r^{2}\ell \xi^{-1}}.
 \]
Suppose that $(W_1,f_1),\ldots,(W_N,f_N)$ are zoom-out, function pairs such that the $W_i$'s are all distinct and of codimension $r$, and for each $i \in [N]$
    \[
    \Pr_{L \in \Grass(\Ff_q^V, 2\ell)}[f_i|_L \equiv T[L] \; | \; L \subseteq W_i] \geq C.
    \]
Then there is a subspace $A \subseteq \Ff_q^V$, a linear function $h': A \xrightarrow[]{} \Ff_q$, and a set of subspaces $\mathcal{W}' \subseteq \{W_1, \ldots, W_N \}$ of size at least $q^{50r\ell \xi^{-1}}$ which satisfy the following,
\begin{itemize}
    \item Each $W_i \in \W'$ is strictly contained in $A$ and has codimension $r' < r$ with respect to $A$.
    \item $\W'$ is $2$-generic with respect to $A$.
    \item For any $W_i \in \W'$, $h'|_{W_i} \equiv f_i$.
\end{itemize}
\end{lemma}
\begin{proof}
    The proof is deferred to~\cref{sec: bounded zoom-outs}.
\end{proof}
To upper bound the number of maximal zoom-outs we  also need the following list decoding property.
\begin{restatable}{lemma}{listdecoding} 
\label{lm: table list decoding}
    Let $T$ be a table on $\Grass(\Ff_q^V, 2\ell)$, let $Q$ be an $r_1$-dimensional subspace, and let $W \supseteq Q$ be a subspace of codimension $r_2$. Suppose that $2\ell$ is sufficiently large and $\dim(W) \geq 20 \ell$. Let $f_1, \ldots, f_m$ be a list of distinct linear functions such that $f_i|_L \equiv T[L]$ for at least $\beta$-fraction of the $2\ell$-dimensional subspaces $L$ such that $Q \subseteq L \subseteq W$, for $\beta \geq 2q^{-2\ell+r_1} + c$, and $c > 0$. Then, $m \leq \frac{4}{c^2}$ and in particular if $\beta > 4q^{-2\ell+r_1}$, then $m \leq \frac{16}{\beta^2}$.
\end{restatable}
\begin{proof}
    The proof is deferred to~\cref{app: list decoding}.
\end{proof}
Combining~\cref{lm: bound on good zoom outs,lm: table list decoding} yields an upper bound on the number of $(C, \zeta)$ maximal zoom-out function pairs with respect to a table $T$ on $Q$.
\begin{thm} \label{th: bounded zoom-out with zoom-in}
Fix $\delta, \Ff_q, \xi, \delta_2, t, \ell$ as in \cref{lm: bound on good zoom outs}, and suppose $|V| \geq q^{\ell}$. Then, the following is true. For any table $T$ on $\Grass(\Ff_q^V, 2\ell)$ such that $|V| \geq q^{\ell}$, any subspace $Q \subseteq \Ff_q^V$ of dimension $r_1 \leq \frac{10}{\delta}$ and any $C \geq q^{-2(1-\delta^3)\ell}$, the number of $(C,\frac{1}{5})$-maximal zoom-out, function pairs with respect to $T$ on $Q$ of codimension at most $\frac{10}{\delta}$ is at most $\frac{160}{\delta} \cdot C^{-2} \cdot  q^{100 \left(t-1\right)! (10/\delta)^{2}\ell \xi^{-1}}$.
\end{thm}
\begin{proof} 
Suppose for the sake of contradiction that $(W_1, f_1), \ldots, (W_M, f_M)$ are $$M > \frac{160}{\delta}C^{-2}q^{100(t-1)!(10/\delta)^2 \ell \xi^{-1}}$$ distinct pairs that are $(C, \frac{1}{5})$ maximal with respect to $T$ on $Q$. By~\cref{lm: table list decoding}, for each $W_i$, there are at most $16C^{-2}$ functions $f: W_i \xrightarrow[]{} \Ff_q$ such that $f|_L \equiv T[L]$ for at least $C$-fraction of the $L \in \Zoom_{2\ell}[Q, W]$. Thus, there are $C^2 M/4$ distinct $W_i$'s appearing in the pairs, and there is a codimension $r_2 \leq \frac{10}{\delta}$ such there are $\frac{M}{\frac{160}{\delta}C^{-2}} = N \geq  q^{100 \left(t-1\right)! (10/\delta)^{2}\ell \xi^{-1}}$ pairs, say, $(W_1, f_1),\ldots, (W_N, f_N)$ zoom-out function pairs that are $(C, \frac{1}{5})$ maximal with respect to $T$ on $Q$ such that the $W_i$'s are all distinct and of codimension $r_2$ in $\Ff_q^V$.

Write $\Ff_q^V = Q \oplus A$. For each $L \subseteq \Ff_q^V$ of dimension $2\ell$ containing $Q$, there is a unique $L' \subseteq A$ such that $L = Q \oplus L'$. Define the table $T'$ that assigns linear functions to each $L' \in \Grass(A, 2\ell - \dim(Q))$ by
 \begin{equation}
    T'[L'] \equiv T[L' \oplus Q]|_{L'}.
\end{equation}  
For each $1 \leq i \leq N$, let $W'_i \subseteq A$ be the unique subspace such that $W_i = W'_i \oplus Q$. We have that $f_i|_{L'} = T'[L']$ for at least $C$-fraction of $L' \in \Grass(W'_i, 2\ell - \dim(Q))$. 

By~\cref{lm: bound on good zoom outs} there exists a subspace $W'_{\amb}$, a linear function $h': W'_{\amb} \xrightarrow[]{} \Ff_q$, and a subcollection $\mathcal{W}' \subseteq \{W'_1, \ldots, W'_m \}$ of size at least $m \geq q^{50(10/\delta) \ell \xi^{-1}}$ such that 
\begin{itemize}
    \item Each $W'_i \in \W'$ has codimension $r'$ where $1 \leq r' \leq r_2$ with respect to $W'_{\amb}$.
    \item $\mathcal{W}'$ is $2$-generic with respect to $W'_{\amb}$.
    \item For any $W'_i \in \W$, $h'|_{W'_i} = f_i|_{W_i'}$.
\end{itemize}
Let $\W = \{W_i \; | \; W'_i \in \W' \}$ and linearly extend $h'$ to the function $h^\star$ on $V^\star = W'_{\amb} \oplus Q$ so that, $h^{\star}|_{W_i} \equiv f_i$ for at least $q^{-r_1}$ of the $W_i$ in $\W$. Note that such an $h^\star$ exists because a random linear extension of $h'$ satisfies this in expectation. It follows that there is a set $\mathcal{V} = \{W_i \in \W \; | \; h^\star|_{W_i} \equiv f_i \}$ of size $|\V| \geq mq^{-r_1} \geq mq^{-\frac{10}{\delta}}$.

Furthermore, because $\W'$ is $2$-generic inside of $W'_{\amb}$, $\W$ is $2$-generic inside of $V^\star$. We will now finish the proof by applying~\cref{lm: nu vs mu} on $\mathcal{V}$. Specifically, let $\nu$ denote the measure over $\Zoom_{2\ell}[Q, V^\star]$ generated by choosing $W_i \in \V$ and then $L \in \Zoom_{2\ell}[Q, V^\star]$ conditioned on $L \subseteq W_i$, let $\mu$ denote the uniform measure over $\Zoom_{2\ell}[Q, V^\star]$, and let 
\[
\Lcal = \{L \in \Zoom_{2\ell}[Q,V^\star] \; | \; h^\star|_L \equiv T[L] \}.
\]
Since $h^\star|_{W_i} \equiv f_i$ for every $W_i \in \V$, we have
\[
 \nu(\Lcal) =  \E_{ W_i \in \V}\left[ \Pr_{L \in \Zoom_{2\ell}[Q, W_i]}[f_i|_L \equiv T[L]] \right] \geq C.
\]
By~\cref{lm: nu vs mu} with $a = r_1, r = r_2,$ and $j = 2\ell$, it follows that 
\[
\mu(\Lcal) \geq C - \frac{3q^{\frac{r_2}{2}(2\ell - r_1)}}{\sqrt{m \cdot q^{-r_1}}} \geq \frac{C}{5}.
\]
 We get that the zoom-out function pair $(V^\star, h^\star)$ satisfies that $V^\star \supsetneq W_i$ and $h^\star|_{W_i} \equiv f_i$ for at least one $i$, and $\Pr_{L \in\Grass(\Ff_q^V, 2\ell)}[h^\star|_L \equiv T[L] \; | \; Q \subseteq L \subseteq V^\star] \geq \frac{C}{5}$. This contradicts the assumption that $(W_i, f_i)$ is $(C, \frac{1}{5})$-maximal with respect to $T$ on $Q$.
\end{proof}

\section{Analysis of the PCP}\label{sec:analysis_of_both}
In this section we analyze the PCP construction $\Psi$ from~\cref{sec:pcp_construct}. As usual, the completeness analysis is straightforward and the soundness analysis
will comprise the bulk of our effort.
\subsection{Completeness}
Suppose that the $\Lin$ instance $(X, \Eq)$ has an assignment $\sigma: X \xrightarrow[]{} \Ff_q$ that satisfies at least $1-\epsilon_1$ of the equations in $\Eq$. Let $\Us \subseteq \U$ be the set of all $U = (e_1, \ldots, e_k)$ where all $k$ equations $e_1, \ldots, e_k$ are satisfied. Then, $|\Us| \geq (1- k\epsilon_1) |\U|$. 
We identify $\sigma$ with the linear function from $\mathbb{F}_q^{X}\to\mathbb{F}_q$, assigning the value 
$\sigma(i)$ to the $i$th elementary basis element $e_i$. 
Abusing notation, we denote this linear map by $\sigma$ 
as well.

For each $U \in \Us$ and vertex $L \oplus H_U$, we set $T_1[L \oplus H_U] \equiv \sigma|_{L \oplus H_U}$. Since $U \in \Us$, these assignments satisfy the side conditions. For all other $U$'s, set $T_1[L \oplus H_U]$ so that the side conditions of $H_U$ are satisfied and $T_1[L \oplus H_U]|_{L} \equiv \sigma|_L$. Such an assignment is possible because $L \cap H_U = \{0\}$. Similarly,
the table $T_2$ is defined as $T_2[R] \equiv \sigma|_{R}$.

Sampling a constraint, note that the constraint is satisfied whenever the $L' \oplus H_{U'}$ chosen in step $3$ of the test satisfies that $U' \in \Us$. As the marginal distribution of $L' \oplus H_{U'}$ is uniform,\footnote{This is true because first a clique is chosen with probability that is proportional to its size and then a vertex is sampled uniformly from the clique.} the distribution of $U'$ is uniform. It follows that the constraint
is satisfied whenever $U'\in \mathcal{U}_{{\sf sat}}$, which 
happens with probability at least $1-k\eps_1$. Thus, 
${\sf val}(\Psi)\geq 1-k\eps_1$.
\subsection{Soundness} \label{sec: PCP soundness proof}
In this section we relate the soundness of the composed PCP to that of the outer PCP. More precisely, we show:
\begin{lemma} \label{lm: soundness outer to inner}
 For all $\delta>0$ there are $r\in\mathbb{N}$ and $\zeta(\delta)>0$ such that the following holds for sufficiently large $\ell \in \mathbb{N}$, sufficiently small $c \in (0, 1)$, and 
 \[
 k = q^{2(1+c)\ell}, \quad  \beta = q^{-2(1+2c/3)\ell},
 \]
as in \eqref{eq: pcp parameters}.
 
Let $G_{\beta, r}^{\otimes k}$ be the parallel repetition of the Smooth Variable versus Equation Game with advice described in~\cref{sec: final outer pcp}, and let $\Psi$ be the composed PCP described in~\cref{sec:pcp_construct} using the sufficiently large $\ell$ above. If $\val(G_{\beta, r}^{\otimes k}) < q^{-\zeta(\delta)\cdot \ell^2}$, then $\val(\Psi) \leq 64q^{-2(1-1000\delta)\ell}$. 
\end{lemma}
The rest of this section is devoted to the proof of~\cref{lm: soundness outer to inner}, and it heavily relies on the tools from~\cref{sec:tools_for_pcp_analysis}. We will show the contrapositive of \cref{lm: soundness outer to inner}. That is, we will show that if there are tables $T_1$ and $T_2$ that are $\epsilon$-consistent for $\epsilon > 64q^{-2\ell(1-1000\delta)}$, then this implies strategies for the two provers that have success probability at least $q^{-\zeta(\delta)\cdot O(\ell^2)}$.

\subsubsection{Clique Consistency}
To start, we will reduce to the case where $T_1$ satisfies a condition called \emph{clique-consistency}.

\begin{definition}
We say an assignment $T$ to $\mathcal{A}$ is clique consistent 
if for every vertex $L_1 \oplus H_{U_1}$ and for every 
$L_2\oplus H_{U_2},L_3\oplus H_{U_3}\in [L_1\oplus H_{U_1}]$, 
the assignments $T[L_2\oplus H_{U_2}]$ and $T[L_3\oplus H_{U_3}]$
satisfy the $1$-to-$1$ constraint between $L_2\oplus H_{U_2}$
and $L_3\oplus H_{U_3}$ specified in~\cref{lm: clique extension}.
\end{definition}

The following lemma shows that if 
$T_1$ and $T_2$ are $\eps$-consistent assignments to 
$\Psi$, then there is a clique-consistent assignment 
$T'_1$ such that $T_1'$ and $T_2$ are $\eps$-consistent.
\begin{lemma}\label{lem:assume_clique_consistent}
    Suppose that the assignments $T_1$ and $T_2 $ are $\epsilon$-consistent, then there is a clique-consistent assignment $T'_1$ such that $T'_1$ and $T_2$ are $\epsilon$-consistent. 
\end{lemma}
\begin{proof}
    Partition $\mathcal{A}$ into cliques, $\mathcal{A} = \textsf{Clique}_1 \sqcup \cdots \sqcup \textsf{Clique}_{m}$. For each $i$, choose a random $L \oplus H_U \in \Cl_i$ uniformly, and for every $L' \oplus H_{U'} \in \Cl_i$ assign $T'_1[L' \oplus H_{U'}]$ in the unique way that is consistent with $T_1[L \oplus H_U]$ and the side conditions of $U'$ as described in~\cref{lm: clique extension}. It is clear that $T'_1$ is clique consistent, and we next analyze the expected fraction
    of constraints that $T'_1$ and $T_2$ satisfy.
    
    Note that an alternative description of sampling a constraint in $\Psi$ proceeds as follows. First choose a clique $\Cl_i$
    with probability that is proportional to by its size, and then choose $L \oplus H_U \in \Cl_i$ in the first step. 
    The rest of the sampling procedure is the same. Let $P(L \oplus H_U)$ be the probability that the test passes conditioned on $L \oplus H_U$ being chosen in the second step. It is clear that every vertex in the clique has equal probability of being chosen, therefore the probability of passing if $\Cl_i$ is chosen is
    \[
    \frac{1}{|\Cl_i|} \sum_{L \oplus H_U \in \Cl_i} P(L \oplus H_U).
    \]
    On the other hand, the expected fraction of constraints satisfied by $T_1'$ and $T_2$ (over the randomness
    of choosing $T_1'$) is
    \[
    \sum_{L \oplus H_U \in \Cl_i} \frac{1}{|\Cl_i|} \cdot P(L \oplus H_U) = \frac{1}{|\Cl_i|} \sum_{L \oplus H_U \in \Cl_i} P(L \oplus H_U).
    \]
    To see this, note that for any $L \oplus H_U$,  $\frac{1}{|\Cl_i|}$ is the probability that $T_1[L \oplus H_U]$ is used to define $T_1'$ on $\Cl_i$. If this is the case, then the probability the test passes on $T_1'$ within $\Cl_i$ is $P(L \oplus H_U)$. 
    
    Since this holds over every clique, it follows that the  expected fraction of constraints satisfied by $T_1'$ equals the fraction of constraints satisfied by $T_1$ and $T_2$. In particular, there is a choice of $T_1'$ such that together with $T_2$ it satisfies at least $\epsilon$ fraction of the constraints.
\end{proof}
Applying~\cref{lem:assume_clique_consistent} we 
conclude that there are clique-consistent assignments to $\Psi$ that are $\eps$-consistent, and henceforth we assume that $T_1$ is clique-consistent to begin with.
We remark that, in the notation of~\cref{sec: constraint graph}, the benefit of having a 
clique-consistent assignment is that the 
constraint that the verifier checks is equivalent to checking that $T_1[L\oplus H_U]|_{R} \equiv T_2[R]$. The latter check 
is a test which is performed within
the space $\mathbb{F}_q^U$ of the first
prover. We will use this fact in the next
section.

\subsubsection{A Strategy for the First Prover}
Let $p(U)$ be the pass probability of tables $T_1$ and $T_2$ conditioned on $U$ being the question to the first prover. That is, $p(U)$ is the probability that a constraint, as described in \cref{sec: constraint graph}, is satisfied, conditioned on $U$ being chosen in step 1 there. As we are assuming that the overall success probability is at least $\epsilon$, $\E_U[p(U)] \geq \epsilon$. By an averaging argument, $p(U) \geq \frac{\epsilon}{2}$ for at least $\frac{\epsilon}{2}$-fraction of the $U$'s. Call such $U$'s good and let $\Ug$ be the set of good $U$'s.

Let $U \in \U$ be the question to the first prover and let $Q$ be the advice. If $U \notin \Ug$, then the first prover gives up, so henceforth assume that $U \in \Ug$. For such $U$, the test of the inner PCP passes with probability at least $\frac{\epsilon}{2}$. More concretely, we have
\[
\Pr_{\substack{L \in \Grass_q(\Ff_q^U, 2\ell), L \cap H_U = \{0\}\\ \dim(R) = 2(1-\delta)\ell}}[T_1[L\oplus H_U]|_{R} \equiv T_2[R] \; | \; R \subseteq L] \geq \frac{\epsilon}{2}.
\]

Next, the first prover chooses an integer $0 \leq r' \leq \frac{10}{\delta}$ uniformly, and takes $Q$ to be the span of the first $r'$-advice vectors (so that $r$ in~\Cref{lm: soundness outer to inner} is taken to be $10/\delta$). By~\cref{th: consistent with side}, there are $r_1, r_2$ satisfying $r_1 + r_2 \leq \frac{10}{\delta}$ such that for at least $q^{-6\ell^2}$ of the $Q \subseteq \Ff_q^U$ of dimension $r_1$, there exists $W_Q \subseteq \Ff_q^U$ containing $Q \oplus H_U$ of codimension $r_2 \leq \frac{10}{\delta}$ and a linear function $g_{Q,W_Q}: W_Q \xrightarrow[]{} \Ff_q$ that satisfies the side conditions on $H_U$ and 
\begin{equation} \label{eq: p1 consistency}
\Pr_{L \in \Zoom_{2\ell}[Q, W_Q], L \cap H_U = \{0\}}[g_{Q,W_Q}|_{L \oplus H_U} \equiv T_1[L \oplus H_U] ] \geq \frac{q^{-2(1-1000\delta^2)\ell}}{2}.
\end{equation}
Henceforth, set $C :=\frac{q^{-2(1-1000\delta^2)\ell}}{2}$. With probability at least $\frac{1}{1+10/\delta}$, the first prover chooses $r' = r_1$, where $r_1$ is the parameter from~\cref{th: consistent with side}. Say that a dimension $r_1$ subspace $Q$ is lucky if there exist $W_Q$ and $g_{Q, W_Q}$ as above satisfying \eqref{eq: p1 consistency}  and let $\Ql$ be the set of all lucky $Q \subseteq \Ff_q^U$. 

Throughout the remainder of the section, for each $r_1$-dimensional $Q$ that satisfies $Q \in \Ql$ and $Q \cap H_U = \{0\}$, we will write $W_Q$ and $g_{Q, W_Q}$ to be the zoom-out and function described that satisfy \eqref{eq: p1 consistency}. Additionally, we will require $W_{Q} = W_{Q'}$ and $g_{Q,W_Q} \equiv g_{Q', W_{Q'}}$ for any two valid $Q, Q'$ such that $Q \oplus H_U = Q' \oplus H_U$. First let us see why this is possible. It is straightforward to verify that $Q \cap H_U = \{0\}$ if and only if $Q' \cap H_U = \{0\}$. Next, note that for any $W \subseteq \Ff_q^U$ we have $W \supseteq Q \oplus H_U$ if and only if $W \supseteq Q' \oplus H_U$. Moreover, for any function $g: W \to \Ff_q$ agreeing with the side conditions on $H_U$, we have
\begin{equation} \label{eq: Q and Q' lucky equivalent}
    \begin{split}
        \Pr_{L \in \Zoom_{2\ell}[Q, W], L \cap H_U = \{0\}}[g|_{L \oplus H_U} \equiv T_1[L \oplus H_U] ] 
        = \Pr_{L \in \Zoom_{2\ell}[Q', W], L \cap H_U = \{0\}}[g|_{L \oplus H_U} \equiv T_1[L \oplus H_U] ],
    \end{split}
\end{equation}
because in each probability, the distribution over $L \oplus H_U$ is the same. It follows that $Q \in \Ql$ if and only if $Q' \in \Ql$. Then, by \eqref{eq: Q and Q' lucky equivalent} we can enforce that the first prover's zoom-out and function satisfy $W_Q = W_{Q'}$ and $g_{Q, W_Q} \equiv g_{Q', W_{Q'}}$. 

Moving on, for each $Q$ such that $Q \in \Ql$ and $Q \cap H_U = \{0 \}$, let us define
\[
\Lcal_Q = \{L \in \Zoom_{2\ell}[Q, W_Q] \; | \; g_{Q, W_Q}|_L \equiv T_1[L \oplus H_U]|_L \}.
\]
For our analysis, we will only consider the case where the first prover chooses $r' = r_1$ and $Q \in \Ql$. If $Q \notin \Ql$ or $Q \cap H_U \neq \{0 \}$ we simply define $\Lcal_Q = \emptyset$, but again we will not worry about such $Q$'s in the analysis. 

Finally, define
\[
\Lcal = \{(x_1, \ldots, x_{2\ell}) \in \left(\Ff_q^{U}\right)^{2\ell}\; | \; \spa(x_1,\ldots, x_{r_1}) \in \Ql, \spa(x_1, \ldots, x_{2\ell})\in \Lcal_{\spa(x_1,\ldots, x_{r_1})} \},
\]
and let $\Qs$ denote the set of $r_1$-dimensional subspaces $Q$ such that 
\begin{equation} \label{eq: smooth}
    \D'_Q(\Lcal) \geq 0.8 \D_Q(\Lcal) - \eta^{20},
\end{equation} 
where $\eta = q^{-100\ell^{100}}$ is the negligible constant appearing in \cref{sec: covering properties}.

If either $Q \notin \Ql$, $Q \cap H_U \neq \{0 \}$, or $Q \notin \Qs$, then the first prover gives up. Otherwise the prover extends the function $g_{Q,W_Q}$ to a linear function on the entire space $\Ff_q^U$ randomly, and we denote this extension by $g: \Ff_q^U \xrightarrow[]{} \Ff_q$. The prover outputs the string $s_{Q,U}$ as their answer where $s_{Q, U} \in \Ff_q^U$ is the unique string such that $g(x) = \langle s_{Q, U}, x \rangle$ for all $x \in \Ff_q^U$. As $g_{Q, W_Q}$, and by extension $g$, respects the side conditions, it follows that $s_{Q,U}$ satisfies the $k$-linear equations of $U$.

At this point, we have fully described the first prover's strategy. With the strategy fresh in mind, it will be helpful to now compute some probabilities related to this strategy. Specifically, let us lower bound $\D_Q(\mc{L})$ for the $Q$ such that $Q \in \Ql$ and $Q \cap H_U = \{0\}$. If $Q$ is additionally smooth, then we also get a lower bound on $\D'_Q(\mc{L})$ by \eqref{eq: smooth}. By definition of $\D_Q$ and $\Lcal$, if $Q \in \Ql$ and $Q \cap H_U = \{0\}$, we have 
\begin{align*}
\D_Q(\Lcal) 
&= \Pr_{x=(x_1,\ldots, x_{2\ell}) \in \left(\Ff_q^{U}\right)^{2\ell}}[x\in \mc{L} \; | \; \spa_{r_1}(x) = Q] \\
&= \Pr_{x \in \left(\Ff_q^{U}\right)^{2\ell}}[g_{Q,W_Q}|_{\spa(x)} \equiv T_1[\spa(x)\oplus H_U]|_{\spa(x)} \land \spa(x) \subseteq W_Q \; | \; \spa_{r_1}(x) = Q] \\
&=  \Pr_{L \in \Grass_q(\Ff_q^U, 2\ell), L \cap H_U = \{0\}}[g_{Q,W_Q}|_{L \oplus H_U} \equiv  T_1[L \oplus H_U] \land L \subseteq W_Q \; | \; Q \subseteq L]  \\
&- \Pr_{x \in \left(\Ff_q^{U}\right)^{2\ell}}[\dim(\spa(x)) < 2\ell \lor \spa(x) \cap H_U \neq \{0\} \; | \; \spa_{r_1}(x) = Q] \\
&\geq  \Pr_{L \in \Grass_q(\Ff_q^U, 2\ell), L \cap H_U = \{0\}}[g_{Q,W_Q}|_{L \oplus H_U} \equiv  T_1[L \oplus H_U] \land L \subseteq W_Q \; | \; Q \subseteq L] - q^{2\ell - 3k} - q^{2\ell - 2k}.
\end{align*}
where the second transition is because, by definition, every $L \in \Lcal_Q$ is contained in $W_Q$. Recalling that $k \gg \ell$, one should think of the last two terms in the last line as negligible.

Continuing, for $Q \in \Ql$, we have
\begin{align}  \label{eq: D_Q to agreement}
     \D_Q(\Lcal) 
     &\geq \Pr_{L \in \Grass_q(\Ff_q^U, 2\ell), L \cap H_U = \{0\}}[g_{Q,W_Q}|_{L \oplus H_U} \equiv  T_1[L \oplus H_U] \; | \; Q \subseteq L \subseteq W_Q]  \notag \\
 &\cdot \Pr_{L \in \Grass_q(\Ff_q^U, 2\ell), L \cap H_U = \{0\}}[L \subseteq W_Q \; | \; Q \subseteq L] - q^{2\ell - 3k} - q^{2\ell - 2k} \notag\\
     &\geq q^{-r_2(2\ell-r_1)} \cdot C - q^{2\ell - 3k} - q^{2\ell - 2k},
\end{align}
where in the third transition uses~\eqref{eq: p1 consistency}
to lower bound the first term by $C$. 

 \subsubsection{A Strategy for the Second Prover}
Let $V$ be the question to the second prover. The second prover will use a table $T'_1$ to derive their strategy. The table $T'_1$ is obtained from $T_1$ as follows. For a question $V$ to the second prover, let $U' \supseteq V$ be an arbitrary question to the first prover. For all $2\ell$-dimensional subspaces $L \subseteq \Ff_q^V$, define
\[
T'_1[L] \equiv T_1[L \oplus H_{U'}]|_L.
\]
If $L \cap H_{U'} \neq \{0\}$, then we arbitrarily define $T'[L]$ to be the $0$ function. Note that provided $|V| \geq 2k$, this will only be the case for a negligible, $q^{2\ell - k}$, fraction of subspaces $L$. 

\vspace{1ex}

\noindent \textbf{Using clique consistency:} we first argue that $T'_1$ is well defined, and for that we make several observations. First, recall that the subspace $L \oplus H_{U'}$ is indeed a subspace of $\Ff_q^{U'}$ because we have $L \subseteq \Ff_q^{V} \subseteq \Ff_q^{U'}$, where $\Ff_q^V$ is viewed as the subspace of $\Ff_q^{U'}$ where the restriction to coordinates in $U' \setminus V$ is $0$. Next, note that all choices of $U' \supseteq V$ lead to the same value of $T'_1[L]$. Indeed, for a fixed $L$, the vertices $L \oplus H_{U'}$ over all $U' \supseteq V$ are in the same clique, 
so since $T_1$ is clique consistent they all receive a consistent value with regards to the mapping in~\cref{lm: clique extension} and hence lead to the same function $T_1[L \oplus H_{U'}]|_L$. Therefore the second prover can construct the table $T'_1$, and we can assume without loss of generality that the $U'$ that was chosen is the same question that the first prover received, i.e.\ $U' = U$.

After constructing $T'_1$, the second prover then chooses a dimension $0 \leq r' \leq \frac{10}{\delta}$ uniformly for the advice $Q$. That is, the set $Q$ to be the span of the first $r'$ advice vectors that they receive. Note that with probability at least $\frac{1}{10/\delta+1}$ the second prover also chooses $r' = r_1$ (that is, they choose the same dimension as the first prover for the advice). The second prover then uniformly chooses a zoom-out function pair $(\Ws, g_{Q, \Ws})$ that is 
\[
\left(\frac{C}{4 \cdot 5^{10/\delta}}, \frac{1}{5} \right)\text{-maximal}
\]
 with respect to $T'_1$ on $Q$ (the second prover gives up if such pair does not exist). 

Finally, the second prover extends the function $g_{Q, \Ws}$ randomly to a linear function on $\Ff_q^{V}$ to arrive at their answer. The resulting function is linear and it is equal to the inner product function $y \xrightarrow[]{} \langle s_{Q,V}, y \rangle$ for some unique string $s_{Q,V} \in \Ff_q^{V}$. The second prover outputs $s_{Q,V}$ as their answer.

\subsubsection{The Success Probability of the Provers}
In order to be successful, a series of events must occur. We go through each one and state the probability that each occurs. At the end this yields a lower bound on the provers' success probability. We remark that the analysis of this sections requires~\cref{lm: basic covering,lm: covering zoom-in}, so recall that $k$ and $\beta$ are set according to~\eqref{eq: pcp parameters} in~\cref{sec: covering properties} so that these lemmas hold.

First, the provers need $U \in \Ug$, which occurs with probability at least $\frac{\epsilon}{2}$. Assuming that this occurs, the provers then both need to choose $r' = r_1$ for the dimension of their zoom-in, which happens with probability at least $\frac{\delta^2}{101}$. If both provers do choose $r' = r_1$, then both prover's set $Q$ to be the span of the first $r_1$ advice vectors that they receive and in particular, receive the same zoom-in $Q$ as advice. Going forward, we will only analyze the provers' success probability assuming that this is the case. 

The provers then need $Q \in \Ql$, $Q \cap H_U = \{0\}$, and $Q \in \Qs$. When analyzing the probability that these three events occur, we need to recall that the advice vectors are actually drawn uniformly from the second prover's space $\Ff_q^V$, rather than $\Ff_q^U$. To this end, let us define the following coupled distribution.

\paragraph{The Advice Distribution, $\D_{{\sf adv}}:$}
\begin{itemize}
    \item Choose $V \subseteq U$ according to the outer PCP.
    \item Choose $x_1, \ldots, x_{r_1} \in \Ff_q^V$ uniformly.
    \item Choose $w_1,\ldots, w_{r_1} \in H_U$ uniformly.
    \item Set $Q = \spa(x_1,\ldots, x_{r_1})$ and $Q' = \spa(x_1+w_1, \ldots, x_{r_1} + w_{r_1})$.
\end{itemize}

Thus, the subspace $Q$ is actually drawn according to the first marginal of $\D_{{\sf adv}}$, which we denote $\D^{1}_{{\sf adv}}$, while the second marginal of $\D_{{\sf adv}}$ is the distribution $\D'_{r_1}$ from \cref{sec: covering advice}. 

The following claim will help us move from statements about $Q'$ to statements about $Q$. It is helpful, because the covering property from \cref{sec: covering advice} will allow us to bound probabilities about $Q'$, while \cref{cl: Q' to Q} will allow us to translate these into statements about $Q$.

\begin{claim} \label{cl: Q' to Q}
Let $(Q,Q')$ be output by $\D_{{\sf adv}}$. Then the following two statements are true.
\begin{itemize}
    \item If $Q' \in \Ql$, $Q' \cap H_U = \{0\}$, then $Q \in \Ql$, $Q \cap H_U = \{0\}$.
    \item $\D_{Q'}(\Lcal) = \D_Q(\Lcal)$ and $\D'_{Q'}(\Lcal) = \D'_Q(\Lcal)$, and as a consequence if $Q' \in \Qs$ then $Q \in \Qs$.
\end{itemize}
\end{claim}
\begin{proof}
    Let $(Q, Q')$ be as described so that $Q = \spa(x_1,\ldots, x_{r_1})$ and $Q' = \spa(x_1 + w_1, \ldots, x_{r_1} + w_{r_1})$ for some $x_i \in \Ff_q^U$ and $w_i \in H_U$. 
    
    We start with the first item. Suppose $Q' \in \Ql$, $Q' \cap H_U = \{0\}$. To see that $Q \in \Ql$, note that $Q \oplus H_U = Q' \oplus H_U$, so as explained earlier we have $W_{Q'} = W_Q$, and $g_{Q, W_Q} \equiv g_{Q', W_{Q'}}$, so the result follows from \eqref{eq: Q and Q' lucky equivalent}. It is also easy to see that $Q\cap H_U = \{0\}$.

    We move onto the second item of showing $\D_{Q'}(\Lcal) = \D_Q(\Lcal)$ and $\D'_{Q'}(\Lcal) = \D'_Q(\Lcal)$. We will first show $\D'_Q(\Lcal) = \D'_{Q'}(\Lcal)$ and then describe how the other equality can be shown similarly. Let $A: Q' \to Q$ be the linear transformation that maps $x_i + w_i$ to $x_i$.  Abusing notation, for $x = (x_1, \ldots, x_{2\ell}) \in \Ff_q^U$ satisfying $\spa_{r_1}(x) = Q$, let us write $A(x) = (A(x_1), \ldots, A(x_{r_1}), x_{r_1+1}, \ldots, x_{2\ell})$. Note that $A$ is a bijection between the following two sets
    \[
   \mathcal{A}_1 = \{x \in \Ff_q^U \; | \; \spa_{r_1}(x) = Q' \} \quad \text{and} \quad \mathcal{A}_{2} = \{y \in \Ff_q^U \; | \;  \spa_{r_1}(y) = Q\}.
    \]    
    Then, 
    \begin{equation} \label{eq: D'_Q and D'_Q' measures}    
\mc{D}'_{Q'}(\Lcal) = \frac{\sum_{x \in \Lcal \cap \mc{A}_1} \mc{D}'(x)}{\sum_{x \in \mathcal{A}_1} \mc{D}'(x)} \quad \text{and} \quad \mc{D}'_{Q}(\Lcal) = \frac{\sum_{y \in \Lcal \cap \mc{A}_2} \mc{D}'(y)}{\sum_{x \in \mathcal{A}_2} \mc{D}'(y)} 
    \end{equation}
Now we claim that the map $A$ is a measure preserving (under measure $\D'$) bijection between $\mathcal{A}_1$ and $\mathcal{A}_2$ and $\Lcal \cap \mc{A}_1$ and $\Lcal \cap \mc{A}_2$.

First note that for any $x \in \mc{A}_1$, we have $A(x) - x \in \left(H_U\right)^{2\ell}$, so it is straightforward to verify that $\D'(x) = \D'(A(x))$. It follows that $A$ is a bijection between $\mathcal{A}_1$ and $\mathcal{A}_2$ which preserves measure under $\D'$, and consequently the denominators in \eqref{eq: D'_Q and D'_Q' measures} are the same.

For the numerators we have noted that if $x \in \Lcal \cap \mc{A}_1$ then $\D'(x)=\D'(A(x))$, and it remains to show that $x \in \Lcal \cap \mc{A}_1$ if and only if $A(x) \in \Lcal \cap \mc{A}_2$. To this end observe that if $x \in \Lcal \cap \mc{A}_1$ then $\spa_{r_1}(x) = Q'$ and 
    \[
    T_1[\spa(x) \oplus H_U]|_{\spa(x)} \equiv g_{Q', W_{Q'}}|_{\spa(x)}.
    \]
 Then, note that $\spa(A(x)) \oplus H_U = \spa(x) \oplus H_U$, $W_Q = W_{Q'}$, and $g_{Q, W_Q} \equiv g_{Q', W_{Q'}}$. Moreover, $g_{Q,W_Q}$ agrees with the side conditions on $H_U$, so the above implies that $g_{Q, W_Q}|_{\spa(x) \oplus H_U} \equiv T_1[\spa(x) \oplus H_U]$, and consequently,
 \[
T_1[\spa(x) \oplus H_U]|_{\spa(A(x))} \equiv g_{Q', W_{Q'}}|_{\spa(A(x))}.
 \]
It follows that $A(x) \in \Lcal \cap \mc{A}_2$ as desired. The other direction follows similarly. Since $x \in \Lcal \cap \mc{A}_1$ if and only if $A(x) \in \Lcal \cap \mc{A}_2$ and the map $A$ is also a bijection from $\Lcal \cap \mc{A}_1$ to $\Lcal \cap \mc{A}_2$ which preserves measure under $\D'$. Thus, the numerators in \eqref{eq: D'_Q and D'_Q' measures} are the same and we can conclude that $\D'_{Q'}(\Lcal) = \D'_Q(\Lcal)$.

The proof that $\D_{Q'}(\Lcal) = \D_{Q}(\Lcal)$ proceeds similarly, except that $x$ and $A(x)$ trivially have the same measure under $\D$, because $\D$ is the uniform distribution over $\Ff_q^U$. 

We conclude that $\D_{Q'}(\Lcal) = \D_Q(\Lcal)$ and $\D'_{Q'}(\Lcal) = \D'_Q(\Lcal)$, so by the definition of $\Qs$ we get that if $Q' \in \Qs$, then $Q \in \Qs$.
\end{proof}

With \cref{cl: Q' to Q} in hand, we may instead lower bound the probability that  $Q' \in \Ql$, $Q' \cap H_U = \{0\}$, and $Q' \in \Qs$ under $\D'_{r_1}$. To this end, we start with the probability that these events occur under $\D_{r_1}$, which is simply the uniform distribution over $\Ff_q^U$. 

By~\cref{th: consistent with side}, the first item occurs with probability at least $q^{-6\ell^2}$. On the other hand the probability that the second item does not occur is at most $\sum_{i=0}^{r_1} \frac{q^{i}q^k}{q^{3k}} \leq q^{r_1+1 - 2k}$, while the probability that the third item does not occur is at most $3\eta^{20}$ by~\cref{lm: smooth Q}. Altogether we get that with probability at least 
\[
q^{-6\ell^2} -q^{r_1+1-2k} -3\eta^{20} \geq  q^{-7\ell^2}
\]
under $\mathcal{D}_{r_1}$, we have $Q' \in \Ql$, $Q' \cap H_U = \{0\}$, and $Q' \in \Qs$. By~\cref{lm: covering zoom-in}, we have that $Q' \in \Ql$, $Q' \cap H_U = \{0\}$, and $Q' \in \Qs$ with probability at least $q^{-8\ell^2}$ under $\mathcal{D}_{r_1}'$. 

Applying \cref{cl: Q' to Q} now yields that with probability at least $q^{-8\ell^2}$ over $(Q, Q') \sim \D_{{\sf adv}}$, we have  $Q \in \Ql$, $Q \cap H_U = \{0\}$, and $Q \in \Qs$ as well. Here, the subspace $Q$ is finally distributed according to the correct distribution. Towards our soundness analysis, this means that with probability at least $q^{-8\ell^2}$, both provers receive an advice subspace $Q$ which satisfies  $Q \in \Ql$, $Q \cap H_U = \{0\}$, and $Q \in \Qs$, and we fix such a $Q$ henceforth.

Recall that $W_Q$ and $g_{Q,W_Q}$ are the zoom-out and function associated with $Q$. Let $r_2 \leq \frac{10}{\delta}$ be the codimension of $W_Q$. By \eqref{eq: D_Q to agreement} and the fact that $Q$ is smooth, we get
\begin{equation} \label{eq: DQ Lcal measure}
        \D'_{Q}(\Lcal) \geq  q^{-r_2(2\ell-r_1)} \cdot \frac{C}{2}.
\end{equation}
Note that by definition of $\mc{L}$ and $\mc{D}'_Q$, the only $x \in \mc{L}$ which contribute to the measure $\mc{D}'_Q(\mc{L})$ are those such that $\spa_{r_1}(x) = Q$ and $\spa(x) \subseteq W_Q$. It will be convenient to let $\mc{L}'_Q$ be precisely these tuples, and to have the following definition for $\mc{L}'_Q$ on hand:
\begin{equation}
\begin{split}
\mc{L}'_Q &= \{x \in \mc{L} \; | \; \spa_{r_1}(x) = Q,\;  \spa(x) \subseteq W_Q \} \\
&=  \{x \in \left(\Ff_q^{U}\right)^{2\ell} \; | \;  \; \spa_{r_1}(x) = Q,\;  g_{Q, W_Q}|_{\spa(x)} \equiv T_1[\spa(x) \oplus H_U]|_{\spa(x)} \}.
\end{split}
\end{equation}
Note that in the definition of $\mc{L}'_Q$, it is implied that $\dim(\spa(x)) = 2\ell$ as otherwise the entry $T_1[\spa(x) \oplus H_U]$ is not defined.

By our previous discussion, we have 
\begin{equation} \label{eq: D'_Q ineq}  
\D'_Q(\Lcal) = \D'_Q(\Lcal'_Q) \geq  q^{-r_2(2\ell-r_1)}\cdot \frac{C}{2}.
\end{equation}

Now, for a randomly chosen question to the second prover, $V \subseteq U$, conditioned on the advice being $Q$, we would like to bound the probability that the first prover's zoom-out function pair, $(g_{Q,W_Q}, W_Q)$, still has good agreement in the second prover's table. If this is the case, then $g_{Q,W_Q}$ remains a candidate function for the second prover, giving the prover's a chance at winning the game. Let $W_Q[V] = W_Q \cap \Ff_q^V$, and consider $V \subseteq U$ chosen conditioned on the advice subspace being $Q$. Note that this is precisely the marginal distribution of $V$ in $\mc{A}(\cdot, Q)$ considered in \cref{sec: V in U condition on Q}. Moreover, $Q, \mc{L}'_Q,$ and $W_Q$ satisfy the setting of \cref{lm: V preserve}. Hence we can apply \cref{lm: V preserve}, which along with \eqref{eq: D'_Q ineq} yields that with probability at least $\beta^{r_2 + 2} \geq \beta^{10/\delta + 2}$ over the second prover's question $V$ conditioned on $Q \subseteq \Ff_q^V$, we have:
\begin{equation} \label{eq: W[V] good}
 \Pr_{\substack{{x'_i \in W_Q[V], w_i \in H_U,}\\ {x_i = x'_i + w_i, \forall i \in [2\ell]}}}\left[x \in \mc{L}'_Q \; | \; \spa_{r_1}(x) = Q\right] \geq \frac{C}{\frac{80}{\delta} \cdot 2^{10/\delta} q^{30/\delta}} := C'.
\end{equation}

We call such $V$ consistent, so that $V$ is consistent with probability at least $\beta^{10/\delta + 2}$. Now we will show that if $V$ is consistent, the zoom-out $W_Q[V]$ is a candidate for the second prover, in the sense that, the $g_{Q,W_Q}$ also agrees with the second prover's table $T'_1$, for a non-trivial fraction of $L \in \Zoom_{2\ell}[Q, W_Q[V]]$. 




Fix a consistent $V$, so that~\eqref{eq: W[V] good} holds. Let us now bound the fraction of agreement that $g_{Q,W_Q}$ has with the table $T'_1$ inside of $\Zoom_{2\ell}[Q, W_Q[V]]$, by relating this quantity to the probability from \eqref{eq: W[V] good}. We have:
\begin{equation} \label{eq: soundness proof chain of equality}
\begin{split}   
    &\Pr_{L' \in \Grass_q(\Ff_q^V, 2\ell)}[T'_1[L'] \equiv g_{Q,W_Q}|_{L'} \; | \; Q \subseteq L' \subseteq W_Q[V]] \\
    &\qquad\qquad\qquad \geq \Pr_{x'_i \in W_Q[V], L' = \spa(x')}[T'_1[L'] \equiv g_{Q,W_Q}|_{L'}\; | \; Q \subseteq L'] - q^{2\ell-k} \\
    &\qquad\qquad\qquad\geq \Pr_{x'_i \in W_{Q}[V], L' = \spa(x')}[T_1[L'\oplus H_U]|_{L'} \equiv g_{Q,W_Q}|_{L'}\; | \; Q \subseteq L'] - q^{2\ell-2k} - q^{2\ell - k}\\
    &\qquad\qquad\qquad= \Pr_{x'_i \in W_{Q}[V], L' = \spa(x')}[T_1[L'\oplus H_U]|_{L'} \equiv g_{Q,W_Q}|_{L'}\; | \; \spa_{r_1}(x') = Q] \\
    &\qquad\qquad\qquad \qquad- q^{2\ell-2k} - q^{2\ell - k}\\
\end{split}
\end{equation}
We note that in the events of interest for the probabilities above, it is implied that the event of interest only occurs if $\dim(\spa(x'))= 2\ell$ as otherwise the table entries for $T'_1$ and $T_1$ are not defined. In the first transition, we are using the fact that the $x_i'$ are not linearly independent with probability at most $q^{2\ell-k}$. In the second transition we are using the fact that $L' \cap H_U \neq \{0\}$ with probability at most $q^{2\ell-k}$, and that $T'_1[L'] = T_1[L' \oplus H_U]|_{L'}$. In the third transition we are using the fact that the distribution over $L'$ in the probabilities of the third and fourth lines are the same.

Continuing, the probability from the last line is equal to 
\begin{equation} \label{eq: span x' to span x}
    \begin{split}
        &\Pr_{x'_i \in W_{Q}[V], L' = \spa(x')}[T_1[L'\oplus H_U]|_{L'} \equiv g_{Q,W_Q}|_{L'}\; | \; \spa_{r_1}(x') = Q] \\
        &= \Pr_{\substack{{x'_i \in W_{Q}[V], w_i \in H_U}\\{x_i = x'_i + w_i, L = \spa(x)}}}[T_1[L\oplus H_U]|_{L} \equiv g_{Q,W_Q}|_{L}\; | \; \spa_{r_1}(x') = Q] \\
        &=  \Pr_{\substack{{x'_i \in W_{Q}[V], w_i \in H_U}\\{x_i = x'_i + w_i, L = \spa(x)}}}[T_1[L\oplus H_U]|_{L} \equiv g_{Q,W_Q}|_{L}\; | \; \spa_{r_1}(x') \oplus H_U = Q \oplus H_U] \\
        &= \Pr_{\substack{{x'_i \in W_{Q}[V], w_i \in H_U}\\{x_i = x'_i + w_i, L = \spa(x)}}}[T_1[L\oplus H_U]|_{L} \equiv g_{Q,W_Q}|_{L}\; | \; \spa_{r_1}(x) \oplus H_U = Q \oplus H_U] \\
        &= \Pr_{\substack{{x'_i \in W_{Q}[V], w_i \in H_U}\\{x_i = x'_i + w_i, L = \spa(x)}}}[T_1[L\oplus H_U]|_{L} \equiv g_{Q,W_Q}|_{L}\; | \; \spa_{r_1}(x) = Q] \\
        &= \Pr_{\substack{{x'_i \in W_{Q}[V], w_i \in H_U}\\{x_i = x'_i + w_i, L = \spa(x)}}}[x \in \Lcal'_Q \; | \; \spa_{r_1}(x) = Q].
    \end{split}
\end{equation}
The first transition is because the distribution over $L' \oplus H_U$ in the first probability is the same as the distribution over $L \oplus H_U$ in the second probability, and the event is determined by $L \oplus H_U$. The reason is that $g_{Q, W_Q}$ agrees with the side conditions on $H_U$, so $T_1[L' \oplus H_U]|_{L'} \equiv g_{Q, W_Q}|_{L'}$ if and only if $T_1[L' \oplus H_U]|_{L' \oplus H_U} \equiv g_{Q, W_Q}|_{L' \oplus H_U}$, and $T_1[L \oplus H_U]|_{L} \equiv g_{Q, W_Q}|_{L}$ if and only if $T_1[L \oplus H_U]|_{L \oplus H_U} \equiv  g_{Q, W_Q}|_{L \oplus H_U}$. The second transition is because in both the second and third probabilities, the subspace $L \oplus H_U$ is a uniformly random $(2\ell + \dim(H_U))$-dimensional subspace of $W_Q[V] \oplus H_U$ containing $Q \oplus H_U$, and again in both cases the event is determined by $L \oplus H_U$. The third transition is because the conditioning in the two probabilities are the same. The fourth transition follows the same reasoning as the second transition, and the final transition is by definition of $\Lcal'_Q$.

Combining \eqref{eq: W[V] good}, \eqref{eq: soundness proof chain of equality} and \eqref{eq: span x' to span x} we get

\begin{equation*}
\begin{split}
    &\Pr_{L' \in \Grass_q(\Ff_q^V, 2\ell)}[T'_1[L'] \equiv g_{Q,W_Q}|_{L'} \; | \; Q \subseteq L' \subseteq W_Q[V]] 
     \geq C' - q^{2\ell-2k} - q^{2\ell - k} 
    \geq \frac{C'}{8}.
\end{split}
\end{equation*}
By~\cref{lm: zoom out contained in maximal}, there exists some $(W'_Q[V], g_{Q, W'_Q[V]})$ that is $\left(\frac{C'}{8 \cdot 5^{r_2}}, \frac{1}{5} \right)$-maximal and satisfies $W'_Q[V] \supseteq W_Q[V]$, $g_{Q, W'_Q[V]}: W'_Q[V] \xrightarrow[]{} \Ff_q$ is linear, and $g_{Q, W'_Q[V]}|_{W_Q[V]} = g_{Q,W_Q}|_{W_Q[V]}$. Furthermore, $\frac{C'}{8 \cdot 5^{r_2}} \geq \frac{C'}{8 \cdot 5^{10/\delta }}$, so this zoom-out function pair is also $\left(\frac{C'}{8 \cdot 5^{10/\delta}}, \frac{1}{5}\right)$-maximal and can potentially be chosen by the second prover. Applying~\cref{th: bounded zoom-out with zoom-in}, the number of $\left(\frac{C'}{8 \cdot 5^{10/\delta}}, \frac{1}{5}\right)$-maximal zoom-out function pairs containing $Q$ that the second prover chooses from is at most
\[
M = \frac{2560}{\delta} \cdot 5^{20/\delta}\cdot \left(C'\right)^{-2} \cdot  q^{100 \left(t-1\right)! (10/\delta)^{2}\ell \xi^{-1}} \leq q^{\zeta(\delta) \ell},
\]
where $\zeta(\delta)$ is some function depending only on $\delta$. Thus, the second prover chooses $(W'_Q[V], g_{Q, W'_Q[V]})$ with probability at least $\frac{1}{M}$. Finally, if the second prover chooses $(W'_Q[V], g_{Q, W'_Q[V]})$, then the provers succeed if both provers extend their functions, $g_{Q,W_Q}|_{W[V]}$ and $g_{Q, W'_Q[V]}$ in the same manner. This occurs with probability at least $q^{-\codim(W[V])} \geq q^{-10/\delta}$. 

Putting everything together, we get that the provers succeed with probability at least
\[
\frac{\epsilon}{2} \cdot \frac{\delta^2}{101} \cdot q^{-8\ell^2} \cdot \frac{C}{5} \cdot \frac{1}{M} \cdot q^{-10/\delta} = q^{-\zeta(\delta)\cdot \ell^2},
\]
where the first term is the probability that $U \in \Ug$, the second term is the probability that both provers choose the same zoom-in dimension, the third term is the probability that $Q \in \Ql$, $Q \cap H_U = \{0\}$, $Q \in \Qs$, the fourth term is the probability that $V$ is consistent, the fifth term is the probability that the second prover chooses a function that extends $g_{Q, W_Q}$, and the final term is the probability that both provers extend their functions in the same manner. This proves~\cref{lm: soundness outer to inner}. \qedhere

\section{Proofs of the Main Theorems}
\subsection{Proof of~\cref{thm:main}} 
\cref{thm:main} follows by applying our PCP construction
from~\cref{sec:pcp_construct} starting with an instance
of $\GapLin$ from~\cref{th: 3lin hardness}. Fix any $\tau, \eps$, which will parameterize our target completeness and soundness in \cref{thm:main} and let $s$ be the value from \cref{th: 3lin hardness}.

We apply the construction of \cref{sec:pcp_construct} with parameters as follows. Set $q= 2$, set some sufficiently small $\delta > 0$ with respect to $\eps$, set $r = \ceil{10/\delta}$, set $\ell$ sufficiently large relative to $\delta^{-1}$ and $s$, set $c$ sufficiently small relative to $\delta$, and set 
\[
k = q^{2(1+c)\ell} \quad \text{and} \quad q^{-2(1+2c/3)\ell}.
\]
Apply this construction to a $\GapLin[1-\eta, s]$ where the completeness value is set with $\eta \leq \tau/k$. The completeness and soundness of the final PCP are as follows.

If the original ${\sf 3Lin}$ instance is at least $1-\eta$ satisfiable, our final PCP satisfies $\val(\Psi) \geq 1 - k\eta \geq 1-\tau$. This gives the desired completeness.

On the other hand, if the original instance is at most $s$ satisfiable for some constant $s > 0$, then by~\cref{claim:soundness_of_outerpcp}, the value of the outer PCP is at most 
\[
\val(G_{\beta, r}^{\otimes k})\leq 2^{
-\Omega\left((1-s)^2q^{-r + \frac{2c}{3}\ell}\right)} < q^{-\zeta(\delta)O(\ell)^2},
\] 
since we take $\ell$ sufficiently large compared to $\delta^{-1}$. By~\cref{lm: soundness outer to inner} it follows that if the original instance is at most $s$ satisfiable, then ${\sf val}(\Psi)\leq 64q^{-2(1-1000\delta)\ell}$. The proof is 
concluded as the alphabet size of $\Psi$ is $O(q^{2\ell})$.

\subsection{Proof of~\cref{thm:QP}}
To show quasi-NP-hardness for approximate Quadratic Programming, we rely on the following result due to~\cite{ABHKS}, 
who show a reduction from $2$-Prover-$1$-Round Games to Quadratic Programming.

\begin{thm}\label{thm: 2p1r to QP}
    There is a reduction from a 2-Prover-1-Round Games, $\Psi$ with graph $G = (L \cup R, E)$ and alphabets $\Sigma_L, \Sigma_R$ to a Quadratic Programming instance $A$ such that:
    \begin{itemize}
        \item The running time of the reduction and the number of variables in $A$ are both polynomial in $|L| + |R|$ and $2^{|\Sigma_L|}$.
        \item If ${\sf val}(\Psi) \geq 1- \eta$, then ${\sf OPT}(A) \geq 1-\eta-\frac{1}{|L| + |R|}$.
        \item If ${\sf val}(\Psi) \leq \epsilon$, then ${\sf OPT}(A) \leq O(\epsilon)$.
    \end{itemize}
\end{thm}

We are now ready to prove~\cref{thm:QP}.
\begin{proof}[Proof of~\cref{thm:QP}]
    Starting with a ${\sf SAT}$ instance of size $n$,  we take the instance of $\GapLin$ from~\cref{th: sat to 3lin} of size $N \leq 2^{O(\log^2 n)}$ 
    with field size $q = 2$ as the starting point of our reduction. Take $\delta > 0$ to be a small constant, $r = \ceil{10/\delta}$, $k = (\log n)^C$ for a large constant $C$ and pick $c$ sufficiently small relative to $\delta$ and $\ell$ and $\beta$ correspondingly
    so that \eqref{eq: pcp parameters} holds. 
    
    This yields a $2^{O(k\log^2 n)}$-time reduction from ${\sf SAT}$ to a 2-Prover-1-Game on $G = (L \cup R, E)$, with alphabets $\Sigma_L, \Sigma_R$ and the following properties:
\begin{itemize}
    \item $|R| + |L| = O(N^k \cdot q^{3k + 2\ell})$.
    \item $|\Sigma_R| \leq |\Sigma_L| = q^{2\ell} = k^{1/(1+c)}$.
    \item The completeness is at least $1 - k \eta$, where $\eta = 2^{-\Theta(\sqrt{\log n})}$.
    \item The soundness is at most $q^{-2(1-1000\delta)\ell}$.
\end{itemize}
Indeed, the first $3$ properties are clear.
For the soundness, as the original 
${\sf 3Lin}$ instance is at most $1-\eps$ 
satisfiable for $\eps =\Omega( 1/\log^3 N)$, 
we get from~\cref{claim:soundness_of_outerpcp} 
that
\[
\val\left(G_{\beta, r}^{\otimes k}\right)
\leq 
2^{-\Omega\left(\eps^{-2}q^{-r + \frac{2c\ell}{3}}\right)} 
\leq q^{-\zeta(\delta)O(\ell)^2}.
\] 
Then the soundness of the composed PCP is at most $q^{-2(1-1000\delta)\ell}$ by~\cref{lm: soundness outer to inner}.
Applying the reduction of~\cref{thm: 2p1r to QP}, we get a reduction from SAT to a Quadratic Programming instance $A$ such that, 
\begin{itemize}
    \item The running time of the reduction and number of variables in $A$ are both polynomial in
    \[
    M = \poly\left(2^{O(\log^2 n) q^{2(1+c)\ell}} (\log n)^{O(q^{2(1+c)\ell})} 2^{q^{2\ell}}\right).
    \]
    \item If the original SAT instance is satisfiable, then 
    \[
    {\sf OPT} \geq 1 - 2^{-\Omega\left(\sqrt{\log n}\right)}.
    \]
    \item If the original SAT instance is not satisfiable, then 
    \[
    {\sf OPT} \leq O\left(q^{-2(1-1000\delta)\ell}\right).
    \]
\end{itemize}
Note that
\[
\log(M) = q^{2(1+c)\ell}O\left(\log^2 n\right),
\]
whereas the gap between the satisfiable and unsatisfiable cases is $\Omega\left(q^{-2(1-1000\delta)\ell}\right) = \frac{1}{\log(M)^{1-O(\delta)}}$. Altogether, this shows that for all $\eps>0$ there is $C>0$ such that unless ${\sf NP} \subseteq {\sf DTIME}\left(2^{\log(n)^{C}}\right)$, there is no $\log(M)^{1-\eps}$-approximation algorithm for Quadratic Programming on $M$ variables.
\end{proof}

\subsection{Proof of~\cref{thm:2CSPs}}
In this section we prove~\cref{thm:2CSPs}, and 
for that we must first establish 
a version of~\cref{thm:main} 
for biregular graphs of bounded degree.
The proof of this requires minor modifications of our construction, as 
well as the right degree reduction technique of Moshkovitz and Raz~\cite{MR}.
\subsubsection{Obtaining a Hard Instance of Bipartite Biregular $2$-CSP}
We first show that the $2$-Prover-$1$-Round game from~\cref{thm:main} can be transformed into a hard instance of biregular, bipartite $2$-CSP with bounded degrees. This version may be useful for future applications, and is formally stated below. Call a bipartite $2$-CSP $(d_1, d_2)$-regular if the left degrees of its underlying graph are all $d_1$, and the right degrees of its underlying graph are all $d_2$. Throughout this section, for a $2$-CSP, $\Psi$, and an assignment to it, $F$, we will use $\val(F)$ to denote the fraction of constraints in $\Psi$ that $F$ satisfies.
\begin{thm} \label{thm: biregular csp}
    For every $\varphi, \epsilon > 0$, and sufficiently large $R \in \mathbb{N}$, there exist $d_1, d_2 \in \mathbb{N}$ such that given a bipartite $(d_1, d_2)$-regular $2$-CSP, $\Psi$, with alphabet size $R$, it is ${\sf NP}$-hard to distinguish the following two cases:
    \begin{itemize}
        \item Completeness: $\val(\Psi) \geq 1- \varphi$,
        \item Soundness: $\val(\Psi) \leq \frac{1}{R^{1-\epsilon}}$.
    \end{itemize}
\end{thm}
To prove~\cref{thm: biregular csp}, we start with 
an instance $\Psi$ from~\cref{thm:main} and show how to modify it to be biregular, while mostly preserving its value. We will do so in two steps, first making $\Psi$ left-regular, and then applying a generic expander-based transformation (similar to that of Moshkovitz and Raz \cite{MR}) to gain bounded-degree right-regularity. 

Fix $\varphi, \epsilon > 0$, and let $\Psi$ be the $2$-Prover-$1$-Round game constructed for~\cref{thm:main}. Recall that this requires us to choose some large enough $\ell$ relative to $\varphi^{-1}, \epsilon^{-1}$, some large enough $q$ relative to $\ell$, and set $R = q^{2\ell}$. We also set $\delta = \frac{\epsilon}{1000}$, $0 < c$ arbitrarily small relative to $\delta$, and $k = q^{2(1+c)\ell}$, $q=2$. Finally, we construct our $2$-Prover-$1$-Round game from a hard instance of $\GapLin$ with the appropriate completeness and soundness, so that it is {\sf NP}-hard to distinguish between, 
\[
\val(\Psi) \geq 1 - \varphi \quad \text{and} \quad \val(\Psi) \leq \frac{1}{q^{2(1-1000\delta)\ell}} = \frac{1}{R^{1-\epsilon}}. 
\]
It is clear that our $2$-Prover-$1$-Round game can equivalently be viewed as an instance of bipartite 2-CSP, so let us analyze the underlying graph. Let $\U$ denote the set of possible questions to the first prover. Recall that the set of left vertices is,
\[
{ \sf Left} = \{L \oplus H_{U} \; | \;  U \in \U,  L \in \Grass_q(\Ff_q^{U}, 2\ell), L \cap H_U = \{0\} \},
\]
while the set of right vertices is
\[
{\ \sf Right} = \{R \in \Grass_q(\Ff_q^{U}, 2(1-\delta)\ell) \; | \; U \in \U \}.
\]
The constraints of $\Psi$ are on edges of the graph and are weighted. We use
$w(\cdot)$ to denote the weight function, which takes as input edges of the form $(L \oplus H_U, R)$. Recall that the weighting is defined according to the process described in~\cref{sec: constraint graph}. That is, the weight of the constraint on edge $(L \oplus H_U, R)$ is exactly the probability that this edge is chosen according to the process in~\cref{sec: constraint graph}. For a fixed $L \oplus H_U$, let us define
\[
 w_{L \oplus H_{U}}(R) = \frac{|\{L' \oplus H_{U'} \in [L \oplus H_{U}] \; | \; R \subseteq L' \}|}{|[L \oplus H_U]|} \cdot \frac{1}{\qbin{2\ell}{2(1-\delta)\ell}},
\]
which is the probability of choosing the edge $(L \oplus H_U, R)$ conditioned on first choosing $L \oplus H_U$. 
Since we choose $L \oplus H_U \in {\sf Left}$ uniformly, it follows that
\[
w(L\oplus H_U, R) = \frac{ w_{L \oplus H_{U}}(R)}{|{\sf Left}|}. 
\]
 Define the neighborhood of a vertex as, 
\[
{\sf nb}(L \oplus H_{U}) = \{R \in {\sf Right} \; | \; w_{L \oplus H_{U}}(R) > 0\}.
\]
Call $L \oplus H_{U}$ trivial if there is an equation $e \in U$ such that for every basis $x_1, \ldots, x_{2\ell} \in \Ff_q^{U}$ of $L$, we have that, for every $i \in [2\ell]$, the point $x_i$ restricted to the variables in $e$ is of the form $(\alpha, \alpha, \alpha)$ for some $\alpha \in \Ff_q$.

\begin{claim} \label{cl: trivial L}
The fraction of $L \oplus H_{U} \in {\sf Left}$ that are trivial is at most  $2q^{-(2 - 2c)\ell}$.
\end{claim}
\begin{proof}
    Fix a $U \in \U$. Note that it suffices to show that at most $2q^{-(2 - 2c)\ell}$ vertices of the form $L \oplus H_U$ are trivial, as for each $U \in \U$, there are an equal number of vertices $L \oplus H_U$.
    
    Write $U = (x_1,\ldots, x_{3k})$, where the $i$th equation in $U$ contains the variables $x_{3i - 2}, x_{3i-1}, x_{3i}$. Call these three coordinates a block, so that each $x \in \Ff_q^{3k}$ consists of $k$ blocks of consecutive coordinates. Let us bound the fraction of $L$ such that $L \oplus H_U$ is trivial. For $y_1,\ldots, y_{2\ell} \in \Ff_q^{3k}$, let $s(y_1,\ldots, y_{2\ell})$ be the number of blocks where $y_1,\ldots, y_{2\ell}$ are all of the form $(\alpha, \alpha, \alpha)$ for some $\alpha \in \Ff_q$. Then 
    \[
    \Pr_{L}[\text{$L \oplus H_U$ is trivial}] \leq 2\Pr_{y_1,\ldots, y_{2\ell}}[s(y_1,\ldots, y_{2\ell}) = 0],
    \]
    where the factor of $2$ accounts for the probability that either $y_1,\ldots, y_{2\ell}$ are not linearly dependent, or $\spa(y_1,\ldots, y_{2\ell}) \cap H_U \neq \{0\}$. Note that the probability that a specific block is trivial is $q^{-4\ell}$, hence by linearity of expectation we get that
    \[
    \E_{y_1,\ldots, y_{2\ell}}[s(y_1,\ldots, y_{2\ell})] = kq^{-4\ell} = q^{-(2 - 2c)\ell},
    \]
    and therefore 
    \[
    \Pr_{x_1,\ldots,x_{2\ell}}[s(x_1,\dots, x_{2\ell}) \geq 1] \leq  q^{-(2 - 2c)\ell}.\qedhere
    \]
\end{proof}
Now, let $\Psi'$ be the instance obtained from $\Psi$ after removing all trivial $L \oplus H_U$ from $\Psi$. Let ${\sf Left}'$ denote the set of left vertices in $\Psi'$ and let $w'(\cdot)$ denote the weight function over edges in $\Psi'$, which is given by choosing $L \oplus H_U \in {\sf Left}'$ uniformly, and then choosing $R \in {\sf nb}(L \oplus H_U)$ with probability proportional to $w_{L \oplus H_U}(R)$. It follows that,
\begin{equation} \label{eq: new weight}  
w'(L \oplus H_U, R) = \frac{w_{L \oplus H_U}(R)}{|{\sf Left}'|}.
\end{equation}

Since so few vertices are removed when going from ${\sf Left}$ to ${\sf Left}'$, the value of $\Psi$ and $\Psi'$ are roughly the same.

\begin{claim} \label{cl: remove trivial}
We have,
\[
 \val(\Psi) - 2q^{-(2-2c)\ell} \leq \val(\Psi') \leq \val(\Psi) + 3q^{-(2-2c)\ell}.
\]
\end{claim}
\begin{proof}
    Fix any assignment to $\Psi'$ and let $E'$ be the set of constraints that it satisfies. Note that $E'$ is also a set of constraints in $\Psi$. Let $w(E')$ and $w'(E')$ be the sum of the weights of the edges in $E'$ under $w$ and $w'$ respectively. For every $(L\oplus H_U, R) \in E'$, we have the upper bound
    \[
   w'(L \oplus H_U, R) = w(L\oplus H_U, R) \cdot \frac{|{\sf Left}'|}{|{\sf Left}|} \leq w(L\oplus H_U, R) \cdot \frac{1}{1-2q^{-(2-2c)\ell}}
    \]
Since this holds for every edge in $E'$, we get that $w'(E') \leq \frac{w(E')}{1 - 2q^{-(2-2c)\ell}}$, and hence 
\[
\val(\Psi') \leq \frac{\val(\Psi)}{1 - 2q^{-(2-2c)\ell}} \leq \val(\Psi) + 3q^{-(2-2c)\ell}
\]

For the lower bound, fix the any assignment to $\Psi$ and let $E$ be the set of constraints that it satisfies in $\Psi$. Once again let $w(E)$ and $w'(E)$ be the sum of the weights of $E$ under $w$ and $w'$ respectively. If the left endpoint of an edge in $E$ is no longer in ${\sf Left'}$, we define its weight under $w'$ to be $0$. Then note that
\[
w'(E) \geq \sum_{(L \oplus H_U, R) \in E, L \oplus H_U \in {\sf Left'}} w(L\oplus H_U, R) \geq w(E) - 2q^{-(2-2c)\ell}.\qedhere
\]
\end{proof}
\cref{cl: remove trivial} allows us to remove trivial vertices without affecting the value much. Towards going from a weighted $2$-CSP to an unweighted $2$-CSP, we will next bound the size of the neighborhoods in $\Psi'$.
\begin{claim} \label{cl: bounded degree}
    For each non-trivial $L\oplus H_U$ we have $|{\sf nb}(L \oplus H_U)| \leq 10^{k}q^{6k\ell}$.
\end{claim}
\begin{proof}
   Let $U = (x_1,\ldots, x_{3k})= (e_1,\ldots, e_k)$ and suppose equation $e_i$ contains variables $(x_{3i-2}, x_{3i-1}, x_{3i})$. Since $L\oplus H_U$ is not trivial, for each $i$, there must be a point $v \in L \subseteq \Ff_q^U$ such that the values of $v$ restricted to the coordinates of variables $(x_{3i-2}, x_{3i-1}, x_{3i})$ are not all equal. Without loss of generality, say that it is $x_{3i}$ for each $1 \leq i \leq k$. It follows that in order to have
   \[
   L \oplus H_U \oplus H_{U'} = L' \oplus H_U \oplus H_{U'},
   \]
    $U'$ must contain an equation with the variable $x_{3i}$ for each $1 \leq i \leq k$. Let $E_i$ denote this set of equations for each $i$. By the regularity assumptions on our {\sf 3Lin} instance, $|E_i| \leq 10$ and $E_i \cap E_j = \emptyset$ for $i \neq j$. It follows that $U'$ must contain exactly one equation from each $E_i$, and that these form all $k$ equations of $U'$, so there are at most $10^k$ possible $U'$ for which there can exist $L' \subseteq U'$, such that $L' \oplus H_{U'} \in [L \oplus H_U]$. The lemma follows from the observation that $|\Grass_q(3k, 2(1-\delta)\ell)| \leq q^{6k\ell}$.
\end{proof}
Performing the same procedure as in \cite[Lemma 3.4]{KR}, we can turn $\Psi'$ into a bipartite, left-regular $2$-CSP instance, while again not affecting its value too much. Let $Q = 10^{k} q^{6k\ell}$ be the upper bound on neighborhood sizes in~\cref{cl: bounded degree}.

\begin{claim} \label{cl: left regular degree}
For any $C \in \mathbb{N}$, there is a polynomial time algorithm that takes $\Psi'$ as input and outputs a bipartite $2$-CSP $\Psi''$ that is left regular with degree $C \cdot Q$ such that
\[
\val(\Psi') - \frac{1}{C} \leq \val(\Psi'') \leq \val(\Psi') + \frac{1}{C}.
\]
\end{claim}
\begin{proof}
We create $\Psi''$ by doing the following for each vertex $L \oplus H_U \in {\sf Left}'$. Let $R_1, \ldots, R_m$ be the vertices in ${\sf nb}(L \oplus H_U)$. For each $2 \leq i \leq m$, add $\lfloor w_{L \oplus H_U}(R_i) \cdot C \cdot Q \rfloor$ edges from $L \oplus H_U$ to $R_i$. Then, add $C \cdot Q - \sum_{i = 2}^m \lfloor w_{L \oplus H_U}(R_i) C \cdot Q \rfloor$ edges from $L \oplus H_U$ to $R_1$. It is clear that $\Psi''$ is left regular with degree $C \cdot Q$, and that for each $R_i$ with $2 \leq i \leq m$, there are at most $w_{L \oplus H_U}(R_i) C \cdot Q$ edges between $L \oplus H_U$ and $R_i$,  while for $R_1$, there are at most $\left(w_{L \oplus H_U}(R_1) + \frac{1}{C}\right)C \cdot Q $ edges between $L \oplus H_U$ and $R_1$.

To bound the value of $\Psi''$ fix a labeling for it, and note that this is also a labeling for $\Psi'$, since the two CSPs have the same vertex set. For every left vertex, $L \oplus H_U$, if $s$-fraction of its neighboring constraints are satisfied in $\Psi'$ (where this fraction is under the weighting in $\Psi'$), then between $(s- 1/C)$ and $(s+1/C)$-fraction of its neighboring constraints are satisfied in $\Psi''$.
\end{proof}
Applying~\cref{cl: left regular degree} with $C = q^{10\ell}$, we obtain a bipartite $2$-CSP, $\Psi''$, that is left regular with degree $q^{10\ell}Q$, that still has nearly the same completeness and soundness as our original instance $\Psi$. We will now create a $2$-regular bipartite CSP from $\Psi''$, using a technique in the spirit of the right-degree reduction of Moshkovitz and Raz \cite{MR}. We need the expander mixing lemma, stated below.
\begin{lemma}[Expander Mixing Lemma]\label{lem:eml}
    Let $G = (A \cup B, E)$ be a bipartite biregular graph with second largest eigenvalue $\lambda$ and vertex degree $D$. Then for any sets of vertices $A_0 \subseteq A, B_0 \subseteq B$ with sizes $|A_0| = \alpha |A|$ and $|B_0| = \beta|B_0|$, we have
    \[
    \left| \frac{e(A_0, B_0)}{|E|} - \alpha\beta \right| \leq \frac{\lambda}{D} \sqrt{\alpha \beta},
    \]
    where $e(A_0, B_0)$ is the number of edges between the sets of vertices $A_0$ and $B_0$
\end{lemma}

\begin{lemma}  \label{lm: Right Degree Reduction}
    For any parameter $D$, there is a polynomial time algorithm that takes as input, a bipartite, left regular $2$-CSP, $\Psi_0$, with left degree $d_{{\sf left}}$, projection constraints, and left and right alphabets $\Sigma_A$ and $\Sigma_B$ respectively, and outputs a bipartite, $(D \cdot d_{{\sf left}}, D)$-regular $2$-CSP, $\Psi'_0$, with projection constraints and the same alphabets, such that 
    \[
    \val(\Psi_0) \leq \val(\Psi'_0) \leq \val(\Psi_0) + O\left(D^{-1/2}\right).
    \]
    \end{lemma}
\begin{proof}
For every pair of integers $N \geq D$, there is a polynomial time algorithm which constructs a biregular, bipartite expander graph, with $N$ vertices on each side, vertex degree $D$, and second eigenvalue $O(D^{1/2})$. When $N$ is sufficiently large relative to $D$, the algorithm is from \cite{alon2021explicit}, and otherwise, a brute force algorithm suffices.

The construction of $\Psi'_0$ is as follows. Fix the parameter $D$ from the lemma statement and let $A$ and $B$ be the left and right sides of $\Psi_0$ respectively, and let $E$ be its set of edges. We first describe how to construct $A'$ and $B'$, which are the left and right sides of $\Psi'_0$. For each $b \in B$, denote by $A_b$ the set of neighbors of $b$ in $\Psi_0$, and by $d(b) = |A_b|$ the degree of $b$. Then construct a $D$-regular bipartite expander graph, $H_b = (A_b, B_b, E_b)$ such that $|A_b| = |B_b| = d(b)$ with second eigenvalue $O\left(D^{1/2}\right)$. The left vertices of $\Psi'_0$ are unchanged, so $A' = A$, but the right vertices are now $B' = \bigcup_{b \in B} B_b$. The edges are $E' = \bigcup_{b \in B} E_b$, and for each $(a, b') \in E'$, if $b' \in B_b$, then the constraint on $(a,b')$ is the same as that on $(a,b)$ in $\Psi_0$. The alphabets of $\Psi'_0$ are also unchanged. It is clear that $\Psi'_0$ is $(Dd_{{\sf left}}, D)$-regular, so it remains to show the bounds on its value.

For the lower bound, fix an assignment $F$ to $\Psi_0$. Consider the following assignment $F'$ to $\Psi'_0$. For $a' \in A' = A$ set $F'(a) = F(a')$, and for $b' \in B'$ set $F'(b') = F(b)$, where $b$ is the vertex in $B$ such that $b' \in B_{b'}$. It is clear that $F'$ satisfies the same fraction of constraints as $F$, so $\val(\Psi'_0) \geq \val(\Psi_0)$.

For the upper bound, let $F'$ be an assignment to $\Psi'_0$. We show how to derive an assignment $F$ with nearly the same value. For each $b \in B$, choose $b' \in B_b$ uniformly at random and set $F(b) = F'(b')$. To analyze the value of $F$, for each $\sigma \in \Sigma_B$, let 
\[
X_{b, \sigma} = \{a \in A_b \; |\; (a, b) \in E, \text{and $F'(a), \sigma$ satisfy the constraint on $(a,b)$ in $\Psi_0$} \}.
\]
For each $b \in B$, let $Y_{b, \sigma} = \{b' \in B_b \; | \; F'(b) = \sigma \}$. Then,
\begin{align*}
    \E_{F}[\val(F)] = \E_{b \in B} \left[\sum_{\sigma \in \Sigma_B} \frac{|X_{b,\sigma}| |Y_{b,\sigma}|}{d(b)^2} \right],
\end{align*}
where the expectation on the left is over the randomness when choosing assignment $F$ and the expectation on the right chooses $b$ proportional to $d(b)$.

By~\cref{lem:eml} on the graph $H_b$, we have that 
\[
\frac{|X_{b,\sigma}||Y_{b,\sigma}|}{d(b)^2} \geq \frac{e(X_{b,\sigma}, Y_{b,\sigma})}{d(b) \cdot D} -  \frac{O(D^{1/2})}{D} \cdot \frac{\sqrt{|X_{b,\sigma}||Y_{b,\sigma}|}}{d(b)} = \frac{e(X_{b,\sigma}, Y_{b,\sigma})}{d(b) \cdot D} -  O\left( D^{-1/2}\right) \cdot \frac{\sqrt{|X_{b,\sigma}||Y_{b,\sigma}|}}{d(b)}
\]
so 
\begin{align*}
  \E_{F}[\val(F)] &\geq \E_{b \in B} \left[\sum_{\sigma \in \Sigma_B} \frac{e(X_{b,\sigma}, Y_{b,\sigma})}{d(b) \cdot D} \right]  - O\left(D^{-1/2}\right) \cdot  \E_{b \in B} \left[\sum_{\sigma \in \Sigma_B} \frac{\sqrt{|X_{b,\sigma}||Y_{b,\sigma}|}}{d(b)}\right]\\
  &\geq \val(F')  + O\left(D^{-1/2}\right) \cdot\E_{b \in B} \left[\frac{1}{d(b)} \cdot \sqrt{\sum_{\sigma \in \Sigma_B}|X_{b, \sigma}|} \cdot  \sqrt{\sum_{\sigma \in \Sigma_B}|Y_{b, \sigma}|}\right] \\
  &\geq \val(F') -O\left(D^{-1/2}\right).
\end{align*}
When transitioning to the second line, we are using the fact that the first term is exactly $\val(F')$ and we are bounding the second term using the Cauchy-Schwarz inequality. In the transition to the last line we are using the fact that ${\sum_{\sigma \in \Sigma_B}|X_{b, \sigma}|} = d(b)$ because the constraints in $\Psi_0$ are projections, and ${\sum_{\sigma \in \Sigma_B}|Y_{b, \sigma}|} = d(b)$.
\end{proof}

We are now ready to complete the proof of~\cref{thm: biregular csp}.

\begin{proof}[Proof of Theorem~\ref{thm: biregular csp}]
Fix $\varphi, \eps > 0$ as in \cref{thm: biregular csp} and take $\Psi$ to be the $2$-CSP from \cref{thm:main} with $\delta$ therein set to $\varphi/2$, $\eps$ therein set to $\eps/2$, and set $\ell$ and therefore $R$ sufficiently large, so that distinguishing between $\val(\Psi) \geq 1- \varphi/2$ and $\val(\Psi) \leq R^{-(1-\eps/2)}$ is NP-hard. Starting from $\Psi$, obtain $\Psi'$ by deleting trivial vertices and applying the transformation in \cref{cl: left regular degree} with $C = q^{10\ell}$, and then let $\Psi''$ be obtained from $\Psi'$ via the transformation in \cref{lm: Right Degree Reduction} with $D = q^{10\ell}$. It is clear that $\Psi''$ is $(C\cdot Q \cdot D, D)$-regular, has alphabet size $R$, and is constructed from $\Psi$ in polynomial time.

It remains to show that the transformations performed preserve the value.  By \cref{cl: remove trivial} and \cref{cl: left regular degree}, we have
\[
 \val(\Psi) - 2q^{-(2-2c)\ell} - q^{-10\ell} \leq \val(\Psi') \leq \val(\Psi) + 3q^{-(2-2c)\ell} + q^{-10\ell},
\]
and by \cref{lm: Right Degree Reduction}, we have
\[
\val(\Psi') \leq \val(\Psi'') \leq \val(\Psi') + O\left(q^{-5\ell}\right),
\]
so altogether we get that
\[
\val(\Psi)   - 2q^{-(2-2c)\ell} - q^{-10\ell} \leq \val(\Psi'') \leq  \val(\Psi) + 3q^{-(2-2c)\ell} + q^{-10\ell} + O\left(q^{-5\ell}\right).
\]
Since we set $\ell$ and, as a result, $R$ sufficiently large, it is clear that if $\val(\Psi) \geq 1-\varphi/2$, then $\val(\Psi'') \geq 1 - \varphi$, and if $\val(\Psi) \leq R^{-(1-\eps/2)}$ then $\val(\Psi'') \leq R^{-(1-\eps)}$. This establishes \cref{thm: biregular csp}.
\end{proof}

\subsubsection{Sparsification}
In \cite{LeeManurangsi}, Lee and Manurangsi show how to conclude~\cref{thm:2CSPs} from~\cref{thm: biregular csp} via a sparsification procedure. We summarize the steps here. Fix the $\eta > 0$ for~\cref{thm:2CSPs}. Set $\varphi = \epsilon = 0.01 \eta$ in~\cref{thm: biregular csp} and let $\Psi$ be the resulting hard bipartite $(d_1,d_2)$-regular 2-CSP, and it is {\sf NP}-hard to distinguish between the cases
\[
\val(\Psi) \geq 1 - \varphi \quad \text{and} \quad \val(\Psi) \leq \frac{1}{R^{1-\epsilon}},
\]
where $R$ is the sufficiently large alphabet size. It is shown in~\cite[Lemma 10]{LeeManurangsi}, stated below, that the degrees of $\Psi$ can be multiplied by arbitrary constants by simply copying vertices. 

 \begin{lemma} \cite[Lemma 10]{LeeManurangsi} \label{lm: LM 1}
     For any integers $d_1, d_2, c_1, c_2$, there is a polynomial time reduction from a bipartite $(d_1, d_2)$-biregular CSP, $\Psi$, to a bipartite $(c_2d_1d_2, c_1d_1d_2)$-biregular CSP $\Psi'$, such that $\val(\Psi) = \val(\Psi')$, and such that the left and right alphabet sizes are preserved. 
 \end{lemma}

It is then shown in~\cite[Theorem 11]{LeeManurangsi}, stated below, that one can perform a subsampling procedure to $\Psi'$ that significantly lowers the degree.
\begin{thm} \cite[Theorem 11]{LeeManurangsi} \label{th: LM}
    For any $0 <\nu_1 < \nu_2 \leq 1$, any positive integer $C$, and any sufficiently large positive integers $d_A, d_B \geq d_0(\varphi, \nu)$, and $R \geq R_0(\delta, \nu, d_A, d_B)$, the following holds: there is a randomized polynomial-time reduction from a bipartite $(d_A C, d_B C)$-biregular $2$-CSP, $\Psi'$, with alphabet size at most $R$, $(d_A, d_B)$-bounded degree $2$-CSP, $\Psi''$, such that, with probability at least $2/3$ the following two items hold
    \begin{itemize}
        \item Completeness: $\val(\Psi'') \geq \val(\Psi') - \nu_1$,
        \item Soundness: If $\val(\Psi') \leq \frac{1}{R^{\nu_2}}$, then $\val(\Psi'') \leq \frac{1}{\nu_2 - \nu_1}\left(\frac{1}{d_A} + \frac{1}{d_B} \right)$ 
    \end{itemize}
\end{thm}
Putting everything together, we can prove~\cref{thm:2CSPs}.

\begin{proof}[Proof of~\cref{thm:2CSPs}]
Recall the values $\eta$ and $d$ from~\cref{thm:2CSPs}. Start with an instance $\Psi$ of $2$-CSP from~\cref{thm: biregular csp} with $\varphi = \epsilon = 0.01\eta$ and sufficiently large alphabet size $R$. Then $\Psi$ is $(d_1, d_2)$-biregular, with sufficiently large alphabet size $R$ relative to $\varphi^{-1}, \epsilon^{-1},$ and $d$. For such a $\Psi$, it is {\sf NP}-hard to distinguish whether $\val(\Psi) = 1-0.01\eta$, or $\val(\Psi) \leq \frac{1}{R^{1-0.01\eta}}$.

Applying~\cref{lm: LM 1} with $c_1 = c_2 = d$ yields, in polynomial time, a $(dd_1d_2, dd_1d_2)$-biregular $2$-CSP, $\Psi'$, with alphabet size $R$ and satisfying $\val(\Psi') = \val(\Psi)$. Next, applying~\cref{th: LM}, with $d_A = d, d_B = d, C = d_1d_2, \nu_1 = 0.01\eta, \nu_2 = 1- \epsilon$, we get (after randomized polynomial time reduction) a $2$-CSP $\Psi''$ with degree at most $d$ such that:
\begin{itemize}
    \item If $\val(\Psi) \geq 1 - \varphi$, then $\val(\Psi'') \geq 1-\varphi - \nu_1 = 1-0.02\eta$.
    \item If $\val(\Psi) \leq \frac{1}{R^{1-\epsilon}}$, then $\val(\Psi'') \leq \frac{1}{1-\epsilon - \nu_1}\left( \frac{1}{d} + \frac{1}{d}\right) = \frac{1}{1-0.02\eta}\cdot \frac{2}{d}$. 
\end{itemize}
Finally note that,
\[
\frac{1-0.02\epsilon}{ \frac{1}{1-0.02\epsilon}\cdot \frac{2}{d}} \geq d\left(\frac{1}{2}-\epsilon \right).
\]
Thus, by~\cref{thm: biregular csp} and the randomized polynomial time reduction above, it follows that unless $\text{NP} \subseteq \text{BPP}$, there is no polynomial time $d\left(\frac{1}{2}-\eta \right)$ approximation algorithm for 2-CSP with degree at most $d$.
\end{proof}

\subsection{Proof of~\cref{thm:rooted_conn}}
Combining our $2$-Prover-$1$-Round Game with the reductions in \cite{laekhanukit2014parameters} we obtain improved hardness of approximation results for Rooted $k$-connectivity on undirected graphs, the vertex-connectivity survivable network design problem, and the vertex-connectivity $k$-route cut problem on undirected graphs. We give an overview of the reduction here. It will be more convenient to refer to $2$-CSPs rather than $2$-Prover-$1$-Round Game for the remainder of this subsection.

In the reduction of \cite{laekhanukit2014parameters} one uses a minimization variant for a $2$-CSP $\Psi$, called Min-Rep. The input in Min-Rep is still a $2$-CSP, but now each alphabet symbol is assigned a cost, and a labeling of $\Psi$ is allowed to assign multiple alphabet symbols to each. The cost of a labeling is the total cost of the alphabet symbols used over all vertices and $\minrep(\Psi)$ is the minimum cost labeling such that every constraint of $\Psi$ is satisfied. Here, a constraint on the edge $(u,v)$ is satisfied as long as one pair of labels from those assigned to $u$ and $v$ respectively satisfy the constraint. We refer to \cite[page 1630]{laekhanukit2014parameters} for a more formal definition. The following lemma \cite{laekhanukit2014parameters} describes approximation preserving reductions from Min-Rep to various connectivity problems.

\begin{thm}\label{thm: lae} \cite[Theorem 3.1]{laekhanukit2014parameters}
There are polynomial-time, approximation-preserving reductions which take an instance of minimum cost label cover with degree $d$ and alphabet size $R$, and output:
\begin{itemize}
    \item an instance of rooted $k$-connectivity on undirected graphs with $k = O(d^3 \cdot R + d^4)$,
    \item an instance of vertex-connectivity survivable network design on undirected graphs with maximum requirement $k = O(d \cdot R + d^2)$,
    \item an instance of vertex-connectivity $k$-route cut on undirected graphs with $k = O(d\cdot R)$
\end{itemize}
\end{thm}
To obtain hardness of approximation results from \cref{thm: lae}, we must first obtain a hardness of approximation result for Min-Rep. This can be done by using the following result due to~\cite{Manurangsi19}.

\begin{lemma}\label{lm: max to min} \cite[Lemma 7]{Manurangsi19} \label{lm: transform to minrep}
For any $\eps, \gamma > 0$, there is a polynomial time randomized reduction that, given any $2$-CSP with projection constraints
$\Psi$, alphabet size $R$, and size $N$, outputs a $2$-CSP with projection constraints $\Psi'$ with the same alphabet and maximum degree at most 
\[
d := 10^6 \left( \frac{2 \log (2 R )}{\sqrt{\gamma}} \right)
\]
such that
\begin{itemize}
    \item (\textbf{Completeness}) if $\val(\Psi) \geq 1 - \eps$, then $\minrep(\Psi') \leq (1 + \varepsilon \cdot d)N$ with probability $0.9$, and,
    \item (\textbf{Soundness}) if $\operatorname{val}(\Pi) < \gamma$, then $\operatorname{Min\text{-}Rep}(\Pi') > \left(\frac{0.06}{\sqrt{\gamma}}\right) N$ with probability $0.9$.
\end{itemize}
\end{lemma}
Using \cref{lm: max to min}, we get the following hardness of approximation result for Min-Rep.
\begin{thm} \label{thm: min-rep hardness}
For every $\eps' > 0$ and sufficiently large $g > 0$, given a $2$-CSP, $\Psi$ with maximum degree $O(g \log g)$ and alphabet size $O(g^{2+\eps'})$, it is NP-hard under randomized reduction to approximate the $\minrep(\Psi)$ to within a factor of $g$.
\end{thm}
\begin{proof}
    Fix sufficiently large $g > 0$ and take $\Psi$ from \cref{thm:main} with $\delta = \frac{1}{g}$, $\eps$ such that $1/(1-\eps) = 1+\eps'/2$, and alphabet $R$ such that the value in the no case is $\gamma := R^{-(1-\eps)} = 2^{-11}g^{-2}$. Let $N$ denote the size of $\Psi$. Then applying \cref{lm: transform to minrep} we get that if $\val(\Psi) \geq 1- \delta$, then $\minrep(\Psi') \leq 2N$, while if $\val(\Psi) \leq \gamma$, then $\minrep(\Psi') >  \left(\frac{0.06}{\sqrt{\gamma}}\right) N$. Therefore, it is NP-hard under randomized reduction to approximate $\minrep(\Psi')$ within factor, 
    \[
    0.03 \cdot \gamma^{-1/2} \geq g.
    \]
    The alphabet size of $\Psi'$ is $R = O(g^{2/(1-\eps)}) = O(g^{2+\eps'})$ and the maximum degree is at most $d = O(g \log g)$.
\end{proof}

Finally, by combining \cref{thm: lae} with \cref{thm: min-rep hardness}, we can show \cref{thm:rooted_conn}.

\begin{proof}[Proof of \cref{thm:rooted_conn}]
    Fix $\eps > 0$ and $g > 0$ sufficiently large, and let $\Psi$ be the $2$-CSP from \cref{thm: min-rep hardness} such that $\minrep(\Psi)$ is hard to approximate within factor $g$, and $\Psi$ has maximum degree $d = O(g\log g)$ and alphabet size $R = O(g^{2+\eps})$. By \cref{thm: lae}, there is a polynomial-time algorithm that produces instances of the following problems which are hard to approximate within factor $g$:
    \begin{itemize}
        \item rooted $k$-connectivity on undirected graphs with $k = O(g^{5+\eps}\log g)$,
        \item vertex-connectivity survivable network design on undirected graphs with maximum requirement $k = O(g^{3+\eps} \log g)$,
        \item vertex-connectivity $k$-route cut on undirected graphs with $k = O(g^{3 + \eps}\log(g))$.
    \end{itemize}
    Expressing $g$ in terms of $k$ in each of the three cases gives the desired result.
\end{proof}
\section{Bounding the Number of Successful Zoom-outs of a Fixed Codimension} \label{sec: bounded zoom-outs}
The goal of this section is to prove~\cref{lm: bound on good zoom outs}. Let us recall some context first. Throughout the section, we work in the second prover's space, $\Ff_q^V$, where $V$ is some question to the second prover (in the outer PCP). 

Accordingly, we make the assumption that $n:= |V| \gg \ell$, say $n \geq 2^{100}q^{\ell}$ to be concrete. Also, as $n = |V|$, we will write the ambient space as $\Ff_q^n$ from now on. We fix 
$T$ to be a table that assigns, to each $L \in \Grass_q(n, 2\ell)$, a linear function on $L$. For ease of notation, we define $\mc{L} := \Grass_q(n, 2\ell)$ throughout this section. Now, let us review the set up of~\cref{lm: bound on good zoom outs}. Recall that we set 
\[
\xi := \delta^{5}, \quad \delta_2 := \xi/100, \quad t := \left(2^{2+10/\delta_2}\right)! \; .
\]
Let $\mathcal{S} = \{W_1, \ldots, W_N\}$ be a set of codimension $r$-subspaces in $\Ff_q^n$ of size 
\[
N \geq q^{100(t-1)! r^2 \ell \xi^{-1}},
\]
where $r \leq \frac{10}{\delta}$. For each $W_i$, let $f_i: W_i \xrightarrow[]{} \Ff_q$ be a linear function such that $f_i|_L \equiv T[L]$ for at least $C$-fraction of the $2\ell$-subspaces $L \subseteq W_i$, where $C \geq q^{-2(1-\xi)\ell}$, and $\xi > 0$.  

\subsection{Step 1: Reducing to a Generic Set of Subspaces} \label{sec: Generic Subspace Proofs}
Instead of working with all of $\mc{S}$, we will want to use only a subset of $\mc{S}$ which is generic. This reduction to a generic set of subspaces allows us to readily use the sampling lemmas from \cref{sec: sampling lemmas}. Applying~\cref{lm: generic}, with parameter $t$ as define, we get that there exists a subspace $W'_{\amb} \subseteq \Ff_q^n$ and a set of 
\[\
m_1 \geq \frac{N^{\frac{1}{(r+1) \cdot (t-1)!}}}{q^{3r}} \geq q^{75r\ell \xi^{-1}}
\]
subspaces $\mathcal{W} = \{W_1, \ldots, W_{m_1}\} \subseteq \mathcal{S}$, such that
    \begin{itemize}
        \item Each $W_i \in \mathcal{W}$ is contained in $W'_{\amb}$ and has codimension $s$ with respect to $W'_{\amb}$, where $s \leq r$.
        \item $\mathcal{W}$ is $t$-generic with respect to $W'_{\amb}$.
    \end{itemize}
We remark that this subspace $W'_{\amb}$ will ultimately be the one used for~\cref{lm: bound on good zoom outs} (as the subspace ``$A$'' there). The remainder of the proof is devoted to finding the appropriate linear function $h': W'_{\amb} \xrightarrow[]{} \Ff_q$, and the set $\W'$ as required by~\cref{lm: bound on good zoom outs}. The set $\W'$ we ultimately find will be a subset of $\mc{W}$ above.
\subsection{Step 2: Finding Local Agreement}
For a subspace $X$ and linear assignment to $X$, $\sigma \in \Ff_q^{X}$, let 

\[\Lc_X = \Zoom_{2\ell}[X, W'_{\amb}] \quad \text{and} \quad \Lc_{X, \sigma} = \{L \in \Lc_X \; | \; T[L]|_X = \sigma \}.\]
Likewise, define 
\[\W_X = \{W_i \in \W \; | \; X \subseteq W_i \} \quad \text{and} \quad \W_{X, \sigma} = \{W_i \in \W_X \; | \; f_i|_X = \sigma \}.
\]
We will use $\mu_{X,\circ}(\cdot)$ to denote the uniform measure over $\mc{L}_X$. The first step of our proof is to find sets $\W_{X, \sigma}$ and $\Lc_{X,\sigma}$ that have strong agreement between them, in the sense of the following lemma. The approach of this first step is similar to that of \cite{IKW, BDN, MZ}. Fix $\gamma>0$ to be a small constant, say $\gamma = 10^{-6}$.
\begin{restatable}{lemma}{excellent}\label{lm: excellent}
    There exists a $2\left(1 - \frac{\xi}{2}\right)\ell$-dimensional subspace $X$ and a linear assignment, $\sigma$, to $X$, such that the following hold: 
    \begin{itemize}
        \item $\mu_{X,\circ}(\Lc_{X,\sigma}) \geq \frac{C}{6}$.
        \item $|\W_{X,\sigma}| \geq \frac{m_1}{q^{10r\ell}}$.
        \item Choosing $L \in \Lc_{X, \sigma}$ uniformly, and $W_i \in \W_{X,\sigma}$ uniformly such that $W_i \supseteq L$, we have
    \[
    \Pr_{L, W_i}[f_i|_L \not \equiv T[L]] \leq 5\gamma.
    \]
    \end{itemize}    
\end{restatable}
\begin{proof}
    Deferred to~\cref{sec:lem8.1}.
\end{proof}
We next state \cref{cor: excellent}, which is an immediate consequence of \cref{lm: excellent}. There are two differences between the two statements though. First, in  \cref{cor: excellent} we require the third condition to hold for every $L \in \Lc_{X,\sigma}$ (instead of for a random $L$ as in \cref{lm: excellent}). Second we require every $L \in \Lc_{X,\sigma}$ to be contained in roughly the same number of $W_i \in \W_{X,\sigma}$. 

\begin{corollary} \label{cor: excellent}
    Taking $\Lc_{X,\sigma}$ and $\W_{X,\sigma}$ from~\cref{lm: excellent}, there is a subset $\Lc'_{X,\sigma} \subseteq \Lc_{X,\sigma}$ such that the following hold. 
    \begin{itemize}
        \item $\mu_{X,\circ}(\Lc'_{X,\sigma}) \geq \frac{C}{12}$
        \item $m_2 := |\W_{X,\sigma}|  \geq \frac{m_1}{q^{10r \ell}}$.
        \item For every $L \in \Lc'_{X, \sigma}$, choosing $W_i \in \W_{X,\sigma}$ uniformly such that $W_i \supseteq L$, we have
    \[
    \Pr_{W_i \supseteq L, W_i \in \mathcal{W}_{X,\sigma}}[f_i|_L \not \equiv T[L]] \leq 12\gamma.
    \]
    \item For every $L \in \Lc'_{X,\sigma}$, 
    \[
     0.95 \cdot |\W_{X,\sigma}|\cdot q^{-\xi \ell \cdot s} \leq N_{\W_{X,\sigma}}(L) \leq 1.05 \cdot |\W_{X,\sigma}|\cdot q^{-\xi \ell \cdot s}.
    \]
    \end{itemize}    
    Recall that $N_{\W_{X,\sigma}}(L)$ is as defined in \eqref{eq: def N_W(L)}, and $s$ in the fourth item is the codimension of each subspace in $\mc{W} \supseteq \mc{W}_{X, \sigma}$.
\end{corollary}
\begin{proof}
    Take $X, \sigma, \Lc_{X,\sigma},$ and $\W_{X,\sigma}$ as guaranteed by~\cref{lm: excellent}, so that $\mu_{X,\circ}(\Lc_{X,\sigma}) \geq \frac{C}{6}$.  We will keep the same $X, \sigma$, but we remove some $L$'s from $\Lc_{X,\sigma}$ to make the third and fourth items hold.
    
    By Markov's inequality, at most $\frac{5}{12}$-fraction of $L \in \mathcal{L}_{X,\sigma}$ violate the third item. By~\cref{lm: deviation bounds} applied to $\W_{X,\sigma}$ with parameters $Q = X$, $j = 2\ell$, $a = 2\left(1 - \frac{\xi}{2}\right) \ell$, $r = s$, and $c = 0.05$, we have that 
    \begin{align*}     
    &\Pr_{L \in \Zoom_{2\ell}[Q, V]}\left[\left|N_{\mc{W}_{X,\sigma}}(L) - q^{-\xi \ell \cdot s} m_2\right| \geq 0.1 q^{-\xi \ell } m_2\right] \\
    &\leq  \Pr_{L \in \Zoom_{2\ell}[Q, V]}\left[\left|N_{\mc{W}_{X,\sigma}}(L) - \beta\right| \geq 0.05 \beta \right]\\
    &\leq \frac{405q^{\xi\ell \cdot s}}{m_2} \\
    &\leq \frac{C}{100},
        \end{align*}
        where in the above we use $\beta = \E_{L \in \Zoom_{2\ell}[Q,V]}[N_{\mc{W}_{X,\sigma}}(L)]$ for convenience, and note from \cref{lm: p_1 bound} that $|\beta - q^{-\xi\ell \cdot s}m_2| \leq 0.01 q^{-\xi\ell \cdot s}m_2$.

     It follows that after removing the $L \in \Lc_{X,\sigma}$ that do not satisfy the third or fourth condition, we arrive at the desired $\Lc'_{X,\sigma}$, which still has measure  at least 
    \[
    \frac{7}{12}\cdot \frac{C}{6} - \frac{C}{100} \geq \frac{C}{12},
    \]
    under $\mu_{X,\circ}(\cdot)$.
\end{proof}

For the rest of the argument, let us fix $X,\sigma$ as well as $W_{X,\sigma}$ and $\Lcal'_{X,\sigma}$ to be 
as in~\cref{cor: excellent}.
\subsection{Step 3: A Global Set with Local Agreement}


The next step is to further refine the set $\Lcal'_{X,\sigma}$ so that the remaining subspaces  ``evenly cover'' a subspace of $W^\star_{\amb} \subseteq W'_{\amb}$ with codimension $\dim(X) + O_{\delta_2}(1)$. To do this, we will reduce to the case where $\Lcal'_{X,\sigma}$ is pseudo-random within some zoom-in $A$ and zoom-out $B$ such that $X \subseteq A \subseteq B$. This is done via the following argument, which can be informally described as:
\begin{enumerate}
    \item While $\Lc'_{X,\sigma}$ is not pseudo-random, there must be some zoom-in or zoom-out on which its measure is significantly higher than its overall measure over the whole space, so consider the restriction to this zoom-in or zoom-out.
    \item This increases the measure of $\Lc'_{X,\sigma}$, and we may repeat until we have a pseudo-random set (within some zoom-in, zoom-out combination).
    \item By choosing the pseudo-randomness parameters suitably, we are able to perform the above process in relatively few times until the restriction of $\Lc_{X,\sigma}'$ that 
    we arrive at is pseudo-random.
    \item As a result, when restricting to the zoom-in, zoom-out combination, the resulting set of subspaces evenly covers a space that is still relatively large in the sense that it contains $q^{-O(\xi^{-1})}$-fraction of the present space (i.e.\ the space obtained after iteratively restricting to zoom-ins and zoom-outs as described in steps 1-3). In particular, this present space is described as $W_{\amb,0}$ below.
\end{enumerate}

We now move on to the formal statement. For a zoom-in $A$ and zoom-out $B$ such that $X \subseteq A \subseteq B \subseteq W'_{\amb}$, write $W'_{\amb} = A \oplus W_{\amb,0}$ and $B = A \oplus W^\star_{\amb}$, where $W^\star_{\amb} \subseteq W_{\amb,0}$. Also define
\[
\W^{\star}_{[A,B]} = \{W^\star_i \; | \; \exists W_i \in \W_{X,\sigma} \text{ s.t } A \oplus W^\star_i = W_i \cap B \}.
\]
It is clear that each $W^\star_i \in \W^{\star}_{[A,B]}$ is contained inside of some $W_i \in \W_{X,\sigma}$, so for each $W^\star_i$ we may associate a function $f^\star_i$ to each $W^\star_i$ as follows. Choose an arbitrary $W_i \in \mc{W}_{X,\sigma}$ such that $A \oplus W^\star_i = W_i \cap B$, and define $f^\star_i := f_i|_{W^\star_i}$.


\begin{restatable}{lemma}{findglobal} 
\label{lm: find global 2}
Keeping the notation above, there is a zoom-in $A$ and a zoom-out $B$ inside of the space $W'_{\amb}$ such that the following holds. There exists a collection of subspaces $\W^{\star} = \{W^{\star}_1, \ldots, W^\star_{m_3} \} \subseteq \W^\star_{[A,B]}$ of codimension $s$ with respect to $W^\star_{\amb}$, a dimension $\ell' \geq \frac{\xi}{3}\ell$, and a set $\mc{L}^\star \subseteq \Grass_q(W^\star_{\amb}, \ell')$, such that the following items hold for the table $T_{\ell'}$, which assigns linear functions to $L \in \Grass_q(W^\star_{\amb}, \ell')$ as follows
    \[
    T_{\ell'}[L] :\equiv T[A \oplus L]|_{L}.
    \]
    \begin{enumerate}
        \item $\mu(\Lcal^{\star}) := \eta \geq \frac{C}{12}$, where the measure here is over $\Grass_q(W^\star_{\amb}, \ell')$.
        \item The set $\Lcal^{\star}$ is $(1, q^{\delta_2 \ell}\eta)$-pseudo-random.
        \item Each $W^\star_i$ has codimension $s \leq r$ inside of $W^\star_{\amb}$, $\W^{\star}$ is $4$-generic, with respect to $W^\star_{\amb}$, and $W^\star_{\amb}$ has codimension at most $10/\delta_2$ with respect to $W'_{\amb}$.
        \item The size of $\W^\star$ satisfies
        \[
        \frac{m_2\cdot q^{-10s/\delta_2}}{2} \leq m_3 \leq m_2.
        \]
        \item For each $W^\star_i \in \W^\star_i$, there is a linear function $f^\star_i: W^\star_i \to \Ff_q$ such that the following holds. For every $L \in \Lc^{\star}$, choosing $W^\star_i \in \W^{\star}$ uniformly such that $W^\star_i \supseteq L$, we have
    \[
    \Pr_{W^\star_i \supseteq L, W^\star_i \in \W^{\star}}[f^\star_i|_L \not \equiv T_{\ell'}[L]] \leq 14\gamma.
    \]
    \item For every $L \in \mathcal{L}^{\star}$,
    \[
    0.8 \cdot p_1\cdot m_3  \leq N_{\W^\star}(L) \leq 1.2 \cdot p_1 \cdot m_3,
    \]
    where $N_{\W^\star}(L)$ is as defined in \eqref{eq: def N_W(L)}, and 
    \[
    p_1 := \Pr_{L \in \Grass_q(W^\star_{\amb}, \ell')}[L \subseteq W],
    \]
    for an arbitrary $W \subseteq W^\star_{\amb}$ of codimension $s$.
    \end{enumerate}
\end{restatable}
\begin{proof}
    Deferred to~\cref{app: proof of find global}.
\end{proof}

As a consequence of pseudo-randomness, we may apply~\cref{lm: fourier even covering}, to get that $\Lc^{\star}$ evenly covers $W^\star_{\amb}$. 

\begin{lemma} \label{lm: even covering}
Setting $Z := \{z \in W^{\star}_{\amb} \; | \; |\mu_{z, \circ}(\Lc^{\star}) - \eta| \leq \frac{\eta}{10} \}$, we have that,
\begin{equation*}
    \frac{|Z|}{|W^\star_{\amb}|} \geq 1 - q^{-\frac{\ell'}{2}}.
\end{equation*}
\end{lemma}
\begin{proof}
    This is immediate by the pseudo-randomness of $\Lc^\star$ and~\cref{lm: fourier even covering}.
\end{proof}
\subsection{Step 4: Local to Global Agreement}
The following lemma establishes agreement, on average, between functions $f^\star_i, f^\star_j$. Specifically it shows that a random pair of functions agrees almost entirely in $W^\star_i \cap W^\star_j \cap Z$.
\begin{restatable}{lemma}{globalagreement} \label{lm: Wi agreement}
Let $\W^\star = \{W^\star_1, \ldots, W^\star_{m_3}\}$ and $f^\star_1,\ldots, f^\star_{m_3}$  be the $4$-generic set of codimension $s \leq r \leq 10/\delta$ subspaces inside of $W'_{\amb}$ and associated linear functions obtained from \cref{lm: find global} respectively. We have
    \[
    \Pr_{\substack{W^\star_i, W^\star_j \in \W^{\star}\\ z \in W^{\star}_i \cap W^\star_j \cap Z}}[f^\star_i(z) \neq f^\star_j(z)] \leq 500\gamma,
    \]
and for every $W^\star_i, W^\star_j \in \W^\star$ we have $|W^\star_i \cap W^\star_j \cap Z| \geq 0.81 \cdot |W^\star_i \cap W^\star_j|$. Here $Z \subseteq W^\star_{\amb}$ is as defined in \cref{lm: even covering}.
\end{restatable}
\begin{proof}
    Deferred to~\cref{app: Wi Agreement}.
\end{proof}

Using~\cref{lm: Wi agreement} we conclude the proof of~\cref{lm: bound on good zoom outs} by using ideas from the Raz-Safra analysis of the Plane versus Plane test~\cite{RS}. Define a graph, $G$ with vertex set $\W^\star$, where $W^\star_i, W^\star_j$ are adjacent if and only if $f^\star_i|_{W^\star_i \cap W^\star_j} = f^\star_j|_{W^\star_i \cap W^\star_j}$. We claim that this graph contains a large clique, and towards this end we first show that the graph is nearly transitive. For a graph $H = (V, E)$, define
\[
\beta(H) = \max_{(u,w) \notin E} \Pr_{v \in V}[(v,u), (v,w) \in E].
\]
A graph $H$ is transitive if $\beta(H) = 0$. It is easy to see that transitive graphs are (edge) disjoint unions of cliques. The following lemma, proved in \cite{RS}, 
asserts that if $H$ is relatively dense and $\beta(H)$ is small, 
then one could remove only a small fraction
of the edges and get a fully transitive graph.

\begin{lemma} \cite[Lemma 2]{RS}\label{lm: transitive delete edges}
   Any graph $H$ on $n$ vertices can be made transitive by deleting at most $\ceil{3 \sqrt{\beta(H)}\cdot n^2}$ edges.
\end{lemma}

To use \cref{lm: transitive delete edges} we first show that the graph $G$ 
we defined is highly transitive.
\begin{claim} \label{cl: transitive}
    We have $\beta(G)\leq \frac{1}{m_3}$.
\end{claim}
\begin{proof}
    Fix a $W^\star_i, W^\star_j$ that are not adjacent. We claim that they can have at most $1$ common neighbor. Suppose for the sake of contradiction that $W^\star_a, W^\star_b$ are distinct common neighbors. Then,
    \[
    f^\star_i|_{W^\star_i \cap W^\star_j \cap W^\star_a} = f^\star_j|_{W^\star_i \cap W^\star_j \cap W^\star_a},
    \qquad
    f^\star_i|_{W^\star_i \cap W^\star_j \cap W^\star_b} = f^\star_j|_{W^\star_i \cap W^\star_j \cap W^\star_b}.
    \]
    It follows that $f^\star_i$ and $f^\star_j$ agree on $W^\star_i \cap W^\star_j \cap W^\star_a + W^\star_i \cap W^\star_j \cap W^\star_b$. However, since $\W^{\star}$ is $4$-generic, we have
    \[
    \codim(W^\star_i \cap W^\star_j \cap W^\star_a + W^\star_i \cap W^\star_j \cap W^\star_b) \leq  3s + 3s - 4s = 2s,
    \qquad
    \codim(W^\star_i \cap W^\star_j) = 2s,
    \]
    and
    \[
    W^\star_i \cap W^\star_j \cap W^\star_a + W^\star_i \cap W^\star_j \cap W^\star_b \subseteq W^\star_i \cap W^\star_j,
    \]
    so it must be the case that  $W^\star_i \cap W^\star_j \cap W^\star_a + W^\star_i \cap W^\star_j \cap W^\star_b = W^\star_i \cap W^\star_j$. This contradicts the assumption that $W^\star_i$ and $W^\star_j$ are not adjacent. Thus, any two non-adjacent vertices can have at most $1$ common neighbor, and the result follows.
\end{proof}

\begin{claim}\label{cl: clique}
    The graph $G$ contains a clique of size at least $\frac{m_3}{4}$
\end{claim}
\begin{proof}
    Applying Markov's inequality to~\cref{lm: Wi agreement}, we have that with probability at least $9/10$ over $W^\star_i$ and $W^\star_j$, we have both $\Pr_{z \in W^\star_i \cap W^\star_j \cap Z}[f^\star_i(z) \neq f^\star_j(z)] \leq 10001\gamma$ and $|W^\star_i \cap W^\star_j \cap Z| \geq 0.81 \cdot |W^\star_i \cap W^\star_j|$. In this case, $f^\star_i$ and $f^\star_j$ agree on at least $(1-10001\gamma)$-fraction of the points in $W^\star_i \cap W^\star_j \cap Z$, which is in turn at least $(1-10001\gamma)\cdot 0.81 > 1/q$-fraction of the points in $W^\star_i \cap W^\star_j$. As $f^\star_i$ and $f^\star_j$ are linear functions, the Schwartz-Zippel lemma implies that such $W^\star_i, W^\star_j$ are adjacent in $G$. Hence, $G$ has at least $9\binom{m_3}{2}/10$ edges. 
    
    By~\cref{cl: transitive,lm: transitive delete edges}, we can delete $\ceil{3m_3^{3/2}}$ edges to make $G$ a union of cliques. Doing so yields a graph on $m_3$ vertices with at least $m_3^2/4$-edges that is a union of cliques. Let $C_1, \ldots, C_{N}$ be the cliques, with $C_1$ being the largest one. We have, 
    \[
    |C_1|\cdot m_3 \geq |C_1| \cdot \sum_{i=1}^N |C_i| \geq \sum_{i=1}^N |C_i|^2 \geq \frac{m_3^2}{4}.
    \]
    It follows that $|C_1| \geq \frac{m_3}{4}$, and that $G$ contains a clique of size at least $\frac{m_3}{4}$.
\end{proof}

Let $\C$ be the clique guaranteed by~\cref{cl: clique} and write $\C = \{W^\star_1, \ldots , W^\star_{\frac{m_3}{4}}\}$. To complete the proof of~\cref{lm: bound on good zoom outs} we will find a linear $h$, such that for all $1 \leq i \leq \frac{m_3}{4}$, $f^\star_i|_{V^{\star}} \equiv h|_{W^\star_{\amb}}$, and then show that this $h$ can be extended to $X \oplus A \oplus V^{\star}$ in a manner that is consistent with many of the \emph{original} $f_i$'s for $1 \leq i \leq m_3/4$. To this end, first define $g: V^{\star} \xrightarrow[]{} \Ff_q$ as follows:

  \begin{equation}
    g(x) =
    \begin{cases}
      f^\star_i(x), & \text{if }\ \exists W^\star_i \in \C, x \in W^\star_i \\
      0, & \text{otherwise.}
    \end{cases}
  \end{equation}
Since $f^\star_i(x) = f^\star_j(x)$ whenever $x \in W^\star_i \cap W^\star_j$, it does not matter which $i$ is chosen if there are multiple $W^\star_i \in \C$ containing $x$. Thus, $g$ is well defined and $g|_{W^\star_i} = f^\star_i$ for all $1 \leq i \leq \frac{m_3}{4}$.

We next show that $g$ is close to a linear function and that this linear function agrees with most of the functions $f^\star_i|_{W^\star_i}$ for $W^\star_i \in \C$. To begin, we show that $g$ passes the standard linearity test with high probability.
\begin{lemma} \label{lm: g pass linearity}
    We have,
    \[
    \Pr_{z_1, z_2 \in W^\star_{\amb}}[g(z_1+z_2) = g(z_1) + g(z_2)] \geq 1 - \frac{25q^{2s}}{m_3}.
    \]
\end{lemma}
\begin{proof}
Note that we have
\[
 \Pr_{z_1, z_2 \in W^\star_{\amb}}[g(z_1+z_2) = g(z_1) + g(z_2)] \geq \Pr_{z_1, z_2 \in W^\star_{\amb}}[\exists W^\star_i \in \mathcal{C}, \text{ s.t. }, z_1,z_2 \in W^\star_i].
\]
For every $z_1, z_2 \in W^\star_{\amb}$ linearly independent, we can let $N(z_1, z_2)$ denote the number of $W^\star_i \in \mathcal{C}$ containing $\spa(z_1, z_2)$. The result them follows from~\cref{lm: deviation bounds} with $a =0$, $j = 2$, $r = s$, $c = 0.99$, and using the fact that $|\mc{C}| = m_3/4$. Indeed,
\begin{align*}
Pr_{z_1, z_2 \in W^\star_{\amb}}[\exists W^\star_i \in \mathcal{C}, \text{ s.t. }, z_1,z_2 \in W^\star_i] &\geq \Pr_{z_1, z_2 \in W^\star_{\amb}}[N(z_1, z_2) > 0 \; | \; \dim(\spa(z_1, z_2)) = 2 ]\\
&\geq 1 - \frac{25q^{2s}}{m_3}.
\end{align*}
In the first transition, we use that if $z_1,z_2$ are linearly dependent then $g(z_1+z_2) = g(z_1) + g(z_2)$.
\end{proof}

We now apply the linearity testing result of Blum, Luby, and Rubinfeld \cite{BLR} and get that $g$ is $\frac{25q^{2s}}{m_3}$-close to a linear function, say $h: W^\star_{\amb} \xrightarrow[]{} \Ff_q$.

\begin{thm} \cite{BLR}
    Suppose $g: W^\star_{\amb} \to \Ff_q$ satisfies, $\Pr_{z_1, z_2}[g(z_1+z_2) = g(z_1) + g(z_2)] \geq 1 - \rho$, for some $\rho < 2/9$. Then, there exists a linear function $h: W^\star_{\amb} \to \Ff_q$ such that $\Pr_{z \in W^\star_{\amb}}[h(z) \neq g(z)] \leq 2\rho$.
\end{thm}

We will conclude by showing that this $h$ is the desired function which agrees with many of the original $f_i$'s. To this end, we first show that it agrees with many of the $f^\star_i$'s that we have (which themselves are restrictions of the original $f_i$'s), and then show that $h$ can be extended to $W'_{\amb}$ in a manner that retains agreement with many of the $f_i$'s.
 
Towards the first step, set $S = \{x \in W^\star_{\amb} \; | \; g(x) \neq h(x) \}$. We show that choosing $W^\star_i \in \mathcal{C}$ randomly, and then a point $x \in W^\star_i$, it is unlikely that $x \in S$. Define the measure $\nu_{\mc{C}}$ over nonzero points in $W^\star_{\amb}$ obtained by choosing $W_i \in \C$ uniformly at random and then $x \in W_i$ nonzero uniformly at random. Let $\mu$ be the uniform measure over $W^\star_{\amb}$, so $\mu(S) \leq \frac{25q^{2s}}{m_3}$. Then $\nu_{\mc{C}}(S)$ is precisely the probability of interest and can be upper bounded using~\cref{lm: nu vs mu}, with parameters $a = 0, j = 1$, codimension $s$, 
\begin{equation} \label{eq: clique agree}
    \nu_{\mc{C}}(S) \leq \mu(S) + \frac{25q^{\frac{s}{2}}}{\sqrt{m_3}} \leq \frac{26q^{\frac{s}{2}}}{\sqrt{m_3}}.
\end{equation}

\begin{lemma}
    We have $h|_{W^\star_i} \equiv f^\star_i$ for at least half of the $W^\star_i \in \mathcal{C}$.
\end{lemma}
\begin{proof}
By Markov's inequality and~\eqref{eq: clique agree} with probability at least $1/2$, over $W^\star_i \in \C$, we have
\[
\frac{\left|W^\star_i \cap S \right|}{\left|W^\star_i \right|} \leq \frac{52q^{\frac{s}{2}}}{\sqrt{m_3}} < 1 - \frac{1}{q},
\]
and $f^\star_i$ and $h|_{W^\star_i}$ agree on more than $1/q$ of the points in $W^\star_i$. Since $f^\star_i$ and $h|_{W^\star_i}$ are both linear, by the Schwartz-Zippel Lemma that $h|_{W^\star_i} \equiv f^\star_i$, and the result follows.
\end{proof}
We are now ready to finish the proof of~\cref{lm: bound on good zoom outs}.
\begin{proof}[Proof of~\cref{lm: bound on good zoom outs}]
Summarizing, we now have linear functions $f^\star_i: W^\star_i \xrightarrow[]{} \Ff_q$ for $1\leq i \leq \frac{m_3}{8}$ and a linear function $h: W^\star_{\amb} \xrightarrow[]{} \Ff_q$ such that $h|_{W^\star_i} = f^\star_i$. Furthermore, for each $f^\star_i, W^\star_i$, there is a $f_i, W_i$ from~\cref{lm: bound on good zoom outs} such that $W_i \cap W^\star_{\amb} = W^\star_i$, $W_i \subseteq W'_{\amb}$, $f_i|_{W^\star_i} = f^\star_i$, and $f_i|_{X} = \sigma$.

We extend $h$ in a manner so that it agrees with many of these original functions $f_i$. To this end, recall that $W^\star_{\amb}$ satisfies, 
\[
A \oplus W^{\star}_{\amb} = B \subseteq W'_{\amb}.
\]
and $\dim(A) + \codim(B) - \dim(X) \leq \frac{10}{\delta_2}$. Therefore, we may choose a random linear function $h': W'_{\amb} \xrightarrow[]{} \Ff_q$ conditioned on $h'|_{W^\star_{\amb}} \equiv h$ and $h'|_{X} \equiv \sigma$. For any $f_i$, we have that
\[
\Pr_{h'}[h'|_{W_i} \equiv f_i] \geq q^{-(\dim(A) + \codim(B) - \dim(X))} \geq q^{-\frac{10}{\delta_2}}.
\]
Indeed there is a $q^{-(\dim(A) - \dim(X))}$ probability that $h
|_{A} \equiv f_i|_A$, as we condition on $h'|_{X} \equiv \sigma \equiv f_i|_X$. Then, extending $h'$ from $B$ to $W'_{\amb}$, there is at least a $q^{-\codim(B)}$ probability that $h'$ is equal to $f_i$ on these extra dimensions. It follows that there is a linear $h': W'_{\amb} \xrightarrow[]{} \Ff_q$ such that $h'|_{W_i} \equiv f_i$ for at least 
\[
\frac{m_3 q^{-\frac{10}{\delta_2}}}{8} \geq q^{50r\ell \xi^{-1}}
\]
 of the pairs $f_i, W_i$ from~\cref{lm: bound on good zoom outs}. Take these $W_i$ to be the set $\W'$ for~\cref{lm: bound on good zoom outs}. As $\W' \subseteq \W$, they are $2$-generic with respect to $W'_{\amb}$ and have codimension $s \leq r$ in $W'_{\amb}$.
\end{proof}

\section{Acknowledgments} We sincerely thank anonymous reviewers for their careful reading of the manuscript and their detailed comments, which greatly improved the writing of the paper.

\bibliographystyle{alpha}
\bibliography{references}
\appendix
\section{Proofs of Lemmas~\ref{lm: pseudorandom edges} and ~\ref{lm: fourier even covering}} \label{app: level d}
In this section we prove~\cref{lm: pseudorandom edges,lm: fourier even covering}. The proofs of these lemmas requires tools from~\cite{DBLP:conf/stoc/EllisKL23,EvraKL} regarding Fourier analysis over the Bilinear Scheme. 

\subsection{Fourier Analysis over the Bilinear Scheme}
The key to proving~\cref{lm: pseudorandom edges} is a level-$d$ inequality for indicator functions on the Bilinear Scheme due to Evra, Kindler, and Lifshitz \cite{EvraKL}. In order to use the result of \cite{EvraKL}, however, we first give some necessary background for Fourier analysis over the Bilinear Scheme, and describe the analogues of zoom-ins, zoom-outs, and pseudo-randomness as in \cite{DBLP:conf/stoc/EllisKL23,EvraKL}. After doing so, we must then find a suitable map from the Grassmann graph to the Bilinear Scheme that (1) preserves the edges of our original bipartite inclusion graph between $2\ell$-dimensional and $2(1-\delta)\ell$ subspaces, and (2) maps zoom-ins and zoom-outs in the Grassmann graph to their analogues over the Bilinear Scheme. 

\paragraph{The Bilinear Scheme:} Let $\Ff_q^{n \times 2\ell}$ 
be the set of $n \times 2\ell$ matrices over $\mathbb{F}_q$. One can define a graph over $\Ff_q^{n \times 2\ell}$ that is similar to the Grassmann graphs by calling $M_1, M_2 \in \Ff_q^{n \times 2\ell}$ adjacent if $\dim(\ker(M_1 - M_2)) \leq s$ for some $s \leq 2\ell$. Such graphs are often referred to as the \emph{Bilinear Scheme}. For our purposes, we will need to work with a bipartite version of this graph between $\Ff_q^{(n-2\ell) \times 2\ell}$ and $\Ff_q^{(n-2(1-\delta)\ell) \times 2(1-\delta)\ell}$.
\skipi
We equip the space $L_2(\Ff_q^{n \times 2\ell})$ with the following inner product:
\[
\langle F, G \rangle = \E_{M \in \Ff_q^{n \times 2\ell}}[F(M) \overline{G(M)}],
\]
where the distribution taken over $M$ is uniform. Let $\omega$ be a primitive $p$th root of unity, where recall $p$ is the characteristic of $\Ff_q$. For $s \in \Ff_q^{n}$ and $x \in \Ff_q^{n}$, let $\chi_s(x) = \omega^{\Tr(s \cdot x)}$ 
where $\Tr\colon\mathbb{F}_q\to\mathbb{F}_p$ is the 
trace map given by $\Tr(\gamma) = \sum_{i=0}^{\log_p(q)-1}\gamma^{p^i}$. A common fact that we will used is that the trace is additive $\Tr(\gamma + \gamma') = \Tr(\gamma) + \Tr(\gamma')$ for $\gamma, \gamma' \in \Ff_q$. The characters 
$\{\chi_S: \Ff_q^{n \times 2\ell} \xrightarrow[]{} \mathbb{C}~|~S = (s_1,\ldots, s_{2\ell}) \in \Ff_q^{n \times 2\ell}\}$, given by
\begin{equation} \label{eq: char def} 
\chi_S(x_1,\ldots, x_{2\ell}) = \prod_{i=1}^{2\ell} \chi_{s_i}(x_i) =  \omega^{\sum_{i=1}^{2\ell} \Tr(s_i \cdot x_i)} =  \omega^{\Tr \left(\sum_{i=1}^{2\ell} (s_i \cdot x_i)\right)},
\end{equation}
form an orthonormal basis of $L_2(\Ff_q^{n \times 2\ell})$. As a result, any $F \in L_2(\Ff_q^{n \times 2\ell})$ can be expressed as,
\[
F = \sum_{S \in \Ff_q^{n \times 2\ell}}\widehat{F}(S) \chi_{S},
\]
where $\widehat{F}(S) = \langle F, \chi_S \rangle$. The level $d$ component of $F$ is given by
$F^{=d} = \sum_{S: \; \rank(S) = d} \widehat{F}(S) \chi_{S}$. If a function $F$ only consists of components up to level $d$, i.e.\ $\widehat{F}(S) = 0$ for all $S$ with $\rank(S) > d$, then we say $F$ is of degree $d$.

We now describe the analogues of zoom-ins and zoom-outs on $\Ff_q^{n \times 2\ell}$.

\begin{definition}
  A zoom-in of dimension $d$ over $\Ff_q^{n \times 2\ell}$ is given by $d$-pairs of vectors $(u_1, v_1), \ldots, (u_r, v_r)$ where each $u_i \in \Ff_q^{2\ell}$ and each $v_i \in \Ff_q^{n}$. Let $U \in \Ff_q^{2\ell \times d}$ and $V \in \Ff_q^{n \times d}$ denote the matrices whose $i$th columns are $u_i$ and $v_i$ respectively. Then the zoom-in on $(U, V)$ is the set of $M \in \Ff_q^{n \times 2\ell}$ such that $MU = V$, or equivalently, $Mu_i = v_i$ for $1 \leq i \leq d$.  
\end{definition}

Next, we define the analog of zoom-outs.
\begin{definition}
 A zoom-out of dimension $d$ is defined similarly, except by multiplication on the left. Given $X \in \Ff_q^{d \times n}$ and $Y \in \Ff_q^{d \times 2\ell}$, whose rows are given by $x_i$ and $y_i$ respectively, the zoom-out $(X, Y)$ is the $M \in \Ff_q^{n \times 2\ell}$ such that $XM = Y$, or equivalently, $x_iM = y_i$ for $1 \leq i \leq d$.   
\end{definition}

Let $\Zoom[(U, V), (X, Y)]$ denote the intersections of the zoom-in on $(U, V)$ and the zoom-out on $(X, Y)$. The codimension of $\Zoom[(U, V), (X, Y)]$ is the sum of the number of columns of $U$ and the number of rows of $X$, which we will denote by $\dim(U)$ and $\codim(X)$. For a zoom-in and zoom-out pair and a Boolean function $F$, we define $F_{(U, V), (X,Y)}: \Zoom[(U,V), (X,Y)] \xrightarrow[]{} \{0,1\}$ to be the restriction of $F$ which is given as 
\[
F_{(U, V), (X,Y)}(M) = F(M) \quad \text{for} \quad M \in  \Zoom[(U,V), (X,Y)].
\]
When $\dim(U) + \codim(X) = d$, we say that the restriction is of size $d$. We define $(d, \epsilon)$-pseudo-randomness in terms of the $L_2$-norms of restrictions of $F$ of size $d$.
Here and throughout, when we consider
restricted functions, the underlying 
measure is the uniform measure
over the corresponding zoom-in and zoom-out
set $\Zoom[(U,V),(X,Y)]$.
\begin{definition}
We say that an indicator function $F \in L_2(\Ff_q^{n \times 2\ell})$ is \emph{$(d, \epsilon)$-pseudo-random} if for all zoom-in zoom-out combinations $\Zoom[(U, V), (X, Y)]$ such that $\dim(U) + \codim(X) = d$, we have
\[
\norm{F_{(U,V), (X, Y)}}_2^2 
\leq \epsilon.
\]
\end{definition}
We note that for Boolean functions $F$, 
$\norm{F_{(U,V), (X, Y)}}_2^2 
= \E_{M \in \Zoom[(U, V), (X, Y)]}[F(M)]$, and hence the 
definition above generalizes the definition we have for Boolean
functions. 

\begin{definition}
We say that an indicator function $F \in L_2(\Ff_q^{n \times 2\ell})$ is \emph{$(d, \epsilon, t)$-pseudo-random} if for all $\Zoom[(U, V), (X, Y)]$ such that $\dim(U) + \codim(X) = d$, we have,
\[
\norm{F_{(U,V), (X, Y)}}_{t/(t-1)} = \left(\E_{M \in \Zoom[(U, V), (X, Y)]}\left[|F(M)|^{\frac{t}{t-1}}\right]\right)^{\frac{t-1}{t}} \leq \epsilon.
\]
\end{definition}

The following result is a combination of two results form~\cite{EvraKL}. It states that 
if a Boolean function $F$ is $(r,\eps)$-pseudo-random, then
its degree $d$ parts are $(r,C_{q,d}\eps^2)$-pseudo-random 
for $d\leq r$. 
\begin{lemma} \cite[Theorem 5.5 + Proposition 3.6]{EvraKL} \label{lm: EKL 1}
     Let $t \geq 4$ be a power of $2$ and let $F: \Ff_q^{n \times 2\ell} \to\{0,1\}$ be a function that is $(d, \epsilon, t)$-pseudo-random. Then $F^{=d}$ is $(r, q^{10 dr+500d^2t}\epsilon^2)$-pseudo-random for all $r \geq d$.
\end{lemma}
\begin{proof}
    This lemma does not actually appear in~\cite{EvraKL}, 
    but it is easy to derive by combining Theorem 5.5 with Proposition 3.6 therein. In~\cite{DBLP:conf/stoc/EllisKL23,EvraKL}, the authors introduce an additional notion of generalized influences and having small generalized influences. We 
    refrain from defining these notions explicitly as it is 
    slightly cumbersome, but roughly speaking, one defines a Laplacian for each zoom-in, zoom-out combination, so that 
    having $(d,\eps)$ small generalized influences means that upon applying these Laplacians on $F$, the $2$-norm squared of the resulting function never exceeds $\eps$. 
    
    With this notion in hand, if a function $F$ is $(d, \epsilon, t)$-pseudo-random, then by~\cite[Theorem 5.5]{EvraKL} we get  that $F^{=d}$ has $(d, q^{500d^2 t}\epsilon^2)$-small generalized influences. Applying \cite[Proposition 3.6]{EvraKL} then implies that $F^{=d}$ is $(r, q^{10dr}\cdot q^{500d^2 t}\epsilon^2)$-pseudo-random for any $r \geq d$, which is the desired result.
\end{proof}

Lastly, we need the following global hypercontractivity result 
also due to~\cite{EvraKL}.\footnote{We remark that earlier results~\cite{DBLP:conf/stoc/EllisKL23} showed similar
statement for $4$-norms, i.e. the case that $t=4$, and 
the result below follows by a form of induction on $t$.
That is, one starts with $F$ and concludes via applying the
case $t=4$ that the function and it $F^2$ is $(d,C_{q,d}\eps)$-pseudo-random.
Then one apply the case $t=4$ on $F^2$ to conclude that $F^4$
is $(d,C_{q,d}'\eps)$-pseudo-random and so on.}
\begin{thm} \cite[Theorem 1.13]{EvraKL} \label{thm: EKL 2}
    Let $t \geq 4$ be a power of $2$ and let $F\in L_2(\Ff_q^{n \times 2 \ell})$ be a function of degree $d$ that is $(d, \epsilon)$-pseudo-random. Then,
    \[
    \norm{F}_t^t \leq q^{200d^2t^2} \norm{F}_2^2 \epsilon^{t/2-1}.
    \]
\end{thm}

Combining~\cref{lm: EKL 1,thm: EKL 2}, we arrive at the following result which bounds the $t$-norm of the level $d$ component of pseudo-random indicator functions. This result will be the key to showing an analogue of~\cref{lm: pseudorandom edges} over the Bilinear Scheme.
\begin{thm} \label{th: EKL}
    Let $t \geq 4$ be a power of $2$. Then if $F: \Ff_q^{n \times 2\ell} \xrightarrow[]{} \{0,1\}$ is $(r, \epsilon)$-pseudo-random, we have
    \[
     \norm{F^{=d}}_t \leq q^{500d^2 t}\epsilon^{\frac{t-2}{t}}
    \]
    for all $d \leq r$.
\end{thm}
\begin{proof}
    Suppose $F$ is $(r, \epsilon)$-pseudo-random, let $t \geq 4$ be a power of $2$, and fix a $d \leq r$. Since $d \leq r$, we also have that $F$ is $(d, \epsilon)$-pseudorandom. Therefore for any size $d$ restriction of $F$, $F_{(U,V), (X,Y)}$ ,we have,
    \[
    \norm{F_{(U,V),(X,Y)}}_{t/(t-1)} = \left(\norm{F_{(U,V),(X,Y)}}_2^2\right)^{\frac{t-1}{t}} \leq \epsilon^{\frac{t-1}{t}}.
    \]
    Thus, $F$ is $(d, \epsilon^{\frac{t-1}{t}}, t)$-pseudo-random, and by~\cref{lm: EKL 1} it follows that $F^{=d}$ is $(d, q^{10d^2 + 500d^2t}\epsilon^{\frac{2t-2}{t}})$-pseudo-random. Clearly, $F^{=d}$ is degree $d$, so applying~\cref{thm: EKL 2} we get that
    \[
    \norm{F^{=d}}_t^{t} \leq q^{200d^2 t^2}\norm{F^{=d}}_2^2\left(q^{10d^2 + 500d^2t} \epsilon^{\frac{2t-2}{t}}\right)^{t/2-1}
    \leq q^{500d^2t^2} \epsilon^{t-2},
    \]
    where we also use the fact that $\norm{F^{=d}}_2^2\leq \norm{F}_2^2 \leq \eps$ because $F$ is an $(r,\eps)$-pseudo-random Boolean function. Taking the $t$-th root of the above inequality completes the proof.
\end{proof}

\subsection{An Analog of~\cref{lm: pseudorandom edges} for 
the Bilinear Scheme}
In this section, we will build up some Fourier analytic tools over the bilinear scheme towards showing \cref{lm: pseudorandom edges}. Specifically, our goal will be to transform the statement from \cref{lm: pseudorandom edges} into a statement about indicator functions over the bilinear scheme, and then apply \cref{thm: EKL 2}.

To this end, we start with some definitions. Say that a function $F \in L_2(\Ff_q^{n \times 2\ell})$ is \emph{basis invariant} if it satisfies $F(M) = F(M A)$ for any full rank $A \in \Ff_q^{2\ell \times 2\ell}$. Basis invariant indicator functions are meant to correspond to indicator functions for sets of subspaces in $\Grass_q(n, 2\ell)$. Indeed, by definition, if two $M_1$ and $M_2$ have the same column image, then $F(M_1) = F(M_2)$. 
Define the adjacency operator $\T: L_2\left(\Ff_q^{n \times 2\ell}\right) \xrightarrow{} L_2\left(\Ff_q^{n \times 2(1-\delta)\ell} \right)$ as follows. For any $H \in L_2(\Ff_q^{n \times 2\ell})$, the function 
$\T H \colon \Ff_q^{n \times 2(1-\delta)\ell}\to\mathbb{C}$ is given by
\[
\T F(M) = \E_{v_1,\ldots, v_{2\delta \ell}}[F\left([M, v_1,\ldots, v_{2\delta \ell}]\right)].
\]
Here, $[M, v_1,\ldots, v_{2\delta \ell}]$ refers to the matrix obtained by adding the columns $v_i$ to $M$ on the right, for all $i \in [2\delta\ell]$. In words, the operator $\T$ averages over extensions of the matrix
$M$ to an $n\times 2\ell$ matrix by adding to it $2\delta\ell$ random columns. This operator is morally the adjacency operator for the bipartite Grassmann graph between $\Grass_q(n, 2\ell)$ and $\Grass_q(n, 2(1-\delta)\ell)$, and as such, it will ultimately help us transform the probability of interest in \cref{lm: pseudorandom edges}, into an inner product over the bilinear scheme. 

Before jumping into \cref{lm: pseudorandom edges} though, we must first show some lemmas about basis invariant functions and the operator $\T$. The first lemma we need is that the level $d$ component of a basis invariant function is also basis invariant, and to show this, the following two identities will be useful. 
\begin{lemma} \label{lm: character identity}
    For any $S = (s_1, \ldots, s_{\ell'}) \in \Ff_q^{n \times \ell'}$, any $M \in \Ff_q^{n \times 2\ell}$, and any matrix $A \in \Ff_q^{2\ell \times \ell'}$ we have, 
    \[
    \chi_S(MA) = \chi_{SA^T}(M).
    \]
\end{lemma}
\begin{proof}
     Letting $v_1,\ldots, v_{2\ell}$ denote the columns of $M$ and $a_{i,j}$ denote the entries of $A$, we have,
    \[
        \chi_S(MA) = \omega^{\sum_{i=1}^{\ell'} \Tr\left(s_i \cdot \left(\sum_{j=1}^{2\ell} v_j a_{j,i}\right)\right)} = \omega^{\sum_{j = 1}^{2\ell}\Tr\left(v_j \cdot \left(\sum_{i=1}^{\ell'} s_i a_{j,i} \right)\right)} = \chi_{SA^T}(M).
        \qedhere
    \]
\end{proof}
\begin{lemma} \label{lm: character basis invariant}
    Let $S = (s_1,\ldots, s_{2\ell}) \in \Ff_q^{n \times 2\ell}$ and let $F \in L_2\left(\Ff_q^{n \times 2\ell} \right)$ be basis invariant. Then for any $A \in \Ff_q^{2\ell \times 2\ell}$ that is full rank, we have $\widehat{F}(SA) = \widehat{F}(S)$.
\end{lemma}

\begin{proof}
    For any matrix full rank $B \in \Ff_q^{2\ell \times 2\ell}$ we have 
    \begin{align*}
        \widehat{F}(S) &= \E_{M \in \Ff_q^{n \times 2\ell}}\left[\overline{\chi_S(M)} F(M) \right] \\
        &= \E_{M \in \Ff_q^{n \times 2\ell}}\left[\overline{\chi_S(M)} F(M B^{-1}) \right] \\
        &= \E_{M \in \Ff_q^{n \times 2\ell}}\left[\overline{\chi_S(MB)} F(M) \right] \\
        &=E_{M \in \Ff_q^{n\times 2\ell}}[\overline{\chi_{SB^T}(M)}F(M)] \\
        &=\widehat{F}(SB^T),     
    \end{align*}
    where we use that $F$ is basis invariant in the third transition and~\cref{lm: character identity} in the fourth transition. Setting $B = A^T$ gives the result.
\end{proof}
Using~\cref{lm: character basis invariant}, we can show that the level $d$ component of a basis invariant function is also basis invariant. 

\begin{lemma} \label{lm: level d basis invariant}
If $F \in L_2\left(\Ff_q^{n \times 2\ell} \right)$ is basis invariant, then $F^{=d}$ is basis invariant as well for any $d$.
\end{lemma} 
\begin{proof}
Fix any $M \in \Ff_q^{n \times 2\ell}$ and $A \in \Ff_q^{2\ell \times 2\ell}$ full rank. We have
\begin{align*}
F^{=d}(MA) &= \sum_{S \in \Ff_q^{n\times 2\ell}, \rank(S) = d}\widehat{F}(S) \chi_S(MA) \\
&= \sum_{S \in \Ff_q^{n\times 2\ell}, \rank(S) = d}\widehat{F}(S) \chi_{SA^T}(M) \\
&=\sum_{S \in \Ff_q^{n\times 2\ell}, \rank(S) = d}\widehat{F}(S(A^T)^{-1}) \chi_{S}(M) \\
&= \sum_{S \in \Ff_q^{n\times 2\ell}, \rank(S) = d}\widehat{F}(S) \chi_{S}(M) \\
&= F^{=d}(M),
\end{align*}
where we use~\cref{lm: character identity} in the second transition, and~\cref{lm: character basis invariant} in the fourth transition.
\end{proof}

Next, we will define another operator $\mc{G}: L_2\left(\Ff_q^{n \times 2(1-\delta)\ell}\right) \xrightarrow{} L_2\left(\Ff_q^{n \times 2\ell} \right)$, which acts like the adjoint of $\T$ for basis invariant functions. The reason that we work with $\mc{G}$ and not the actual adjoint of $\T$ is that we want to define this operator in a specific, convenient way so that we can ultimately study what the composed operator $\mc{G} \T$ looks like. The operator $\mc{G}$ is given by
\[
\mc{G} H(M) = \E_{A \in \Ff_q^{2\ell \times 2(1-\delta)\ell}}[H(MA) \; | \; \rank(A) = 2(1-\delta)\ell].
\]
The next lemma shows that $\mc{G}$ indeed acts like the adjoint of $\T$ for basis invariant functions.
\begin{lemma} \label{lm: adjoint}
For $F \in L_2\left(\Ff_q^{n \times 2\ell}\right)$ that is basis invariant and $G \in L_2\left(\Ff_q^{n \times 2(1-\delta)\ell}\right)$, we have 
\[
\langle \T F, G \rangle = \langle F, \mc{G} H\rangle.
\]
\end{lemma}
\begin{proof}
Let $J \in \Ff_q^{2\ell \times 2(1-\delta)\ell}$ be the matrix whose restriction to the first $2(1-\delta)\ell$ rows is the identity matrix $I_{2(1-\delta)\ell \times 2(1-\delta)\ell}$ and whose remaining rows are all $0$. We have 
\begin{align*}
\langle \T F, H \rangle &= \E_{M' \in \Ff_q^{n \times 2(1-\delta)\ell}, v_1,\ldots, v_{2\delta \ell}\in \Ff_q^n }\left[F \left([M', v_1,\ldots, v_{2\delta \ell}] \right) \cdot \overline{H(M')}\right] \\
&= \E_{M'\in \Ff_q^{n \times 2(1-\delta)\ell}, v_i \in \Ff_q^n, A \in \Ff_q^{2\ell \times 2\ell} }\left[F\left([M', v_1,\ldots, v_{2\delta\ell}] A \right) \cdot \overline{H(M')}\; | \; \rank(A) = 2\ell \right] \\
&= \E_{M \in \Ff_q^{n \times 2\ell}, A \in \Ff_q^{2\ell \times 2\ell}}[F(M) \cdot \overline{H(MA^{-1}J)} \; |\; \rank(A) = 2\ell],
\end{align*}
where in the second transition we used
the fact that $F$ is basis invariant, 
and in the third one we made a change 
of variables 
$M = [M',v_1,\ldots,v_{2\delta\ell}]A$.
Now note that $A^{-1} J$ is the matrix $A^{-1}$ restricted to its first $2(1-\delta)\ell$ columns and hence in the final distribution, $A^{-1}J$ is a uniformly random matrix in $\Ff_q^{n \times 2(1-\delta)\ell}$ with rank $2(1-\delta)\ell$. It follows that,
\[
\langle \T F, H \rangle = \E_{M \in \Ff_q^{n \times 2\ell}, B \in \Ff_q^{2\ell \times 2(1-\delta)\ell}}[F(M) \cdot \overline{H(MB)} \; | \; \rank(B) = 2(1-\delta)\ell] = \langle F, \mc{G} H \rangle.
\qedhere
\]
\end{proof}

We will want to understand the operator $\mc{G} \T$, and 
towards this end we define the operator
\[
\Phi F (M) = \E_{\substack{B \in \Ff_q^{n \times 2\delta \ell}\\ C \in \Ff_q^{2\delta \ell \times 2\ell}\\ \rank(C) = 2\delta\ell}} \left[F(M + BC) \right].
\]
The operator $\Phi$ is meant to be a version of $\mc{G} \T$ which is easier to work with and acts in the same way on basis invariant functions. The upshot of working with $\Phi$ is that it is an averaging operator with respect to a Cayley graph over $\mathbb{F}_q^{n\times 2\ell}$, so each character $\chi_S$ is an eigenvector of $\Phi$ and the eigenvalues have an 
explicit formula (see~\cite{spielman2019spectral} for a detailed discussion).

\begin{lemma} \label{lm: phi}
If $F \in L_2\left(\Ff_q^{n \times 2\ell} \right)$ is basis invariant, then $\mc{G} \T F = \Phi F$.
\end{lemma}
\begin{proof}
By definitions
\[
\mc{G}\T F(M) = \E_{\substack{R' \in \Ff_q^{2\ell \times 2(1-\delta)\ell},\\ w_1, \ldots, w_{2\delta\ell} \in \Ff_q^n}}[F\left([MR', w_1,\ldots, w_{2\delta \ell}]\right)~|~\rank(R') = 2(1-\delta)\ell].
\]
We can also view $M' = [MR', w_1,\ldots, w_{2\delta \ell}]$ as being sampled as follows. Choose $R' \in \Ff_q^{2\ell \times 2(1-\delta)\ell}$ with linearly independent columns, extend $R'$ to a matrix $R \in \Ff_q^{2\ell \times 2\ell}$ with linearly independent columns randomly by adding $2\delta\ell$ columns on the right, sample a random matrix $[0,\dots, 0, w_1,\ldots, w_{2\delta\ell}] \in \Ff_q^{n \times 2\ell}$, and output, 
\[
M' = MR + [0,\dots, 0, w_1,\ldots, w_{2\delta\ell}].
\]
Furthermore, under this distribution, it is clear that $R \in \Ff_q^{2\ell \times 2\ell}$ is a uniformly random matrix with linearly independent columns. Therefore, 
\begin{align*}
\mc{G}\T  F(M) &= \E_{\substack{R \in \Ff_q^{2\ell \times 2\ell},\\ w_1,\ldots, w_{2\delta\ell}\in \Ff_q^n}}[ F(MR + [0,\dots, 0, w_1,\ldots, w_{2\delta\ell}])~|~\rank(R) = 2\ell] \\
 &= \E_{\substack{R \in \Ff_q^{2\ell \times 2\ell}, \\w_1,\ldots, w_{2\delta\ell}\in \Ff_q^n}}[ F(M + [0,\dots, 0, w_1,\ldots, w_{2\delta\ell}]R^{-1})~|~\rank(R) = 2\ell], 
\end{align*}
where we are using the fact that $F$ is basis invariant and $R$ is invertible. In the last expectation, note that the distribution over $[0,\dots, 0, w_1,\ldots, w_{2\delta\ell}]R^{-1}$ is the same as that over $BC$ where $B \in \Ff_q^{n \times 2 \delta \ell}$ is uniformly random, and $C \in \Ff_q^{2\delta \ell \times 2\ell}$ is uniformly random conditioned on having linearly independent rows. More precisely, it is equal to $BC$ where $B = [w_1,\ldots, w_{2\delta\ell}]$, and $C$ is the last $2\delta\ell$ rows of $R^{-1}$.
It follows that
\[
\mc{G}\T  F(M) = \E_{B \in \Ff_q^{n \times 2\delta \ell}, C \in \Ff_q^{2\delta \ell \times 2\ell}} \left[F(M + BC) \; | \; \rank(C) = 2\delta \ell \right].\qedhere
\]
\end{proof}
The following lemma gives upper bound on the eigenvalues of $\Phi$.
\begin{lemma} \label{lm: eigenvalue calcs}
Suppose that $\rank(S) = t$. If $t=0$, then $\chi_S$
is an eigenvector of $\Phi$ of eigenvalue $1$. 
If $t>0$, $\chi_S$ is an eigenvector of $\Phi$ of eigenvalue which is at most $3q^{t-n} + q^{-t(2\delta\ell-1)}$ in absolute value.
\end{lemma}
\begin{proof}
Fix $S$. We argued earlier that $\chi_S$ is an eigenvector
of $\Phi$, and we denote the corresponding eigevalue by $\lambda = \Phi \chi_S(0)$. If $t=0$ the statement is clear, so we assume 
that $t>0$ henceforth. 

Find $A\in\mathbb{F}_q^{2\ell\times 2\ell}$ of 
full rank so that $SA^T = (v_1,\ldots,v_t,0,0,\ldots,0)$
where $v_1,\ldots,v_t$ are linearly independent. Thus, 
as the distribution of $C$ is invariant under multiplying 
by $A^T$ from the right, we get that
\[
\lambda
=\Phi \chi_S(0)
= \E_{B, C}[\chi_S(BCA^T)~|~\rank(C) = 2\delta \ell]
= \E_{B, C}[\chi_{SA}(BC)~|~\rank(C) = 2\delta \ell]
=\Phi \chi_{SA}(0),
\]
where we used~\cref{lm: character identity}. We may  therefore assume that $S = (v_1,\ldots,v_t,0,\ldots,0)$
for linearly independent $v_1,\ldots,v_t$. 
Applying symmetry again, we conclude that
\[
\lambda = 
\E_{\substack{v_1,\ldots,v_t\\\text{linearly independent}}}
\left[\Phi \chi_{(v_1,\ldots,v_t,\vec{0})}(0)\right]
=\E_{\substack{v_1,\ldots,v_t\\\text{linearly independent}}}
\left[\E_{B,C} \omega^{\sum\limits_{i=1}^{t} \Tr(v_i\cdot {\sf col}_i(BC))}\right],
\]
and interchanging the order of expectations we get that
\[
\lambda
=\E_{B,C}\left[\E_{\substack{v_1,\ldots,v_t\\\text{linearly independent}}}
\omega^{\sum\limits_{i=1}^{t} \Tr(v_i\cdot {\sf col}_i(BC))}\right],
\]
Denote $w_i = {\sf col}_i(BC)$, and inspect these vectors. 
\begin{claim}\label{claim:what_if_non_zero_w}
If $w_i\neq 0$ for some $i$, then
\[
\left|\E_{\substack{v_1,\ldots,v_t\\\text{linearly independent}}}
\left[\omega^{\sum\limits_{i=1}^{t} \Tr(v_i\cdot {\sf col}_i(BC))}\right]\right|\leq 2q^{t-n}.
\]
\end{claim}
\begin{proof}
We first claim that if $v_1, \ldots, v_t$ are chosen uniformly, then the left hand side is $0$, or equivalently
\[
\E_{\substack{v_1,\ldots,v_t\\\text{uniform}}}
\left[\omega^{\Tr\left(\sum\limits_{i=1}^{t} v_i\cdot w_i\right)}\right] = 0.
\]
To see this, it suffices to show that $\sum_{i=1}^{t} v_i \cdot w_i$ takes every value in $\Ff_q$ with equal probability, and we focus on showing this. Fix $i$ such that $w_i\neq 0$ and suppose the $j$th entry, $w_{i,j}$ is nonzero. We can fix all entries of the $v_1, \ldots, v_t$ uniformly except for $v_{i,j}$, and then for each $\alpha \in \Ff_q$, there is exactly one choice of $v_{i,j}$ that will result in  $\sum_{i=1}^{t} v_i \cdot w_i = \alpha$.  

Thus, if we took the distribution over $v_1,\ldots,v_t$ to be uniformly 
and independently chosen, then the magnitude of the above expectation would be 
$0$. Hence, we conclude that the above expectation is at most 
twice the probability randomly chosen $v_1,\ldots,v_t$ are not
linearly independent, which is at most $q^{t-n}$.
\end{proof}
By~\cref{claim:what_if_non_zero_w} we conclude that
$\lambda\leq 2q^{t-n} + \Pr_{B,C}[w_i = 0, \; \forall i=1,\ldots,t]$, and we next 
bound this probability. Recalling the definition of $w_i$, we 
have that
\[
w_i =\sum\limits_{j=1}^{2\delta\ell} C(j,i) {\sf col}_j(B).
\]
Consider the $2\delta\ell\times t$ minor of $C$ and call it $C'$. 
First we upper bound the probability that $\rank(C') = 0$. 
Note that the distribution of $C$ is the same as of 
$A|_{2\delta\ell\times 2\ell}$ where  $A\in\mathbb{F}_q^{2\ell\times 2\ell}$ is a random
invertible matrix. Thus, 
$C'$ has the same distribution as of $A|_{2\delta\ell\times t}$,
and the probability that $C' = 0$ is at most
\[
\frac{q^{2\ell - 2\delta\ell}}{q^{2\ell} - 1}
\cdot
\frac{q^{2\ell - 2\delta\ell}}{q^{2\ell} - q}
\cdots
\frac{q^{2\ell - 2\delta\ell}}{q^{2\ell} - q^{t-1}}
\leq q^{-t(2\delta\ell-1)}.
\]

It remains to bound the probability 
that $w_i$ are all $0$ in the case that $\rank(C')\geq 1$. 
In this case, assume without loss of generality that the first
column of $C'$ is non-zero. Thus, it follows
that over the randomness of $B$, the vector $w_1$ is uniformly
chosen from $\mathbb{F}_q^n$, and so the probability it is the all $0$ vector is at most $q^{-n}$. Combining, we get that 
$\lambda\leq 3q^{t-n} + q^{-t(2\delta\ell-1)}$.
\end{proof}

Finally, using~\cref{lm: character identity} again, we note that $\mc{G}$ does not increase the level of a function:
\begin{equation} \label{eq: decompose inner product}
    \mc{G} \chi_S(M) = \E_{A \in \Ff_q^{2\ell \times 2(1-\delta)\ell}}[\chi_S(MA)] = \E_{A}[\chi_{SA^T}(M)].
\end{equation}
This observation has the following implication:
\begin{lemma} \label{lm: degree inner}
Let $F \in L_2\left(\Ff_q^{n \times 2\ell} \right)$ be basis invariant and let $G \in L_2 \left(\Ff_q^{n \times 2(1-\delta)\ell} \right)$. Then, $\langle \T F^{=d},  H \rangle = \langle \T F^{=d} ,  H^{=d} \rangle$.

\end{lemma}
\begin{proof}
    Using~\eqref{eq: decompose inner product}, we have, 
    \begin{align*}
        \mc{G} H^{=j}(M) = \sum_{S \in \Ff_q^{n \times 2(1-\delta)\ell}, \rank(S) = j} \widehat{H}(S) \cdot  \E_{A}\left[  \chi_{SA^T}(M) \; | \; \rank(A) = 2(1-\delta)\ell \right].
    \end{align*}
    Since $\rank(S) = j$, it follows that the $\rank(SA^T)$ is at most $j$, so it follows that for $j < d$, we have $\langle F^{=d}, \mc{G} H^{=j} \rangle = 0$. As a result,
    \begin{equation}\label{eq:parsevaling_bipartite}
     \langle \T F^{=d},  H \rangle = \langle F^{=d}, \mc{G} H \rangle = \sum_{j=d}^{2(1-\delta)\ell}  \langle F^{=d}, \mc{G} H^{=j} \rangle.
    \end{equation}
    In the first equality we used~\cref{lm: adjoint} and the fact that 
    $F^{=d}$ is basis invariant by~\cref{lm: level d basis invariant}.
    Next we have,
    \[
    \T F^{=d}(M) = \sum_{S \in \Ff_q^{n \times 2\ell}, \rank(S) = d} \widehat{F}(S) \chi_S(M') = \sum_{S \in \Ff_q^{n \times 2\ell}, \rank(S) = d} \widehat{F}(S) \chi_{S'}(M'),
    \]
    where both $M'$ and $S'$ are obtained from $M$ by removing the last $2\delta \ell$ columns. It follows that $\T F^{=d}$ has level at most $d$, so using~\eqref{eq:parsevaling_bipartite} we get
    \[
     \langle \T F^{=d},  H\rangle  =  \sum_{j=d}^{2(1-\delta)\ell}  \langle F^{=d}, \mc{G} H^{=j} \rangle 
     =  \sum_{j=d}^{2(1-\delta)\ell}  \langle \T F^{=d}, H^{=j} \rangle
     =\langle \T F^{=d}, H^{=d} \rangle.\qedhere
    \]
\end{proof}

We are now ready to state and prove an analog of~\cref{lm: pseudorandom edges} for basis invariant functions on the Bilinear
scheme.
\begin{lemma} \label{lm: pseudorandom edges over bilinear scheme}
    Let $F \in L_2(\Ff_q^{n \times 2\ell})$ and $H \in L_2(\Ff_q^{n \times 2(1-\delta)\ell})$ be basis invariant indicator functions with $\E[F] = \alpha, \E[H] = \beta$. If $F$ is $(r,\eps)$ pseudo-random, then for all $t \geq 4$ that are powers of $2$, we have
$$\langle \T F, H \rangle \leq q^{O_{t,r}(1)} \beta^{\frac{t-1}{t}} \eps^{\frac{t-2}{t}} + q^{-r\delta\ell}\sqrt{\alpha\beta}.$$
\end{lemma}
\begin{proof}
 Using the degree decomposition of $F$ and~\cref{lm: degree inner}, we can write
\[
\inner{\mathcal{T} F}{H} = \sum_{d=0}^{2\ell} \inner{\mathcal{T} F^{=d}}{H^{=d}}.
\]
We first bound the contribution from terms in the summation with $d>r$ using Cauchy-Schwarz. For $d > r$,
\begin{align*}
|\inner{\mathcal{T} F^{=d}}{H^{=d}}|^2 &\leq \norm{\mathcal{T} F^{=d}}_2^2 \norm{H^{=d}}_2^2\\
&=\norm{H^{=d}}_2^2 \inner{F^{=d}}{\mc{G} \mathcal{T} F^{=d}} \\
&= \norm{H^{=d}}_2^2 \inner{F^{=d}}{\Phi F^{=d}} \\
&\leq \left(q^{-2d\delta \ell}+3q^{d-n} \right)\norm{F^{=d}}_2^2 \norm{H^{=d}}_2^2 \\
&\leq \left(q^{-2d\delta \ell}+3q^{d-n} \right) \alpha \beta,
\end{align*}
where the third transition uses~\cref{lm: phi,lm: level d basis invariant}, and the fourth transition uses~\cref{lm: eigenvalue calcs}. Thus, the total contribution from the $d>r$ terms is 
\[
\sum_{d = r+1}^{2\ell} \left|\langle \T F^{=d}, H^{=d}\rangle \right| \leq \sum_{d = r+1}^{2\ell}  2q^{-d\delta \ell}  \sqrt{\alpha \beta} \leq  q^{-r\delta \ell} \sqrt{\alpha \beta}.
\]

Next, we bound the contribution from $d\leq r$ by bounding each term separately. Fix a $d \leq r$. By~\cref{lm: degree inner} and H\"{o}lder's inequality we have
\[\left| \inner{\mathcal{T} F^{=d}}{H^{=d}}\right| = \left|\inner{\mathcal{T} F^{=d}}{H} \right|
\leq 
\norm{\mathcal{T} F^{=d}}_{t}\norm{H}_{t/(t-1)}
\leq 
\norm{F^{=d}}_{t}\beta^{(t-1)/t} \leq q^{500d^2 t} \beta^{\frac{t-1}{t}}\epsilon^{\frac{t-2}{t}},
\] 
where in the last inequality we use the fact that $F$ is $(r,\epsilon)$-pseudo-random, so by~\cref{th: EKL} 
\[
\norm{F^{=d}}_{t} \leq  q^{500d^2 t}\epsilon^{\frac{t-2}{t}}.
\]
Altogether, this shows
\[
\inner{\mathcal{T} F}{H} \leq q^{O_{t,r}(1)} \beta^{\frac{t-1}{t}} \eps^{\frac{t-2}{t}} + q^{-r\delta\ell}\sqrt{\alpha\beta}.
\qedhere
\]
\end{proof}

\subsection{Reduction to the Bilinear Scheme}
We are now ready to prove~\cref{lm: pseudorandom edges}. Let $\mc{L} \subseteq \Grass_q(n, 2\ell)$ and $\mc{R}\subseteq \Grass_q(n, 2(1-\delta)\ell)$ be as in \cref{lm: pseudorandom edges} and suppose that $\mc{L}$ is $(r, \epsilon)$-pseudo-random. Define the related Boolean functions $F \in L_2\left(\Ff_q^{n \times 2\ell}\right)$, $H \in L_2(\Ff_q^{n \times 2(1-\delta)\ell})$ as follows:
\begin{equation} \label{eq: define F from L}
    F(M')=
    \begin{cases}
      1  & \text{if } \im(M') \in \Lcal, \\
      0 & \text{otherwise},
    \end{cases}
    \qquad\qquad 
    G(M)=
    \begin{cases}
     1 & \text{if } \im(M) \in \Rcal, \\
      0 & \text{otherwise}.
    \end{cases}
  \end{equation}
  In the above, $\im(\cdot)$ is the usual definition of matrix image and refers to the span of the columns of the matrix. Implicit in the definitions is the fact that $F$ evaluates to $0$ if the columns of $M'$ are not linearly independent, as in this case their span is not even a dimension $2\ell$ subspace, and likewise $H$ evaluates to $0$ if the columns of $M$ are not linearly independent. We note that $F$ and $G$ are both basis invariant 

Towards applying \cref{lm: pseudorandom edges over bilinear scheme}, we want to first show that the pseudo-randomness of $\mc{L}$ carries over to $F$. To this end we begin with the following lemma that
simplifies the type of zoom-ins and zoom-out combinations we have to consider when showing $F$ is pseudo-random.
\begin{lemma} \label{lm: level d intermediate lemma 1}
    For any $(U, V), (X, Y)$ such that $\dim(U) + \codim(X) < 2\ell$ and $\Zoom[(U,V), (X,Y)]$ is nonempty, there are $r'$ linearly independent columns of $V$, say $v_1, \ldots, v_{r'} \in \Ff_q^n$ and a subset of linearly independent rows of $X$, say $X'' \in \Ff_q^{s' \times n}$, such that $r' \leq \dim(U)$, $s' \leq \codim(X)$ and
    \[
     \norm{F_{(U,V), (X,Y)}}_2^2 \leq 2 \cdot \E_{M' \in \Ff_q^{n \times (2\ell-r')}}[F\left([v_1,\ldots, v_{r'}, M']\right) \; | \; X'' M' = 0],
    \]
    where $[v_1,\ldots, v_{r'}, M'] \in \Ff_q^{n\times 2\ell}$ is the matrix whose first $r'$ columns are $v_1, \ldots, v_{r'}$, and remaining columns are $M'$.
\end{lemma}
\begin{proof}
Let $r = \dim(U)$ and $s = \codim(X)$. First note that we can assume that the columns of $U$ and $V$ respectively are both nonzero and linearly independent. Indeed, otherwise say $u_i = 0$, then either $v_i = 0$, in which case the $i$th columns of $U$ and $V$ can be removed, or $v_i \neq 0$ and $\Zoom[(U, V), (X, Y)]$ is an empty set. Otherwise, if, say, $v_i = 0$, then either $u_i = 0$ and again we can ignored the $i$th columns, or $u_i \neq 0$ and $\Zoom[(U, V), (X, Y)]$ consists of matrices whose columns are not linearly independent. In this case $F$ is identically $0$ on $\Zoom[(U,V), (X,Y)]$ and the statement is trivially true. Similarly, if the columns of $V$ are not linearly independent, then  $\Zoom[(U, V), (X, Y)]$ consists of matrices whose columns are not linearly independent, and again $F$ is identically $0$ on $\Zoom[(U,V), (X,Y)]$. Finally, if the columns of $U$ are not linearly independent, then either $\Zoom[(U, V), (X, Y)]$ is empty or there must be some $i$ such that both $u_i$ and $v_i$ are linear combinations of the other columns in $U$ and $V$ respectively, with the same coefficients. In this case, we can remove the $i$th columns of $U$ and $V$ without changing $\Zoom[(U, V), (X, Y)]$.

Now suppose that the columns of $U$ and $V$ are nonzero and linearly independent, and let $A \in \Ff_q^{2\ell \times 2\ell}$ be a full rank matrix such that the column vectors $A U_i = e_i$ for $1 \leq i \leq r$. Let $Y' \in \Ff_q^{d \times (2\ell - r)}$ denote the last $2\ell-r$ rows of $YA^{-1}$. We have
\begin{align*}
    \norm{F_{(U,V), (X,Y)}}_2^2 &= \E_{M \in \Ff_q^{n \times 2\ell}}[F(M ) \; | \; MU = V, XM = Y] \\ 
   &=  \E_{M \in \Ff_q^{n \times 2\ell}}[F(M A^{-1}) \; | \; MU = V, XM = Y] \\
   &=  \E_{M \in \Ff_q^{n \times 2\ell}}[F(M) \; | \; MAU = V, XMA = Y] \\
   &= \E_{M' \in \Ff_q^{n \times (2\ell - r')}}[F([v_1,\ldots, v_r, M']) \; | \; X M' = Y'].
\end{align*}
In the second equality we use the fact that $F$ is basis invariant and in the last equality we use the fact that because $\Zoom[(U,V), (X,Y)]$ is nonempty,  the matrix product $X[v_1, \ldots, v_r]$ has the same first $r$ rows as $Y$.

To complete the proof, we show how to reduce to the case that $Y'$ is the zero matrix. First note that, using the same reasoning as we did for $U$ and $V$, we can assume that the nonzero rows $Y'$ are linearly independent and the rows of $X$ are linearly independent. Suppose that $y'_1, \ldots, y'_{a} \in \Ff_q^{2\ell-r}$ are the nonzero (and linearly independent) rows of $Y'$, while the remaining rows are $y'_{a+1}, \ldots, y'_s = 0$. Let $Y'' \in \Ff_q^{a \times (2\ell-r)}$ be the first $a$ rows of $Y'$, which are nonzero, let $X' \in \Ff_q^{a \times n}$ denote the first $a$ rows of $X$, and let $X'' \in \Ff_q^{(2\ell - a) \times n}$ denote rows $a+1$ through $s$ of $X$. For any $Z= (z_1, \ldots, z_a)\in \Ff_q^{a \times (2\ell-r)}$ with $a$ linearly independent rows let $A_Z \in \Ff_q^{(2\ell-r) \times (2\ell-r)}$ be the full rank matrix such that $Y'' A_Z = Z$. Then, for any linearly independent $z_1, \ldots, z_a \in \Ff_q^{2\ell-r}$, 
\begin{align*}
    &\E_{M' \in \Ff_q^{n \times (2\ell - r')}}[F([v_1,\ldots, v_r, M']) \; | \; X M' = Y']\\
    &\qquad\qquad=   \E_{M' }[F([v_1,\ldots, v_r, M']) \; | \; X' M' = Y'', X'' M' = 0] \\
    &\qquad\qquad=  \E_{M' }[F([v_1,\ldots, v_r, M'  A_Z^{-1}]) \; | \; X'M' = Y'', X'' M' = 0]  \\
    &\qquad\qquad=\E_{M' }[F([v_1,\ldots, v_r, M'  ]) \; | \; X' M' = Z, X'' M' = 0].
\end{align*}
In the second transition, we used the fact that $F$ is basis invariant and multiplied
its input by the matrix whose top left 
$r\times r$ minor is the identity, its
bottom right $(2\ell-r)\times (2\ell-r)$ 
is $A_Z^{-1}$, and the rest of the entries 
are $0$.
Since the above holds for any $Z$ with $a$ many linearly independent rows, letting $E$ denote the event that $X' M'$ has $a$ linearly independent rows, it follows that
\[
 \E_{M' \in \Ff_q^{n \times (2\ell - r')}}[F([v_1,\ldots, v_r, M']) \; | \; X M' = Y'] = \E_{M' }[F([v_1,\ldots, v_r, M'  ]) \; | \; E \land  X'' M' = 0],
\]
and 
\begin{align*}
    \E_{M' \in \Ff_q^{n \times (2\ell - r')}}[F([v_1,\ldots, v_r, M'  ]) \; | \; X''M' = 0] &\geq \Pr_{M' \in \Ff_q^{n \times (2\ell - r')}}[E \; | \; X''M' = 0]\\
    &\cdot \E_{M'}[F([v_1,\ldots, v_r, M']) \; | \; E\land X'' M' = 0].
\end{align*}
Finally since $\Pr[E \; | \; X''M' = 0] \geq \frac{1}{2}$ (as it is the probability of choosing $a < 2\ell - r$ linearly independent vectors in $\Ff_q^{2\ell-r}$), we have,
\begin{align*}
    \norm{F_{(U,V), (X,Y)}}_2^2 &= \E_{M' \in \Ff_q^{n \times (2\ell - r')}}[F([v_1,\ldots, v_r, M']) \; | \; E\land X'' M' = 0] \\
    &\leq 2 \E_{M'}[F([v_1,\ldots, v_r, M'  ]) \; | \; X''M' = 0],
\end{align*}
and the proof is concluded.
\end{proof}
As an immediate consequence of~\cref{lm: level d intermediate lemma 1} we have the following result.
\begin{lemma} \label{lm: preserve pseudorandom}
   If $\mc{L} \subseteq \Grass_q(n,2\ell)$ is $(r, \epsilon)$-pseudo-random, then $F$ defined from $\mc{L}$ as in \eqref{eq: define F from L} is $(r, 2\epsilon)$-pseudo-random.
\end{lemma}
\begin{proof}
    Fix any $(U, V)$ and $(X, Y)$ such that $\dim(U) + \codim(X) = r$. Using~\cref{lm: level d intermediate lemma 1}, there are linearly independent $v_1, \ldots, v_{r'} \in \Ff_q^{n}$ and $X' \in \Ff_q^{s' \times n}$ with linearly independent rows such that 
    \begin{align*}    
     &\norm{F_{(U,V), (X,Y)}}_2^2\\ 
     &\qquad\qquad\leq 2 \cdot \E_{M' \in \Ff_q^{n \times (2\ell-r')}}[F\left([v_1,\ldots, v_{r'}, M']\right) \; | \; X' M' = 0] \\
     &\qquad\qquad\leq 2 \cdot \E_{M' \in \Ff_q^{n \times (2\ell-r')}}[F\left([v_1,\ldots, v_{r'}, M']\right) \; | \; X' M' = 0, \dim(\im([v_1,\ldots, v_r, M']))=2\ell ] \\
     &\qquad\qquad=  2 \cdot \Pr_{M' \in \Ff_q^{n \times (2\ell-r')}}[\im([v_1,\ldots, v_{r'}, M']) \in \mc{L} \; | \; X' M' = 0, \dim(\im([v_1,\ldots, v_r, M'])) = 2\ell ],
    \end{align*}
    where in the second transition we are using the fact that $F(M) = 0$ for all $M$ such that $\dim(\im(M))< 2\ell$, and in the third transition we are using the definition of $F$. We will bound the final term by using the pseudo-randomness of $F$. 
    
    Choosing $M' \in \Ff_q^{n \times (2\ell-r')}$ uniformly conditioned on $X'M' = 0$, and $\dim(\im([v_1,\ldots, v_r, M']))=2\ell$, we claim that $\im([v_1,\ldots, v_r, M'])$ is a uniformly random $2\ell$-dimensional subspace in $\Zoom[Q, Q + H]$, where $H$ is the codimension $s$ subspace that is dual to the rows of $X'$ and $Q = \spa(v_1,\ldots, v_r)$.
    To see why, first note that it is clear $\im([v_1,\ldots, v_r, M']) \in \Zoom[Q, Q + H]$. Additionally, each $L \in \Zoom[Q, Q + H]$ has an equal number of $M'\in \Ff_q^{n \times (2\ell-r')}$ such that 
    \[
    L = \im([v_1,\ldots, v_r, M']),
    \]
    and therefore has an equal chance of being selected. It follows that choosing $M' \in \Ff_q^{n \times (2\ell-r')}$ uniformly conditioned on $X'M' = 0$, and $\dim(\im([v_1,\ldots, v_r, M']))$, $\im([v_1,\ldots, v_r, M'])$ is a uniformly random $2\ell$-dimensional subspace in $\Zoom[Q, Q + H]$. As a result, 
    \begin{align*}
     &\norm{F_{(U,V), (X,Y)}}_2^2 \\
     &\qquad\qquad\leq 2 \cdot \Pr_{M' \in \Ff_q^{n \times (2\ell-r')}}[\im([v_1,\ldots, v_{r'}, M'])\in \mc{L} \; | \; X' M' = 0, \dim(\im([v_1,\ldots, v_r, M']))=2\ell ] \\
     &\qquad\qquad= 2 \cdot \Pr_{L \in \Zoom[Q, Q + H]}[L \in \mc{L}] \\
     &\qquad\qquad\leq 2\epsilon,
    \end{align*}
    where in the last transition we use the fact that the set of subspaces $\mc{L}$ is $(r,\epsilon)$-pseudo-random and $\dim(Q) + \codim(Q + H) \leq r' + s \leq r$.
\end{proof}

The next lemma shows that the value of the probability of interest in \cref{lm: pseudorandom edges}
is not much larger than the value $\langle \T F, G \rangle$:
\begin{lemma} \label{lm: edge count preserved}
We have 
\[
\Pr_{L \in \Grass_q(n, 2\ell), R \in \Grass_{q}(n, 2(1-\delta)\ell), L \supseteq R}[L \in \mc{L}, R \in \mc{R}]\leq  2\langle \T F, G \rangle .
\]
\end{lemma}
\begin{proof}
 We have,
    \begin{align*}
       &\langle \T F, G \rangle  = \E_{M \in \Ff_q^{n\times 2\ell}, A \in \Ff_q^{2\ell \times 2(1-\delta)\ell}}[F(M) \cdot G(MA) \; | \; \rank(A) = 2(1-\delta)\ell] \\
       &\quad\geq \Pr_{M,A}[\rank(M)= 2\ell, \rank(MA)= 2(1-\delta)\ell \; | \; \rank(A) = 2(1-\delta)\ell] \\
       &\quad\cdot \E_{M, A}[F(M) \cdot G(MA) \; | \; \rank(M)= 2\ell, \rank(MA)= 2(1-\delta)\ell, \rank(A) = 2(1-\delta)\ell] \\
       &\quad\geq \frac{1}{2}\E_{M, A}[F(M) \cdot G(MA) \; | \; \rank(M)= 2\ell, \rank(MA)= 2(1-\delta)\ell, \rank(A) = 2(1-\delta)\ell] \\
       &\quad= \frac{1}{2} \E_{M, A}[F(\im(M)) \cdot G(\im(MA)) \; | \; \rank(M)= 2\ell, \rank(MA)= 2(1-\delta)\ell, \rank(A) = 2(1-\delta)\ell].
    \end{align*}    
To finish the proof, notice that in the conditional distribution $(\im(M), \im(MA))$ in the last term, $\im(M)$ is a uniform  $L \in  \Grass_{q}(n, 2\ell)$ and $\im(MA)$ is a uniform $R \in \Grass_q(n, 2(1-\delta)\ell)$ such that $R \subseteq L$. Thus, the final expectation is exactly the probability
\[
p := \Pr \limits_{L \in \Grass_q(n, 2\ell), R \in \Grass_{q}(n, 2(1-\delta)\ell), L \supseteq R}[L \in \mc{L}, R \in \mc{R}].
\]
Therefore, 
\[
     \langle \T F, G \rangle \geq \frac{1}{2}\E_{L, R}[F(L) \cdot G(R) \; | \; L \supseteq R] = \frac{p}{2}\qedhere
\]
\end{proof}

We now prove~\cref{lm: pseudorandom edges} by combining~\cref{lm: preserve pseudorandom,lm: edge count preserved}.
\begin{proof}[Proof of~\cref{lm: pseudorandom edges}]
Take $\mc{L} \subseteq \Grass_q(n, 2\ell)$ and $\mc{R} \subseteq \Grass_q(n, 2(1-\delta)\ell)$ with fractional sizes $\alpha$ and $\beta$ as in the lemma statement. Using these sets, define the associated functions $F \in L_2\left(\Ff_q^{n\times 2\ell}\right)$ and $G \in L_2\left(\Ff_q^{n\times 2(1-\delta)\ell}\right)$ as in~\eqref{eq: define F from L}. It is clear that
\[
\norm{F}_2^2 \leq  \alpha, \quad \norm{G}_2^2 \leq \beta.
\]
Furthermore, by~\cref{lm: edge count preserved}, we have that
\[
 \Pr_{L \supseteq R}[L \in \mc{L}, R \in \mc{R}] \leq 2\langle \T F, G \rangle.
\]
By~\cref{lm: preserve pseudorandom} $F$ is $(r, 2\epsilon)$-pseudo-random, and applying~\cref{lm: pseudorandom edges over bilinear scheme} we get that 
\[
 \Pr_{L \supseteq R}[L \in \mc{L}, R \in \mc{R}]  \leq 2\langle \mathcal{T} F, G \rangle \leq q^{O_{t,r}(1)} \beta^{\frac{t-1}{t}} \eps^{\frac{t-2}{t}} + q^{-r\delta\ell}\sqrt{\alpha\beta}.
\qedhere
\]
\end{proof} 
\subsection{Proof of~\cref{lm: fourier even covering}} \label{app: even covering}

We will show that if a set of $\ell'$-dimensional subspaces $\Lc^{\star} \subseteq \Grass_q(W, \ell')$ is pseudo-random, then it must ``evenly cover'' the space $W$ in the sense that there are very few points $z \in W$ such that $\mu_z(\Lc^\star)$ significantly deviates from $\mu(\Lc^\star)$. We will require the following result from \cite{EvraKL}.
\begin{thm}\cite[Theorem 1.14]{EvraKL} \label{th: EKL other}
    If $F \in L_2\left(\Ff_q^{n \times \ell'} \right)$ is a Boolean function which is $(1, \epsilon)$-global, then for all $t$ that are powers of $2$, we have
    \[
    \norm{F^{=1}}_2^2 \leq q^{460t} \norm{F}_2^{2} \epsilon^{1-2/t}.
    \]
\end{thm}

We are now ready to prove~\cref{lm: fourier even covering}.
\begin{proof}[Proof of~\cref{lm: fourier even covering}]
    Let $\dim(W) = n$, let $F' \in L_2\left(\Ff_q^{n \times \ell'}\right)$ be the function associated with $\mc{L}^\star$ given by 
    \begin{equation*}
    F'(x_1,\ldots, x_{\ell'})=
    \begin{cases}
      1 & \text{if}\ \spa(x_1,\ldots, x_{2\ell}) \in \mc{L}^\star, \\
      0 & \text{otherwise}.
    \end{cases}
  \end{equation*}
  By Lemma~\ref{lm: preserve pseudorandom}, $F'$ is $(1, 2q^{c \ell'}\eta)$-pseudo-random. For any point $x \in W$, we have 
\[
\mu_z(\Lcal^\star) = \E_{x_1,\ldots, x_{\ell'-1} \in W}\left[F'(x_1,\ldots, x_{\ell'-1}, z) \; | \; \dim(\spa(x_1,\ldots, x_{\ell'-1}, z)) = \ell' \right],
\]
so it follows that 
\[
\left|\mu_z(\Lcal^\star) -  \E_{x_1,\ldots, x_{\ell'-1}}\left[F'(x_1,\ldots, x_{\ell'-1}, z) \right]\right| \leq  \frac{q^{\ell'}}{q^n} \quad \text{and} \quad \left|\mu(\Lc^\star) - \norm{F'}_2^2 \right| \leq \frac{q^{\ell'}}{q^n}.
\]
Thus 
\begin{equation}\label{eq:lvl_1_wt_use}
\left| \E_{x_1,\ldots, x_{\ell'-1}}\left[F'(x_1,\ldots, x_{\ell'-1},z)\right] -  \norm{F'}_2^2 \right| \geq \left|\mu_z(\Lc^\star) - \mu(\Lc^\star) \right| - \frac{q^{\ell'}}{q^n}.
\end{equation}
We will now relate this quantity to the level $1$ weight of $F'$. Note that
  \begin{align*}
  \E_{x_1,\ldots, x_{\ell'-1}}\left[F'(x_1,\ldots, x_{\ell'-1}, z) \right] &=  \sum_{S = (s_1,\ldots, s_{\ell'}) \in W}\widehat{F'}(S)  \E_{x_1,\ldots, x_{\ell'-1}}[\chi_S(x_1,\ldots, x_{\ell'-1}, z)] \\
  &=  \sum_{S = (s_1,\ldots, s_{\ell'}) \in W}\widehat{F'}(S) \chi_{s_{\ell'}}(z)\prod_{i=1}^{\ell'-1} \E_{x_i}\left[\chi_{s_i}(x_i)\right].
  \end{align*}
Now note that $\E_{x_i}\left[\chi_{s_i}(x_i)\right] = 0$ if $s_i$ is not the zero vector, and $\E_{x_i}\left[\chi_{s_i}(x_i)\right] = 1$ if $s_i$ is the zero vector. Thus, 
\begin{align*}
     \E_{x_1,\ldots, x_{\ell'-1}}\left[F'(x_1,\ldots, x_{\ell'-1}, z) \right] &= \widehat{F'}(0, \ldots, 0) + \sum_{a \in W, a \neq 0} \widehat{F'}(0,\ldots,0, a)\chi_a(z),
\end{align*}
and using the fact that $\widehat{F'}(0,\ldots, 0) = 
\E[F'] = \norm{F'}_2^2$, 
\[
\left(\E_{x_1,\ldots, x_{\ell'-1}}\left[F'(x_1,\ldots, x_{\ell'-1}, z) \right]-  \norm{F'}_2^2 \right)^2 = \left( \sum_{a \in W, a \neq 0} \widehat{F'}(0,\ldots,0, a)\chi_a(z) \right)^2 .
\]
Therefore, we get by~\eqref{eq:lvl_1_wt_use} that
\begin{align*}
    \E_{z \in W}[(\mu_z(\Lc^\star) - \eta)^2] &\leq 
    \E_{z\in W}\left[\left| \E_{x_1,\ldots, x_{\ell'-1}}\left[F'(x_1,\ldots, x_{\ell'-1},z)\right] -  \norm{F'}_2^2 \right|^2\right] + 5\frac{q^{\ell'}}{q^n} \\
    &\leq \E_z\left[\left|\sum_{a \in W, a \neq 0} \widehat{F'}(0,\ldots,0, a)\chi_a(z) \right|^2\right] + 5\frac{q^{\ell'}}{q^n} \\
    &= \sum_{a \in W, a \neq 0}  |\widehat{F'}(0,\ldots,0, a)|^2+ 5\frac{q^{\ell'}}{q^n}.
\end{align*}
We now examine the summation in the last line. Since $F'$ is basis invariant, using~\cref{lm: character identity}, we have that for all $\alpha_1,\ldots, \alpha_{\ell'}\in \Ff_q$ that are not all zero,
\begin{align*}
    \widehat{F'}\left(\alpha_1 a,\ldots, \alpha_{\ell'} a\right) =  \widehat{F'}\left(0,\ldots, 0, a\right).
\end{align*}
It follows that 
\[
 \sum_{a \in W, a \neq 0}  |\widehat{F'}(0,\ldots,0, a)|^2 = \frac{1}{q^{\ell'}-1}\sum_{\rank(S) = 1}|\widehat{F'}(S)|^2 = \frac{\norm{F'^{=1}}_2^2}{q^{\ell'}-1}.
\]
Using~\cref{th: EKL other} along the fact that $F'$ is $(1, 2q^{c \ell'}\eta)$-pseudo-random, we get
\[
    \E_{z \in W}[(\mu_z(\Lc^\star) - \eta)^2] \leq \frac{\norm{F'^{=1}}_2^2}{q^{\ell'}-1}+ 5\frac{q^{\ell'}}{q^n}
\leq \frac{q^{461 t} q^{c \ell'} \eta^{2-\frac{2}{t}}}{q^{\ell'}-1},
\]
for any $t \geq 4$ that is a power of $2$, where we also used the fact that $\norm{F'}_2^2\leq 2\mu(\mc{L})\leq 2\eta$ by~\Cref{lm: preserve pseudorandom}. By Markov's inequality it follows that
\[
\frac{\left|\overline{Z}\right|}{q^n} \cdot \frac{\eta^2}{100} \leq \frac{100 q^{461t} q^{c \ell'}}{\eta^{\frac{2}{t}} q^{\ell'}}\eta^2 \leq  \frac{100 q^{461t}}{q^{\left(1 - c - \frac{2C}{t} \right)\ell'}}\eta^2 \leq 
\frac{q^{-\frac{\ell'}{2}}}{100}\eta^2.
\]
In the third term, we set $C:= \frac{6}{\xi}$, and use the fact that $\ell' \geq \frac{\xi}{3} \ell$, so  $\eta \geq q^{-2\ell} \geq q^{-C\ell'}$. Then for the last transition we take $t$ to be the smallest power of $2$ that satisfies $1 - c + \frac{2C}{t} \geq \frac{2}{3}$, and use the fact that $\ell$, and hence $\ell'$, is sufficiently large. Dividing by $\eta^2$ finishes the proof.
\end{proof}
\section{Proof of~\cref{th: consistency many zoom in}}\label{app: inner pcp}
In order to prove~\cref{th: consistency many zoom in} we will find the subspaces $Q$ one at a time by using~\cref{th: consistency}. We maintain a list $\mathcal{Q}$ of all $Q$'s collected thus far. Each time a new subspace $Q$ is added to $\mathcal{Q}$, we randomize the assignment $T_1[L]$ for all $2\ell$-dimensional $L \supset Q$. At a high level, the effect of this randomization is that there is only a little agreement between any linear function and the assignments on subspaces containing $Q$, thus these entries are essentially ``deleted''.

Formally, we construct the set $\mathcal{Q}$ of subspaces as follows. Initially set $\widetilde{T}_1 = T_1$, $\mathcal{Q} = \emptyset$, and $\mc{X} = \emptyset$. Recall that initially $\widetilde{T}_1$ and $T_2$ are $\epsilon$-consistent for $\epsilon \geq 2q^{-2\ell(1-1000\delta)}$. While $\widetilde{T}_1$ and $T_2$ are at least $\epsilon/2$-consistent, do the following.

\begin{enumerate}
    \item Let $Q \subset W$ be subspaces guaranteed by~\cref{th: consistency}. That is, $Q$ and $W$ satisfy $\dim(Q) + \codim(W) = r$ and there exists linear $g_{Q,W}: W \to \Ff_q$ such that  \[
    \Pr_{L \in \Grass(n, 2\ell)}[g_{Q,W}|_{L} = \widetilde{T}_1[L] \; | \; Q \subseteq L \subseteq W] \geq \epsilon' := q^{-2(1-1000\delta^2)\ell}.
    \]
    \item Set $\mathcal{Q} \xleftarrow[]{} \mathcal{Q} \cup \{Q\}$.
    \item Set $\mc{X} \xleftarrow[]{} \mc{X} \cup \{L \; | \; Q \subseteq L \subseteq W\}$.
    \item For each $L\in \mc{X}$ independently, choose $\widetilde{T}_1[L]$ uniformly among all linear functions on $L$.
\end{enumerate}
We have the following claim regarding the re-assignment phase.
\begin{claim} \label{cl: reassignment}
With probability at least $1-e^{-\Omega(q^{\ell n})}$ over the random assignment $\widetilde{T}_1$, for every pair of subspaces $Q \subset W$ such that $\dim(Q) + \codim(W) = r$ and $\mu_{Q,W}(\mc{X}) \geq q^{-2\ell}$, and every linear function $g_{Q,W}: W \xrightarrow[]{} \Ff_q$, we have
    \[
    \Pr_{L \in \mc{X}, L \in \Zoom_{2\ell}[Q,W]}[g_{Q,W}|_L \equiv \widetilde{T}_1[L]] \leq q^{1-2\ell}
    \]

\end{claim}
\begin{proof}
    Fix $Q, W, g_{Q,W}$ such that $\dim(Q) + \codim(W) = r$ and $\mu_{Q,W}(\mc{X}) \geq q^{-2\ell}$.
 Let $\mc{A} = \Zoom_{2\ell}[Q,W] \cap \mc{X}$ and for each $L \in \mc{A}$, let $Z_L$ denote the indicator variable that takes value $1$ if $g_{Q,W}|_L = \widetilde{T}_1 [L]$ and $0$ otherwise, where $\widetilde{T}_1$ is the randomly chosen assignment over $\mc{X}$. Over randomly chosen $\widetilde{T}_1$,
 the expectation of $Z_L$ is $q^{-2\ell}$, so by a Chernoff bound, we have
    \[
    \Pr_{\widetilde{T}}\left[\frac{1}{|\mc{A}|}\sum_{L \in \mc{A}}Z_L \geq q^{1-2\ell}  \right]  \leq e^{-q^2\left(q^{-2\ell}|\mc{A}|\right)/6}.
    \]
    By assumption, $|\mc{A}| \geq q^{-2\ell}\left|\{L \in \Zoom_{2\ell}[Q,W], \dim(L) = 2\ell\}\right| \geq q^{-2\ell}q^{(2\ell-r)(n-r-2\ell)}$. Thus, using a union bound over all $Q, W, g_{Q,W}$, the probability that there exist a bad triple is at most,

    \[
    (r+1)q^{nr}q^n e^{-q^{-4\ell+2}q^{(2\ell-r)(n-r-2\ell)}/6} 
    \leq e^{-\Omega(q^{\ell n})}.\qedhere
    \]
\end{proof}
We now analyze the process. 
Note that when we find a new triplet $(L,Q,g_{Q,W})$, then by Chernoff's bound with probability 
$1-e^{-\Omega(q^{\ell n})}$ over the randomization step,
the probability 
$\Pr_{L \in \Grass_q(n, 2\ell)}[g_{Q,W}|_{L} = \widetilde{T}_1[L] \; | \; Q \subseteq L \subseteq W]$ drops from at least $\eps'$ to 
at most $q^{1-2\ell}$. In that case, the measure of $\mc{X}$ increases by 
at least
\[
(\eps' - q^{1-2\ell})q^{-rn}
\geq q^{-O(rn)}.
\]
Doing a union bound over the steps, it follows that with probability $1-e^{-\Omega(q^{\ell n})}q^{O(rn)} = 1-o(1)$
the process terminates within $q^{O(rn)}$ steps.

Note that it is possible that the same subspace $Q$ is added multiple times (with different zoom-outs) in the process above, so we clarify that $\mathcal{Q}$ is considered as a set without repeats. Also note that with probability $1-o(1)$, for each 
$Q\in\mathcal{Q}$, $W$ and $g_{Q,W}$ found in the process it 
holds that 
\begin{equation}\label{eq:randomize_T_1}
    \Pr_{L \in \Grass(n, 2\ell)}[g_{Q,W}|_{L} = T_1[L] \; | \; Q \subseteq L \subseteq W] \geq \frac{1}{2}\epsilon'
\end{equation}
(the point being is that the agreement now is compared
to the original $T_1$ and not to $\widetilde{T}_1$). Indeed, considering the step $Q,W$ and $g_{Q,W}$ were found, 
$g_{Q,W}$ had agreement at least $\eps'$ with $\widetilde{T}_1$
on $\Zoom_{2\ell}[Q,W]$ at that point, and by Claim~\ref{cl: reassignment} with probability $1-e^{-\Omega(q^{\ell n})}$ 
at most $q^{1-2\ell}\leq \eps'/2$ of that agreement
came from $L\in \mc{X}$. Thus, by union bound over all of the steps,
with probability $1-q^{O(rn)}e^{q^{-\Omega(\ell n)}} = 1-o(1)$
it follows that~\eqref{eq:randomize_T_1} holds for every 
$Q,W$ and $g_{Q,W}$ found throughout the process.

The following claim shows that at the end of the process the number of $Q$'s found in the process is large, thereby finishing the proof of~\cref{th: consistency many zoom in}.
\begin{claim}
    There exists some $0 \leq r_1 \leq r$ such that $\mathcal{Q}$ contains at least a $q^{-5\ell^2}$-fraction of all $r_1$-dimensional subspaces.
\end{claim}
\begin{proof}
    At the end of the process, the consistency has dropped by at least $\epsilon'/2$, so the probability over edges $(L, R)$ that $L$ was reassigned must be at least $\epsilon'/2$. For each $0 \leq r_1 \leq r$, let $N_{r_1}$ be the number of $Q$ of dimension $r_1$ in $\mathcal{Q}$. 
    
    For each $Q$ of dimension $r_1$, the fraction of $2\ell$-dimensional $L$'s that are reassigned due to $Q$ being added to $\mathcal{Q}$ is at most the fraction of $2\ell$-dimensional subspaces that contain $Q$. This is,
    \[
    \frac{\qbin{n}{2\ell-r_1}}{\qbin{n}{2\ell}} \leq \frac{q^{n(2\ell-r_1)}}{q^{2\ell(n-2\ell)}} = q^{4\ell^2 - r_1n}.
    \]
    It follows that there must be some $r_1$ such that 
    $N_{r_1}q^{4\ell^2 - r_1n} \geq \frac{\epsilon}{2(r+1)}$, and rearranging gives that 
    \[
    N_{r_1} \geq \frac{\epsilon}{2(r+1)q^{4\ell^2}}q^{r_1n} \geq q^{-5\ell^2}\qbin{n}{r_1}.
    \]
    Thus there exists an $r_1$ such that $\mathcal{Q}$ contains at least a $q^{-5\ell^2}$-fraction of all $r_1$-dimensional subspaces.
\end{proof}

\section{Proof of~\cref{lm: basic covering}}\label{app: covering}  \label{app: covering basic}
Fix a question $U$ to the first prover. Recall from~\eqref{eq: pcp parameters} that we set
\[
k = q^{2(1+c)\ell} \quad \text{and} \quad \beta = q^{-2(1+2c/3)\ell},
\]
where $0 < c < 1$ is some small constant close to $0$ and set $\eta = q^{-100\ell^{100}}$. Let us recall the distributions $\D$ and $\D'$, as well as ~\cref{lm: basic covering}.

\noindent $\D:$
\begin{itemize}
    \item Choose $x_1, \ldots, x_{2\ell} \in \Ff_q^U$ uniformly.
    \item Output the list $(x_1,\ldots , x_{2\ell})$.
\end{itemize}

\noindent $\D':$
\begin{itemize}
    \item Choose $V \subseteq U$ according to the outer PCP.
    \item Choose $x'_1,  \ldots, x'_{2\ell} \in \Ff_q^V$ uniformly, and lift these vectors to $\Ff_q^U$ by inserting $0$'s into the missing coordinates.
    \item Choose $w_{1}, \ldots, w_{2\ell} \in H_U$ uniformly, and set $x_i = x'_i + w_i$ for $1 \leq i \leq 2\ell$.
    \item Output the list $(x_1, \ldots, x_{2\ell})$.
\end{itemize}

\basiccovering*
\begin{proof}
 For $x_1, \ldots, x_{2\ell} \in \Ff_q^{|U|}$, let us view $x_1, \ldots, x_{2\ell}$ as the rows of a $2\ell \times 3k$ matrix, and split the columns of this matrix into $k$ blocks - each consisting of $3$ consecutive columns. The three columns of each one of these $k$ blocks correspond to the three variables coming from one equation of the outer PCP. Let $s(x_1, \ldots, x_{2\ell})$ be the number of blocks where exactly two of the columns are equal, and set $p = 3q^{-2\ell} - 3q^{-4\ell}$ to be the probability that a fixed block has exactly two columns equal to each other. The intention is that $s(x_1, \ldots, x_{2\ell})$ equals the number of equations where we drop variables in the outer PCP.
 
 Let $s'(x_1,\ldots, x_{2\ell})$ be the number of blocks where all $3$ columns are equal, and let $p' = q^{-4\ell}$ be the probability that a fixed block has all three columns equal.
We define the sets $E_1,E_2,E$ as follows:
 \begin{equation}  \label{eq: E1}
 E_1 = \left\{(x_1, \ldots, x_{2\ell}) \in \left(\Ff_q^{U}\right)^{2\ell} \; | \; \left|s(x_1,\ldots, x_{2\ell}) - pk\right| > 50 \sqrt{p k \log(1 / \eta)} \right\},
 \end{equation}

 \begin{equation} \label{eq: E2}
 E_2 = \left\{(x_1, \ldots, x_{2\ell}) \in \left(\Ff_q^{U}\right)^{2\ell} \; | \; s'(x_1,\ldots, x_{2\ell}) > \ell^{100} \right\},
 \end{equation}
 and $E = E_1 \cup E_2$. We first argue that the event $E$ has small probability under both $\mathcal{D}$ and $\mathcal{D}'$. Note that
 \[
\E_{\D}[s(x_1,\ldots, x_{2\ell})] = pk \quad \text{and} \quad \E_{\D'} [s(x_1,\ldots, x_{2\ell})] = (1-\beta)pk + \beta(1-q^{-2\ell})k,
 \]
 and in particular,
 \[
\left|\E_{\D'} [s(x_1,\ldots, x_{2\ell})] - pk \right|\leq \beta k.
 \]
 
 By a Chernoff bound,
\begin{equation*} 
    \D(E_1) =    \Pr_{\D}\left[|s(x_1,\ldots, x_{2\ell}) - pk| > 50 \sqrt{p k \log(1 / \eta)} \right] \leq \eta^{50},
\end{equation*}
where recall that $\eta = q^{-100\ell^{100}}$. Also, by our setting of $\beta k$, we have $\beta k = q^{2c\ell/3}$, while $p k = \Omega(q^{2c\ell})$, so the same Chernoff bound holds for $\D'$:
\begin{equation*} 
     \D'(E_1) = \Pr_{\D'}\left[|s(x_1,\ldots, x_{2\ell}) - pk| > 50 \sqrt{p k \log(1 / \eta)} \right] \leq \eta^{50}.
\end{equation*}

For the event $E_2$ we have,
\[
\D(E_2) = \Pr_{D}[s'(x_1,\ldots, x_{2\ell}) >\ell^{100}] \leq \binom{k}{\ell^{100}} p'^{\ell^{100}} \leq (kp')^{\ell^{100}} \leq \eta^{100},
\]
where in the middle term, the first factor is the number of ways to choose $\ell^{100}$ blocks and the second factor is the probability that all of these blocks have all three columns equal. Similarly, 
\[
\D'(E_2) = \Pr_{D'}[s'(x_1,\ldots, x_{2\ell}) >\ell^{100}] 
\leq 
\binom{k}{\ell^{100}} \left((1-\beta)p' + \beta \frac{1}{q^{2\ell}}\right)^{\ell^{100}} \leq k^{\ell^{100}} \left(2q^{-4\ell}\right)^{\ell^{100}} \leq \eta^{100}.
\]
Putting everything together, we get that
\begin{equation}
    \D(E) \leq \D(E_1) + \D(E_2) \leq \eta^{40} \quad \text{and} \quad \D'(E) \leq \D'(E_1) + \D'(E_2) \leq \eta^{40}.
\end{equation}
We next show that the probability measures 
$\mathcal{D}$ and $\mathcal{D}'$ assign roughly the same
measure to each $x\not\in E$. 

Fix $(x_1, \ldots, x_{2\ell}) \notin E$. It is clear that $\mc{D} (x_1, \ldots, x_{2\ell}) = q^{-2\ell\cdot 3k}$, where we use $|U| = 3k$. Let $s = s(x_1, \ldots, x_{2\ell})$ and $s' =s'(x_1, \ldots, x_{2\ell})$. Then,
\begin{align} \label{eq: covering equation}
            \mc{D}'(x_1, \ldots, x_{2\ell}) &=  ((1-\beta)q^{-3 \cdot 2\ell})^{k-s-s'} \left((1-\beta)q^{-3 \cdot 2\ell}+ \frac{\beta}{3}q^{-4 \ell} \right)^{s} \left((1-\beta)q^{-3\cdot 2\ell} + \beta q^{-4\ell} \right)^{s'}\notag\\
    &= q^{-2\ell \cdot 3k} (1-\beta)^{k-s-s'}\left(1-\beta + \frac{\beta}{3}q^{2\ell}\right)^s (1-\beta + \beta q^{2\ell})^{s'}
\end{align}
In the first equality, the first term is the probability of choosing the blocks that have three distinct columns. Then, $(1-\beta)$ is the probability that no variables are dropped, and $q^{-3(2\ell)}$ is the probability of choosing those three particular $x_i$'s in that block. The second term is the probability of choosing the blocks that have exactly two equal columns. Then, $(1-\beta)q^{-3(2\ell)}$ is the probability of having no variables dropped and choosing the three $x_i$'s, and $ \frac{\beta}{3}q^{-4\ell}$ is the probability of first having the variable dropped in the column that is not equal to the other two, and then choosing the correct values for the remaining two column values. The third term similarly corresponds to blocks that have all columns being equal.

We start by showing
\[
\frac{\mc{D}'(x_1, \ldots, x_{2\ell})}{\mc{D} (x_1, \ldots, x_{2\ell})} \geq \frac{1}{1.1}.
\]
Using~\eqref{eq: covering equation},
    \begin{align*}
        \frac{\mc{D}'(x_1, \ldots, x_{2\ell}) }{\mc{D}(x_1, \ldots, x_{2\ell})} &=   (1-\beta)^{k-s-s'}\left(1-\beta + \frac{\beta}{3}q^{2\ell}\right)^s (1-\beta + \beta q^{2\ell})^{s'}\\
        &\geq (1-\beta)^{k-s}\left(1-\beta + \frac{\beta}{3}q^{2\ell}\right)^s \\
        &= (1-\beta)^{k-s}\left(1 + \beta \left(\frac{q^{2\ell}}{3}-1\right)\right)^s\\
        &\geq (1-\beta(k-s))\left(1 + \beta s\left(\frac{q^{2\ell}}{3}-1\right)\right)\\
        &\geq 1-\beta k + \beta s \frac{q^{2\ell}}{3}.
    \end{align*}
where in the fourth transition we use the bound $(1+z)^n \geq 1 + nz$ which holds for all $n \geq 1$ and $z \geq -1$.

Now write $s = p k - v$ and let us analyze the last line. Plugging this in and using the definition of $p$ we get that
    \begin{equation}\label{eq: covering equation 2}
       1  -\beta k +  \frac{q^{2\ell}}{3}\beta s =1 -\beta k +  \frac{q^{2\ell}}{3} \beta (pk - v) 
         = 1-\beta k\left(1 -\frac{q^{2\ell}}{3} p \right) -\frac{q^{2\ell}}{3} \beta v 
         =1 -\beta k q^{-2\ell} -\frac{q^{2\ell}}{3} \beta v .
    \end{equation}
Hence, plugging in our values for $\beta, k$ and $p$, we get
\begin{align*}
        \frac{\mc{D}'(x_1, \ldots, x_{2\ell}) }{\mc{D}(x_1, \ldots, x_{2\ell})} &\geq 1 -\beta k q^{-2\ell} -\frac{q^{2\ell}}{3} \beta v = 1 - q^{-2\ell + 2c\ell/3} - \frac{q^{-4c\ell/3} v}{3} \geq \frac{1}{1.1},
   \end{align*}
In the last transition we use $v \leq 50 \sqrt{p k \log(1 / \eta)} \leq q^{(c+o(1))\ell}$ (where above and henceforth the $o(1)$ terms are as $\ell$ goes to infinity) and the fact that $\ell$ is sufficiently large.

For the other direction, we show 
\[
\frac{\mc{D}'(x_1, \ldots, x_{2\ell})}{\mc{D} (x_1, \ldots, x_{2\ell})}  \leq \frac{1}{0.9},
\]
in nearly the same fashion. First note that 
\begin{equation} \label{eq: bound s'}
\left(\frac{1-\beta + \beta q^{2\ell}}{1-\beta}\right)^{s'} \leq \left(\frac{1-\beta + \beta q^{2\ell}}{1-\beta}\right)^{\ell^{100}} = \left(1 + \frac{\beta}{1-\beta}q^{2\ell}\right)^{\ell^{100}} \leq 
1+o(1),
\end{equation}
By~\eqref{eq: covering equation}, we have 
\begin{align*}
        \frac{\mc{D}'(x_1, \ldots, x_{2\ell}) }{\mc{D}(x_1, \ldots, x_{2\ell})} &= (1-\beta)^{k-s-s'}\left(1-\beta + \frac{\beta}{3}q^{2\ell}\right)^s (1-\beta + \beta q^{2\ell})^{s'}\\
        &\leq (1+o(1)) \cdot (1-\beta)^{k-s}\left(1-\beta + \frac{\beta}{3}q^{2\ell}\right)^s \\
        &= (1+o(1))\cdot (1-\beta)^{k-s}\left(1 + \beta \left(\frac{q^{2\ell}}{3}-1\right)\right)^s\\
        &\leq (1+o(1))\cdot \exp\left(-\beta(k-s) +  \left(\frac{q^{2\ell}}{3}-1\right) \beta s \right) \\
         &= (1+o(1))\cdot \exp\left(-\beta k +  \frac{q^{2\ell}}{3} \beta s  \right).\\
\end{align*}
where in the first transition we used~\eqref{eq: bound s'}, in the fourth transition we use the fact that $1 + z \leq e^{z}$. Writing $s = pk+v$ and using~\eqref{eq: covering equation 2} (but with $+v$ instead of $-v$) and $v \leq O\left(q^{(c/2+o(1))\ell}\right)$ we have 
\[
 \frac{\mc{D}'(x_1, \ldots, x_{2\ell}) }{\mc{D}(x_1, \ldots, x_{2\ell})}  \leq 
 (1+o(1))\exp\left(-\beta k q^{-2\ell} +\frac{q^{2\ell}}{3} \beta v\right) \leq 
 (1+o(1))\exp\left(O(q^{-c\ell/3-o(1)})   \right) \leq \frac{1}{0.9}.\qedhere
\]
\end{proof}
\section{List Decoding Bound} \label{app: list decoding}
In this section we prove~\cref{lm: table list decoding}, which is a direct consequence of the generic list decoding bound of~\cite{GRS}:

\begin{thm}\label{th: generic list decoding} \cite[Theorem 15]{GRS}
    Let $\mathcal{C} \subseteq \Sigma^N$ be a code with alphabet size $M := |\Sigma|$, blocklength $N$, and relative distance $1-\gamma$. Let $\delta > 0$ and $R \in \Sigma^N$. Suppose that $C_1, \ldots, C_m \in \Sigma^N$ are distinct codewords from $\mathcal{C}$ that each differ from $R$ on at most $(1-\delta)$-fraction of entries. If
    \[
    \delta > \frac{1}{M} + \sqrt{\left(\gamma - \frac{1}{M}\right)\left(1 - \frac{1}{M}\right)},
    \]
    then
    \[
    m \leq \frac{1}{(\delta- 1/M)^2 - (1-1/M)(\gamma - 1/M)}.
    \]
\end{thm}
Proving~\cref{lm: table list decoding} is simply a matter of translating to the notation of~\cref{th: generic list decoding}, and below are the details.
\begin{proof}[Proof of~\cref{lm: table list decoding}]
Let $\Sigma = \Ff_q^{2\ell-r_1}$ and write $Q$ as $Q := \spa(z_1, \ldots, z_{r_1})$.  Define a code $\mathcal{C} \subseteq \{P\colon W^{2\ell - r_1} \to \Ff_q^{2\ell - r_1} \}$ consisting of $P: W^{2\ell - r_1} \to \Ff_q^{2\ell - r_1}$ such that there exists $v \in W$ satisfying 
\[
P(x_1, \ldots, x_{2\ell-r_1})
=(v \cdot x_1, \ldots, v \cdot x_{2\ell-r_1})
\]
for all $x_1, \ldots, x_{2\ell-r_1} \in W^{2\ell-r_1}$.

Note that for distinct $v, w \in W$ we have that $v \cdot x = w \cdot x$ for at most $1/q$-fraction of $x \in W$. Thus, the relative distance of $\mc{C}$ is $1- q^{-2\ell+r_1}$. We would like the table $T$ corresponds to a word, say $P'$, and $f_1, \ldots, f_m$ correspond to $m$ codewords in $\C$, say $C_1, \ldots, C_m$. A slight issue is that $T$ is only defined over $2\ell$-dimensional subspaces of $L \in \Zoom_{2\ell}[Q,W]$, while $P$ has an entry for every $(2\ell-r_1)$-tuple of points in $W$. To resolve this, note that nearly every $(2\ell-r_1)$-tuple of points combined with $z_1, \ldots, z_{r_1}$ span an $L \in \Zoom_{2\ell}[Q,W]$. Thus, define $P$ as follows. If $(z_1, \ldots, z_{r_1}, x_1,\ldots, x_{2\ell-r_1})$ are linearly independent, then let $L$ be the span of $(z_1, \ldots, z_{r_1}, x_1,\ldots, x_{2\ell-r_1})$ and define
\[
P'(x_1,\ldots, x_{2\ell-r_1}) = \left(T[L](x_1), \ldots, T[L](x_{2\ell-r_1})\right).
\]
Otherwise, define $P'(x_1,\ldots, x_{2\ell-r_1})$ arbitrarily. Note that the fraction of tuples $(x_1,\ldots, x_{2\ell-r_1})$ such that  $(z_1, \ldots, z_{r_1}, x_1,\ldots, x_{2\ell-r_1})$ are not linearly independent is at most,
\[
\sum_{i=r_1+1}^{2\ell}\frac{q^{i-1}}{q^n} \leq q^{2\ell - n},
\]
so nearly all of the entries in $P'$ correspond to table entries in $T$. For the functions $f_1, \ldots, f_m$ we define $C_i$ corresponding to $f_i$ by
\[
C_{i}(x_1,\ldots, x_{2\ell-r_1}) = (f_i(x_1), \ldots, f_i(x_{2\ell-r_1})).
\]
As each $f_i$ agrees with $T$ on at least $\beta$-fraction of the entries, we have that $P'$ and $C_i$ agree on at least $\beta$-fraction of the entries $(x_1,\ldots, x_{2\ell-r_1})$ such that  $(z_1, \ldots, z_{r_1}, x_1,\ldots, x_{2\ell-r_1})$ are linearly independent, so
\[
\delta(P, C_i) \leq 1 - \beta \cdot (1-q^{2\ell - n}) \leq 1 - \frac{\beta}{2}
\]
for each $1 \leq i \leq m$. Finally, note that the alphabet size of $\mathcal{C}$ is $\left|\Ff_q^{2\ell-r_1}\right| = q^{2\ell-r_1}$. To bound $m$, we can apply~\cref{th: generic list decoding} with $\delta = \frac{\beta}{2} \geq q^{-2\ell+r_1} + \frac{c}{2}$, $M = q^{2\ell-r_1}$, and $\gamma = q^{-2\ell+r_1}$. We first note that the condition of~\cref{th: generic list decoding} is indeed satisfied:
\[
\delta \geq q^{-2\ell+r_1} + \frac{c}{2} > q^{-2\ell+r_1} + 0.
\]
Thus~\cref{th: generic list decoding} implies that
$m \leq \frac{4}{c^2}$.
\end{proof}
\section{Missing Proofs from Section~\ref{sec: bounded zoom-outs}}\label{app: bounded zoom-outs}
 This section contains the missing proofs from~\cref{sec: bounded zoom-outs}, and we begin 
 by recalling some notation. We recall the notation $\mu_{X, W}, \mu_{X, \circ}$, and $\mu_{\circ, W}$ from \eqref{eq: measure in zoom-in zoom-out}. Throughout this section, for a subspace $L$ and a set of constant codimension subspaces $\mathcal{W}$, we use the notation $N_\W(L) = |\{W \in \W \; | \; L \subseteq W \}|$ from~\eqref{eq: def N_W(L)}.
 
 Let us also recall the setup in \cref{sec: bounded zoom-outs}. We have the following items:
\begin{itemize}
    \item $\delta > 0$ is a fixed small constant, $\ell$ is taken sufficiently large relative to $1/\delta$, $r$ is a codimension parameter and is at most $10/\delta$, and
    \[
\xi := \delta^{5}, \quad \delta_2 := \xi/100, \quad t := \left(2^{2+10/\delta_2}\right)! \; .
\]
    \item $W'_{\amb}$ is the ambient space and has dimension at least $n - 10/\delta$. Also, $n \geq 2^{100} q^{\ell}$ and $\ell$ is thought of as going to infinity relative to $1/\delta$. Broadly, we just want that $n$ along with the dimension of $W'_{\amb}$ is ``much larger'' than $\ell$.
    \item $\mc{W} = \{W_1,\ldots, W_{m_1}\}$ is a set of $t$-generic subspaces of codimension $s$ with respect to $W'_{\amb}$, where $s \leq r$, and $m_1 \geq q^{75r \ell \xi^{-1}}$.
    \item $T$ is a table defined over $\Grass(\Ff_q^n, 2\ell) \supseteq \Grass(W'_{\amb}, 2 \ell)$. For this section, we mainly work in $W'_{\amb}$ and thus focus on $T$ over $\Grass(W'_{\amb}, 2 \ell)$. 
    \item For each $i \in [m_1]$, we have a linear function $f_i: W_i \to \Ff_q$ such that 
    \[
    \Pr_{L \in \Grass(W, 2\ell)}[f_i|_L \equiv T[L]] \geq C \geq q^{-2(1-\xi)\ell}.
    \]
\end{itemize}
We are now ready to move onto the missing proofs for \cref{sec: bounded zoom-outs}.

\subsection{Proof of~\cref{lm: excellent}}\label{sec:lem8.1}
This subsection is devoted to the proof of ~\cref{lm: excellent} which is restated below for convenience. We recall the following definitions for a subspace $X$ (of dimension less than $2\ell$) and a linear assignment $\sigma$ to $X$:
\[\Lc_X = \Zoom_{2\ell}[X, W'_{\amb}] \quad \text{and} \quad \Lc_{X, \sigma} = \{L \in \Lc_X \; | \; T[L]|_X \equiv \sigma \}.\]
Also,
\[\W_X = \{W_i \in \W \; | \; X \subseteq W_i \} \quad \text{and} \quad \W_{X, \sigma} = \{W_i \in \W_X \; | \; f_i|_X \equiv \sigma \}.
\]
Finally, for the remainder of this subsection, we remind the reader that $\gamma = 10^{-6}$ is a small constant and $C \geq q^{-2(1-\xi)\ell}$, and $\xi > 0$ is the small constant fixed at the start of \cref{sec: bounded zoom-outs} (and recalled at the start of \cref{app: bounded zoom-outs}).

\excellent* 

For each $2\left(1-\frac{\xi}{2}\right)\ell$-dimensional subspace $X$ and linear assignment, $\sigma$, to $X$, let us define
\[
p_{X,\sigma} := \Pr_{\substack{W_i \in \W_{X, \sigma}\\L \in \Zoom_{2\ell}[X,W_i]}}[L \in \Lc_{X, \sigma} \land f_i|_L \neq T[L] ],
\qquad
q_{X,\sigma} := \Pr_{\substack{W_i \in \W_{X, \sigma}\\ L \in \Zoom_{2\ell}[X,W_i]}}[L \in \Lc_{X, \sigma} ],
\]
where in both probabilities $X$ and $\sigma$ are fixed, and $W_i \in \W_{X,\sigma}$ is chosen  uniformly and $L \in \Zoom_{2\ell}[X, W_i]$ is chosen uniformly. The intention behind these values is that for a fixed $(X,\sigma)$, the quantity $p_{X,\sigma}$ should reflect how much disagreement there is between the table $T$ and the functions $f_i$ for $W_i \in \W_{X,\sigma}$, on subspaces $L \in \Lc_{X,\sigma}$, while $q_{X,\sigma}$ should reflect the size of $\Lc_{X,\sigma}$. Note that if $L \in \Lc_{X,\sigma}$ and $W_i \in \W_{X,\sigma}$, then by definition we already have $f_i|_X \equiv T[L]|_X \equiv \sigma$. Therefore we would expect that in fact $f_i$ and $T[L]$ also agree on $L$ (since it is only larger than $X$ by $\xi \ell$ dimensions), meaning that we expect $p_{X,\sigma}$ to typically be small. Also, for each $W_i$, there are at least a $C$-fraction of $L \in \Grass(W_i, 2\ell)$ for which $f_i|_L \equiv T[L]$, so we would expect $q_{X,\sigma}$ to be $\Omega(C)$ for a non-trivial fraction of $(X, \sigma)$. In the following claim, we formalize this intuition and show that there indeed exists an $(X,\sigma)$ for which $p_{X,\sigma}$ is small, $q_{X,\sigma}$ is large, and additionally the set $\W_{X,\sigma}$ is large. 

\begin{remark}
The idea of looking for such $(X,\sigma)$ was first introduced in \cite{IKW} where they call these $(X, \sigma)$-excellent and was used again in \cite{BDN, MZ} to analyze lower dimensional subspace versus subspace tests, which is similar in spirit to what we are ultimately trying to show in~\cref{lm: bound on good zoom outs}.
\end{remark}

\begin{claim} \label{cl: excellent 1}
    There exists $(X,\sigma)$ where $X$ is a $2\left(1-\frac{\xi}{2}\right)$-dimensional subspace and $\sigma$ is a linear function on $X$ such that:
    \begin{itemize}
        \item $m_2 := |\W_{X,\sigma}| \geq \frac{m_1}{q^{10r\ell}}$.
        \item $q_{X,\sigma} \geq \frac{C}{2}$.
        \item $p_{X,\sigma} \leq \gamma \cdot q_{X,\sigma}$.
    \end{itemize}
\end{claim}
\begin{proof}
Consider the following process which outputs $W_i, L, X, \sigma$ such that $W_i$ is uniform in $\W$, $L \in \Grass(W_i, 2\ell)$ is uniform, $X \in \Grass\left(L, 2\left(1-\frac{\xi}{2}\right)\ell\right)$ is uniform, and $\sigma$ is the assignment of $f_i$ to $X$, i.e\ $\sigma :\equiv f_i|_X$.

\begin{enumerate}
    \item Choose $(X, \sigma)$ with probability proportional to $|\W_{X, \sigma}|$. 
    \item Choose $W_i \in \W_{X, \sigma}$ uniformly.
    \item Choose a $2\ell$-dimensional subspace $L$ uniformly conditioned on $X \subseteq L \subseteq W_i$.
\end{enumerate}

Notice that the marginal distribution over $(W_i, L)$ above is equivalent to that of choosing $W_i \in \W$ uniformly and $L \subset W_i$ uniformly. Furthermore, in the distribution above, conditioned on an $L$, $X \subseteq L$ is also uniform. Moreover, $f_i|_L \equiv T[L]$ only if $L \in \Lc_{X, \sigma}$, as $f_i$ and $T[L]$ must agree on $X \subseteq L$ in order to agree on $L$. Therefore,

\begin{equation} \label{eq: qxsig}
\E_{X, \sigma}[q_{X,\sigma}] \geq \Pr_{W_i \in \W, L \subseteq W_i}[f_i|_L \equiv T[L]] \geq C \geq \frac{1}{q^{2(1-\xi)\ell}}.
\end{equation}
On the other hand,
\begin{equation} \label{eq: pxsig}
    \E_{X,\sigma}[p_{X,\sigma}] \leq 
    \Pr_{X \subseteq L}[f_i|_X \equiv T[L]|_X \land f_i|_L \neq T[L] ] \leq \frac{1}{q^{2(1-\xi/2)\ell}}.
\end{equation}
Here the distribution of $(X,\sigma)$ is proportional to the sizes $|\W_{X,\sigma}|$ and the second inequality is by the Schwartz-Zippel lemma. Indeed, by the Schwartz-Zippel lemma, if $f_i|_L$ and $T[L]$ are distinct then they agree on at most $1/q$-fraction of points $z$ in $L$. Therefore, the middle term is bounded by the probability that $2(1-\xi/2)\ell$ uniformly random, linearly independent points are all chosen in this $1/q$-fraction.

Then,
\[
\E_{X, \sigma}[q_{X,\sigma} - \gamma p_{X,\sigma}] \geq C- \frac{\gamma}{q^{2(1-\xi/2)\ell}} \geq 0.99C
\]
By Markov's inequality, we get that for at least $(0.49C)$-fraction of $(X,\sigma)$ we have, $q_{X,\sigma} - \gamma p_{X,\sigma} \geq \frac{C}{2}$, and hence for at least $(0.49C)$-fraction of $(X,\sigma)$ we have both $q_{X,\sigma} \geq C/2$ and $p_{X,\sigma} \leq \gamma q_{X,\sigma}$.



Next we wish to argue that for most of these $(X, \sigma)$'s, $|\W_{X,\sigma}|$ is large. First note that the total number of $(X, \sigma)$'s is $\qbin{n}{2(1-\xi/2)\ell}q^{2\left(1-\frac{\xi}{2}\right)\ell}$. For a fixed $(X, \sigma)$, the probability that it is chosen is precisely, 
\begin{equation} \label{eq: prob choose x sig}
    \frac{|\W_{X,\sigma}|}{m_1} \cdot \frac{1}{\qbin{n-s}{2(1-\xi/2)\ell}}.
\end{equation}
Therefore, 
\begin{equation} \label{eq: many wxsig}
    \Pr_{X,\sigma}\left[|\W_{X,\sigma}| \leq \frac{m_1}{q^{10r\ell}}\right] \leq \frac{1}{q^{10r\ell}} \cdot \frac{\qbin{n}{2(1-\xi/2)\ell}q^{2\ell}}{\qbin{n-r}{2(1-\xi/2)\ell}} \leq \frac{1}{q^{10r\ell}} \cdot q^{4r\ell} \cdot q^{2\ell} \leq \frac{1}{q^{5r\ell}}.
\end{equation}
In the first transition we union bound over the probability of choosing $(X,\sigma)$ for all $(X,\sigma)$ such that $|\W_{X,\sigma}| \leq \frac{m_1}{q^{10r\ell}}$. Using ~\eqref{eq: prob choose x sig}, this probability is at most,
\[
\frac{1}{q^{10r\ell}} \cdot  \frac{1}{\qbin{n-s}{2(1-\xi/2)\ell}} \leq \frac{1}{q^{10r\ell}} \cdot  \frac{1}{\qbin{n-r}{2(1-\xi/2)\ell}}
\]
for each $(X,\sigma)$ (recall that $s \leq r$), and we have to union bound over at most $\qbin{n}{2(1-\xi/2)\ell}q^{2\left(1-\frac{\xi}{2}\right)\ell}$ many $(X,\sigma)$'s. 
Altogether, it follows that with probability at least
\[
0.49C  - \frac{1}{q^{5r\ell}} > 0,
\]
over $(X,\sigma)$, we have, $q_{X,\sigma} \geq \tau$, $p_{X,\sigma} \leq \gamma \tau$, and $|\W_{X,\sigma}| \geq \frac{m_1}{q^{10r\ell}}$, which establishes the claim.
\end{proof}
Taking the $(X,\sigma)$ given by~\cref{cl: excellent 1} almost works for~\cref{lm: excellent}. However, notice that while the probability of interest for the third item there looks similar to $p_{X,\sigma}$, it has a different distribution over $L$ and $W_i$. Indeed, in~\cref{lm: excellent} one first chooses $L \in \Lc_{X,\sigma}$ and then $W_i \in \W_{X,\sigma}$ containing $L$, whereas for $p_{X,\sigma}$, we are first choosing $W_i \in \W_{X,\sigma}$ and then $L\subseteq W_i$ (and in particular we do not condition on $L$ being in the set $\Lc_{X,\sigma}$). Intuitively, we expect the following three quantities to be roughly equal
\[
\Pr_{L \in \Lc_{X,\sigma}, W_i \in \W_{X,\sigma}}[f_i|_L \neq T[L] \; | \; W_i\supseteq L] \approx \Pr_{W_i \in \W_{X,\sigma}, L \subseteq W_i}[f_i|_L \neq T[L] \; | \; L \in \Lc_{X,\sigma}] \approx \frac{p_{X,\sigma}}{q_{X,\sigma}}.
\]
Once we establish the above, the third item of \cref{cl: excellent 1} yields $\frac{p_{X,\sigma}}{q_{X,\sigma}} \leq \gamma$, which completes the proof of \cref{lm: excellent}. Thus, the bulk of the transition from~\cref{cl: excellent 1} to~\cref{lm: excellent} is in converting from the distribution used in $p_{X,\sigma}$ to that required by the third item of~\cref{lm: excellent} without losing too much. The rest of the argument is devoted to this goal.

\begin{proof}[Proof of~\cref{lm: excellent}]
Fix $(X, \sigma)$ from~\cref{cl: excellent 1} and define $\W_{X,\sigma}$ and $\Lc_{X, \sigma}$ accordingly. Let $m_2 := |\W_{X,\sigma}| \geq \frac{m_1}{q^{10r\ell}}$. In order to lower bound  $\mu_{X, \circ}(\Lc_{X,\sigma})$ we use~\cref{lm: nu vs mu} on the collection of subspaces $\W_{X,\sigma}$ with parameters $j = 2\ell$, $a = 2\left(1 - \frac{\xi}{2}\right) \ell$. Indeed, consider the measure $\nu_{ \W_{X,\sigma}}$ over $\Zoom_{2\ell}[X, W'_{\amb}]$, which, recall, is obtained by choosing $W_i \in \W_{X,\sigma}$ uniformly and then a uniform subspace in $\Zoom_{2\ell}[X, W_i]$, and note that this measure is precisely the measure 
corresponding to $q_{X,\sigma}$. Thus, applying~\cref{lm: nu vs mu} we get
\[
\mu_{X,\circ}(\Lcal_{X,\sigma}) \geq \nu_{\W_{X,\sigma}}(\Lcal_{X,\sigma}) - \frac{3q^{\frac{s}{2}\xi \ell}}{\sqrt{m_2}}=
q_{X,\sigma} - \frac{3q^{\frac{s}{2}\xi \ell}}{\sqrt{m_2}} \geq \frac{C}{6}.
\]
Define $\tau := q_{X,\sigma}$ so that $\tau \geq C/2$ and $p_{X,\sigma} < \gamma \cdot \tau$.
Note that the first two conditions of~\cref{lm: excellent} are already satisfied, so it remains to check the third condition. 

To this end it will be helpful to have in mind the bipartite graph with parts $\W_{X,\sigma}$ and $\Lc_{X,\sigma}$ and edges $(W_i, L)$ if $L \subseteq W_i$. For each $L' \in \Lc_{X}$ and $W_i \in \W_{X,\sigma}$ define the following degree-like quantities:
\begin{itemize}
    \item $N_{\mc{W}_{X,\sigma}}(L') := |\{W_i \in \W_{X, \sigma} \; | \; W_i \supseteq L' \}|$,
    \item $e_{L'} := |\{W_i \in \W_{X, \sigma} \; | \; W_i \supseteq L', f_i|_{L'} \neq T[L'] \}|$,
    \item $d_i  := |\{L \in \Lc_{X,\sigma} \; | \; L \subseteq W_i \}|$,
    \item $e_i := |\{L \in \Lc_{X,\sigma} \; | \; L \subseteq W_i,   f_i|_{L} \neq T[L]\}|$.
\end{itemize}
Note that the first quantity is the same as in~\eqref{eq: def N_W(L)}. We also let $D = |\{L \; | \; X \subseteq L \subseteq W_i, \dim(L) = 2\ell \}|$, where the $W_i \in \W_{X,\sigma}$ is arbitrary (the value is the same regardless which we pick). Then we clearly have that $\E_{L\in \Lc_X}[N_{\mc{W}_{X,\sigma}}(L)] = \frac{m_2D}{|\Lc_X|}$ and the probability that we are interested in can be expressed as:
\[
\Pr_{L \in \Lc_{X,\sigma}, W_i \in \W_{X,\sigma} }[f_i|_L \neq T_1[L] \; | \; L \subseteq W_i] =  \E_{L \in \Lc_{X,\sigma}} \left[\frac{e_L}{N_{\mc{W}_{X,\sigma}}(L)}  \right].
\]
Since $q_{X,\sigma} = \tau$ and $p_{X,\sigma} \leq \gamma \cdot \tau$, we have
\begin{equation*} 
 q_{X,\sigma} \cdot m_2 \cdot D =  \sum_{L \in \Lc_{X,\sigma}} N_{\mc{W}_{X,\sigma}}(L) = \sum_{W_i \in \W_{X,\sigma}} d_i = m_2 \cdot D \cdot \tau,
\end{equation*}
and
\begin{equation}\label{eq: excellent 2}
p_{X,\sigma} \cdot m_2 \cdot D = \sum_{L \in \Lc_{X,\sigma}} e_L =\sum_{W_i \in \W_{X,\sigma}} e_i \leq m_2 \cdot  D \cdot \gamma \cdot \tau.
\end{equation}
By~\cref{lm: deviation bounds}, along with the fact that $\E_{L\in \Lc_X}[N_{\mc{W}_{X,\sigma}}(L)] = \frac{m_2D}{|\Lc_X|}$, and the very loose bound $\frac{D}{|\Lc_X|}  \geq \frac{1}{q^{r\cdot \ell}}$, 
we have

\begin{equation} \label{eq: excellent 3}
    \Pr_{L \in \Lc_X}\Biggl[N_{\mc{W}_{X,\sigma}}(L) \leq 0.9\E_{L\in\Lc_X}[N_{\mc{W}_{X,\sigma}}(L)]  \Biggr] =   \Pr_{L \in \Lc_X}\Biggl[N_{\mc{W}_{X,\sigma}}(L) \leq \frac{0.9 \cdot m_2 D }{|\mc{L}_X|}  \Biggr] \leq \frac{101 q^{r \cdot \xi \ell}}{m_2}.
\end{equation}
We conclude that 
\begin{align*}
     \E_{L \in \Lc_{X,\sigma}} \left[\frac{e_L}{N_{\mc{W}_{X,\sigma}}(L)}  \right] &\leq \Pr_{L \in \Lc_{X, \sigma}}\left[N_{\mc{W}_{X,\sigma}}(L) \leq  \frac{0.9 \cdot m_2 \cdot D}{|\Lc_X|}\right] +  \E_{L \in \Lc_{X,\sigma}}\left[\frac{e_L}{0.9\cdot m_2\cdot D/|\Lc_X|} \right]\\
     &\leq \frac{101 q^{r\cdot \xi \ell}/m_2}{\Pr_{L \in \Lc_X}[L \in \Lc_{X,\sigma}]}+  \E_{L \in \Lc_{X,\sigma}}\left[\frac{e_L}{0.9\cdot m_2 \cdot D/|\Lc_X|} \right]\\
     &\leq \frac{101 q^{r\cdot \xi \ell}/m_2}{C/6} +  \E_{L \in \Lc_{X,\sigma}}\left[\frac{e_L|\Lc_X|}{0.9\cdot m_2 \cdot D} \right] \\
     &\leq \frac{606 q^{r\cdot \ell}}{m_2 \cdot C} + \frac{m_2 \cdot  D \cdot \gamma \cdot \tau}{0.9\cdot m_2 \cdot D} \cdot \frac{|\Lc_X|}{|\Lc_{X,\sigma}|}\\
     &\leq  \frac{606 q^{r\cdot \ell}}{m_2  \cdot C} + \frac{\gamma \tau}{0.9} \cdot \frac{3}{\tau} \\
     &\leq 5\gamma,
\end{align*}
where in the second transition we used~\eqref{eq: excellent 3} and in the fourth transition  we used~\eqref{eq: excellent 2}.
\end{proof}

\subsection{Proof of~\cref{lm: find global 2}} \label{app: proof of find global}
Take the $X, \sigma, \Lc'_{X,\sigma},$ and $\W_{X,\sigma}$ from~\cref{cor: excellent}, and recall $W'_{\amb}$ is the ambient space and $\delta_2 = \frac{\xi}{100}$.  We next recall a few notations: for a zoom-in $A$ and zoom-out $B$ such that $X \subseteq A \subseteq B \subseteq W'_{\amb}$, we write $W'_{\amb} = A \oplus W_{\amb,0}$ and $B = A \oplus W^\star_{\amb}$, where $W^\star_{\amb} \subseteq W_{\amb,0}$. Now define
\[
\W^{\star}_{[A,B]} := \{W^\star_i \; | \; \exists W_i \in \W_{X,\sigma} \text{ s.t } A \oplus W^\star_i = W_i \cap B \}.
\]
It is clear that each $W^\star_i \in \W^{\star}_{[A,B]}$ is contained inside of some $W_i \in \W_{X,\sigma}$, so for each $W^\star_i$, we may define $f^\star_i \equiv f_i|_{W^\star_i}$.
If there are multiple such $i$, we choose one arbitrarily to define $f^\star_i$. With this in mind, we restate~\cref{lm: find global 2}.

\findglobal*

The following result finds the zoom-in and zoom-out pair as required for~\cref{lm: find global 2}, modulo a few minor alterations.

\begin{lemma} \label{lm: find global}
We can find a zoom-in $A$ and a zoom-out $B$ such that $X \subseteq A \subseteq B \subseteq W'_{\amb}$, 
such that the following hold.
    \begin{itemize}
        \item $\dim(A) + \codim(B) \leq \dim(X) + \frac{10}{\delta_2}$.
        \item The set of subspaces $\mathcal{L}' :=  \mc{L}'_{X, \sigma} \cap \Zoom_{2\ell}[A, B]$ satisfies $\eta := \mu_{A,B}(\Lc')\geq \frac{C}{12}$ and $\mc{L}'$ is $(1, q^{\delta_2\ell} \eta)$-pseudo-random .\footnote{By $(1, q^{\delta_2\ell} \eta)$-pseudo-random in $\Zoom_{2\ell}[A, B]$ we mean that $\Lcal'$ does not increase its fractional size to $q^{\delta_2\ell} \eta$ when restricted to any zoom-in containing $A$ or any zoom-out contained in $B$.}
    \end{itemize}
\end{lemma}

\begin{proof}
Set $A_0 = X$,  $\Lc_0 = \Lc_{X,\sigma}$, $B_0 = V$, and $\eta_0 = \mu_{X,\circ}(\Lc_{X,\sigma}) \geq \frac{C}{12}$. Now do the following.

\begin{enumerate}
    \item Set $i = 0$, and initialize $A_0, \Lc_0, B_0, \eta_0$ as above.
    \item If $\Lc_i$ is $(1, q^{\delta_2 \ell}\eta_i)$-pseudo-random inside of $\Zoom_{2\ell}[A_i, B_i]$, then stop.
    \item Otherwise, there exist $A \subseteq B$ such that, $A_i \subseteq A \subseteq B \subseteq B_{i}$, $\dim(A) + \codim( B) = \dim(A_i) + \codim(B_i) + 1$, and $\mu_{A, B}(\Lc_i) \geq q^{\delta_2 \ell} \eta _i$. 
    \item Set $A_{i+1} := A$, $B_{i+1} := B$, $\Lc_{i+1} :=\mc{L}_i \cap \Zoom_{2\ell}[A_{i+1}, B_{i+1}]$, and let $\eta_{i+1} := \mu_{A_{i+1}, B_{i+1}}(\Lc_{i+1})$.
    \item Increment $i$ by $1$ and return to step $2$.
\end{enumerate}
Suppose this process terminates on iteration $j$. We claim that taking $\Lc' = \Lc_j$, $A = A_j$, and $B = B_j$, satisfies the requirements of the lemma. 

First, note that by construction $\eta_{i+1} \geq q^{\delta_2 \ell}\eta_i$. Therefore, we perform at most $\ceil[\Big]{\frac{\log(12/C)}{\log(q^{\delta_2 \ell})}} \leq   \frac{10}{\delta_2}$ iterations before stopping, so $j \leq \frac{10}{\delta_2}$. By construction $\Lc_{j}$ is $(1, q^{\delta_2\ell}\eta_j)$-pseudo-random in $\Zoom_{2\ell}[A_j, B_j]$ and has fractional size $\eta_j \geq \mu_{X,\sigma}(L_{X, \sigma}) \geq \frac{C}{12}$ in $\Zoom_{2\ell}[A_j, B_j]$. Moreover, $\dim(A_j) + \codim(B_j) =  \dim(X) + j \leq \dim(X) + \frac{10}{\delta_2}$, so the conditions of the lemma are satisfied.
\end{proof}
Take $A, B$ and $\Lc'$ given by~\cref{lm: find global}, as well as $\mc{W}_{X,\sigma}$ from~\cref{lm: excellent}. We are now ready to prove \cref{lm: find global 2}.
\begin{proof}[Proof of~\cref{lm: find global 2}]
With $A$ and $B$ set as above, we now construct the $\Lc^\star, \W^\star,$ that satisfy~\cref{lm: find global 2}. Let $W^{\star}_{\amb}$ be a subspace such that $A \oplus W^{\star}_{\amb} = B$, set $\ell' = 2\ell - \dim(A)$, and let
    \[
    \Lc^{\star} = \{L^{\star} \in \Grass_{q}(W^{\star}_{\amb}, \ell') \; | \; L^{\star} \oplus A \in \Lc'\}.
    \]
    For each $L^\star \in \Lc^\star$, let $L'$ denote the corresponding subspace such that $A \oplus L^\star = L' \in \Lc'$. Note the correspondence,
    \begin{equation}  \label{eq: L bijection}
    L^\star \in \Lc^\star \longleftrightarrow L' = A \oplus L^\star \in \Lc',
    \end{equation}
    is a bijection between $\Lc^\star$ and $\Lc'$ because every subspace in $\mc{L}'$ contains $A$. Abusing notations we let $T$ denote both the original table on $2\ell$-dimensional subspaces, as well as the new table on $\Grass(W^\star_{\amb}, \ell')$, given by $T[L^\star] \equiv T[L']|_{L^\star}$ for each $L^\star \in \Grass(W^\star_{\amb}, \ell')$. It will always be clear, based on the argument in $T[\cdot]$, which assignment we are referring to. 
    
    We obtain $\W^\star$ in a similar way as $\Lc^\star$, however, some care will be needed to ensure that it is $4$-generic.
    Starting with $\mc{W}_{X,\sigma}$, we set 
    \[
    \mc{W}_{X,\sigma, \mc{L}'} = \{W \in \W_{X,\sigma} \; | \; \exists L \in \Lc', W \supseteq L \}.
    \]
    That is, $\mc{W}_{X,\sigma, \mc{L}'}$, consists of the subspaces in $\W_{X,\sigma}$ containing at least one $L \in \Lc'$. Since $A \subseteq L$ for all $L \in \Lc'$, we are guaranteed at every $W \in \W_{X,\sigma, \Lc'}$, contains $A$. Now, take the collection of subspaces $\{W \cap B \; | \; W \in \W_{X,\sigma, \Lc'}\}$. This collection may not be $4$-generic, and moreover may even contain duplicate subspaces, but we take  $\widetilde{\W}^{\star}$ to be its \emph{largest} subset that is $4$-generic with codimension $s$ with respect to $B$ (note that being $4$-generic implies that every subspace is distinct). Finally, set
    \[
    \W^\star = \{W^\star_i \subseteq W^\star_{\amb} \;| \; A \oplus W^\star_i \in \widetilde{\W}^{\star} \},
    \]
    and set $m_3 := \left|\W^\star \right|$. For each $W^\star_i$, choose an arbitrary $W_i \in \mc{W}_{X,\sigma}$ such that $W^\star_i = W_i \cap B$ and set $f^\star_i :\equiv f_i|_{W^\star_i}$. Summarizing, we have the following chain of relations, which took us from $\mc{W}_{X,\sigma}$ to $\W^\star$:
    
        \begin{alignat}{10}\label{eq: W chain}
 &\W_{X,\sigma} &\stackrightarrow{\supseteq L} \hspace{0.3cm}
 &\W_{X,\sigma, \mc{L}'} &\hspace{0.3cm}\stackrightarrow{\text{$\cap B$ and
 make $4$-generic}} \hspace{0.6cm}
 &\widetilde{\W}^\star &\stackleftrightarrow{\text{Subtract subspace $A$}} &\W^\star \\
    &W_i &\hspace{0.2cm}\stackrightarrow{} \hspace{0.3cm}
    &W_i &\stackrightarrow{}\hspace{1.35cm}
    &W_i \cap B &\hspace{0.6cm}\stackrightarrow{} \hspace{0.6cm}
    &W^\star_i \text{ s.t } W^\star_i \oplus A = W_i \cap B 
\end{alignat}
  It will be helpful to refer back to the chains of relations above. The double arrow transitions are bijections, while in the single arrow transitions subspaces are being removed. The second line shows what a generic member of each set looks like, where $W_i$ are the original subspaces in $\W_{X,\sigma}$. It is clear from~\eqref{eq: W chain} that $\W^\star \subseteq \W_{[A,B]}$. When going from $\mc{W}_{X,\sigma,\mc{L}'}$, some subspaces are removed from the collection $\{W \cap B \; | \; W \in W_{X,\sigma,\mc{L}'}\}$, and we use \cref{lm: keep half generic iterated} to upper bound the number of removed subspaces. In particular, since $\mc{W}_{X,\sigma}$ is $2^{2 + 10/\delta_2}$ generic and since by \cref{lm: find global} $B$ has codimension at most 
  \begin{equation} \label{eq: w star codim}
      10/\delta_2 + (\dim(X) - \dim(A)) \leq 10/\delta_2
  \end{equation} in $W'_{\amb}$, we get that
\begin{equation} \label{eq: subspaces lost}
 0 \leq |\mc{W}_{X,\sigma, \Lcal'}| -  |\widetilde{\W}^\star| \leq 2^{2 + 10/\delta_2}.
\end{equation}
We remark that this argument also establishes the desired codimension of $W^\star_{\amb}$ inside of $W'_{\amb}$ for property 3. We now verify that the six properties of~\cref{lm: find global 2} hold.
\paragraph{Property 1.} The subspaces of $\Lc^{\star}$ are of dimension $\ell'$ inside $W^\star_{\amb}$, and 
\[
\ell' = 2\ell - \dim(A) \geq 2\ell - \dim(X) - \frac{10}{\delta_2} \geq \frac{\xi}{3}\ell.
\]
Also, $\mu(\Lcal^{\star}) = \eta$ is the same as the measure of $\Lcal'$ inside $\Zoom_{2\ell}[A, B]$ due to the bijection between $\Lc^\star$ and $\Lc'$ noted in~\eqref{eq: L bijection}. Therefore $\mu(\Lcal^{\star}) = \eta \geq \frac{C}{12}$ by the second part of~\cref{lm: find global}. 

\paragraph{Property 2.} Since $\Lcal'$ does not increase its measure to $q^{\delta_2\ell}\eta$ on any zoom-in containing $A$ or zoom-out inside $B$, it follows that $\Lcal^\star$ is $(1, q^{\delta_2 \ell} \eta)$-pseudo-random.

\paragraph{Property 3.} By construction, $\widetilde{\W}^\star$ is $4$-generic inside of $B$. Since $B = A \oplus W^\star_{\amb}$, and all $W_i \in \widetilde{\W}^\star$ contain $A$, it follows that $\W^\star$ is $4$-generic inside of $W^\star_{\amb}$. Finally, the codimension of $W^\star_{\amb}$ inside of $W'_{\amb}$ was established in \cref{eq: w star codim}.

\vspace{0.3cm} \noindent Before showing the remaining properties, we will state a few useful inequalities. Fix an $L^\star \in \Lc^\star$ and let $L' = L^\star \oplus A$, so that $L' \in \Lc' \subseteq \Lc_{X,\sigma}$. We will use the fact that, by construction of $\W_{X,\sigma, \Lc'}$,
\[
N_{\W_{X,\sigma}}(L') = N_{\W_{X,\sigma, \Lc'}}(L'),
\qquad
N_{\W^\star}(L^\star) = N_{\widetilde{\W}^\star}(L').
\]

Fix $L' \subseteq B$. Note that for every $W \in \mc{W}_{X,\sigma}$ containing $L'$, the subspace $W \cap B$ contains $L'$ as well. Thus, for each $W \in \mc{W}_{X,\sigma}$ contributing to $N_{\W_{X,\sigma}}(L')$, the corresponding $W \cap B$ also contributes to $N_{\widetilde{\mc{W}}^\star}(L')$ unless it was removed in the second transition of~\eqref{eq: W chain}. Hence,

\[
N_{\mc{W}_{X,\sigma}}(L') - \left(|\mc{W}_{X,\sigma , \Lcal'}| - |\widetilde{\W}^\star|\right) \leq N_{\widetilde{\W}^\star}(L') \leq N_{\mc{W}_{X,\sigma}}(L'),
\]
so we have
\begin{equation} \label{eq: property tmp}  
  N_{\W_{X,\sigma}}(L')- \left(|\mc{W}_{X,\sigma, \Lcal'}| - |\widetilde{\W}^\star|\right)\leq N_{\mc{W}^\star}(L^\star) \leq N_{\W_{X,\sigma}}(L').
\end{equation}
Thus, combining this with the upper bound from~\eqref{eq: subspaces lost} as well as the bounds on $N_{\W_{X,\sigma}(L')}$ from \cref{cor: excellent}, we have,
\begin{equation} \label{eq: property sandwich}
  0.95 \cdot m_2 \cdot q^{-\xi \ell \cdot s} - 2^{2 + 10/\delta_2} \leq N_{\W^\star}(L^\star) \leq 1.05 \cdot m_2 \cdot q^{-\xi \ell \cdot s}.
\end{equation}
Since we fixed $L^\star \in \Lc^\star$ arbitrarily, note that~\eqref{eq: property sandwich} holds for every $L^\star \in \mc{L}^\star$.

\paragraph{Property 5.} Fix $L^\star \in \mc{L}^\star$ and let $L' \in \mc{L}'$ be the subspace such that $L' = A \oplus L^\star$. Let $G_{\mc{W}_{X,\sigma}}(L') := | \{W_i \in \mc{W}_{X,\sigma} \; | \; f_i|_{L'} \not \equiv T[L'] \}$, and let $G_{\mc{W}^\star}(L^\star) := | \{W^\star_i \in \mc{W}^\star \; | \; f^\star_i|_{L^\star} \not \equiv T[L^\star] \}$. By the third part of \cref{cor: excellent}, we have that
\begin{equation} \label{eq: property 5 prev assumption}  
\frac{G_{\mc{W}_{X,\sigma}}(L')}{N_{\mc{W}_{X,\sigma}}(L')} \leq 12\gamma.
\end{equation}

To show property 5, we must bound $G_{\mc{W}^\star}(L^\star)/N_{\mc{W}^\star}(L^\star)$, and we will accomplish this by relating it to the above quantity. Recall that for every $W^\star_i \in \mc{W}^\star$ and associated function $f^\star_i$, there is a corresponding $W_i \in \mc{W}_{X,\sigma}$ such that $W^\star_i = W_i \cap B$, and $f^\star_i = f_i|_{W^\star_i}$. In addition, $T[L']|_{L^\star} = T[L^\star]$, so if $f^\star_i|_{L^\star} \not \equiv T[L^\star]$, then $f_i|_{L'} \not \equiv T[L']$. Hence,
\[
G_{\W^\star}(L^\star) \leq  G_{\mc{W}_{X,\sigma}}(L'),
\]
and
\begin{equation}    \label{eq: property 5 tmp 1}
\frac{G_{\mc{W}^\star}(L^\star)}{N_{\mc{W}^\star}(L^\star)} \leq \frac{G_{\mc{W}_{X,\sigma}}(L')}{N_{\mc{W}_{X,\sigma}}(L')} \cdot \frac{N_{\mc{W}_{X,\sigma}}(L')}{N_{\mc{W}^\star}(L^\star)} \leq 12\gamma \cdot \frac{N_{\mc{W}_{X,\sigma}}(L')}{N_{\mc{W}^\star}(L^\star)},
\end{equation}
where we use \eqref{eq: property 5 prev assumption} in the last transition. It remains to bound the ratio on the right hand side above:
\begin{equation} \label{eq: property 5 tmp 2}
\frac{N_{\mc{W}_{X,\sigma}}(L')}{N_{\mc{W}^\star}(L^\star)} \leq 1 + \frac{N_{\mc{W}_{X,\sigma}}(L') - N_{\mc{W}^\star}(L^\star)}{N_{\mc{W}^\star}(L^\star)} \leq 1 + \frac{2^{2 + 10/\delta_2}}{N_{\mc{W}^\star}(L^\star)} \leq 1.1,
\end{equation}
where in the second inequality we are using \eqref{eq: property tmp} along with the upper bound on $|\mc{W}_{X,\sigma, \Lcal'}| - |\widetilde{\W}^\star|$ from \eqref{eq: subspaces lost}. Combining \eqref{eq: property 5 prev assumption}, \eqref{eq: property 5 tmp 1}, and \eqref{eq: property 5 tmp 2}, we get
\[
\frac{G_{\mc{W}^\star}(L^\star)}{N_{\mc{W}^\star}(L^\star)}  \leq 12 \gamma \cdot 1.1 \leq 14\gamma,
\]
as desired.
\paragraph{Property 6.} The only difference between property 6 and \eqref{eq: property sandwich} is that property 6 is expressed in terms of $m_3$, while \eqref{eq: property sandwich} is expressed in terms of $m_2$. From \eqref{eq: property sandwich}, we already have that all of the $N_{\W^\star}(L^\star)$ over $L^\star \in \mc{L}^\star$ are within a multiplicative factor of $1.11$ of each other, so to conclude property 6, we just need to show that there is at least one $L^\star$ such that $N_{\mc{W}^\star}(L^\star)$ is in the interval $[0.95 \cdot p_1 \cdot m_3, 1.05 \cdot p_1 \cdot m_3]$ which follows by \cref{lm: deviation bounds}.

\paragraph{Property 4.}
Since there exists $L^\star$ such that the lower bound from \eqref{eq: property sandwich} and the upper bound from property $6$ hold, we have:
\[
m_3 \geq \frac{1}{1.2 \cdot p_1} \cdot \left(0.95 \cdot m_2 \cdot q^{-\xi \ell \cdot s} - 2^{2+10/\delta_2} \right) \geq \frac{0.6 \cdot m_2}{q^{\xi \ell \cdot s} \cdot p_1}
\]
By \cref{lm: p_1 bound}, we have that $p_1 \leq 1.01 q^{-s \ell'}$, so,
\[
m_3 \geq \frac{0.6 \cdot m_2}{1.01 q^{\xi \ell \cdot s - s \ell'}} \geq \frac{0.6 \cdot m_2}{1.01 q^{s(\xi \ell - 2\ell + \dim(A)  )}} =  \frac{0.6 \cdot m_2}{1.01 q^{s(\dim(A) - \dim(X))}} \geq \frac{m_2 q^{-10s/\delta_2}}{2}.
\qedhere
\]
\end{proof}

\subsection{Proof of~\cref{lm: Wi agreement}}\label{app: Wi Agreement}
This subsection is dedicated to proving~\cref{lm: Wi agreement}, but before going into the proof we 
must first show a basic Fourier analytic fact. We will use some of the set-up from~\cref{app: level d}. The proof of~\cref{lm: Wi agreement} is given after this necessary fact is established.

\subsubsection{A Necessary Fourier Analytic Fact}
 Let $\mathcal{A} \subseteq \Grass(n,j)$ and let $\eta = \mu(\mathcal{A})$. We assume that $n$ is sufficiently large relative to $q$ and $j$ and that $\eta$ is not too small, say
 \begin{equation} \label{eq: fourier fact assumptions}
     n > q^j \quad \text{and} \quad \eta \geq 100jq^{j-n}
 \end{equation}
 to be concrete. Define $F\colon \Ff_q^{n\times j}\to\{0,1\}$ by
\begin{equation*}
F(x_1,\ldots, x_{j}) =
    \begin{cases}
      1, & \text{if}\ \spa(x_1,\ldots, x_{j}) \in \mathcal{A}, \\
      0, & \text{otherwise.}
    \end{cases}
\end{equation*}

In the next lemma, we use the characters $\chi_S$ for $S = (s_1, \ldots, s_j) \in \Ff_q^{n \times j}$ defined as in \eqref{eq: char def}, except we replace the dimension parameter $2\ell$ therein by $j$. We will use the notation $S\subseteq W^{\perp}$ to mean that $s_i\in W^{\perp}$
for each $i$.

\begin{lemma} \label{lm: char product}
    Fix a subspace $W \subseteq \Ff_q^n$, then for any $S = (s_1,\ldots, s_{j}) \in \Ff_q^{n\times j}$ we have,
\begin{equation*}
\E_{x_1,\ldots, x_{j} \in W}[\chi_S(x_1,\ldots, x_{j})] =
    \begin{cases}
      1, & \text{if}\ S \subseteq W^\perp, \\
      0, & \text{if}\ S \subsetneq W^\perp.
    \end{cases}
\end{equation*}
\end{lemma}
\begin{proof}
       If $S \subseteq W^{\perp}$, then for any $x \in W$, we have $s_i \cdot x = 0$ for all $1 \leq i \leq j$, so the first case follows.
    
    Now suppose $S \subsetneq W^{\perp}$, and without loss of generality say that $s_1 \notin W^{\perp}$. We can write,
    \[
    \E_{x_1,\ldots,x_{j} \in W}[\chi_S(x)] = \E_{x_1 \in W}\left[\omega^{\Tr(x_1 \cdot s_1)}\right] \E_{x_2,\ldots, x_j \in W}\left[\omega^{\sum_{i=2}^j \Tr(x_i \cdot s_i)} \right].
    \]
    We will show that $\E_{x_1 \in W}\left[\omega^{\Tr(x_1 \cdot s_1)}\right] = 0$. Notice that it is sufficient to show that $x_1 \cdot s_1$ takes each value in $\Ff_q$ with equal probability over uniformly random $x_1 \in W$. First, since $s_1 \notin W^{\perp}$, $\Pr_{x_1 \in W}[x_1 \cdot s_1 = 0] = \frac{1}{q}$. Next note for any $\alpha \neq 0$,
    \[
\Pr_{x_1 \in W}[x_1 \cdot s_1 = 1] = \Pr_{x_1 \in W}[(\alpha x_1) \cdot s_1 = \alpha] = \Pr_{x_1 \in W}[x_1 \cdot s_1 = \alpha].
    \]
    Therefore, $x_1 \cdot s_1$ takes each of the $q-1$ nonzero values in $\Ff_q$ with probability $\frac{1}{q}$ over uniform $x_1 \in W$, and this concludes the proof.
\end{proof}

We can now show the necessary fact, asserting that if zooming out to $W$ changes the measure of $\mathcal{A}$ considerably, then we can attribute it to a Fourier character $S$ only containing elements from $W^{\perp}$.
\begin{lemma}\label{lm: Fourier fact}
  Suppose the parameters $j, q, n, \eta$ satisfy \eqref{eq: fourier fact assumptions} and let $F$ be defined as above.  If $W \subseteq \Ff_q^n$ has codimension $r$ and satisfies
    \[
    \left|\mu_{\circ,W}(\mathcal{A}) - \eta\right| \geq 0.01\eta,
    \]
    then there is a nonzero $S \in \Ff_q^{n \times j}$ such that $S \subseteq W^\perp$ and $\left|\widehat{F}(S)\right| \geq \frac{\eta}{20q^{r \cdot j}}$.
\end{lemma}
\begin{proof}
    Note that, $\mu_{\circ,W}(\mathcal{A}) = \Pr \limits_{x_1,\ldots, x_{j}\in \Ff_q^n}[\spa(x_1,\ldots, x_{j}) \in \mathcal{A} \; | \;\dim(\spa(x)) = j] $, so 
    \begin{equation}\label{eq:basic_Fourier_fact}
    \left|\mu_{\circ,W}(\mathcal{A}) - \E_{x\subseteq W}\left[F(x) \right]\right| \leq j \cdot q^{j-n},
    \end{equation}
    where the term on the right hand side bounds the probability that uniformly random $x_1, \ldots, x_j$ do not satisfy the conditioning $\dim(\spa(x)) = j$. Using the Fourier decomposition of $F$ we have
    \[
     \E_{x\subseteq W}\left[F(x)\right] = \widehat{F}(0) + \sum_{0 \neq S \subseteq W^{\perp}}\widehat{F}(S)\E_{x \subseteq W}[\chi_S(x)] +  \sum_{0 \neq S \subsetneq W^{\perp}}\widehat{F}(S)\E_{x \subseteq W}[\chi_S(x)].
    \]
    Combining this with~\eqref{eq:basic_Fourier_fact} and~\cref{lm: char product} and using the fact that $\widehat{F}(0) = \eta$, we get
    \[
    \left|\mu_{\circ,W}(\mathcal{A}) - \eta - \sum_{0\neq S \subseteq W^\perp}\widehat{F}(S)  \right| \leq j \cdot q^{j-n}.
    \]
    By the triangle inequality we conclude that 
    \[
    \left|\mu_{\circ,W}(\mathcal{A}) - \eta \right| \leq \left|\sum_{0\neq S \subseteq W^\perp}\widehat{F}(S)  \right|  + j \cdot q^{j-n},
    \]
    and finally by the assumption in the lemma statement we have, 
    \[
    \left|\sum_{0\neq S \subseteq W^\perp}\widehat{F}(S)  \right| \geq 0.1 \cdot \eta - j \cdot q^{j-n}.
    \]
    Since $j \cdot q^{j-n} \leq 0.01 \cdot \eta$ by \eqref{eq: fourier fact assumptions} and there are at most $q^{r \cdot j}$ tuples $S =(s_1,\ldots, s_{j}) \subseteq W^\perp$, the result follows.
\end{proof}

\subsubsection{The Proof of~\cref{lm: Wi agreement}}
We are now ready to return to~\cref{lm: Wi agreement}. Take $\Lc^\star, \W^\star, W^\star_{\amb}$ from~\cref{lm: find global 2}, and recall that 
\[
Z := \{z \in W^{\star}_{\amb} \; | \; |\mu_{z,\circ}(\Lc^{\star}) - \eta| \leq \frac{\eta}{10} \}, \; \eta := \mu(\Lc^\star), \; \text{and} \; m_3 := |\W^\star|.
\]
Also recall that each $W^\star_i \in \W^\star$ has an associated linear function $f^\star_i: W^\star_i \to \Ff_q$ and  $T_{\ell'}$ is a table assigning a linear function to each subspace in $\Grass(W^\star_{\amb}, \ell')$, where $\ell' \geq \frac{\xi}{3} \ell$ is a dimension and $\xi > 0$ should be thought of as a small constant.

We also restate the six items from \cref{lm: find global 2} which the objects above satisfy
    \begin{enumerate}
        \item $\mu(\Lcal^{\star}) := \eta \geq \frac{C}{12}$, where the measure here is over $\Grass_q(W^\star_{\amb}, \ell') \supseteq \Lcal^\star$.
        \item The set $\Lcal^{\star}$ is $(1, q^{\delta_2 \ell}\eta)$-pseudo-random.
        \item Each $W^\star_i$ has codimension $s \leq r$ inside of $W^\star_{\amb}$, $\W^{\star}$ is $4$-generic, with respect to $W^\star_{\amb}$, and $W^\star_{\amb}$ has codimension at most $10/\delta_2$ with respect to $W'_{\amb}$. As a consequence, $\dim(W^\star_{\amb})$ is much larger than $\ell'$. Concretely:
        \[
        \dim(W^\star_{\amb}) \geq \dim(W'_{\amb}) - 10/\delta_2 \geq n - 10/\delta - 10/\delta_2 \geq 2^{99}q^{\ell'},
        \]
        where here we use $n \geq 2^{100}q^{\ell}$.
        \item The size of $\W^\star$ satisfies
        \[
        \frac{m_2\cdot q^{-10s/\delta_2}}{2} \leq m_3 \leq m_2.
        \]
        \item For each $W^\star_i \in \W^\star_i$, there is a linear function $f^\star_i: W^\star_i \to \Ff_q$ such that the following holds. For every $L \in \Lc^{\star}$, choosing $W^\star_i \in \W^{\star}$ uniformly such that $W^\star_i \supseteq L$, we have
    \[
    \Pr_{W^\star_i \supseteq L, W^\star_i \in \W^{\star}}[f^\star_i|_L \not \equiv T_{\ell'}[L]] \leq 14\gamma.
    \]
    \item For every $L \in \mathcal{L}^{\star}$,
    \[
    0.8 \cdot p_1\cdot m_3  \leq N_{\W^\star}(L) \leq 1.2 \cdot p_1 \cdot m_3,
    \]
    where $N_{\W^\star}(L)$ is as defined in \eqref{eq: def N_W(L)}, and 
    \[
    p_1 := \Pr_{L \in \Grass_q(W^\star_{\amb}, \ell')}[L \subseteq W],
    \]
    for an arbitrary $W \subseteq W^\star_{\amb}$ of codimension $s$.
    \end{enumerate}
    
With this context in mind, the remainder of the section is devoted to the proof of \cref{lm: Wi agreement}, restated below for convenience.

\globalagreement*

Let us start by setting up some notation. For an arbitrary fixed point $z \in W^\star_{\amb}$, let $D$ denote the number of $\ell'$-dimensional subspaces $L \subseteq W^\star_{\amb}$ containing $z$. We note that $D$ does not depend on which point $z$ is fixed. Also let,
\begin{equation}
    \begin{split}
        & \Lc^\star_z = \{L \in \Lc^\star \; | \; z \in L \}, \\
        &\W^\star_z = \{ W^\star_i \in \W^\star \; | \; z \in W^\star_i\},\\
        &m_z = \left|\W^\star_z\right|, \\
        &N_{2, \W^\star_z}(L) = |\{(i,j) \; | \; W^\star_i \cap W^\star_j \supseteq L, \;  W^\star_i, W^\star_j \in \W^\star_z\}|, \\
        &N_{2,\W^\star}(L) = |\{(i,j) \; | \;  W^\star_i \cap W^\star_j \supseteq L,  W^\star_i, W^\star_j \in \W^\star\}|.
    \end{split}
\end{equation}
We remark that in the last two equations, we do not require $i$ and $j$ to be distinct. Now, for an arbitrary pair of distinct $W^\star_i$ and $W^\star_j$ which both contain some point $z$, define 
\begin{equation} 
p_1 = \Pr_{L \in \Grass(W^\star_{\amb}, \ell')}[L \subseteq W^{\star}_i] \quad \text{and} \quad p_2 = \Pr_{L \in \Grass(W^\star_{\amb}, \ell')}[L \subseteq W^\star_i \cap W^\star_j \; | \; z \in L].
\end{equation}
Note that the above quantities depend neither on the specific identity of the distinct pair of subspaces $W^\star_i, W^\star_j$, nor on the identity of the point $z$. Also, we point out that $p_1$ here is the same as defined in Property 6 of \cref{lm: find global 2}. A straightforward computation shows that 
\begin{equation} \label{eq: p2 p1 ratio}
\frac{p_2}{p_1^2} \geq \frac{q^{2s}}{2},
\end{equation}
where recall $s = \codim(W^\star_i)$ in $W^\star_{\amb}$. We start by removing all $z \in Z$ that do not satisfy 
\begin{equation} \label{eq: m_z bounds}
1.1\cdot q^{-s} m_3 \geq m_z \geq 0.9 \cdot q^{-s} m_3.
\end{equation}
We will abuse notation and still call the resulting set $Z$. By~\cref{lm: deviation bounds}, the fraction of $z$ removed is at most $\frac{2q^s}{m_3}$, so combined with \cref{lm: even covering} we still have 
\begin{equation} \label{eq: new mu z bound} 
\mu(Z) \geq 1 - q^{-\ell' / 2} - \frac{2q^s}{m_3}.
\end{equation}
For the remainder of the section we have that all $z \in Z$ satisfy~\eqref{eq: m_z bounds}.

Consider the following two distributions over triples $(z, W_i^\star, W^\star_j) \in Z \times \mc{W}^\star \times \mc{W}^\star$. The first is $\mathcal{D}_1$, generated by choosing $z \in Z$ uniformly and $W^\star_i, W^\star_j \in \W^\star$ uniformly conditioned on $z \in W^\star_i \cap W^\star_j$. The second is $\mathcal{D}'_1$, generated by choosing $z \in Z$ uniformly, $L \in \Lc^{\star}_{z}$ uniformly, and then $W^\star_i, W^\star_j \in \W^\star$ uniformly conditioned on $L \subseteq W^\star_i \cap W^\star_j$. We have
\begin{equation}\label{eq:def_of_11'}
    \begin{split}
        \mathcal{D}_1(z, W^\star_i, W^\star_j) &:= \frac{1}{|Z|} \cdot \frac{1}{|\{i',j' \; | \; W^\star_{i'}, W^\star_{j'} \in \W^\star, z \in W^\star_{i'} \cap W^\star_{j'}\}|} = \frac{1}{|Z| \cdot m_z^2},\\
        \mathcal{D}'_1(z, W^\star_i, W^\star_j) &= \frac{1}{|Z|} \cdot \frac{|\{L \in \Lc^{\star}_z \; | \; L \subseteq W^\star_i \cap W^\star_j \}|}{| \Lc^\star_z|} \cdot \E_{L \in \Lc^\star_z, L \subseteq W^\star_i \cap W^\star_j} \left[ \frac{1}{N_{2, \W^\star}(L)} \right].
    \end{split}
\end{equation}
Again, we do not require $i, j$ or $i', j'$ to be distinct in any of the definitions above. 

We will show that, for most $(z,W_i^{\star}, W_j^{\star})$ and up to multiplicative constants, the distribution $\mc{D}_1'$ assigns at least as much as the distribution $\mc{D}_1$. For that, we use the fact that for all $z, W^\star_i, W^\star_j$:
\begin{align}\label{eq: D1 bounds}
\frac{1}{|Z|}\cdot \frac{1}{1.21m_3^2q^{-2s}} \leq \mathcal{D}_1(z, W^\star_i, W^\star_j) 
&\leq \frac{1}{|Z|} \cdot \frac{1}{0.81 \cdot m_3^2 \cdot q^{-2s}},
\end{align}
\begin{equation} \label{eq: D'1 bounds}
\mathcal{D}'_1(z, W^\star_i, W^\star_j) \geq \frac{1}{|Z|} \cdot \frac{\mu_{z, W^\star_i \cap W^\star_j}(\Lc^\star) \cdot D \cdot p_2}{1.1 \cdot \eta \cdot D} \cdot \frac{1}{1.44 \cdot p_1^2 \cdot m_3^2} \geq \frac{1}{|Z|} \cdot \frac{\mu_{z, W^\star_i \cap W^\star_j}(\Lc^\star)}{3.2 \cdot \eta \cdot m_3^2 \cdot q^{-2s}}.
\end{equation}
To get \eqref{eq: D1 bounds} we plugged the bounds for $m_z$ from \eqref{eq: m_z bounds} into the definition of $\D_1$ from \eqref{eq:def_of_11'}. To get the first inequality in \eqref{eq: D'1 bounds} we use the fact that, due to the definition of $Z$, we have $|\Lc^\star_z | \leq 1.1 \eta \cdot D$ for all $z \in Z$. We also use the fact that $N_{2, \W^\star}(L) = N_{\W^\star}(L)^2$, along with the upper bound on $N_{\W^\star}(L)$ from the sixth property of~\cref{lm: find global 2}. The second inequality in \eqref{eq: D'1 bounds} follows from simplifying and using $p_2/p_1^2 \geq q^{2s}/2$ from \eqref{eq: p2 p1 ratio}.

Now, call a triplet $(z, W^\star_i, W^\star_j)$ bad if 
\[
\mu_{z, W^\star_i \cap W^\star_j}(\Lc^{\star}) \leq \frac{4}{5} \eta.
\]
If the triplet $(z, W^\star_i, W^\star_j)$ is \emph{not} bad, then by the above inequalities
\begin{equation}\label{eq: D1 and D'1}
\mathcal{D}_1(z, W^\star_i, W^\star_j) \leq \frac{1}{|Z|} \cdot \frac{1}{0.81 \cdot m_3^2q^{-2s}} \leq 6 \cdot \frac{1}{|Z|} \cdot \frac{4\eta/5}{3.2 \cdot \eta \cdot m_3^2 \cdot q^{-2s}} \leq 6 \cdot \mathcal{D}'_1(z, W^\star_i, W^\star_j).
\end{equation}
We will now start by showing that there are very few bad triplets.
\begin{claim} \label{cl: few bad triples}
    The following holds:
    \begin{enumerate}
        \item For each $z\in Z$, the number of $i, j$ such that
        $\mu_{z, W^\star_i \cap W^\star_j}(\Lc^{\star}) \leq \frac{4}{5} \eta$ is at most  $\frac{10^6q^{4s\ell}}{\eta^2}m_3$.
        \item for  every $W^\star_i, W^\star_j \in \W^\star$, we have $\mu_{\circ, W^\star_i \cap W^\star_j}(Z) \geq 0.9 \mu(Z)$.
    \end{enumerate}
  
\end{claim}
\begin{proof}
    Fix a point $z\in Z$, let $F_z$ be the restriction of $F$ to the zoom-in of $z$, where $F(x_1,\ldots, x_{\ell'}) = 1$ if $\spa(x_1,\ldots, x_{\ell'}) \in \Lc^\star$ and $0$ otherwise. Let $\eta' = \mu_{z,\circ}(\Lc^{\star})$. Since $z \in Z$, we have $\eta' \in [0.9\eta, 1.1\eta]$. Note that if $i,j$ satisfy$\mu_{z, W^\star_i \cap W^\star_j}(\Lc^{\star}) \leq \frac{4}{5} \eta$, then (using $\eta' \in [0.9\eta, 1.1\eta]$)
    \[
    \left|\mu_{z, W^\star_i \cap W^\star_j}(\Lc^{\star}) - \eta'\right| \geq \frac{\eta'}{20}.
    \]
    By~\cref{lm: Fourier fact} it follows that there is $S = (s_1, \ldots, s_{\ell'-1})$ 
    such that $\spa(s_1, \ldots, s_{\ell'-1}) \subseteq \left(W^\star_i \cap W^\star_j\right)^\perp $ and
    \begin{equation}\label{eq:distinctive_coef}
    |\widehat{F_z}(S)| \geq \frac{\eta'}{400q^{2s(\ell'-1)}}.
    \end{equation}
    Since by Parseval's inequality the sum of 
    $\left|\widehat{F_z}(S)\right|^2$ is at most $\|F_z\|_2^2\leq 1$,
    there are at most $\frac{160000q^{4s\ell'}}{\eta'^2}$ tuples $S$ satisfying~\eqref{eq:distinctive_coef}. Now consider the bipartite graph whose left side consists of these tuples $S = (s_1, \ldots, s_{\ell'-1})$,  right side consists of $W^\star_i \cap W^\star_j$, and the edges are between pairs that satisfy 
    \[
    \spa(s_1, \ldots, s_{\ell'-1}) \subseteq \left(W^\star_i \cap W^\star_j\right)^\perp .
    \]
    It follows that the number of edges in this graph is an upper bound on the number of bad triples containing $z$. Since $\W^\star$ is $4$-generic, we have 
    \[
    (W^\star_i \cap W^\star_j)^\perp \cap (W^\star_{i'} \cap W^\star_{j'})^\perp = \{0\}
    \] 
  for all $i,j,i', j'$ distinct. Therefore, any two neighbours of a vertex on the left must either have their $i$ or $j$ be equal, and hence the maximum degree of a vertex on the left side is at most $2m_3$. As a result, the graph has at most $2m_3\cdot \frac{160000q^{4s\ell}}{\eta'^2} \leq \frac{10^6q^{4s\ell}}{\eta^2}m_3$ edges, where we also use that $\eta' \geq 0.9\eta$. This completes the proof of the 
  first assertion of the claim. 

  For the second part of the lemma, note that $\mu(\overline{Z}) \leq \frac{q^{-\ell'}}{2} + \frac{2q^s}{m_3}$ by~\eqref{eq: new mu z bound}. Therefore, for any $W^\star_i \cap W^\star_j$, we have,
  \[
  \mu_{\circ, W^\star_i \cap W^\star_j}(\overline{Z}) \leq q^{2s}\cdot \left( \frac{q^{-\ell'}}{2} + \frac{2q^s}{m_3}\right).
  \]
  It follows that,
  \[
  \mu_{\circ, W^\star_i \cap W^\star_j}(Z) \geq 1 - q^{2s}\cdot \left( \frac{q^{-\ell'}}{2} + \frac{2q^s}{m_3}\right) \geq 0.9 \mu(Z).
  \qedhere
  \]
\end{proof}

\begin{lemma} \label{lm: tv bound 1}
Let $E$ be any event defined with respect to $(z, W^\star_i, W^\star_j)$. Then,
\[
\mathcal{D}_1(E) \leq 6\mathcal{D}'_1(E) + \gamma.
\]
\end{lemma}
\begin{proof}
    If the triple $(z, W^\star_i, W^\star_j)$ is not bad, then $\mathcal{D}_1(z,W^\star_i, W^\star_j) \leq 6\mathcal{D}'_1(z,W^\star_i, W^\star_j)$. Otherwise, we can use the generic bound $\mathcal{D}_1(z, W^\star_i, W^\star_j) \leq \frac{1}{|Z|\cdot 0.81 \cdot m_3^2q^{-2s}}$ from~\eqref{eq: D1 bounds}. By the bound on the number of bad triples per $z$ in~\cref{cl: few bad triples}, it follows that 
    \begin{align*}
    \mathcal{D}_1(E) &\leq 6\mathcal{D}'_1(E) + |Z| \cdot \frac{10^6q^{4s\ell}}{\eta^2}m_3\frac{1}{|Z|\cdot 0.81 \cdot m_3^2q^{-2s}} \\
    &=6\mathcal{D}'_1(E) + \frac{10^7q^{4s\ell+2s}}{\eta^2 \cdot m_3} \\
    &\leq 6\D'_1(E) + \gamma.
    \end{align*}
    Note that in the last transition we are using the fact that $m_3$ is large by the fourth property in~\cref{lm: find global 2}.
\end{proof}

Now let $\mathcal{D}_2$ be the distribution obtained by choosing $W^\star_i, W^\star_j \in \W$ uniformly, and then choosing $z \in W^\star_i \cap W^\star_j \cap Z$ uniformly. We have
\[
\mathcal{D}_2(z, W^\star_i, W^\star_j) = \frac{1}{m_3^2} \cdot \frac{1}{|W^\star_i \cap W^\star_j \cap Z|}.
\]

Using essentially the same proof, we get the following lemma.

\begin{lemma} \label{lm: tv bound 2}
Let $E$ be any event defined with respect to $(z, W^\star_i, W^\star_j)$. Then $\mathcal{D}_2(E) \leq 2\mathcal{D}_1(E)$.
\end{lemma}
\begin{proof}
By~\cref{cl: few bad triples}, we have $\mu_{\circ, W^\star_i \cap W^\star_j}(Z) > 0.9 \mu(Z)$ for all $i,j$, or equivalently
\[|W^\star_i \cap W^\star_j \cap Z|\geq 0.9 \cdot |Z|q^{-2s}.
\]
Thus, for all $i,j$ and all $z$,
\[
\D_2(z, W^\star_i, W^\star_j) \leq \frac{1}{0.9|Z|\cdot q^{-2s}}\cdot \frac{1}{m_3^2} \leq \frac{2}{1.21|Z|} \cdot \frac{1}{m_3^2q^{-2s}} \leq 2D_1(z,W^\star_i, W^\star_j),
\]
where we use~\eqref{eq: D1 bounds} for the third transition. 
\end{proof}

We are now ready to prove~\cref{lm: Wi agreement}.
\begin{proof}[Proof of~\cref{lm: Wi agreement}]
    By the fifth property in~\cref{lm: find global 2}, for every $L \in \Lc^\star$, we have
    \[
    \Pr_{W_i^\star \supseteq L, W^\star_i \in \W^\star}[f_i|_L \neq T_{\ell'}[L]] \leq 14\gamma.
    \]
    Let $E$ denote the event over $(z, W^\star_i, W^\star_j)$ that $f_i(z) \neq f_j(z)$. It follows that,
    \[
    \D'_1(E) \leq \Pr_{\substack{W_i^\star, W^\star_j \supseteq L\\ W^\star_i, W^\star_j \in \W^\star}}[f_i|_L \neq T_{\ell'}[L] \; \lor \; f_j|_L\neq T_{\ell'}[L]] \leq 2\cdot\Pr_{W_i^\star \supseteq L, W^\star_i \in \W^\star}[f_i|_L \neq T_{\ell'}[L]] \leq 28\gamma.
    \]
    Putting~\cref{lm: tv bound 1,lm: tv bound 2} together,
    \[
    \D_2(E) \leq 2(6\D'_1(E) +\gamma) = 338\gamma,
    \]
    proving the first part of~\cref{lm: Wi agreement}. 
    
    For the second part, recall from the second part of~\cref{cl: few bad triples} that for every pair $i,j$, we have 
    \[
    \left|W^\star_i \cap W^\star_j \cap Z\right| \geq 0.9 \cdot \left|Z\right|q^{-2s} \geq 0.81\left|W^\star_{\amb}\right|q^{-2s} = 0.81\cdot \left|W^\star_i \cap W^\star_j\right|.\qedhere
    \]
\end{proof}

\end{document}